\documentclass[10pt,a4paper,twoside]{article}

\usepackage[ngerman,USenglish]{babel} 

\usepackage{amsmath}                
\usepackage{amsfonts}               
\usepackage{amssymb}                
\usepackage{amsopn}                 
\usepackage{IEEEtrantools}
\allowdisplaybreaks[3] 
\usepackage{amsthm}                
\usepackage{bbm}                    
\usepackage{mathrsfs}               
\usepackage{array}                  
\usepackage{calc}                   
\usepackage{graphicx}        
\usepackage{psfrag}                 
\usepackage{subfigure}              
\usepackage{wrapfig}                
\usepackage{color}           
\usepackage{fancyhdr}               
\usepackage{verbatim}               
\usepackage{moreverb}               
\usepackage{exscale}                
\usepackage{enumerate}
\usepackage{cite}

\usepackage{tikz}
\usetikzlibrary{circuits.logic.US,circuits.logic.IEC}
\usetikzlibrary{arrows,positioning} 
\usetikzlibrary{plotmarks}
\usetikzlibrary{decorations.pathreplacing}
\usepackage{float}
\usepackage{wasysym}
\newtheorem{theorem}{Theorem}
\newtheorem{lemma}[theorem]{Lemma}
\newtheorem{corollary}[theorem]{Corollary}
\newtheorem{proposition}[theorem]{Proposition}

\newtheorem{defin}[theorem]{Definition}
\newtheorem{remark}[theorem]{Remark}
\newtheorem{claim}[theorem]{Claim}
\newtheorem{example}[theorem]{Example}

\fussy
\large \normalsize
\usepackage{hyphenat}

\usepackage{xspace}
\usepackage{bm}
\usepackage{fancyref}
\usepackage{textcomp}
\usepackage[Symbol]{upgreek}

\newcommand{\safemath}[2]{\newcommand{#1}{\ensuremath{#2}\xspace}}
\renewcommand{\safemath}[2]{\newcommand{#1}{\ensuremath{#2}\xspace}}

\safemath{\veca}{\mathbf{a}}
\safemath{\vecb}{\mathbf{b}}
\safemath{\vecc}{\mathbf{c}}
\safemath{\vecd}{\mathbf{d}}
\safemath{\vece}{\mathbf{e}}
\safemath{\vecf}{\mathbf{f}}
\safemath{\vecg}{\mathbf{g}}
\safemath{\vech}{\mathbf{h}}
\safemath{\veci}{\mathbf{i}}
\safemath{\vecj}{\mathbf{j}}
\safemath{\veck}{\mathbf{k}}
\safemath{\vecl}{\mathbf{l}}
\safemath{\vecm}{\mathbf{m}}
\safemath{\vecn}{\mathbf{n}}
\safemath{\veco}{\mathbf{o}}
\safemath{\vecp}{\mathbf{p}}
\safemath{\vecq}{\mathbf{q}}
\safemath{\vecr}{\mathbf{r}}
\safemath{\vecs}{\mathbf{s}}
\safemath{\vect}{\mathbf{t}}
\safemath{\vecu}{\mathbf{u}}
\safemath{\vectU}{\mathbf{U}}
\safemath{\vecv}{\mathbf{v}}
\safemath{\vecw}{\mathbf{w}}
\safemath{\vecx}{\mathbf{x}}
\safemath{\vectX}{\mathbf{X}}
\safemath{\vecy}{\mathbf{y}}
\safemath{\vecz}{\mathbf{z}}

\safemath{\bmuB}{\mathbf{B}}
\safemath{\bmuP}{\mathbf{P}}

\safemath{\setA}{\mathcal{A}}
\safemath{\setB}{\mathcal{B}}
\safemath{\setC}{\mathcal{C}}
\safemath{\setD}{\mathcal{D}}
\safemath{\setE}{\mathcal{E}}
\safemath{\setF}{\mathcal{F}}
\safemath{\setG}{\mathcal{G}}
\safemath{\setH}{\mathcal{H}}
\safemath{\setI}{\mathcal{I}}
\safemath{\setJ}{\mathcal{J}}
\safemath{\setK}{\mathcal{K}}
\safemath{\setL}{\mathcal{L}}
\safemath{\setM}{\mathcal{M}}
\safemath{\setN}{\mathcal{N}}
\safemath{\setO}{\mathcal{O}}
\safemath{\setP}{\mathcal{P}}
\safemath{\setQ}{\mathcal{Q}}
\safemath{\setR}{\mathcal{R}}
\safemath{\setS}{\mathcal{S}}
\safemath{\setT}{\mathcal{T}}
\safemath{\setU}{\mathcal{U}}
\safemath{\setV}{\mathcal{V}}
\safemath{\setW}{\mathcal{W}}
\safemath{\setX}{\mathcal{X}}
\safemath{\setY}{\mathcal{Y}}
\safemath{\setZ}{\mathcal{Z}}
\safemath{\emptySet}{\varnothing}

\safemath{\rndX}{X}
\safemath{\rndY}{Y}
\safemath{\rndZ}{Z}

\safemath{\opE}{\mathbb{E}}
\DeclareMathOperator{\Exop}{\opE}

\newcommand{\ind}[1]{\mathbbm{1}_{#1}}				

\safemath{\reals}{\mathbb R}
\safemath{\positivereals}{\reals_{+}}
\safemath{\integers}{\mathbb Z}
\safemath{\posint}{\integers_{+}}
\safemath{\naturals}{\mathbb N}
\safemath{\posnaturals}{\naturals_{+}}
\safemath{\complexset}{\mathbb C}
\safemath{\rationals}{\mathbb Q}

\newcommand{\channel}[2]{W ( #1 | #2 )} 
\newcommand{\channeltild}[2]{\widetilde W ( #1 | #2 )} 

\safemath{\oprobf}{\hat{p}}	
\safemath{\probf}{p} 

\newcommand{\muti}[2]{I ( #1 , #2 )} 
\newcommand{\bigmuti}[2]{I \bigl( #1 , #2 \bigr)} 
\newcommand{\condmuti}[3]{I ( #1 , #2 | #3 )} 
\newcommand{\ent}[1]{H ( #1 )} 
\newcommand{\bigent}[1]{H \bigl( #1 \bigr)} 
\newcommand{\condent}[2]{H ( #1 | #2 )} 
\newcommand{\bin}[1]{\bm{ \mathcal B}_{#1}}		
\newcommand{\indexset}[1]{\bm{ \mathcal V}_{#1}}		
\newcommand{\indexsetre}[1]{\setV_{#1}}		
\safemath{\indexsetsre}{\{ \indexsetre \nu \}}
\safemath{\binind}{\setI}			
\safemath{\pool}{\bm{ \mathcal P}}			
\safemath{\poolre}{\setP}			
\newcommand{\poolel}[1]{\bmuP ( #1 )}	
\safemath{\rcode}{\bm{ \mathcal C}}			
\newcommand{\cwind}[2]{U_{#1} ( #2 )}	
\safemath{\eptyp}{\setT^{(n)}_\epsilon}	
\newcommand{\idset}[1]{\bm{ \mathcal D}_{#1}}	
\newcommand{\idsetre}[1]{\setD_{#1}}	
\safemath{\dist}{\mathbb{P}}			
\safemath{\tdist}{\tilde \dist}			
\newcommand{\distof}[1]{\dist [ #1 ]}	
\newcommand{\bigdistof}[1]{\dist \bigl[ #1 \bigr]}	
\newcommand{\Bigdistof}[1]{\dist \Bigl[ #1 \Bigr]}	
\newcommand{\biggdistof}[1]{\dist \biggl[ #1 \biggr]}	
\newcommand{\Biggdistof}[1]{\dist \Biggl[ #1 \Biggr]}	
\newcommand{\distsubof}[2]{\dist_{#1} [ #2 ]}	
\newcommand{\bigdistsubof}[2]{\dist_{#1} \bigl[ #2 \bigr]}	
\newcommand{\Bigdistsubof}[2]{\dist_{#1} \Bigl[ #2 \Bigr]}	
\newcommand{\biggdistsubof}[2]{\dist_{#1} \biggl[ #2 \biggr]}	
\newcommand{\Biggdistsubof}[2]{\dist_{#1} \Biggl[ #2 \Biggr]}	
\safemath{\ry}{\setY}			
\safemath{\rz}{\setZ}			
\safemath{\rw}{\setW}			
\newcommand{\channely}[2]{W_\ry ( #1 | #2 )} 
\newcommand{\channelz}[2]{W_\rz ( #1 | #2 )} 
\newcommand{\channelytild}[2]{\widetilde W_\ry ( #1 | #2 )} 
\newcommand{\channelztild}[2]{\widetilde W_\rz ( #1 | #2 )} 
\newcommand{\sigof}[1]{\sigma ( #1 )}	
\newcommand{\card}[1]{| #1 |}			
\newcommand{\Ex}[2]{\ensuremath{\Exop_{#1} [#2]}} 	
\newcommand{\bigEx}[2]{\ensuremath{\Exop_{#1} \bigl[#2\bigr]}} 	
\newcommand{\BigEx}[2]{\ensuremath{\Exop_{#1} \Bigl[#2\Bigr]}} 	
\newcommand{\biggEx}[2]{\ensuremath{\Exop_{#1} \biggl[#2\biggr]}} 	
\newcommand{\BiggEx}[2]{\ensuremath{\Exop_{#1} \Biggl[#2\Biggr]}} 	

\begin{document}
%
%
\selectlanguage{USenglish}
\pagenumbering{arabic}
%
\title{Identification via the Broadcast Channel}
  
\author{Annina Bracher and Amos Lapidoth}

\maketitle

\huge
\begin{abstract}
  \setcounter{page}{1}
  \normalsize
  \vspace{0.5cm}
  
  \let\thefootnote\relax\footnotetext{The results in this paper were presented in part at the IEEE International Symposium on Information Theory (ISIT), Honolulu, USA, Jun.\ 2014.}
  \let\thefootnote\relax\footnotetext{A.~Bracher is with Swiss Reinsurance Company Ltd, Mythenquai~50, 8022 Zurich, Switzerland (e-mail: annina\_bracher@swissre.com).\newline A.~Lapidoth is with the Signal and Information Processing Laboratory, ETH Zurich, 8092 Zurich, Switzerland (e-mail: lapidoth@isi.ee.ethz.ch).}

The identification (ID) capacity region of the two-receiver broadcast channel (BC) is shown to be the set of rate-pairs for which, for some distribution on the channel input, each receiver's ID rate does not exceed the mutual information between the channel input and the channel output that it observes. Moreover, the capacity region's interior is achieved by codes with deterministic encoders. The results are obtained under the average-error criterion, which requires that each receiver reliably identify its message whenever the message intended for the other receiver is drawn at random. They hold also for channels whose transmission capacity region is to-date unknown. Key to the proof is a new ID code construction for the single-user channel. Extensions to the BC with one-sided feedback and the three-receiver BC are also discussed: inner bounds on their ID capacity regions are obtained, and those are shown to be in some cases tight.

\end{abstract}
\normalsize

\section{Introduction} \label{sec:intro}

In Shannon's classical transmission problem the encoder transmits a message from a message set $\setM$ of size $|\setM|$ over a discrete memoryless channel (DMC) $\channel y x$, and the receiver guesses the transmitted message based on the channel's outputs. The guess can be any of the $| \setM |$ messages in the set $\setM$, and the receiver thus faces a hypothesis-testing problem with $| \setM |$ hypotheses. Loosely speaking, we say that a transmission scheme is reliable if, irrespective of the transmitted message $m$, the receiver guesses correctly with high probability. Ahlswede and Dueck's identification-via-channels problem \cite{ahlswededueck89} is different. Here the encoder sends an identification (ID) message from a set $\setM$, and $\card \setM$ receiving parties observe the channel outputs. Each party is focused on a different message $m^\prime \in \setM$. The \emph{$m^\prime$-focused receiving party} must guess whether or not Message~$m^\prime$ was sent. It thus faces a hypothesis-testing problem with only two hypotheses. Loosely speaking, we say that an identification scheme is reliable if, for every possible transmitted ID message $m \in \setM$ and for every $m^\prime \in \setM$ (possibly equal to $m$), the $m^\prime$-focused receiving party guesses correctly with high probability. That is, if $m^\prime$ equals the transmitted ID message $m$, then the $m^\prime$-focused receiving party guesses with high probability that $m^\prime$ was sent, and otherwise it guesses with high probability that $m^\prime$ was not sent.\footnote{The corresponding error events are called \emph{missed identification} and \emph{wrong identification}: a missed identification occurs if $m^\prime = m$ and the $m^\prime$-focused receiving party guesses that $m^\prime$ was not sent, and a wrong identification occurs if $m^\prime \neq m$ and the $m^\prime$-focused receiving party guesses that $m^\prime$ was sent. The identification scheme is reliable if the maximum probabilities of missed and wrong identification are small, where the maximum is w.r.t.\ $m$ for the probability of missed identification and w.r.t.\ the distinct pair $m, \, m^\prime$ for the probability of wrong identification.}

In Shannon's problem the number of messages that can be transmitted reliably is exponential in the number of channel uses, and the transmission rate is thus defined as the logarithm of the number of transmission messages normalized by the blocklength~$n$. In Ahlswede and Dueck's ID problem the number of identifiable messages is double-exponential, and the ID rate is thus defined as the \emph{iterated} logarithm of the number of ID messages normalized by~$n$. The suprema of achievable rates for the two problems are identical: both the transmission and the ID capacity equal $C$, where $C = \max_{P} \muti {P}{W}$ \cite{shannon48, ahlswededueck89, hanverdu92}.

The two problems also differ in the role of randomization at the encoder. Whether or not stochastic encoders are allowed does not influence the transmission capacity. However, stochastic encoders are essential for achieving the ID capacity. Such encoders associate with each ID message a distribution on the channel-input sequence and send ID Message~$m$ by generating the channel-input sequence according to the distribution associated with $m$. If we only allow deterministic encoders, then the number of identifiable messages grows only exponentially in the blocklength.\footnote{For ID codes with deterministic encoders, the ID rate is defined as the logarithm of the number of ID messages normalized by $n$, and the supremum of all achievable ID rates is the logarithm of the number of distinct probability mass functions (PMFs) $\channel {\cdot}{x}$ on the channel output that are induced by the different channel-input symbols $x \in \setX$ \cite{ahlswededueck89}.} Throughout this paper we allow stochastic encoders, but for our main achievability result (Theorem~\ref{th:IDBC}) they are unnecessary.\\

The present paper studies identification via a two-receiver broadcast channel (BC) $\channel {y,z} x$ whose transmitting terminal is Terminal~$\setX$ and whose receiving terminals are $\ry$ and $\rz$. The sender wishes to send two ID messages, one to each receiving terminal. The received sequence at Terminal~$\ry$ is observed by different parties, each of which is focused---among all the possible ID messages intended for Terminal~$\ry$---on a different ID message. Likewise for Terminal~$\rz$. We show that the ID capacity region of the BC is the set of rate-pairs for which, for some distribution on the channel input, each receiver's ID rate does not exceed the mutual information between the channel input and the channel output that it observes (Theorem~\ref{th:IDBC}). The converse we provide is a strong converse.

Our results are obtained under the average-error criterion. Under this criterion, the ID messages $M_\ry$ and $M_\rz$ to the two receiving terminals are assumed to be independent with each being uniform over its message set ($\setM_\ry$ or $\setM_\rz$), and each receiver must identify the message intended for it reliably in expectation over the ID message intended for the other receiving terminal. Loosely speaking, we thus say that an identification scheme is reliable under the average-error criterion if the following two requirements are met: 1) for all (possibly equal) $m_\ry, \, m_\ry^\prime \in \setM_\ry$, if the ID message that is sent to Terminal~$\ry$ is $m_\ry$ and the ID message that is sent to Terminal~$\rz$ is drawn uniformly over $\setM_\rz$, then the $m^\prime_\ry$-focused receiving party guesses correctly with high probability whether or not $m_\ry$ is equal to $m_\ry^\prime$; and 2) likewise for all $m_\rz, \, m_\rz^\prime \in \setM_\rz$.\footnote{The average-error criterion for identification via the BC should not be confused with the average-error criterion for identification via the DMC. On the DMC the average-error criterion requires that for every $m^\prime \in \setM$ the probability of wrong identification associated with the pair $m, \, m^\prime$ be small on average over all possible realizations $m \neq m^\prime$ of the transmitted ID message. Han and Verd\'u showed that under this criterion the ID capacity is infinite whenever $C > 0$ \cite{hanverdu92}. This holds because the stochastic encoder can associate the same distribution on the channl-input sequence with an infinite number of ID messages while guaranteeing that the probability of missed identification and the average (but not the maximum) probability of wrong identification be small at each receiving party. The average-error criterion for the BC, which we consider in this paper, is different: For Terminal~$\ry$ it requires that the probability of wrong identification associated with any distinct pair $m_\ry, \, m_\ry^\prime \in \setM_\ry$ be small; the term ``average'' refers to the fact that the probabilities of missed and wrong identification at Terminal~$\ry$ are defined on average over all possible realizations $m_\rz \in \setM_\rz$ of the ID message that is sent to Terminal~$\rz$. Likewise for Terminal~$\rz$.}

Identification via the BC was previously studied in \cite{verbovenmeulen90, biliksteinberg01,oohama03,ahlswede08} under a different criterion, namely, the maximum-error criterion. Under this criterion each receiver must identify its message reliably irrespective of the realization of the ID message intended for the other receiver. Loosely speaking, we thus say that an identification scheme is reliable under the maximum-error criterion if for all transmitted ID message-pairs $( m_\ry, m_\rz ) \in \setM_\ry \times \setM_\rz$ the following two requirements are met: 1) for every $m_\ry^\prime \in \setM_\ry$ (possibly equal to $m_\ry$), the $m^\prime_\ry$-focused receiving party guesses correctly with high probability whether or not $m_\ry$ is equal to $m_\ry^\prime$; and 2) likewise for every $m_\rz^\prime$- focused receiving party at Terminal~$\setZ$.

The maximum-error ID capacity region of the BC is still unknown (but see \cite{ahlswede08} and our discussion in Section~\ref{sec:maxError} of the case where an additional constraint is imposed on the decay to zero as a function of the blocklength of the probability of error). Clearly, the average-error ID capacity region is an outer bound, but whether this bound is tight is unknown. To-date, the best known inner bound on the maximum-error ID capacity region of the BC is the ``common-randomness capacity region'' of the BC \cite{ahlswede08}. This inner bound is achieved by a common-randomness ID code, which---like that of \cite{ahlswededueckfb89} for the DMC---uses a transmission code to establish common randomness between the encoder and each decoder. As we shall see, the average-error ID capacity region of the BC typically exceeds this inner bound (Remark~\ref{re:commRandIBStrCont}), but this, of course, does not imply that it exceeds the maximum-error ID capacity region. We do know that the capacity regions differ when only deterministic encoders are allowed, because, unlike the maximum-error ID capacity region (or, for that matter, the single-user channel), all rate-pairs in the interior of the average-error ID capacity region can be achieved by deterministic encoders (Remark~\ref{re:detIDCodeOptBC}). This is perhaps not surprising, because to each receiver such a deterministic encoder appears stochastic: the transmitted sequence depends not only on the ID message addressed to it but also on the random ID message (of positive rate) addressed to the other terminal.\\

To derive our capacity region, we introduce a new capacity-achieving ID code construction for the single-user channel. Our coding scheme for the BC builds on this by making it appear to each receiver as though we were using an instance of the new single-user ID code on its marginal channel. We next describe the new single-user coding scheme, which is reminiscent of \cite{ahlswededueck89} but with an important twist that is key to our results. We then describe our scheme for the BC. 

For a DMC $\channel y x$ the new scheme can be described as follows: Fix an input distribution $P$, an ID rate $R < \muti {P}{W}$, and some blocklength $n$. The scheme associates with each ID message $m$ a multiset we call ``the $m$-th bin'' and whose elements are $n$-tuples (not necessarily distinct) of channel inputs.\footnote{A multiset is a generalized set that allows multiple instances of its elements, e.g., $\{ 1, 2, 3, 4 \}$ and $\{ 1, 1, 2, 3, 4, 4, 4 \}$ are different multisets. The size of a multiset is the number of elements that it contains. The size of the multiset $\{ 1, 2, 3, 4 \}$ is thus four and that of $\{ 1, 1, 2, 3, 4, 4, 4 \}$ is seven. If $X$ is chosen uniformly at random from a multiset, then $\distof {X = x}$ is proportional to the number of instances of $x$ in the set. For example, if $X$ is chosen uniformly at random from the multiset $\{ 1, 1, 2, 3, 4, 4, 4 \}$, then $\distof {X = 1} = 2/7$.} To send the $m$-th ID message, the (stochastic) encoder sends a random element of this bin. At the receiver's side, the $m^\prime$-focused receiving party guesses that $m^\prime$ was sent if at least one element of the $m^\prime$-th bin is jointly typical with the received $n$-tuple of channel outputs. To construct the bins we use a random coding argument, with each bin having expected size $e^{ n \tilde R }$, where $\tilde R$ exceeds the ID rate $R$, but is smaller than $\muti {P}{W}$,
\begin{IEEEeqnarray}{C}
R < \tilde R < \muti {P}{W}. \label{eq:constraintDMC}
\end{IEEEeqnarray}

The bins are constructed at random from a size $e^{n R_\poolre}$ multiset that we call ``pool'' and whose elements are $n$-length input sequences. Here $R_\poolre$ can be any number exceeding $\tilde R$, possibly even exceeding $\muti {P}{W}$, so, by \eqref{eq:constraintDMC},
\begin{IEEEeqnarray}{C}
\tilde R < \muti {P}{W} \quad \textnormal{and} \quad R < \tilde R < R_\poolre.
\end{IEEEeqnarray}
We construct every bin by randomly selecting its elements from the pool, with the $n$-tuples in the pool being selected for inclusion in the $m$-th bin independently each with probability $e^{-n (R_\poolre - \tilde R)}$. Since the pool is of size $e^{n R_\poolre} \!$, each bin is a multiset of expected size $e^{n \tilde R}$. The elements of the pool are drawn independently $\sim P^n$. As we shall see, the generated ID code is with high probability reliable (Section~\ref{sec:IDCodingTechnique}).

Our above scheme is reminiscent of the one in \cite{ahlswededueck89}: every ID message is associated with a bin, and in both schemes the bins are chosen at random from a pool. The main difference is that in our scheme the pool need not constitute a codebook that is reliable in Shannon's sense. Indeed, our pool is of size $e^{n R_\poolre} \!$, where $R_\poolre$ can exceed $\muti {P}{W}$ or even $C$. This flexibility in choosing $R_\poolre$ will be critical on the BC.

The scheme we propose for the BC $\channel {y,z} x$ is motivated by the single-user scheme. Denote by $\channely y x = \sum_z \channel {y,z} x$ and $\channelz z x = \sum_y \channel {y,z} x$ the marginal channels. Fix an input distribution $P$, positive ID rates
\begin{IEEEeqnarray*}{C}
0 < R_\ry < \muti {P}{W_\ry}, \\
0 < R_\rz < \muti {P}{W_\rz},
\end{IEEEeqnarray*}
and some blocklength $n$. We first consider the receivers' side, because in their decoding the receivers follow the single-user scheme. Like the single-user scheme, the scheme for the BC associates with each ID message $m_\ry \in \setM_\ry$ a multiset we call the $m_\ry$-th bin and whose elements are $n$-tuples of channel inputs, and likewise with each ID message $m_\rz \in \setM_\rz$. The $m_\ry^\prime$-focused receiving party at Terminal~$\ry$ guesses that $m_\ry^\prime$ was sent if at least one element of the $m_\ry^\prime$-th bin is jointly typical with the sequence it observes, and likewise at Terminal~$\rz$. The encoding, however, is different from the single-user scheme. In fact, our encoder for the BC is deterministic: it maps each ID message-pair $(m_\ry,m_\rz)$ to an $n$-tuple of channel inputs we call the ``$(m_\ry, m_\rz)$-codeword.'' (The $(m_\ry, m_\rz)$-codeword is in the intersection of the $m_\setY$-th and the $m_\setZ$-th bins, whenever the intersection is not empty.) 
We design the codewords and the bins using a random coding argument.

Our goal in designing the codewords and the bins is that to each receiver it would appear as though its intended ID message were sent over its marginal channel using the single-user scheme. More precisely, we want the following to hold: 1) if the ID message that is sent to Terminal~$\ry$ is $m_\ry \in \setM_\ry$ and the ID message that is sent to Terminal~$\rz$ is drawn uniformly over $\setM_\rz$, then the transmitted codeword is nearly uniformly distributed over the $m_\ry$-th bin (in terms of Total-Variation distance); and 2) likewise for $m_\rz \in \setM_\rz$. If 1) and 2) hold, then to each receiver it nearly appears as though we were using an instance of the new single-user ID code on its marginal channel: if we view the ID message that is sent to Terminal~$\rz$ as uniformly-drawn, then the encoder communicates with Terminal~$\ry$ ``essentially'' using our reliable single-user scheme, and likewise with Terminal~$\setZ$. To prove that the design goal can be met, we shall use a random coding argument.

The bins are constructed as in the single-user scheme: We construct all the bins---those associated with an ID message $m_\ry \in \setM_\ry$ or $m_\rz \in \setM_\rz$---from a multiset we call pool. The pool has size $e^{n R_\poolre}$, and each bin associated with an ID message $m_\ry \in \setM_\ry$ or $m_\rz \in \setM_\rz$ has expected size $e^{n \tilde R_\ry}$ or $e^{n \tilde R_\rz}$, respectively. The pool and the bins are generated as in the single-user construction, and $R_\poolre$, $\tilde R_\ry$, and $\tilde R_\rz$ meet similar constraints, so
\begin{IEEEeqnarray*}{C}
\tilde R_\ry < \muti {P}{W_\ry} \quad \textnormal{and} \quad R_\ry < \tilde R_\ry < R_\poolre, \\
\tilde R_\rz < \muti {P}{W_\rz} \quad \textnormal{and} \quad R_\rz < \tilde R_\rz < R_\poolre.
\end{IEEEeqnarray*}
Additionally, we impose the constraint
\begin{IEEEeqnarray}{C}
R_\poolre < \tilde R_\ry + \tilde R_\rz. \label{eq:additionalConstraintBC}
\end{IEEEeqnarray}
(The constraints can all be met, because $R_\ry$ and $R_\rz$, and thus also $\muti {P}{W_\ry}$ and $\muti {P}{W_\rz}$, are positive.) The additional constraint \eqref{eq:additionalConstraintBC} has no counterpart in the single-user setting. It restricts the size of the pool in order to guarantee that with high probability the $m_\ry$-th bin and the $m_\rz$-th bin intersect and that consequently the $(m_\ry,m_\rz)$-codeword will be in both bins. If the $(m_\ry,m_\rz)$-codeword is not in this intersection, then, to at least one of the two receivers, it won't appear as though the $n$-tuple of channel inputs were drawn uniformly over the bin associated with its intended ID message. And if this happens to too many pairs $(m_\ry,m_\rz)$, our scheme will fail.

As to the design of the codewords, if the $m_\ry$-th and the $m_\rz$-th bins intersect, then we draw the $(m_\ry,m_\rz)$-codeword uniformly at random from the intersection, and otherwise we draw it uniformly at random from the pool. As we shall see, the generated ID code meets our design goals with high probability (see Section~\ref{sec:DPIDBC}; key to the proof is that the size of each bin is exponential in $n$ while the cardinalities of $\setM_\ry$ and $\setM_\rz$ are double-exponential).



The flexibility afforded by our single-user scheme to choose a pool of size $e^{n R_\poolre} \!$, where $R_\poolre$ can be larger than $\muti {P}{W_\ry}$ or $\muti {P}{W_\rz}$, is crucial to our BC scheme. To see why, consider for now a BC $\channel {y,z} x$ and an input distribution $P$ for which $$\muti {P}{W_\rz} < \muti {P}{W_\ry}.$$ If the pool had been of size $e^{n R_\poolre}$ for some $R_\poolre \leq \muti {P}{W_\rz}$, then at most $\exp \bigl( \exp \bigl(n \muti {P}{W_\rz}\bigr) \bigr)$ different bins could have been constructed from the pool, and the BC scheme would have thus failed for $R_\ry > \muti {P}{W_\rz}$, because in this case the number of possible ID messages intended for Receiver~$\ry$ would have exceeded the number of different bins. The pool rate $R_\poolre$ must therefore exceed $\muti {P}{W_\rz}$, and hence the pool cannot consist of a codebook that is reliable in the Shannon sense on the marginal channel $\channelz z x$. It is the possibility of choosing $R_\poolre > \muti {P}{W_\rz}$ that allows our BC scheme to achieve every rate-pair $( R_\ry, R_\rz )$ satisfying
\begin{IEEEeqnarray}{C}
0 < R_\ry < \muti {P}{W_\ry} \quad \textnormal{ and } \quad  0 < R_\rz < \muti {P}{W_\rz}, \label{eq:introAchBC}
\end{IEEEeqnarray}
even when $R_\ry > \muti {P}{W_\rz}$.\\


The average-error criterion, which we consider in this paper, is suitable whenever the receivers' ID messages are independent and uniform over their supports. As we shall see, we can adapt our coding scheme to solve for the capacity region of a more general scenario where the receivers' ID messages are not independent but have a common part. In this scenario the ID message intended for Terminal~$\ry$ is a tuple comprising a private message of rate $R_\ry$ and a common message of rate $R$, and likewise for Terminal~$\rz$.\footnote{One can view the common-message setting of the transmission problem via the BC as a scenario where the encoder conveys one message to each receiver, but each receiver's message comprises a private and a common part.} The common messages are identicial, and the private messages are independent, uniformly distributed on their supports, and independent of the common message. We assume that all rates are positive and require that each receiver identify its message reliably in expectation over the other receiver's private message. For this scenario, we show that the ID capacity region of the BC is the set of rate-triples $(R,R_\ry,R_\rz)$ satisfying
\begin{IEEEeqnarray}{C}
0 < R, R_\ry < \muti {P}{W_\ry} \quad \textnormal{ and } \quad  0 < R, R_\rz < \muti {P}{W_\rz} \label{eq:introAchBCCM}
\end{IEEEeqnarray}
for some input distribution $P$ (Theorem~\ref{th:IDBCCM}).\footnote{The assumption that $R > 0$ is not needed; it only ensures that there is a common message. The assumption that $R_\ry, \, R_\rz > 0$ is, however, needed: if $R_\ry$, say, is zero, then the imposed average-error criterion will turn into a maximum-error criterion for Receiver~$\setZ$.} Comparing \eqref{eq:introAchBCCM} and \eqref{eq:introAchBC} we see that the common message appears to come for free at all rates up to $\min \bigl\{ \muti {P}{W_\ry}, \muti {P}{W_\rz} \bigr\}$. This can be explained as follows. The ID rate is the iterated logarithm of the number of ID messages normalized by the blocklength~$n$, and for $n$ sufficiently large and for all nonnegative real numbers $R_1$ and $R_2$ $$\exp ( \exp (n R_1) ) \exp ( \exp (n R_2) ) \approx \exp \bigl( \exp \bigl( n \max \{ R_1, R_2 \} \bigr) \bigr).$$ Comparing \eqref{eq:introAchBCCM} and \eqref{eq:introAchBC} we see that the common message appears to come for free at all rates up to $$\min \bigl\{ \muti {P}{W_\ry}, \muti {P}{W_\rz} \bigr\}.$$ A reason for this is that the ID rate of a pair of ID messages is not equal to the sum of the messages' ID rates.


We also discuss extensions to the BC with more than two receivers and the two-receiver BC with one-sided feedback: We inner-bound the ID capacity region of the three-receiver BC (Theorem~\ref{th:ibBC3Rec}) and show that the bound is tight if no receiver is ``much more capable'' than the other two (see Remark~\ref{re:3RecTight} for more details). The ID capacity region of the two-receiver BC with one-sided feedback is established for the case where the channel outputs are independent conditional on the channel input (Corollary~\ref{co:ICBC1FBMarkov}).\\

The rest of this paper is structured as follows. We conclude this section with some notation and with the concentration inequalities that we shall need. Section~\ref{sec:IDCodingTechnique} is dedicated to the new ID code for the DMC. Section~\ref{sec:IDBC} studies identification via the BC. Section~\ref{sec:maxError} compares the average- and the maximum-error criterion. The extensions are presented in Section~\ref{sec:extensions}, and the paper concludes with a brief summary.

\subsection{Notation and Terminology}

On the single-user channel we denote the channel-input alphabet by $\setX$ and the channel-output alphabet by $\setY$. On the two-receiver BC $\setX$ is the channel-input alphabet, $\ry$ is the channel-output alphabet at Terminal~$\setY$, and $\rz$ is the channel-output alphabet at Terminal~$\rz$. All these alphabets are finite. We write $( \setX, \channel y x, \setY )$ or $\channel y x$ for a DMC of transition law $\channel y x$ and $( \setX, \channel {y,z} x, \setY \times \setZ )$ or $\channel {y,z} x$ for a BC of transition law $\channel {y,z} x$. We denote the marginal channel of the BC $\channel {y,z} x$ to Terminal~$\ry$ by $\channely y x$, i.e., $\channely y x = \sum_z \channel {y,z} x$; and likewise $\channelz z x = \sum_y \channel {y,z} x$.

Random variables are denoted by upper-case letters and their realization or the elements of their supports by lower-case letters, e.g., $Y$ denotes the random output of the DMC and $y \in \setY$ a value it may take. The terms \emph{pool} and \emph{bin} are used for indexed multisets of $n$-tuples from $\setX^n$. Pools and bins are denoted by calligraphic letters, and in boldface if they are random, e.g., $\pool$ denotes a random pool and $\poolre$ a possible realization. Sequences are denoted by boldface lower- or upper-case letters depending on whether they are deterministic or random, e.g., $\poolel j$ denotes the $j$-th $n$-tuple in the random pool $\pool$, and $\vecx$ is an $n$-tuple from $\setX^n$. The positive integer $n \in \naturals$ stands for the blocklength, and, unless otherwise specified, sequences are of length $n$. We denote the positive real numbers by $\reals^+$ and the nonnegative real numbers by $\reals^+_0$, so $\reals^+_0 = \reals^+ \cup \{ 0 \}$.

Variables that occur at Time~$i$ have the subscript $i$, so $\rndY_i$ is the Time-$i$ channel output. Sequences of variables that occur in the time-range $j$ to $i$ bear a subscript $j$ and a superscript $i$, where the subscript $j = 1$ may be dropped, e.g., $\rndY_{4}^5$ denotes the forth and fifth output, and $\rndY^n$ denotes all the outputs through Time~$n$.

The set of PMFs on $\setX$ is denoted $\mathscr P ( \setX )$, and its generic element $P$. If the input $X$ of the channel $\channel y x$ is of PMF $P$, then $P \times W$ denotes the joint distribution of $X$ and the channel output $Y$, i.e., $$( P \times W ) ( x,y ) = P ( x ) \channel y x, \quad x \in \setX, \, y \in \setY,$$ and $P W$ denotes the corresponding distribution of $Y$, i.e., $$( P W ) ( y ) = \sum_{x \in \setX} ( P \times W ) ( x,y  ) = \sum_{x \in \setX} P ( x ) \channel y x, \quad y \in \setY.$$

The set of $\epsilon$-typical sequences of length~$n$ w.r.t.\ $P$ is denoted $\setT^{( n )}_{\epsilon} (P)$, i.e., $$\setT^{( n )}_{\epsilon} ( P ) = \Biggl\{ \vecx \in \setX^n \colon \biggl| \frac{N ( x | \vecx )}{n} - P (x) \biggr| \leq \epsilon P (x), \, \forall \, x \in \setX \Biggr\},$$ where $N (x|\vecx)$ is the number of components of the $n$-tuple $\vecx$ that equal $x$. We often write $\setT^{( n )}_{\epsilon}$ instead of $\setT^{( n )}_{\epsilon} (P)$ when $P$ is clear from the context. The empirical type of an $n$-tuple $\vecx \in \setX^n$ is denoted $P_\vecx$, so $P_\vecx (x) = N (x|\vecx) / n, \, x \in \setX$, and $\setT^{(n)}_P$ is the set of all elements of $\setX^n$ of empirical type $P$. We denote the set of $n$-types on $\setX^n$ by $\Gamma^{(n)}$, so $$\Gamma^{(n)} = \bigl\{ P \in \mathscr P (\setX) \colon \setT^{(n)}_P \neq \emptyset \bigr\}.$$ For a given DMC $\channel y x$ and for every $\vecx \in \setX^n$ and $P \in \mathscr P (\setX)$, we denote by $\setT^{( n )}_\epsilon ( P \times W | \vecx )$ the set of $n$-tuples $\vecy \in \setY^n$ that are jointly $\epsilon$-typical with $\vecx$ w.r.t.\ $P \times W$, i.e., $$\setT^{( n )}_\epsilon ( P \times W | \vecx ) = \bigl\{ \vecy \in \setY^n \colon ( \vecx,\vecy ) \in \eptyp (P \times W) \bigr\}.$$ Similarly, for a given BC $\channel {y,z} x$, $\setT^{( n )}_\epsilon ( P \times W_\ry | \vecx )$ is the set of $n$-tuples $\vecy \in \setY^n$ that are jointly $\epsilon$-typical with $\vecx$ w.r.t.\ $P \times W_\ry$, i.e., $$\setT^{( n )}_\epsilon ( P \times W_\ry | \vecx ) = \bigl\{ \vecy \in \setY^n \colon ( \vecx,\vecy ) \in \eptyp (P \times W_\ry) \bigr\};$$ and $\setT^{( n )}_\epsilon ( P \times W_\rz | \vecx )$ is the set of $n$-tuples $\vecz \in \setZ^n$ that are jointly $\epsilon$-typical with $\vecx$ w.r.t.\ $P \times W_\rz$.

A generic probability measure on a measurable space $( \Omega, \setF )$ is denoted $\dist$. If $\dist_1$ and $\dist_2$ are two probability measures on the same measurable space $( \Omega, \setF )$, then the Total-Variation distance $d ( \dist_1, \dist_2 )$ between $\dist_1$ and $\dist_2$ is $$d ( \dist_1, \dist_2 ) = \sup_{\setA \in \setF} \distsubof 1 \setA - \distsubof 2 \setA.$$  We shall only encounter measurable spaces $( \Omega, \setF )$ for which $\Omega$ is finite and $\setF = 2^\Omega$. On such spaces $$d ( \dist_1, \dist_2 ) = \frac{1}{2} \sum_{\omega \in \Omega} \bigl| \dist_1 ( \omega ) - \dist_2 ( \omega ) \bigr|.$$


\subsection{Some Useful Bounds}

We use the following multiplicative Chernoff bounds (see, e.g., \cite[Theorems~4.4 and 4.5]{mitzenmacherupfal05}):\footnote{The bound \eqref{eq:multChernDeltaGeq1} is not stated in \cite{mitzenmacherupfal05}. It is, however, a direct consequence of \cite[Theorem~4.4]{mitzenmacherupfal05} and the fact that $$e^\delta / (1+\delta)^{1+\delta} < e^{-\delta / 3}, \quad \delta \geq 1.$$}

\begin{proposition}\label{pr:multChernoff}
If $S_1, \ldots, S_n$ are independent binary random variables and $$\mu = \BiggEx {}{\sum^n_{i = 1} S_i},$$ then for all $0 < \delta < 1$
\begin{subequations}
\begin{IEEEeqnarray}{rCl}
\Biggdistof {\sum^n_{i = 1} S_i \leq (1 - \delta) \mu} & \leq & \exp \biggl\{ -\frac{\delta^2 \mu}{2} \biggr\}, \label{eq:multChernDeltaSm1Sm} \\
\Biggdistof {\sum^n_{i = 1} S_i \geq (1 + \delta) \mu} & \leq & \exp \biggl\{ -\frac{\delta^2 \mu}{3} \biggr\}, \label{eq:multChernDeltaSm1La}
\end{IEEEeqnarray}
\end{subequations}
and for all $\delta \geq 1$
\begin{IEEEeqnarray}{rCl}
\Biggdistof {\sum^n_{i = 1} S_i \geq (1 + \delta) \mu} & \leq & \exp \biggl\{ -\frac{\delta \mu}{3} \biggr\}. \label{eq:multChernDeltaGeq1}
\end{IEEEeqnarray}
\end{proposition}

We make frequent use of Hoeffding's inequality:

\begin{proposition}\cite[Theorem~2]{hoeffding63}\label{pr:hoeffding}
If $S_1, \ldots, S_n$ are independent random variables satisfying $S_i \in [a_i, b_i], \, i \in \{ 1, \ldots, n \}$, where $a_i, \, b_i \in \reals$, then for all $t > 0$
\begin{IEEEeqnarray}{rCl}
\Biggdistof {\frac{1}{n} \sum^n_{i = 1} \bigl( S_i - \Ex {}{S_i} \bigr) \geq t} & \leq & \exp \biggl\{ -\frac{2 n^2 t^2}{\sum^n_{i = 1} (b_i - a_i)^2} \biggr\}.
\end{IEEEeqnarray}
\end{proposition}

More general versions of this inequality can be found in \cite[Corollary~2.4.7]{dembozeitouni98} or \cite[Theorem~3.24]{rosspekoz07}.

\section{A Capacity-Achieving ID Code for the DMC}\label{sec:IDCodingTechnique}

In this section we present our capacity-achieving ID code for the DMC $\bigl( \setX, \channel y x, \setY \bigr)$. We begin with the basic definitions of an ID code \cite{ahlswededueck89} and with the capacity theorem.

\begin{defin}\label{def:IDCodeDMC}
Fix a finite set $\setM$, a blocklength $n \in \naturals$, and positive constants $\lambda_1, \, \lambda_2$. Associate with every ID message $m \in \setM$ a PMF $Q_{m}$ on $\setX^n$ and an ID set $\idsetre {m} \subset \setY^n$. The collection of tuples $\{ Q_{m}, \idsetre m \}_{m \in \setM}$ is an $( n, \setM, \lambda_1, \lambda_2 )$ ID code for the DMC $\channel y x$ if the maximum probability of missed identification
\begin{IEEEeqnarray}{rCl}
p_{\textnormal{missed-ID}} & = & \max_{m \in \setM} ( Q_m W^n ) ( Y^n \notin \idsetre m )
\end{IEEEeqnarray}
and the maximum probability of wrong identification
\begin{IEEEeqnarray}{rCl}
p_{\textnormal{wrong-ID}} & = & \max_{m \in \setM} \max_{m^\prime \neq m} ( Q_m W^n ) ( Y^n \in \idsetre {m^\prime} )
\end{IEEEeqnarray}
satisfy
\begin{IEEEeqnarray}{rCl}
p_{\textnormal{missed-ID}} & \leq & \lambda_1, \\
p_{\textnormal{wrong-ID}} & \leq & \lambda_2.
\end{IEEEeqnarray}
A rate $R$ is achievable if for every positive $\lambda_1$ and $\lambda_2$ and for every sufficiently-large blocklength $n$ there exists an $( n, \setM, \lambda_1, \lambda_2 )$ ID code for the DMC with
\begin{IEEEeqnarray*}{rrCll}
& \tfrac{1}{n} \log \log | \setM | & \geq & R \quad &\textnormal{if } R > 0, \\*[-0.625\normalbaselineskip]
\smash{\left\{
\IEEEstrut[6.39\jot]
\right.} \nonumber
\\*[-0.625\normalbaselineskip]
& |\setM| & = & 1 \quad &\textnormal{if } R = 0.
\end{IEEEeqnarray*}
The ID capacity $C$ of the DMC is the supremum of all achievable rates.
\end{defin}

The ID capacity was established in \cite{ahlswededueck89,hanverdu92}: Ahlswede and Dueck \cite{ahlswededueck89} proved the direct part and a soft converse, which holds for error probabilities that decay exponentially in the blocklength. The strong converse, which holds for all probabilities of missed and wrong identification satisfying $\lambda_1 + \lambda_2 < 1$, is due to Han and Verd\'u \cite{hanverdu92}.

\begin{theorem}\cite[Theorem~1]{ahlswededueck89} and \cite[Theorem~2]{hanverdu92} \label{th:IDDMC}
The ID capacity $C$ of the DMC $\channel y x$ is
\begin{IEEEeqnarray}{rCl}
C & = & \max_{P} \muti{P}{W}.
\end{IEEEeqnarray}
\end{theorem}


Fix any positive ID rate $R$ satisfying
\begin{equation}
0 < R < \max_{P} \muti{P}{W},
\end{equation}
and let $\setM$ be a size-$\exp (\exp (n R))$ set of possible ID messages. We assume that $\max_P \muti P W$ is positive, because rate $R = 0$ is always achievable (see Definition~\ref{def:IDCodeDMC}). We next describe our random code construction and show that, for every positive $\lambda_1$ and $\lambda_2$ and for every sufficiently-large blocklength $n$, it produces with high probability an $( n,\setM,\lambda_1,\lambda_2 )$ ID code for the DMC $\channel y x$.

\subparagraph*{Code Generation:} Choose a PMF $P$ on $\setX$ for which $$R < \muti {P}{W},$$ and fix an expected bin rate $\tilde R$ and a pool rate $R_\poolre$ satisfying
\begin{IEEEeqnarray}{C}
R < \tilde R < \muti {P}{W} \quad \textnormal{and} \quad \tilde R < R_\poolre. \label{eq:IDCodeConstraintsBinAndPoolRate}
\end{IEEEeqnarray}
Draw $e^{n R_\poolre}$ $n$-tuples $\sim P^n$ independently and place them in a pool $\pool$. Index the $n$-tuples in the pool by the elements of a size-$e^{n R_\poolre}$ set $\setV$, e.g., $\{ 1, \ldots, e^{n R_\poolre} \}$, and denote by $\poolel v$ the $n$-tuple in $\pool$ that is indexed by~$v \in \setV$. Associate with each ID message $m \in \setM$ an index-set $\indexset m$ and a bin $\bin m$ as follows. Select each element of $\setV$ for inclusion in $\indexset m$ independently with probability $e^{-n( R_\poolre - \tilde R )}$, and let Bin~$\bin m$ be the multiset that contains all the $n$-tuples in the pool that are indexed by $\indexset m$, $$\bin m = \bigl\{ \poolel v, \, v \in \indexset m \bigr\}.$$ (Bin~$\bin m$ is thus of expected size $e^{n \tilde R}$.)

Reveal the pool $\pool$, the index-sets $\{ \indexset m \}_{m \in \setM}$, and the corresponding bins $\{ \bin m \}_{m \in \setM}$ to all parties. The encoding and decoding are determined by
\begin{equation}
\rcode = \bigl( \pool, \{ \indexset m \}_{m \in \setM} \bigr). \label{eq:rcodeDMC}
\end{equation}\\

For the purpose of illustration, the pool and the bins are depicted in Figure~\ref{fig:idCodeConstruction}. As mentioned in Section~\ref{sec:intro}, our code is similar to the one in \cite{ahlswededueck89}: every ID message is associated with a bin, and in both schemes the bins are chosen at random from a pool. The main difference is that in our scheme the pool need not constitute a codebook that is reliable in Shannon's sense. Indeed, our pool is of size $e^{n R_\poolre} \!$, where $R_\poolre$ can exceed $\muti {P}{W}$ or even $C$.

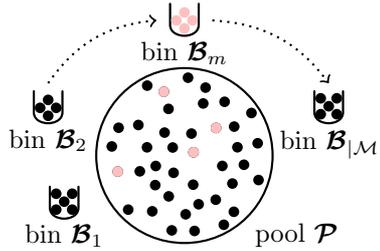
\begin{figure}[!ht]
\begin{center}
\def\pgfsysdriver{pgfsys-dvipdfm.def}
\begin{tikzpicture}
	\tikzstyle{pools}=[draw,thick,circle,text width = 4em, text height = 4em];
	\def\circledia{6em}
	\def\circlediam{3em}
	
	\node (pool) [pools] {};

\foreach \x/\y in {0.3944/0.3473,-0.1464/0.1013,-0.1986/0.5869,-0.3639/-0.3717,
-0.1579/-0.9408,0.0448/0.9791,0.9244/-0.3076,0.0153/-0.7388,0.1088/-0.3667,
0.4912/-0.8705,0.7422/-0.1157,-0.9057/0.1145,-0.7010/-0.6800,-0.4010/-0.7010,
0.6609/0.1043,-0.2567/0.8029,-0.3574/-0.1313,0.0877/0.2765,0.2339/0.6553,
-0.1959/-0.5130,-0.2479/0.2889,-0.4665/0.1035,-0.5745/0.3845,0.1128/0.0408,
0.1261/-0.9481,0.3905/-0.5552,0.5983/-0.3518,0.9135/0.1546,0.2635/0.8546,
-0.8208/0.3497,-0.4639/0.8796,0.4725/-0.1305,0.6205/0.3641,-0.8310/-0.201,
0.1768/-0.6123,0.5008/0.6763,-0.5561/-0.1982,-0.1355/-0.2170,0.8361/0.5351,
0.6167/-0.7071}
	\fill [black] (pool.center) + (\x*\circlediam,\y*\circlediam) circle (2pt);
	
	\path (pool.center) + (1.4*\circlediam,-1*\circlediam) node [text centered] {pool $\pool$};

	\path (pool.center) + (0,0.85*\circledia) node (binmz) [text centered] {\huge{$\cup$}};
	\path (binmz.south) + (0,-0.05*\circledia) node [text centered] {bin~$\bin m$};
	\fill [pink] (binmz.center) + (0*\circlediam,-0.1*\circlediam) circle (2pt);
	\fill [pink] (binmz.center) + (0.1*\circlediam,0*\circlediam) circle (2pt);
	\fill [pink] (binmz.center) + (-0.1*\circlediam,0*\circlediam) circle (2pt);
	\fill [pink] (binmz.center) + (0*\circlediam,0.1*\circlediam) circle (2pt);

	\foreach \x/\y in {0.3944/0.3473,-0.2567/0.8029,0.1128/0.0408,-0.8310/-0.201}
	\fill [pink] (pool.center) + (\x*\circlediam,\y*\circlediam) circle (2pt);

	\path (pool.center) + (-0.75*\circledia,-0.3*\circledia) node (bin1) [text centered] {\huge{$\cup$}};
	\path (bin1.south) + (0,-0.05*\circledia) node [text centered] {bin~$\bin 1$};
	\path (pool.center) + (-0.85*\circledia,0.3*\circledia) node (bin2) [text centered] {\huge{$\cup$}};
	\path (bin2.south) + (0,-0.05*\circledia) node [text centered] {bin~$\bin 2$};
	\path (pool.center) + (0.9*\circledia,0.3*\circledia) node (binl) [text centered] {\huge{$\cup$}};
	\path (binl.south) + (0.025*\circledia,-0.075*\circledia) node [text centered] {bin~$\bin {|\setM|}$};
	\draw [thick, ->, dotted, bend right = -25] (bin2.north) to (binmz.west);
	\draw [thick, ->, dotted, bend right = -25] (binmz.east) to (binl.north);

	\fill [black] (bin1.center) + (0.1*\circlediam,-0.08*\circlediam) circle (2pt);
	\fill [black] (bin1.center) + (-0.1*\circlediam,-0.08*\circlediam) circle (2pt);
	\fill [black] (bin1.center) + (0*\circlediam,0.02*\circlediam) circle (2pt);
	\fill [black] (bin1.center) + (0.1*\circlediam,0.12*\circlediam) circle (2pt);
	\fill [black] (bin1.center) + (-0.1*\circlediam,0.12*\circlediam) circle (2pt);

	\fill [black] (bin2.center) + (0*\circlediam,-0.1*\circlediam) circle (2pt);
	\fill [black] (bin2.center) + (0.1*\circlediam,0*\circlediam) circle (2pt);
	\fill [black] (bin2.center) + (-0.1*\circlediam,0*\circlediam) circle (2pt);
	\fill [black] (bin2.center) + (0*\circlediam,0.1*\circlediam) circle (2pt);
	
	\fill [black] (binl.center) + (0.1*\circlediam,-0.08*\circlediam) circle (2pt);
	\fill [black] (binl.center) + (-0.1*\circlediam,-0.08*\circlediam) circle (2pt);
	\fill [black] (binl.center) + (0*\circlediam,0.02*\circlediam) circle (2pt);
	\fill [black] (binl.center) + (0.1*\circlediam,0.12*\circlediam) circle (2pt);
	\fill [black] (binl.center) + (-0.1*\circlediam,0.12*\circlediam) circle (2pt);
    		
\end{tikzpicture}
\caption[ID code construction for the DMC]{ID code construction for the DMC.}
\label{fig:idCodeConstruction}
\end{center}
\end{figure}

\subparagraph*{Encoding:} To send ID Message $m \in \setM$, the encoder draws some $V$ uniformly at random from $\indexset m$ and transmits the sequence $\poolel V$. ID Message~$m$ is thus associated with the PMF
\begin{IEEEeqnarray}{rCl}
\bm Q_m ( \vecx ) & = & \frac{1}{| \indexset m |} \sum_{v \in \indexset m} \ind {\vecx = \poolel v}, \quad \vecx \in \setX^n, \quad \indexset m \neq \emptyset. \label{eq:distIDMsgDMC}
\end{IEEEeqnarray}
If $\indexset m$ is empty, then the encoder chooses $V = v^\star$ and transmits $\poolel {v^\star}$, where $v^\star$ is an arbitrary but fixed element of $\setV$, so
\begin{IEEEeqnarray}{rCl}
\bm Q_m ( \vecx ) & = & \ind {\vecx = \poolel {v^\star}}, \quad \vecx \in \setX^n, \quad \indexset m = \emptyset. \label{eq:distIDMsgDMC2}
\end{IEEEeqnarray}

\subparagraph*{Decoding:} In this section $\setT^{( n )}_\epsilon$ is short for $\setT^{( n )}_\epsilon ( P \times W )$, and the function $\delta ( \cdot )$ maps every nonnegative real number $u$ to $u \ent{P \times W}$. The decoders choose $\epsilon > 0$ sufficiently small so that $2 \delta ( \epsilon ) < \muti {P}{W} - \tilde R$. The $m^\prime$-focused party guesses that $m^\prime$ was sent if, and only if, (iff) for some index $v \in \indexset {m^\prime}$ the $n$-tuple $\poolel v$ in Bin~$\bin {m^\prime}$ is jointly $\epsilon$-typical with the channel-output sequence $Y^n$, i.e., iff $\bigl( \poolel v, Y^n \bigr) \in \eptyp$ for some $v \in \indexset {m^\prime}$. The set $\idset {m^\prime}$ of output sequences that result in the guess ``$m^\prime$ was sent'' is thus
\begin{IEEEeqnarray}{rCl}
\idset {m^\prime} & = & \Bigl\{ \vecy \in \setY^n \colon \exists \, v \in \indexset {m^\prime} \textnormal{ s.t.\ } \bigl( \poolel v, \vecy \bigr) \in \eptyp \Bigr\} \\
& = & \bigcup_{v \in \indexset {m^\prime}} \setT^{( n )}_\epsilon \bigl( P \times W \bigl| \poolel v \bigr). \label{eq:DefIDSet}
\end{IEEEeqnarray}

\subparagraph*{Analysis of the Probabilities of Missed and Wrong Identification:} We first note that $\rcode$ (together with the fixed blocklength $n$, the fixed element $v^\star$ of $\setV$, and the chosen $\epsilon$) fully specifies the encoding and guessing rules. That is, the randomly constructed ID code $\{ \bm Q_m, \idset m \}_{m \in \setM}$ is fully specified by $\rcode$. Let $\dist$ be the distribution of $\rcode$, and let $\Exop$ denote expectation w.r.t.\ $\dist$. Subscripts indicate conditioning on the event that some of the chance variables assume the values indicated by the subscripts, e.g., $\dist_{\indexsetre m}$ denotes the distribution conditional on $\indexset m = \indexsetre m$, and $\Exop_{\indexsetre m}$ denotes the expectation w.r.t.\ $\dist_{\indexsetre m}$.\\

The maximum probabilities of missed and wrong identification of the randomly constructed ID code $\{ \bm Q_m, \idset m \}_{m \in \setM}$ are the random variables
\begin{subequations}
\begin{IEEEeqnarray}{rCl}\label{bl:randomErrorPrIDDMC}
P_{\textnormal{missed-ID}} & = & \max_{m \in \setM} ( \bm Q_m W^n ) ( Y^n \notin \idset m ), \label{eq:randomPrMissedIDDMC} \\
P_{\textnormal{wrong-ID}} & = & \max_{m \in \setM} \max_{m^\prime \neq m} ( \bm Q_m W^n ) ( Y^n \in \idset {m^\prime} ). \label{eq:randomPrWrongIDDMC}
\end{IEEEeqnarray}
\end{subequations}
They are fully specified by $\rcode$. How we upper-bound these probabilities depends on the size of the index-sets and of their pairwise intersections. For every distinct pair $m, \, m^\prime \in \setM$ denote the intersection of the index-sets $\indexset m$ and $\indexset {m^\prime}$ by $\indexset {m,m^\prime}$, so
\begin{IEEEeqnarray}{C}
\indexset {m,m^\prime} = \indexset m \cap \indexset {m^\prime}. \label{eq:IDCodeDefIntersectionVmmprime}
\end{IEEEeqnarray}
The expected size of $\indexset {m, m^\prime}$ is $e^{n (2 \tilde R - R_\poolre)}$ ($= e^{n R_\poolre} e^{ - 2 n (R_\poolre - \tilde R)}$) and is thus, by \eqref{eq:IDCodeConstraintsBinAndPoolRate}, exponentially smaller than the expected size of the index-sets $\indexset m$ and $\indexset {m^\prime}$, which is $e^{n \tilde R}$. The following lemma upper-bounds the probability that the size of the index-sets deviates from its mean $e^{n \tilde R}$ or that the pairwise intersections are large compared to $e^{n \tilde R}$. To state the lemma, we first introduce the set $\setG_\mu$ comprising the realizations $\{ \indexsetre m \}_{m \in \setM}$ of the index-sets $\{ \indexset m \}_{m \in \setM}$ satisfying that for every distinct pair $m, \, m^\prime \in \setM$ the following three inequalities hold:
\begin{subequations} \label{bl:setGProperties}
\begin{IEEEeqnarray}{rCl}
| \indexsetre m | & > & ( 1 - \delta_n) e^{n \tilde R}, \label{eq:indexsetmSuffLarge} \\
| \indexsetre {m^\prime} | & < & ( 1 + \delta_n) e^{n \tilde R}, \label{eq:indexsetmprimeSuffSmall} \\
| \indexsetre {m,m^\prime} | & < & e^{n ( \tilde R - \mu / 2 ) + \log 2}, \label{eq:indexsetIntersectionSuffSmall}
\end{IEEEeqnarray}
\end{subequations}
where $\mu$ is fixed and satisfies
\begin{IEEEeqnarray}{C}
0 < \mu < \min \bigl\{ R_\poolre - \tilde R, \tilde R - R \bigr\}, \label{eq:IDCodeMu}
\end{IEEEeqnarray}
and
\begin{IEEEeqnarray}{C}
\delta_n = e^{-n \mu / 2}. \label{eq:IDCodeDeltan}
\end{IEEEeqnarray}

\begin{lemma}\label{le:prEventEToZero}
The probability that $\{ \indexset m \}_{m \in \setM}$ is not in $\setG_\mu$ converges to zero as the blocklength~$n$ tends to infinity:
\begin{IEEEeqnarray}{l}
\lim_{n \rightarrow \infty} \bigdistof {\{ \indexset m \}_{m \in \setM} \notin \setG_\mu} = 0.
\end{IEEEeqnarray}
\end{lemma}

\begin{proof}
See Appendix~\ref{app:lePrEventEToZero}.
\end{proof}

To prove that for every choice of $\lambda_1, \, \lambda_2 > 0$ and $n$ sufficiently large the collection of tuples $\{ \bm Q_m, \idset m \}_{m \in \setM}$ is with high probability an $( n,\setM,\lambda_1,\lambda_2 )$ ID code for the DMC $\channel y x$, we prove the following stronger result:

\begin{claim}\label{re:convExpFast}
The maximum probability of missed identification, $P_{\textnormal{missed-ID}}$, and the maximum probability of wrong identification, $P_{\textnormal{wrong-ID}}$, of the randomly constructed ID code $\{ \bm Q_m, \idset m \}_{m \in \setM}$ converge in probability to zero exponentially in the blocklength~$n$, i.e.,
\begin{IEEEeqnarray}{l}
\exists \, \tau > 0 \textnormal{ s.t.\ } \lim_{n \rightarrow \infty} \bigdistof { \max \{ P_{\textnormal{missed-ID}}, P_{\textnormal{wrong-ID}} \} \geq e^{-n \tau} } = 0. \label{bl:toShowIDDMC}
\end{IEEEeqnarray}
\end{claim}

\begin{proof}
Fix some $\mu$ satisfying \eqref{eq:IDCodeMu}, and choose $\delta_n$ as in \eqref{eq:IDCodeDeltan}. We upper-bound $P_{\textnormal{missed-ID}}$ and $P_{\textnormal{wrong-ID}}$ differently depending on whether or not $\{ \indexset \nu \}$ is in $\setG_\mu$, where $\{ \indexset \nu \}$ is short for $\{ \indexset \nu \}_{\nu \in \setM}$. If $\{ \indexset \nu \} \notin \setG_\mu$, then we upper-bound them by one to obtain for every $\tau > 0$
\begin{IEEEeqnarray}{l}
\bigdistof { \max \{ P_{\textnormal{missed-ID}}, P_{\textnormal{wrong-ID}} \} \geq e^{-n \tau} } \nonumber \\
\quad \leq \bigdistof { \{ \indexset \nu \} \notin \setG_\mu } + \sum_{ \indexsetsre \in \setG_\mu } \bigdistof { \{ \indexset \nu \} = \indexsetsre } \, \bigdistsubof { \indexsetsre } { \max \{ P_{\textnormal{missed-ID}}, P_{\textnormal{wrong-ID}} \} \geq e^{-n \tau} }.
\end{IEEEeqnarray}
By Lemma~\ref{le:prEventEToZero} the first term on the RHS converges to zero as the blocklength~$n$ tends to infinity, and it thus suffices to show that
\begin{IEEEeqnarray}{l}
\exists \, \tau > 0 \textnormal{ s.t.\ } \lim_{n \rightarrow \infty} \max_{\indexsetsre \in \setG_\mu}  \bigdistsubof { \indexsetsre } { \max \{ P_{\textnormal{missed-ID}}, P_{\textnormal{wrong-ID}} \} \geq e^{-n \tau} } = 0. \label{bl:toShowIDDMCSetG}
\end{IEEEeqnarray}

\begin{remark}\label{re:drawIndexSetsAtRandom}
As we shall see, \eqref{bl:toShowIDDMCSetG} does indeed hold, and we could have therefore simplified our random code construction considerably by drawing only the pool $\pool$ at random while fixing the index-sets $\indexsetsre \in \setG_\mu$. This is correct, but the main purpose of our random code construction for the DMC is to pave the way for the one for the BC, and there we shall need to draw the index-sets at random.
\end{remark}

Henceforth we assume that $n$ is large enough so that the following two inequalities hold:
\begin{subequations}\label{bl:nSuffLarge}
\begin{IEEEeqnarray}{rCl}
( 1 - \delta_n) e^{n \tilde R} & \geq & 1, \label{eq:nSuffLargeIndexSetsNonempty} \\
\delta_n + e^{-n \mu / 2 + \log 2} & \leq & 1/2, \label{eq:nSuffLargeIndexSetWithoutIntersectionLarge}
\end{IEEEeqnarray}
\end{subequations}
where $\delta_n$ is defined in \eqref{eq:IDCodeDeltan}. (This is possible, because $\delta_n$ converges to zero as $n$ tends to infinity and $\tilde R, \, \mu > 0$.)\\

To establish \eqref{bl:toShowIDDMCSetG}, we first show that
\begin{IEEEeqnarray}{l}
\exists \, \tau > 0 \textnormal{ s.t.\ } \lim_{n \rightarrow \infty} \max_{\indexsetsre \in \setG_\mu} \bigdistsubof {\indexsetsre}{P_{\textnormal{missed-ID}} \geq e^{-n \tau}} = 0, \label{eq:toShowIDDMCPrMissedID}
\end{IEEEeqnarray}
and we then show that
\begin{IEEEeqnarray}{l}
\exists \, \tau > 0 \textnormal{ s.t.\ } \lim_{n \rightarrow \infty} \max_{\indexsetsre \in \setG_\mu} \bigdistsubof { \indexsetsre }{P_{\textnormal{wrong-ID}} \geq e^{-n \tau}} = 0. \label{eq:toShowIDDMCPrWrongID}
\end{IEEEeqnarray}
The Union-of-Events bound, \eqref{eq:toShowIDDMCPrMissedID}, and \eqref{eq:toShowIDDMCPrWrongID} imply \eqref{bl:toShowIDDMCSetG} and hence \eqref{bl:toShowIDDMC}.

To conclude the proof, it remains to establish \eqref{eq:toShowIDDMCPrMissedID} and \eqref{eq:toShowIDDMCPrWrongID}. We start by establishing \eqref{eq:toShowIDDMCPrMissedID}. To this end fix any realization $\indexsetsre$ in $\setG_\mu$. Rather than directly upper-bounding the maximum over $m \in \setM$ of $( \bm Q_m W^n ) ( Y^n \notin \idset m )$ under $\dist_{\indexsetsre}$, we first consider $( \bm Q_m W^n ) ( Y^n \notin \idset m )$ for a fixed $m \in \setM$. (This $\sigof \rcode$-measurable random variable with support $[0,1]$ can be viewed as the probability---associated with the randomly constructed ID code---that the $m$-focused party erroneously guesses that $m$ was not sent.) By \eqref{eq:indexsetmSuffLarge} (which holds because $\indexsetsre \in \setG_\mu$) and \eqref{eq:nSuffLargeIndexSetsNonempty}, $\indexsetre m$ is nonempty, and $\bm Q_m$ is hence given by \eqref{eq:distIDMsgDMC}. This implies that $\dist_{\indexsetsre}$-almost-surely the random variable $( \bm Q_m W^n ) ( Y^n \notin \idset m )$ is upper-bounded by
\begin{IEEEeqnarray}{l}
( \bm Q_m W^n ) ( Y^n \notin \idset m ) \nonumber \\
\quad \stackrel{(a)}= \sum_{\vecx \in \setX^n} \frac{1}{| \indexsetre m |} \sum_{v \in \indexsetre m} \ind {\vecx = \poolel v} \, W^n (Y^n \notin \idset m | \vecx) \\
\quad \stackrel{(b)}\leq \frac{1}{| \indexsetre m |} \sum_{v \in \indexsetre m} W^n \Bigl(Y^n \notin \setT^{( n )}_\epsilon \bigl( P \times W \bigl| \poolel v \bigr) \Bigl| \poolel v \Bigr), \label{eq:sumOfIndRVsLeMissedID}
\end{IEEEeqnarray}
where $(a)$ follows from \eqref{eq:distIDMsgDMC}; and $(b)$ follows from \eqref{eq:DefIDSet}, which implies that $\dist_{\indexsetsre}$-almost-surely $$\setT^{( n )}_\epsilon \bigl( P \times W \bigl| \poolel v \bigr) \subseteq \idset m, \quad v \in \indexsetre m.$$ There is an inequality in $(b)$, because the $m$-focused party may guess correctly even if $\vecy$ is not jointly typical with $\poolel v$: it also guesses correctly when $\vecy$ is jointly typical with $\poolel {v^\prime}$ for some $v^\prime$ in $\indexsetre m$ other than $v$.

Let
\begin{subequations}\label{bl:betaAlpha}
\begin{IEEEeqnarray}{rCl}
\beta_n & = & ( P \times W )^n \Bigl( ( X^n,Y^n ) \notin \eptyp \Bigr), \\
\alpha_n & = & \max \bigl\{ 2 \beta_n, e^{-n \mu/2} \bigr\}, \label{eq:blBetaAlphaAlpha}
\end{IEEEeqnarray}
\end{subequations}
and note that \eqref{eq:blBetaAlphaAlpha} implies that
\begin{IEEEeqnarray}{rCl}
\alpha_n - \beta_n & \geq & e^{-n \mu / 2} / 2. \label{eq:alphaMinusBeta}
\end{IEEEeqnarray}
Moreover, since $\beta_n$ decays exponentially and $\mu > 0$, there must exist a positive constant $\tau > 0$ and some $\eta_0 \in \naturals$ for which
\begin{IEEEeqnarray}{rCl}
\alpha_n \leq e^{ - n \tau }, \quad  n \geq \eta_0. \label{eq:alphaExpSmallInBlocklength}
\end{IEEEeqnarray}

Under $\dist_{\indexsetsre}$ the $[0,1]$-valued random variables $$\biggl\{ W^n \Bigl(Y^n \notin \setT^{( n )}_\epsilon \bigl( P \times W \bigl| \poolel v \bigr) \Bigl| \poolel v \Bigr) \biggr\}_{v \in \setV}$$ are IID and have mean $\beta_n$, because the pool was drawn independently of the index-sets, so $\bigl\{ \poolel v \bigr\}_{v \in \setV}$ are IID $\sim P^n$ also under $\dist_{\indexsetsre}$. Consequently, Hoeffding's inequality (Proposition~\ref{pr:hoeffding}) implies that
\begin{IEEEeqnarray}{l}
\Biggdistsubof {\indexsetsre} {\frac{1}{|\indexsetre m|} \sum_{v \in \indexsetre m} W^n \Bigl(Y^n \notin \setT^{( n )}_\epsilon \bigl( P \times W \bigl| \poolel v \bigr) \Bigl| \poolel v \Bigr) \geq \alpha_n } \nonumber \\
\quad \leq e^{ - 2 \, | \indexsetre m | ( \alpha_n - \beta_n )^2 } \\
\quad \leq \exp \Bigl\{ -  (1 - \delta_n) e^{ n ( \tilde R - \mu ) - \log 2}  \Bigr\}, \quad \indexsetsre \in \setG_\mu, \label{eq:hoeffSumOfIndRvsLeMissedID}
\end{IEEEeqnarray}
where in the second inequality we used \eqref{eq:indexsetmSuffLarge} (which holds because $\indexsetsre \in \setG_\mu$) and \eqref{eq:alphaMinusBeta}. Having obtained \eqref{eq:hoeffSumOfIndRvsLeMissedID} for every fixed $m$, we are now ready to tackle the maximum over $m$ and prove \eqref{eq:toShowIDDMCPrMissedID}: for every $\tau > 0$ and $\eta_0 \in \naturals$ satisfying \eqref{eq:alphaExpSmallInBlocklength} and for all $n$ exceeding $\eta_0$
\begin{IEEEeqnarray}{l}
\max_{\indexsetsre \in \setG_\mu} \bigdistsubof {\indexsetsre} { P_{\textnormal{missed-ID}} \geq e^{ - n \tau } } \nonumber \\
\quad \stackrel{(a)}\leq \max_{\indexsetsre \in \setG_\mu} \bigdistsubof {\indexsetsre} { P_{\textnormal{missed-ID}} \geq \alpha_n } \\
\quad \stackrel{(b)}= \max_{\indexsetsre \in \setG_\mu} \bigdistsubof {\indexsetsre} {\exists \, m \in \setM \colon ( \bm Q_m W^n ) ( Y^n \notin \idset m ) \geq \alpha_n } \\
\quad \stackrel{(c)}\leq \max_{\indexsetsre \in \setG_\mu} \sum_{m \in \setM} \bigdistsubof {\indexsetsre}{ ( \bm Q_m W^n ) ( Y^n \notin \idset m ) \geq \alpha_n } \\
\quad \stackrel{(d)}\leq \max_{\indexsetsre \in \setG_\mu} \sum_{m \in \setM} \Biggdistsubof {\indexsetsre}{\frac{1}{| \indexsetre m |} \sum_{v \in \indexsetre m} W^n \Bigl(Y^n \notin \setT^{( n )}_\epsilon \bigl( P \times W \bigl| \poolel v \bigl) \Bigl| \poolel v \Bigr) \geq \alpha_n } \\
\quad \stackrel{(e)}\leq \sum_{m \in \setM} \exp \Bigl\{ - (1 - \delta_n) e^{ n ( \tilde R - \mu ) - \log 2} \Bigr\} \\
\quad \stackrel{(f)}\leq | \setM | \exp \Bigl\{ - e^{n ( \tilde R - \mu ) - 2 \log 2 } \Bigr\} \\
\quad \stackrel{(g)}\rightarrow 0 \, ( n \rightarrow \infty ), \label{eq:maxPrMissedIDDMCConvZero}
\end{IEEEeqnarray}
where $(a)$ holds by \eqref{eq:alphaExpSmallInBlocklength}, because $n$ exceeds $\eta_0$; $(b)$ follows from \eqref{eq:randomPrMissedIDDMC}; $(c)$ follows from the Union-of-Events bound; $(d)$ follows from \eqref{eq:sumOfIndRVsLeMissedID}; $(e)$ holds by \eqref{eq:hoeffSumOfIndRvsLeMissedID}; $(f)$ follows from \eqref{eq:nSuffLargeIndexSetWithoutIntersectionLarge}, which implies that $\delta_n \leq 1/2$; and $(g)$ holds because $\card {\setM} = \exp (\exp (n R))$ and $\mu < \tilde R - R$.\\

Having established \eqref{eq:toShowIDDMCPrMissedID}, it remains to establish \eqref{eq:toShowIDDMCPrWrongID} in order to conclude the proof. To this end fix any realization $\indexsetsre$ in $\setG_\mu$. We begin by upper-bounding  $( \bm Q_m W^n ) ( Y^n \in \idset {m^\prime} )$ under $\dist_{\indexsetsre}$ for fixed distinct $m, \, m^\prime \in \setM$. Later we will maximize over such $m, \, m^\prime$. (The $\sigof \rcode$-measurable random variable $( \bm Q_m W^n ) ( Y^n \in \idset {m^\prime} )$ with support $[0,1]$ can be viewed as the probability---associated with the randomly constructed ID code---that the $m^\prime$-focused party erroneously guesses that $m^\prime$ was sent when in fact $m$ was sent.) By \eqref{eq:indexsetmSuffLarge} (which holds because $\indexsetsre \in \setG_\mu$) and \eqref{eq:nSuffLargeIndexSetsNonempty}, $\indexsetre m$ is nonempty, and $\bm Q_m$ is hence given by \eqref{eq:distIDMsgDMC}. This implies that $\dist_{\indexsetsre}$-almost-surely the random variable $( \bm Q_m W^n ) ( Y^n  \in \idset {m^\prime} )$ is upper-bounded by
\begin{IEEEeqnarray}{l}
( \bm Q_m W^n ) ( Y^n \in \idset {m^\prime} ) \nonumber \\
\quad \stackrel{(a)}= \sum_{\vecx \in \setX^n} \frac{1}{| \indexsetre m |} \sum_{v \in \indexsetre m} \ind {\vecx = \poolel v} \, W^n (Y^n \in \idset {m^\prime} | \vecx) \\
\quad = \frac{1}{|\indexsetre m|} \sum_{v \in \indexsetre m} W^n \bigl( Y^n \in \idset {m^\prime} \bigl| \poolel v \bigr) \\
\quad \stackrel{(b)}\leq \frac{|\indexsetre {m,m^\prime}|}{|\indexsetre m|} + \frac{1}{|\indexsetre m|} \sum_{v \in \indexsetre {m} \setminus \indexsetre {m,m^\prime}} W^n \bigl( Y^n \in \idset {m^\prime} \bigl| \poolel v \bigr), \label{eq:probWrongIDRatioAndProbInOtherWTSet}
\end{IEEEeqnarray}
where $(a)$ follows from \eqref{eq:distIDMsgDMC}; and $(b)$ holds because $$W^n \bigl( Y^n \in \idset {m^\prime} \bigl| \poolel v \bigr) \leq 1, \quad v \in \setV.$$

We consider the two terms on the RHS of \eqref{eq:probWrongIDRatioAndProbInOtherWTSet} separately, beginning with $|\indexsetre {m,m^\prime}| / |\indexsetre m|$. Because $\indexsetsre \in \setG_\mu$,
\begin{IEEEeqnarray}{rCl}
\frac{|\indexsetre {m,m^\prime}|}{|\indexsetre {m}|} \stackrel{(a)}< \frac{e^{n (\tilde R - \mu / 2 ) + \log 2}}{(1 - \delta_n) e^{n \tilde R}} \stackrel{(b)}\leq e^{-n \mu / 2 + 2 \log 2},
\label{eq:probWrongIDRatio}
\end{IEEEeqnarray}
where $(a)$ follows from \eqref{eq:indexsetmSuffLarge} and \eqref{eq:indexsetIntersectionSuffSmall}; and $(b)$ follows from \eqref{eq:nSuffLargeIndexSetWithoutIntersectionLarge}, which implies that $\delta_n \leq 1/2$. We next consider the second term in \eqref{eq:probWrongIDRatioAndProbInOtherWTSet}, namely, $$\frac{1}{|\indexsetre m|} \sum_{v \in \indexsetre m \setminus \indexsetre {m,m^\prime}} W^n \bigl( Y^n \in \idset {m^\prime} \bigl| \poolel v \bigr).$$ The cardinality of $\idset {m^\prime}$ is $\dist_{\indexsetsre}$-almost-surely upper-bounded by
\begin{IEEEeqnarray}{rCl}
|\idset {m^\prime}| & \stackrel{(a)}= & \biggl| \bigcup_{v \in \indexsetre {m^\prime}} \eptyp \bigl( P \times W \bigl| \poolel v \bigr) \biggr| \nonumber \\
& \leq & \sum_{v \in \indexsetre {m^\prime}} \Bigl| \eptyp \bigl( P \times W \bigl| \poolel v \bigr) \Bigr| \\
& \stackrel{(b)}\leq & (1 + \delta_n) e^{n ( \tilde R + \condent {W}{P} + \delta (\epsilon) )}, \label{eq:cardBsetL}
\end{IEEEeqnarray}
where $(a)$ follows from \eqref{eq:DefIDSet}; and $(b)$ follows from $$\Bigl| \eptyp ( P \times W | \vecx ) \Bigr| \leq e^{n ( \condent {W}{P} + \delta (\epsilon) ) }, \quad \vecx \in \setX^n,$$ and from \eqref{eq:indexsetmprimeSuffSmall} (which holds because $\indexsetsre \in \setG_\mu$).

Let
\begin{subequations}\label{bl:gammaEta}
\begin{IEEEeqnarray}{rCl}
\gamma_n & = & (1 + \delta_n) e^{- n ( \muti {P}{W} - \tilde R - 2 \delta (\epsilon) )}, \label{eq:blGammaEtaGamma} \\
\kappa_n & = & \max \bigl\{ 2 \gamma_n, e^{-n \mu / 2} \bigr\}, \label{eq:blGammaEtaEta}
\end{IEEEeqnarray}
\end{subequations}
and note that \eqref{eq:blGammaEtaEta} implies that
\begin{IEEEeqnarray}{rCl}
\kappa_n - \gamma_n \geq e^{-n \mu / 2} / 2. \label{eq:etaMinusGamma}
\end{IEEEeqnarray}

Fix a realization $\idsetre {m^\prime}$ of $\idset {m^\prime}$ for which $\distsubof {\indexsetsre} {\idset {m^\prime} = \idsetre {m^\prime}} > 0$. From \eqref{eq:DefIDSet} it follows that all output sequences in $\idsetre {m^\prime}$ are of approximate type $P W$, i.e., that
\begin{IEEEeqnarray}{rCl}
\idsetre {m^\prime} & \subseteq & \eptyp ( P W ). \label{eq:formBsetL}
\end{IEEEeqnarray}
And from \eqref{eq:cardBsetL} it follows that
\begin{IEEEeqnarray}{rCl}
| \idsetre {m^\prime} | & \leq & (1 + \delta_n) e^{n ( \tilde R + \condent {W}{P} + \delta (\epsilon) )}. \label{eq:cardBsetL2}
\end{IEEEeqnarray}
The next computation is under $\dist_{\indexsetsre,\idsetre {m^\prime}}$, where we condition not only on $\{ \indexset \nu \} = \indexsetsre$ but also on $\idset {m^\prime} = \idsetre {m^\prime}$. The $n$-tuples in the pool $\bigl\{ \poolel v \bigr\}_{v \in \setV \setminus \indexsetre {m^\prime}}$ that are not indexed by $\indexsetre {m^\prime}$ are IID $\sim P^n$ also under $\dist_{\indexsetsre,\idsetre {m^\prime}}$, because the pool was drawn independently of the index-sets, and because by \eqref{eq:DefIDSet} $\idset {m^\prime}$ depends only on $\bigl\{ \poolel v \bigr\}_{v \in \indexset {m^\prime}}$. Hence, under $\dist_{\indexsetsre,\idsetre {m^\prime}}$ the  $[0,1]$-valued random variables $$\Bigl\{ W^n \bigl( Y^n \in \idset {m^\prime} \bigl| \poolel v \bigr) \Bigr\}_{ v \in \setV \setminus \indexsetre {m^\prime} }$$ are IID of mean
\begin{IEEEeqnarray}{l}
\BigEx{\indexsetsre, \idsetre {m^\prime}}{W^n \bigl( Y^n \in \idset {m^\prime} \bigl| \poolel v \bigr)} \nonumber \\
\quad \stackrel{(a)}= \sum_{\vecy \in \idsetre {m^\prime}} ( P W )^n ( \vecy ) \\
\quad \stackrel{(b)}\leq |\idsetre {m^\prime}| \, e^{ - n ( \ent {P W} - \delta (\epsilon) ) } \\
\quad \stackrel{(c)}\leq (1 + \delta_n) e^{-n ( \muti {P}{W} - \tilde R - 2 \delta (\epsilon) )} \\
\quad \stackrel{(d)}= \gamma_n, \label{eq:meanRVsWrongIDLe}
\end{IEEEeqnarray}
where $(a)$ holds because $\idset {m^\prime} = \idsetre {m^\prime}$ and $\bigl\{ \poolel v \bigr\}_{v \in \setV \setminus \indexsetre {m^\prime}}$ are IID $\sim P^n$ under $\dist_{\indexsetsre, \idsetre {m^\prime}}$; $(b)$ holds because $$( P W )^n ( \vecy ) \leq e^{ - n ( \ent{P W} - \delta (\epsilon) ) }, \quad \vecy \in \eptyp ( P W ),$$ and by \eqref{eq:formBsetL}; $(c)$ follows from \eqref{eq:cardBsetL2}; and $(d)$ holds by \eqref{eq:blGammaEtaGamma}. Consequently, Hoeffding's inequality (Proposition~\ref{pr:hoeffding}) implies that
\begin{IEEEeqnarray}{l}
\Biggdistsubof {\indexsetsre, \idsetre {m^\prime}}{\frac{1}{| \indexsetre m \setminus \indexsetre {m,m^\prime} |} \sum_{v \in \indexsetre m \setminus \indexsetre {m,m^\prime}} W^n \bigl( Y^n \in \idset {m^\prime} \bigl| \poolel v \bigr) \geq \kappa_n } \nonumber \\
\quad \stackrel{(a)}\leq \exp \bigl\{ - 2 \, | \indexsetre m \setminus \indexsetre {m,m^\prime} | \, ( \kappa_n - \gamma_n  )^2 \bigr\} \\
\quad \stackrel{(b)}\leq \exp \bigl\{ - | \indexsetre m \setminus \indexsetre {m,m^\prime} | \, e^{ - n \mu - \log 2} \bigr\} \\
\quad \stackrel{(c)} \leq \exp \Bigl\{ - e^{ n ( \tilde R - \mu ) - 2 \log 2} \Bigr\}, \quad \indexsetsre \in \setG_\mu, \, \distsubof {\indexsetsre} {\idset {m^\prime} = \idsetre {m^\prime}} > 0, \label{eq:probInOtherWTSet}
\end{IEEEeqnarray}
where $(a)$ holds because $\indexsetre m \setminus \indexsetre {m,m^\prime}$ is a subset of $\setV \setminus \indexsetre {m^\prime}$; $(b)$ follows from \eqref{eq:etaMinusGamma}; and $(c)$ follows from
\begin{IEEEeqnarray}{l}
\card {\indexsetre m \setminus \indexsetre {m,m^\prime}} \stackrel{(d)}>  (1 - \delta_n) e^{n \tilde R} - e^{n ( \tilde R - \mu / 2 ) + \log 2} \stackrel{(e)}\geq e^{n \tilde R - \log 2},
\end{IEEEeqnarray}
where $(d)$ is due to \eqref{eq:indexsetmSuffLarge} and \eqref{eq:indexsetIntersectionSuffSmall} (which hold because $\indexsetsre \in \setG_\mu$), and $(e)$ is due to \eqref{eq:nSuffLargeIndexSetWithoutIntersectionLarge}. By \eqref{eq:probInOtherWTSet} and because $| \indexsetre m \setminus \indexsetre {m,m^\prime} | \leq |\indexsetre m|$, the probability that the second term in \eqref{eq:probWrongIDRatioAndProbInOtherWTSet} exceeds $\kappa_n$ is upper-bounded by
\begin{IEEEeqnarray}{l}
\Biggdistsubof {\indexsetsre}{\frac{1}{|\indexsetre m|} \sum_{v \in \indexsetre m \setminus \indexsetre {m,m^\prime}} W^n \bigl( Y^n \in \idset {m^\prime} \bigl| \poolel v \bigr) \geq \kappa_n } \nonumber \\
\quad \leq \Biggdistsubof {\indexsetsre}{\frac{1}{| \indexsetre m \setminus \indexsetre {m,m^\prime} |} \sum_{v \in \indexsetre m \setminus \indexsetre {m,m^\prime}} W^n \bigl( Y^n \in \idset {m^\prime} \bigl| \poolel v \bigr) \geq \kappa_n } \\
\quad = \sum_{\idsetre {m^\prime}} \distsubof {\indexsetsre} {\idset {m^\prime} = \idsetre {m^\prime}} \, \Biggdistsubof {\indexsetsre, \idsetre {m^\prime}}{\frac{1}{| \indexsetre m \setminus \indexsetre {m,m^\prime} |} \sum_{v \in \indexsetre m \setminus \indexsetre {m,m^\prime}} W^n \bigl( Y^n \in \idset {m^\prime} \bigl| \poolel v \bigr) \geq \kappa_n } \\
\quad \leq  \exp \Bigl\{ - e^{ n ( \tilde R - \mu ) - 2 \log 2} \Bigr\}, \quad \indexsetsre \in \setG_\mu. \label{eq:probInOtherWTSet2}
\end{IEEEeqnarray}

Having obtained \eqref{eq:probWrongIDRatioAndProbInOtherWTSet}, \eqref{eq:probWrongIDRatio}, and \eqref{eq:probInOtherWTSet2} for every fixed distinct $m, \, m^\prime$, we are now ready to tackle the maximum over $m, \, m^\prime $ and prove \eqref{eq:toShowIDDMCPrWrongID}: Let
\begin{IEEEeqnarray}{rCl}
\omega_n = e^{-n \mu / 2 + 2 \log 2} + \kappa_n, \label{eq:omegaPrWrongIDDMC}
\end{IEEEeqnarray}
and note that, by \eqref{bl:gammaEta}, because $\mu > 0$, because $\delta_n$ converges to zero as $n$ tends to infinity, and because $2 \delta ( \epsilon ) < \muti {P}{W} - \tilde R$, there must exist a positive constant $\tau > 0$ and some $\eta_0 \in \naturals$ for which
\begin{IEEEeqnarray}{rCl}
\omega_n \leq e^{ - n \tau }, \quad  n \geq \eta_0. \label{eq:omegaExpSmallInBlocklength}
\end{IEEEeqnarray}
For every $\tau > 0$ and $\eta_0 \in \naturals$ satisfying \eqref{eq:omegaExpSmallInBlocklength} and for all $n$ exceeding $\eta_0$
\begin{IEEEeqnarray}{l}
\max_{\indexsetsre \in \setG_\mu} \bigdistsubof { \indexsetsre }{ P_{\textnormal{wrong-ID}} \geq e^{-n \tau} } \nonumber \\
\quad \stackrel{(a)}\leq \max_{\indexsetsre \in \setG_\mu} \bigdistsubof { \indexsetsre }{ P_{\textnormal{wrong-ID}} \geq \omega_n } \\
\quad \stackrel{(b)}= \max_{ \indexsetsre \in \setG_\mu} \bigdistsubof {\indexsetsre } { \exists \, m, m^\prime \in \setM, \, m \neq m^\prime \colon ( \bm Q_m W^n ) ( Y^n \in \idset {m^\prime} ) \geq \omega_n } \\
\quad \stackrel{(c)}\leq \max_{ \indexsetsre \in \setG_\mu} \sum_{m \in \setM} \sum_{m^\prime \neq m} \bigdistsubof{ \indexsetsre }{ ( \bm Q_m W^n ) ( Y^n \in \idset {m^\prime} ) \geq \omega_n} \\
\quad \stackrel{(d)}\leq \max_{ \indexsetsre \in \setG_\mu} \sum_{m \in \setM} \sum_{m^\prime \neq m} \Biggdistsubof{ \indexsetsre }{ \frac{|\indexsetre {m,m^\prime}|}{|\indexsetre m|} + \frac{1}{|\indexsetre m|} \sum_{v \in \indexsetre {m} \setminus \indexsetre {m,m^\prime}} W^n \bigl( Y^n \in \idset {m^\prime} \bigl| \poolel v \bigr) \geq \omega_n} \\
\quad \stackrel{(e)}\leq \max_{ \indexsetsre\in \setG_\mu} \sum_{m \in \setM} \sum_{m^\prime \neq m} \Biggdistsubof { \indexsetsre }{\frac{|\indexsetre {m,m^\prime}|}{|\indexsetre m|} \geq e^{-n \mu / 2 + 2 \log 2}} \nonumber \\
\qquad + \max_{ \indexsetsre \in \setG_\mu} \sum_{m \in \setM} \sum_{m^\prime \neq m} \Biggdistsubof{ \indexsetsre }{\frac{1}{|\indexsetre m|} \sum_{v \in \indexsetre m \setminus \indexsetre {m,m^\prime}} W^n \bigl( Y^n \in \idset {m^\prime} \bigl| \poolel v \bigr) \geq \kappa_n} \\
\quad \stackrel{(f)}\leq | \setM |^2 \exp \Bigl\{ - e^{n ( \tilde R - \mu ) - 2 \log 2 } \Bigr\} \\
\quad \stackrel{(g)}\rightarrow 0 \, ( n \rightarrow \infty ), \label{eq:maxPrWrongIDDMCConvZero}
\end{IEEEeqnarray}
where $(a)$ holds by \eqref{eq:omegaExpSmallInBlocklength}, because $n$ exceeds $\eta_0$; $(b)$ follows from \eqref{eq:randomPrWrongIDDMC}; $(c)$ follows from the Union-of-Events bound; $(d)$ follows from \eqref{eq:probWrongIDRatioAndProbInOtherWTSet}; $(e)$ follows from \eqref{eq:omegaPrWrongIDDMC} and the Union-of-Events bound; $(f)$ holds by \eqref{eq:probWrongIDRatio} and \eqref{eq:probInOtherWTSet2}; and $(g)$ holds because $\card {\setM} = \exp (\exp (n R))$ and $\mu < \tilde R - R$. 
\end{proof}

\section{Identification via the BC}\label{sec:IDBC}

In this section we establish the ID capacity region of the two-receiver BC $( \setX, \channel {y,z}{x}, \setY \times \setZ )$ under the average-error criterion, which requires that each receiver identify the message intended for it reliably in expectation over the uniform ID message intended for the other receiver. 
We begin with the basic definitions of an average-error ID code for the BC $\channel {y,z}{x}$:

\begin{defin}\label{def:IDCodeBC}
Fix finite sets $\setM_\ry$ and $\setM_\rz$, a blocklength $n \in \naturals$, and positive constants $\lambda^{\ry}_1, \lambda^{\ry}_2, \lambda^{\rz}_1, \lambda^{\rz}_2$. Associate with every ID message-pair $( m_{\ry}, m_{\rz} ) \in \setM_{\ry} \times \setM_{\rz}$ a PMF $Q_{m_\ry,m_\rz}$ on $\setX^n$, with every $m_\ry \in \setM_\ry$ an ID set $\setD_{m_{\ry}} \subset \setY^n$, and with every $m_\rz \in \setM_\rz$ an ID set $\setD_{m_\rz} \subset \setZ^n$. The collection of tuples $\bigl\{ Q_{m_\ry,m_\rz}, \setD_{m_\ry}, \setD_{m_\rz} \bigr\}_{(m_\ry,m_\rz) \in \setM_\ry \times \setM_\rz}$ is an $\bigl( n, \setM_\ry, \setM_\rz, \lambda^{\ry}_1, \lambda^{\ry}_2, \lambda^{\rz}_1, \lambda^{\rz}_2 \bigr)$ ID code for the BC $\channel {y,z} x$ if the maximum probabilities of missed identification at Terminals~$\ry$ and $\rz$
\begin{subequations}
\begin{IEEEeqnarray}{rCl}
p^\ry_{\textnormal{missed-ID}} & = & \max_{m_\ry \in \setM_\ry} \frac{1}{\card {\setM_\rz}} \sum_{m_\rz \in \setM_\rz} \bigl( Q_{m_\ry,m_\rz} W^n \bigr) \bigl( Y^n \notin \setD_{m_\ry} \bigr), \\
p^\rz_{\textnormal{missed-ID}} & = & \max_{m_\rz \in \setM_\rz} \frac{1}{\card {\setM_\ry}} \sum_{m_\ry \in \setM_\ry} \bigl( Q_{m_\ry,m_\rz} W^n \bigr) \bigl( Z^n \notin \setD_{m_\rz} \bigr)
\end{IEEEeqnarray}
\end{subequations}
satisfy
\begin{subequations}\label{bl:maxPrMissedIDBC}
\begin{IEEEeqnarray}{rCl}
p^\ry_{\textnormal{missed-ID}} & \leq & \lambda^{\ry}_1, \\
p^\rz_{\textnormal{missed-ID}} & \leq & \lambda^{\rz}_1,
\end{IEEEeqnarray}
\end{subequations}
and the maximum probabilities of wrong identification at Terminals~$\ry$ and $\rz$
\begin{subequations}
\begin{IEEEeqnarray}{rCl}
p^\ry_{\textnormal{wrong-ID}} & = & \max_{m_\ry \in \setM_\ry} \max_{m^\prime_\ry \neq m_\ry } \frac{1}{\card {\setM_\rz}} \sum_{m_\rz \in \setM_\rz} \bigl( Q_{m_\ry,m_\rz} W^n \bigr) \bigl( Y^n \in \setD_{m^\prime_\ry} \bigr), \\
p^\rz_{\textnormal{wrong-ID}} & = & \max_{m_\rz \in \setM_\rz} \max_{m^\prime_\rz \neq m_\rz } \frac{1}{\card {\setM_\ry}} \sum_{m_\ry \in \setM_\ry} \bigl( Q_{m_\ry,m_\rz} W^n \bigr) \bigl( Z^n \in \setD_{m^\prime_\rz} \bigr)
\end{IEEEeqnarray}
\end{subequations}
satisfy
\begin{subequations}\label{bl:maxPrWrongIDBC}
\begin{IEEEeqnarray}{rCl}
p^\ry_{\textnormal{wrong-ID}} & \leq & \lambda^{\ry}_2, \\
p^\rz_{\textnormal{wrong-ID}} & \leq & \lambda^{\rz}_2.
\end{IEEEeqnarray}
\end{subequations}
A rate-pair $( R_\ry, R_\rz )$ is called achievable if for every positive $\lambda^{\ry}_1$, $\lambda^{\ry}_2$, $ \lambda^{\rz}_1$, and $\lambda^{\rz}_2$ and for every sufficiently-large blocklength~$n$ there exists an $\bigl( n, \setM_\ry, \setM_\rz, \lambda^{\ry}_1, \lambda^{\ry}_2, \lambda^{\rz}_1, \lambda^{\rz}_2 \bigr)$ ID code for the BC with
\begin{IEEEeqnarray*}{rrCll}
& \tfrac{1}{n} \log \log | \setM_\ry | & \geq & R_\ry \quad &\textnormal{if } R_\ry > 0,
\\*[-0.625\normalbaselineskip]
\smash{\left\{
\IEEEstrut[6.39\jot]
\right.} \nonumber
\\*[-0.625\normalbaselineskip]
& |\setM_\ry| & = & 1 \quad &\textnormal{if } R_\ry = 0, \\[0.5\normalbaselineskip]
& \tfrac{1}{n} \log \log | \setM_\rz | & \geq & R_\rz \quad &\textnormal{if } R_\rz > 0,
\\*[-0.625\normalbaselineskip]
\smash{\left\{
\IEEEstrut[6.39\jot]
\right.} \nonumber
\\*[-0.625\normalbaselineskip]
& |\setM_\rz| & = & 1 \quad &\textnormal{if } R_\rz = 0.
\end{IEEEeqnarray*}
The ID capacity region $\setC$ of the BC is the closure of the set of all achievable rate-pairs.
\end{defin}

Equivalently, we can define an ID code for the BC $\channel {y,z}{x}$ as follows:

\begin{remark}\label{re:eqFormDefIDCodeBC}
Given a collection of PMFs $\bigl\{ Q_{m_\ry,m_\rz} \bigr\}_{(m_\ry,m_\rz) \in \setM_\ry \times \setM_\rz}$ on $\setX^n$, define the mixture PMFs on $\setX^n$
\begin{subequations} \label{bl:PMFsMarginalCodesBC}
\begin{IEEEeqnarray}{rCl}
Q_{m_\ry} & = & \frac{1}{\card {\setM_\rz}} \sum_{m_\rz \in \setM_\rz} Q_{m_\ry,m_\rz}, \quad m_\ry \in \setM_\ry, \\
Q_{m_\rz} & = & \frac{1}{\card {\setM_\ry}} \sum_{m_\ry \in \setM_\ry} Q_{m_\ry,m_\rz}, \quad m_\rz \in \setM_\rz.
\end{IEEEeqnarray}
\end{subequations}
The collection of tuples $\bigl\{ Q_{m_\ry,m_\rz}, \setD_{m_\ry}, \setD_{m_\rz} \bigr\}_{(m_\ry,m_\rz) \in \setM_\ry \times \setM_\rz}$ is an $\bigl( n, \setM_\ry, \setM_\rz, \lambda^{\ry}_1, \lambda^{\ry}_2, \lambda^{\rz}_1, \lambda^{\rz}_2 \bigr)$ ID code for the BC $\channel {y,z} x$ if, and only if, (iff) the following two requirements are met: 1) $\bigl\{ Q_{m_\ry}, \setD_{m_\ry} \bigr\}_{m_\ry \in \setM_\ry}$ is an $\bigl( n, \setM_\ry, \lambda^{\ry}_1, \lambda^{\ry}_2 \bigr)$ ID code for the marginal channel $\channely y x$; and 2) $\bigl\{Q_{m_\rz}, \setD_{m_\rz} \bigr\}_{m_\rz \in \setM_\rz}$ is an $\bigl( n, \setM_\rz, \lambda^{\rz}_1, \lambda^{\rz}_2 \bigr)$ ID code for $\channelz z x$.
\end{remark}

Our main result is a single-letter characterization of the ID capacity region of the BC:

\begin{theorem}\label{th:IDBC}
The ID capacity region $\setC$ of the BC $\channel {y,z}{x}$ is the set of all rate-pairs $( R_\ry, R_\rz ) \in ( \reals^+_0 )^2$ that for some PMF $P$ on $\setX$ satisfy
\begin{subequations}\label{bl:capacityBC}
\begin{IEEEeqnarray}{rCl}
R_\ry & \leq & \muti{P}{W_\ry}, \\
R_\rz & \leq & \muti{P}{W_\rz}.
\end{IEEEeqnarray}
\end{subequations}
\end{theorem}

We prove the direct part in Section~\ref{sec:DPIDBC} and the converse part in Section~\ref{sec:CVIDBC}. In fact, we shall establish the following stronger results:

\begin{remark}\label{re:convExpFastBC}
The ID capacity region $\setC$ of the BC $\channel {y,z} x$ is achievable even if we require that the maximum probabilities of missed and wrong identification decay exponentially in the blocklength~$n$. And for all sufficiently-large~$n$, rate-pairs outside this region can be achieved only if $\lambda^{\ry}_1 + \lambda^{\ry}_2 + \lambda^{\rz}_1 + \lambda^{\rz}_2 \geq 1$.
\end{remark}

\begin{proof}
This follows from Claims~\ref{cl:toShowIDBC} and \ref{cl:toShowIDBCConv} ahead.
\end{proof}

In contrast to transmission via the BC, Theorem~\ref{th:IDBC} implies that for identification via the BC there is no trade-off between Receiver~$\setY$ and Receiver~$\setZ$'s rate. An intuitive explanation for this is that in transmission via the BC the message to the other receiver hurts because it is like noise, whereas here this effect is offset by the benefits afforded by randomization.\\

Recall that to achieve the ID capacity of a DMC requires stochastic encoders; deterministic encoders cannot achieve any positive ID rate \cite{ahlswededueck89}. On the BC this is not true:

\begin{remark}\label{re:detIDCodeOptBC}
Every rate-pair in the interior of the ID capacity region $\setC$ of the BC $\channel {y,z} x$ can be achieved using ID codes with deterministic encoders.
\end{remark}

\begin{proof}
The encoder we construct in Section~\ref{sec:DPIDBC} ahead to prove the direct part of Theorem~\ref{th:IDBC} is deterministic: it maps every ID message-pair to a channel-input sequence that is fully determined by the random code construction.
\end{proof}


As a corollary to Theorem~\ref{th:IDBC}, we next observe that the ID capacity region of the BC is convex. This requires proof, because the ID rate is the iterated logarithm of the number of ID messages normalized by the blocklength~$n$, and we therefore cannot invoke a time-sharing argument \cite[Remark~2]{verbovenmeulen90}.

\begin{corollary}
The ID capacity region of the BC $\channel {y,z} x$ is convex.
\end{corollary}

\begin{proof}
It suffices to show that the rate region in Theorem~\ref{th:IDBC} is convex. Given two PMFs $P^{(0)}_X$ and $P^{(1)}_X$ on $\setX$ and some $\alpha \in [ 0,1 ]$, let $P_U$ be the Bernoulli distribution with parameter $\alpha$; let the transition law $P_{X | U}$ be $P^{(U)}_X$; and draw $(U,X) \sim P_U \times P_{X|U}$. Denote the resulting law of $X$ by $P_X$. Then,
\begin{IEEEeqnarray}{l}
\alpha \bigmuti {P^{(1)}_X}{W_\ry} + ( 1 - \alpha ) \bigmuti {P^{(0)}_X}{W_\ry} \nonumber \\
\quad = \condmuti {P_{X|U}}{W_\ry}{P_U} \\
\quad \leq \muti {P_{U,X}}{W_\ry} \\
\quad \stackrel{(a)}= \muti {P_X}{W_\ry}, \label{eq:capRegBCConvex1}
\end{IEEEeqnarray}
where $(a)$ holds since $U$, $X$, and $Y$ form a Markov chain in that order. Likewise,
\begin{IEEEeqnarray}{l}
\alpha \bigmuti {P^{(1)}_X}{W_\rz} + ( 1 - \alpha ) \bigmuti {P^{(0)}_X}{W_\rz} \leq \muti {P_X}{W_\rz}. \label{eq:capRegBCConvex2}
\end{IEEEeqnarray}
Inequalities \eqref{eq:capRegBCConvex1} and \eqref{eq:capRegBCConvex2} combine to prove that the rate region in Theorem~\ref{th:IDBC} is convex.
\end{proof}

We next prove Theorem~\ref{th:IDBC}: Section~\ref{sec:DPIDBC} establishes the direct part and Section~\ref{sec:CVIDBC} a strong converse.

\subsection{The Direct Part of Theorem~\ref{th:IDBC}}\label{sec:DPIDBC}

In this section we prove the direct part of Theorem~\ref{th:IDBC} by fixing any input distribution $P \in \mathscr P (\setX)$ and any positive ID rate-pair $(R_\ry, R_\rz)$ satisfying
\begin{subequations}\label{bl:ratePairAch}
\begin{IEEEeqnarray}{rCcCl}
0 & < & R_\ry & < & \muti {P}{W_\ry}, \label{eq:ratePairAchRy} \\
0 & < & R_\rz & < & \muti {P}{W_\rz} \label{eq:ratePairAchRz}
\end{IEEEeqnarray}
\end{subequations}
and showing that the rate-pair $( R_\ry, R_\rz )$ is achievable. We assume that both $\muti {P}{W_\ry}$ and $\muti {P}{W_\rz}$ are positive; when they are not, the result follows from Theorem~\ref{th:IDDMC}. Let $\setM_\ry$ be a size-$\exp (\exp (n R_\ry))$ set of possible ID messages for Terminal~$\ry$, and let $\setM_\rz$ be a size-$\exp (\exp (n R_\rz))$ set of possible ID messages for Terminal~$\rz$. We next describe our random code construction and show that, for every positive $\lambda^{\ry}_1$, $\lambda^{\ry}_2$, $\lambda^{\rz}_1$, and $\lambda^{\rz}_2$ and for every sufficiently-large blocklength~$n$, it produces with high probability an $\bigl( n, \setM_\ry, \setM_\rz, \lambda^{\ry}_1, \lambda^{\ry}_2, \lambda^{\rz}_1, \lambda^{\rz}_2 \bigr)$ ID code for the BC $\channel {y,z} x$. The scheme that we propose builds on our code construction for the single-user channel in Section~\ref{sec:IDCodingTechnique} by making it appear to each receiver as though we were using an instance of the single-user ID code on its marginal channel. 

\subparagraph*{Code Generation:} Fix an expected bin rate $\tilde R_\ry$ for Terminal~$\ry$, an expected bin rate $\tilde R_\rz$ for Terminal~$\rz$, and a pool rate $R_\poolre$ satisfying
\begin{subequations}\label{bl:binAndPoolRateConstraintsBC}
\begin{IEEEeqnarray}{rCcCl}
R_\ry & < & \tilde R_\ry & < & \muti {P}{W_\ry}, \\
R_\rz & < & \tilde R_\rz & < & \muti {P}{W_\rz}, \\
&& \tilde R_\ry & < & R_\poolre, \\
&& \tilde R_\rz & < & R_\poolre, \\
R_\poolre & < & \tilde R_\ry + \tilde R_\rz. && \label{eq:blBinAndPoolRateConstraintsBCUBPoolRate}
\end{IEEEeqnarray}
\end{subequations}
This is possible by \eqref{bl:ratePairAch}. Draw $e^{n R_\poolre}$ $n$-tuples $\sim P^n$ independently and place them in a pool $\pool$. Index the $n$-tuples in the pool by the elements of a size-$e^{n R_\poolre}$ set $\setV$, e.g., $\{ 1, \ldots, e^{n R_\poolre} \}$, and denote by $\poolel v$ the $n$-tuple in $\pool$ that is indexed by~$v \in \setV$. For each receiving terminal $\Psi \in \{ \ry, \rz \}$ associate with each ID message $m_\Psi \in \setM_\Psi$ an index-set $\indexset {m_\Psi}$ and a bin $\bin {m_\Psi}$ as follows. Select each element of $\setV$ for inclusion in $\indexset {m_\Psi}$ independently with probability $e^{-n( R_\poolre - \tilde R_\Psi )}$, and let Bin~$\bin {m_\Psi}$ be the multiset that contains all the $n$-tuples in the pool that are indexed by $\indexset {m_\Psi}$, $$\bin {m_\Psi} = \bigl\{ \poolel v, \, v \in \indexset {m_\Psi} \bigr\}.$$ (Bin~$\bin {m_\Psi}$ is thus of expected size $e^{n \tilde R_\Psi}$.) Associate with each ID message-pair $(m_\ry,m_\rz) \in \setM_\ry \times \setM_\rz$ an index $V_{m_\ry,m_\rz}$ as follows. If $\indexset {m_\ry} \cap \indexset {m_\rz}$ is not empty, then draw $V_{m_\ry,m_\rz}$ uniformly over $\indexset {m_\ry} \cap \indexset {m_\rz}$. Otherwise draw $V_{m_\ry,m_\rz}$ uniformly over $\setV$. Reveal the pool $\pool$, the index-sets $\bigl\{ \indexset {m_\ry} \bigr\}_{m_\ry \in \setM_\ry}$ and $\bigl\{ \indexset {m_\rz} \bigr\}_{m_\rz \in \setM_\rz}$, the corresponding bins $\bigl\{ \bin {m_\ry} \bigr\}_{m_\ry \in \setM_\ry}$ and $\bigl\{ \bin {m_\rz} \bigr\}_{m_\rz \in \setM_\rz}$, and the indices $\bigl\{ V_{m_\ry,m_\rz} \bigr\}_{(m_\ry,m_\rz) \in \setM_\ry \times \setM_\rz}$ to all parties. The encoding and decoding are determined by
\begin{IEEEeqnarray}{l}
\rcode = \Bigl( \pool, \bigl\{ \indexset {m_\ry} \bigr\}_{m_\ry \in \setM_\ry}, \bigl\{ \indexset {m_\rz} \bigr\}_{m_\rz \in \setM_\rz}, \bigl\{ V_{m_\ry, m_\rz} \bigr\}_{ (m_\ry, m_\rz ) \in \setM_\ry \times \setM_\rz } \Bigr). \label{eq:randCodeBC}
\end{IEEEeqnarray}

\subparagraph*{Encoding:} To send ID Message-Pair~$(m_\ry, m_\rz) \in \setM_\ry \times \setM_\rz$, the encoder transmits the sequence $\poolel {V_{m_\ry,m_\rz}}$. ID Message-Pair~$(m_\ry,m_\rz)$ is thus associated with the $\{ 0,1 \}$-valued PMF
\begin{IEEEeqnarray}{rCl}
\bm Q_{m_\ry,m_\rz} (\vecx) & = & \ind {\vecx = \poolel {V_{m_\ry,m_\rz}}}, \quad \vecx \in \setX^n. \label{eq:defPMFRndCodeBC}
\end{IEEEeqnarray}
Note that once the code \eqref{eq:randCodeBC} has been constructed, the encoder is deterministic: it maps ID Message-Pair~$(m_\ry,m_\rz)$ to the $(m_\ry,m_\rz)$-codeword $\poolel {V_{m_\ry,m_\rz}}$. 

\subparagraph*{Decoding:} In this section the function $\delta (\cdot)$ maps every nonnegative real number $u$ to $u  \ent {P \times W}$. The decoders choose $\epsilon > 0$ sufficiently small so that $2 \delta ( \epsilon ) < \muti {P}{W_\ry} - \tilde R_\ry$ and $2 \delta ( \epsilon ) < \muti {P}{W_\rz} - \tilde R_\rz$. The $m^\prime_\ry$-focused party at Terminal~$\ry$ guesses that $m^\prime_\ry$ was sent iff for some index $v \in \indexset {m_\ry^\prime}$ the $n$-tuple $\poolel v$ in Bin~$\bin {m_\ry^\prime}$ is jointly $\epsilon$-typical with the Terminal-$\ry$ output-sequence $Y^n$, i.e., iff $( \poolel v, Y^n ) \in \eptyp (P \times W_\ry)$ for some $v \in \indexset {m_\ry^\prime}$. The set $\idset {m^\prime_\ry}$ of Terminal-$\ry$ output-sequences $\vecy \in \setY^n$ that result in the guess ``$m^\prime_\ry$ was sent'' is thus
\begin{IEEEeqnarray}{rCl}
\idset {m^\prime_\ry} & = & \bigcup_{v \in \indexset {m^\prime_\ry}} \setT^{ ( n  )}_\epsilon \bigl( P \times W_\ry \bigl| \poolel v \bigr). \label{eq:DefIDSetMy}
\end{IEEEeqnarray}
Likewise, the $m^\prime_\rz$-focused party at Terminal~$\rz$ guesses that $m^\prime_\rz$ was sent iff $( \poolel v, Z^n ) \in \eptyp (P \times W_\rz)$ for some $v \in \indexset {m_\rz^\prime}$. The set $\idset {m^\prime_\rz}$ of Terminal-$\rz$ output-sequences $\vecz \in \setZ^n$ that result in the guess ``$m^\prime_\rz$ was sent'' is thus
\begin{IEEEeqnarray}{rCl}
\idset {m^\prime_\rz} & = & \bigcup_{v \in \indexset {m^\prime_\rz}} \setT^{( n )}_\epsilon \bigl( P \times W_\rz \bigl| \poolel v \bigr). \label{eq:DefIDSetMz}
\end{IEEEeqnarray}

\subparagraph*{Analysis of the Probabilities of Missed and Wrong Identification:} We first note that $\rcode$ of \eqref{eq:randCodeBC} (together with the fixed blocklength $n$ and the chosen $\epsilon$) fully specifies the encoding and guessing rules. That is, the randomly constructed ID code
\begin{IEEEeqnarray}{l}
\bigl\{ \bm Q_{m_\ry,m_\rz}, \idset {m_\ry}, \idset {m_\rz} \bigr\}_{(m_\ry,m_\rz) \in \setM_\ry \times \setM_\rz} \label{eq:randCodeBC2}
\end{IEEEeqnarray}
is fully specified by $\rcode$. Let $\dist$ be the distribution of  $\rcode$, and let $\Exop$ denote expectation w.r.t.\ $\dist$. Subscripts indicate conditioning on the event that some of the chance variables assume the values indicated by the subscripts, e.g., $\dist_{\indexsetre {m_\ry}}$ denotes the distribution conditional on $\indexset {m_\ry} = \indexsetre {m_\ry}$, and $\Exop_{\indexsetre {m_\ry}}$ denotes the expectation w.r.t.\ $\dist_{\indexsetre {m_\ry}}$.\\

The maximum probabilities of missed and wrong identification of the randomly constructed ID code are the random variables
\begin{subequations}
\begin{IEEEeqnarray}{rCl}
P^\ry_{\textnormal{missed-ID}} & = & \max_{m_\ry \in \setM_\ry} \frac{1}{\card {\setM_\rz}} \sum_{m_\rz \in \setM_\rz} \bigl( \bm Q_{m_\ry,m_\rz} W^n \bigr) \bigl( Y^n \notin \idset {m_\ry} \bigr), \\
P^\rz_{\textnormal{missed-ID}} & = & \max_{m_\rz \in \setM_\rz} \frac{1}{\card {\setM_\ry}} \sum_{m_\ry \in \setM_\ry} \bigl( \bm Q_{m_\ry,m_\rz} W^n \bigr) \bigl( Z^n \notin \idset {m_\rz} \bigr), \\
P^\ry_{\textnormal{wrong-ID}} & = & \max_{m_\ry \in \setM_\ry} \max_{m^\prime_\ry \neq m_\ry } \frac{1}{\card {\setM_\rz}} \sum_{m_\rz \in \setM_\rz} \bigl( \bm Q_{m_\ry,m_\rz} W^n \bigr) \bigl( Y^n \in \idset {m^\prime_\ry} \bigr), \\
P^\rz_{\textnormal{wrong-ID}} & = & \max_{m_\rz \in \setM_\rz} \max_{m^\prime_\rz \neq m_\rz } \frac{1}{\card {\setM_\ry}} \sum_{m_\ry \in \setM_\ry} \bigl( \bm Q_{m_\ry,m_\rz} W^n \bigr) \bigl( Z^n \in \idset {m^\prime_\rz} \bigr).
\end{IEEEeqnarray}
\end{subequations}
They are fully specified by $\rcode$, because they are fully specified by the randomly constructed ID code \eqref{eq:randCodeBC2}, which is in turn fully specified by $\rcode$. To prove that for every choice of $\lambda^{\ry}_1, \, \lambda^{\ry}_2, \, \lambda^{\rz}_1, \, \lambda^{\rz}_2 > 0$ and $n$ sufficiently large the collection of tuples \eqref{eq:randCodeBC2} is with high probability an $\bigl( n, \setM_\ry, \setM_\rz, \lambda^{\ry}_1, \lambda^{\ry}_2, \lambda^{\rz}_1, \lambda^{\rz}_2 \bigr)$ ID code for the BC $\channel {y,z} x$, we prove the following stronger result:

\begin{claim}\label{cl:toShowIDBC}
The probabilities $P^\ry_{\textnormal{missed-ID}}$, $P^\rz_{\textnormal{missed-ID}}$, $P^\ry_{\textnormal{wrong-ID}}$, and $P^\rz_{\textnormal{wrong-ID}}$ of the randomly constructed ID code \eqref{eq:randCodeBC2} converge in probability to zero exponentially in the blocklength~$n$, i.e.,
\begin{IEEEeqnarray}{l}
\exists \, \tau > 0 \textnormal{ s.t.\ } \lim_{n \rightarrow \infty} \Bigdistof { \max \bigl\{ P^\ry_{\textnormal{missed-ID}}, P^\rz_{\textnormal{missed-ID}}, P^\ry_{\textnormal{wrong-ID}}, P^\rz_{\textnormal{wrong-ID}} \bigr\} \geq e^{-n \tau} } = 0. \label{eq:toShowIDBC}
\end{IEEEeqnarray}
\end{claim}

\begin{proof}
We will prove that
\begin{IEEEeqnarray}{l}
\exists \, \tau > 0 \textnormal{ s.t.\ } \lim_{n \rightarrow \infty} \Bigdistof { \max \bigl\{ P^\ry_{\textnormal{missed-ID}}, P^\ry_{\textnormal{wrong-ID}} \bigr\} \geq e^{-n \tau} } = 0. \label{eq:toShowIDBC2}
\end{IEEEeqnarray}
By swapping $\rz$ and $\ry$ throughout the proof it will then follow that \eqref{eq:toShowIDBC2} also holds when we replace $\ry$ with $\rz$, and \eqref{eq:toShowIDBC} will then follow using the Union-of-Events bound.

To prove \eqref{eq:toShowIDBC2} we consider for each $m_\ry \in \setM_\ry$ two distributions on the set $\setV$, which indexes the pool $\pool$. We fix some $v^\star \in \setV$ and define for every $m_\ry \in \setM_\ry$ the PMFs on $\setV$
\begin{subequations}\label{bl:remIndMzAndUnifBin}
\begin{IEEEeqnarray}{rCl}
\bm P_V^{(m_\ry)} (v) & = & \frac{1}{|\setM_\rz|} \sum_{m_\rz \in \setM_\rz} \ind {v = V_{m_\ry,m_\rz}}, \quad v \in \setV, \\ \label{eq:remIndMzBin}
\tilde {\bm P}_V^{(m_\ry)} (v) & = & \begin{cases} \frac{1}{|\indexset {m_\ry}|} \sum_{v^\prime \in \indexset {m_\ry}} \ind {v = v^\prime} &\textnormal{if } \indexset {m_\ry} \neq \emptyset, \\ \ind {v = v^\star} &\textnormal{otherwise}, \end{cases} \quad v \in \setV. \label{eq:remUnifBin}
\end{IEEEeqnarray}
\end{subequations}
The latter PMF is reminiscent of the distribution we encountered in \eqref{eq:distIDMsgDMC} and \eqref{eq:distIDMsgDMC2} in the single-user case. The former is related to the BC setting when we view $M_\setZ$ as uniform over $\setM_\setZ$. As we argue next, to establish \eqref{eq:toShowIDBC2} it suffices to show that the two are similar in the sense that
\begin{IEEEeqnarray}{l}
\exists \, \tau > 0 \textnormal{ s.t.\ } \lim_{n \rightarrow \infty} \biggdistof { \max_{m_\ry \in \setM_\ry} d \Bigl( \bm P_V^{(m_\ry)}, \tilde {\bm P}_V^{(m_\ry)} \Bigr) \geq e^{-n \tau} } = 0. \label{eq:toShowIDBC4}
\end{IEEEeqnarray}

To see why, let us define for every $m_\ry \in \setM_\ry$ the PMFs on $\setX^n$
\begin{subequations}\label{bl:encDistsMzUnifDMC}
\begin{IEEEeqnarray}{rCl}
\bm Q_{m_\ry} (\vecx) & = & \frac{1}{\card {\setM_\rz}} \sum_{m_\rz \in \setM_\rz} \bm Q_{m_\ry,m_\rz} (\vecx), \quad \vecx \in \setX^n, \\
\tilde {\bm Q}_{m_\ry} (\vecx) & = & \begin{cases} \frac{1}{|\indexset {m_\ry}|} \sum_{v^\prime \in \indexset {m_\ry}} \ind {\vecx = \poolel {v^\prime}} &\textnormal{if } \indexset {m_\ry} \neq \emptyset, \\ \ind {\vecx = \poolel {v^\star}} &\textnormal{otherwise}, \end{cases} \quad \vecx \in \setX^n.
\end{IEEEeqnarray}
\end{subequations}
The collection of tuples $\bigl\{ \bm Q_{m_\ry}, \idset {m_\ry} \bigr\}_{m_\ry \in \setM_\ry}$ can be viewed as a randomly constructed ID code for the DMC $W_\ry (y|x)$ with maximum probability of missed identification
\begin{IEEEeqnarray}{l}
\max_{m_\ry \in \setM_\ry} \bigl( \bm Q_{m_\ry} W^n \bigr) \bigl( Y^n \notin \idset {m_\ry} \bigr) \nonumber \\
\quad = \max_{m_\ry \in \setM_\ry} \frac{1}{\card {\setM_\rz}} \sum_{m_\rz \in \setM_\rz} \bigl( \bm Q_{m_\ry,m_\rz} W^n \bigr) \bigl( Y^n \notin \idset {m_\ry} \bigr) \\
\quad = P^\ry_{\textnormal{missed-ID}}
\end{IEEEeqnarray}
and maximum probability of wrong identification
\begin{IEEEeqnarray}{l}
\max_{m_\ry \in \setM_\ry} \max_{m^\prime_\ry \neq m_\ry } \bigl( \bm Q_{m_\ry} W^n \bigr) \bigl( Y^n \in \idset {m^\prime_\ry} \bigr) \nonumber \\
\quad = \max_{m_\ry \in \setM_\ry} \max_{m^\prime_\ry \neq m_\ry } \frac{1}{\card {\setM_\rz}} \sum_{m_\rz \in \setM_\rz} \bigl( \bm Q_{m_\ry,m_\rz} W^n \bigr) \bigl( Y^n \in \idset {m^\prime_\ry} \bigr) \\
\quad = P^\ry_{\textnormal{wrong-ID}}.
\end{IEEEeqnarray}
And $\bigl\{ \tilde {\bm Q}_{m_\ry}, \idset {m_\ry} \bigr\}_{m_\ry \in \setM_\ry}$ has the same law as the randomly constructed ID code $\{ \bm Q_m, \idset m \}_{m \in \setM}$ of Section~\ref{sec:IDCodingTechnique} for the DMC $W = W_\ry$ with blocklength~$n$, fixed element $v^\star$ of $\setV$, decoding parameter $\epsilon$, size-$\exp (\exp (n R_\ry))$ set $\setM_\ry$ of possible ID messages, expected bin rate $\tilde R_\ry$, and pool rate $R_\poolre$. (Note that $\epsilon$, $R_\ry$, $\tilde R_\ry$, and $R_\poolre$ are eligible for the random code construction in Section~\ref{sec:IDCodingTechnique}, because $\epsilon$ is positive and sufficiently small so that $2 \epsilon \ent {P \times W_\ry} < \muti {P}{W_\ry} - \tilde R_\ry$, and because of \eqref{bl:ratePairAch} and \eqref{bl:binAndPoolRateConstraintsBC}.) Let $\tilde P^\ry_{\textnormal{missed-ID}}$ and $\tilde P^\ry_{\textnormal{wrong-ID}}$ denote the maximum probabilities of missed and wrong identification of the randomly constructed ID code $\bigl\{ \tilde {\bm Q}_{m_\ry}, \idset {m_\ry} \bigr\}_{m_\ry \in \setM_\ry}$, i.e.,
\begin{subequations}
\begin{IEEEeqnarray}{rCl}
\tilde P^\ry_{\textnormal{missed-ID}} & = & \max_{m_\ry \in \setM_\ry} \bigl( \tilde {\bm Q}_{m_\ry} W^n \bigr) \bigl( Y^n \notin \idset {m_\ry} \bigr), \\
\tilde P^\ry_{\textnormal{wrong-ID}} & = & \max_{m_\ry \in \setM_\ry} \max_{m^\prime_\ry \neq m_\ry } \bigl( \tilde {\bm Q}_{m_\ry} W^n \bigr) \bigl( Y^n \in \idset {m^\prime_\ry} \bigr).
\end{IEEEeqnarray}
\end{subequations}
By Claim~\ref{re:convExpFast} on the single-user channel
\begin{IEEEeqnarray}{l}
\exists \, \tau > 0 \textnormal{ s.t.\ } \lim_{n \rightarrow \infty} \Bigdistof { \max \bigl\{ \tilde P^\ry_{\textnormal{missed-ID}}, \tilde P^\ry_{\textnormal{wrong-ID}} \bigr\} \geq e^{-n \tau} } = 0. \label{eq:knowFromDMCIDBC}
\end{IEEEeqnarray}
And by definition of the Total-Variation distance
\begin{subequations}\label{bl:ubPrDecErrorIDBCDMCAndTVDist}
\begin{IEEEeqnarray}{rCl}
P^\ry_{\textnormal{missed-ID}} & \leq & \tilde P^\ry_{\textnormal{missed-ID}} + \max_{m_\ry \in \setM_\ry} d \bigl( \bm Q_{m_\ry} W_\ry^n, \tilde {\bm Q}_{m_\ry} W_\ry^n \bigr), \\
P^\ry_{\textnormal{wrong-ID}} & \leq & \tilde P^\ry_{\textnormal{wrong-ID}} + \max_{m_\ry \in \setM_\ry} d \bigl( \bm Q_{m_\ry} W_\ry^n, \tilde {\bm Q}_{m_\ry} W_\ry^n \bigr).
\end{IEEEeqnarray}
\end{subequations}
For every $\tau_1$, $\tau_2$, and $\tau < \min \{ \tau_1, \tau_2 \}$ we have for all sufficiently-large $n$,
\begin{IEEEeqnarray}{rCl}
e^{-n \tau_1} + e^{-n \tau_2} & \leq & e^{- n \tau}.
\end{IEEEeqnarray}
This, combined with the Union-of-Events bound, \eqref{eq:knowFromDMCIDBC}, and \eqref{bl:ubPrDecErrorIDBCDMCAndTVDist}, implies that to establish \eqref{eq:toShowIDBC2} it suffices to show that
\begin{IEEEeqnarray}{l}
\exists \, \tau > 0 \textnormal{ s.t.\ } \lim_{n \rightarrow \infty} \biggdistof { \max_{m_\ry \in \setM_\ry} d \bigl( \bm Q_{m_\ry} W_\ry^n, \tilde {\bm Q}_{m_\ry} W_\ry^n \bigr) \geq e^{-n \tau} } = 0. \label{eq:toShowIDBC3}
\end{IEEEeqnarray}
Consequently, to prove our claim that \eqref{eq:toShowIDBC4} implies \eqref{eq:toShowIDBC2}, we only have to show that \eqref{eq:toShowIDBC4} implies \eqref{eq:toShowIDBC3}. To that end, define the conditional PMF
\begin{IEEEeqnarray}{rCl}
\bm P_{X^n|V} (\vecx|v) & = & \ind {\vecx = \poolel v}, \quad ( \vecx, v ) \in \setX^n \times \setV, \label{eq:BCCondPMFXnGivV}
\end{IEEEeqnarray}
and note that for every $m_\ry \in \setM_\ry$
\begin{subequations}
\begin{IEEEeqnarray}{rCl}
\bigl( \bm Q_{m_\ry} W_\ry^n \bigr) (\vecy) & = & \Bigl( \bm P_V^{(m_\ry)} \bm P_{X^n|V} W_\ry^n \Bigr) (\vecy), \quad \vecy \in \setY^n, \\
\bigl( \tilde { \bm Q }_{m_\ry} W_\ry^n \bigr) (\vecy) & = & \Bigl( \tilde {\bm P}_V^{(m_\ry)} \bm P_{X^n|V} W_\ry^n \Bigr) (\vecy), \quad \vecy \in \setY^n,
\end{IEEEeqnarray}
\end{subequations}
where we used \eqref{bl:remIndMzAndUnifBin}, \eqref{bl:encDistsMzUnifDMC}, and \eqref{eq:BCCondPMFXnGivV}, and in the first equality also \eqref{eq:defPMFRndCodeBC}. We can now upper-bound $d \bigl( \bm Q_{m_\ry} W_\ry^n, \tilde {\bm Q}_{m_\ry} W_\ry^n \bigr)$ by
\begin{IEEEeqnarray}{l}
d \bigl( \bm Q_{m_\ry} W_\ry^n, \tilde {\bm Q}_{m_\ry} W_\ry^n \bigr) \nonumber \\
\quad = d \Bigl( \bm P_V^{(m_\ry)} \bm P_{X^n|V} W_\ry^n, \tilde {\bm P}_V^{(m_\ry)} \bm P_{X^n|V} W_\ry^n \Bigr) \\
\quad \leq d \Bigl( \bm P_V^{(m_\ry)}, \tilde {\bm P}_V^{(m_\ry)} \Bigr), \label{eq:toShowIDBC3DataProcessingTotVarDist}
\end{IEEEeqnarray}
where the last inequality follows from the Data-Processing inequality for the Total-Variation distance \cite[Lemma~1]{cannoneronservedio15}. From \eqref{eq:toShowIDBC3DataProcessingTotVarDist} we conclude that \eqref{eq:toShowIDBC4} implies \eqref{eq:toShowIDBC3} and hence also \eqref{eq:toShowIDBC2}.\\

Having established that \eqref{eq:toShowIDBC4} implies \eqref{eq:toShowIDBC2}, it remains to prove \eqref{eq:toShowIDBC4}. Before we do that, we give an intuitive explanation why \eqref{eq:toShowIDBC4} holds. Fix $m_\ry \in \setM_\ry$ and a realization $\indexsetre {m_\ry}$ of the corresponding index-set $\indexset {m_\ry}$, and assume that $\indexsetre {m_\ry} \approx e^{n \tilde R_\ry}$. For every $m_\rz \in \setM_\rz$, the probabilitiy that the intersection of $\indexsetre {m_\ry}$ and $\indexset {m_\rz}$ is empty is very small, and if the intersection is nonempty, then, by our random construction of $\indexset {m_\rz}$ and $V_{m_\ry,m_\rz}$, the codeword-index $V_{m_\ry,m_\rz}$ is drawn uniformly at random from $\indexsetre {m_\ry}$. Because $\indexsetre {m_\ry}$ is exponential in $n$ and the cardinality of $\setM_\setZ$ is double-exponential in $n$, and because, by our random construction of $\{ \indexset {m_\rz} \}_{m_\rz \in \setM_\rz}$ and $\{ V_{m_\ry,m_\rz} \}_{m_\rz \in \setM_\rz}$, the codeword-indices $\{ V_{m_\ry,m_\rz} \}_{m_\rz \in \setM_\rz}$ are drawn independently of each other, \eqref{eq:toShowIDBC4} can be derived using concentration inequalities.

To prove \eqref{eq:toShowIDBC4} rigorously, fix some $\mu$ satisfying
\begin{IEEEeqnarray}{l}
0 < \mu < \tilde R_\ry - R_\ry, \label{eq:muIDBCPf}
\end{IEEEeqnarray}
and let
\begin{IEEEeqnarray}{l}
\delta_n = e^{-n \mu / 2}. \label{eq:deltanIDBCPf}
\end{IEEEeqnarray}
Introduce the set $\setH^\ry_\mu$ comprising the realizations $\{ \indexsetre \nu \}_{\nu \in \setM_\ry}$ of the index-sets $\{ \indexset {\nu} \}_{\nu \in \setM_\ry}$ satisfying that
 \begin{IEEEeqnarray}{rCl}
|\indexsetre \nu| > (1 - \delta_n) e^{n \tilde R_\ry}, \, \forall \, \nu \in \setM_\ry. \label{eq:IDBCSetH}
\end{IEEEeqnarray}
We upper-bound $\max_{m_\ry \in \setM_\ry} d \bigl( \bm P_V^{(m_\ry)}, \tilde {\bm P}_V^{(m_\ry)} \bigr)$ differently depending on whether or not $\{ \indexset \nu \}$ is in $\setH^\ry_\mu$, where $\{ \indexset \nu \}$ is short for $\{ \indexset {\nu} \}_{\nu \in \setM_\ry}$. If $\{ \indexset \nu \} \notin \setH^\ry_\mu$, then we upper-bound it by one (which is an upper bound on the Total-Variation distance between any two probability measures) to obtain for every $\tau > 0$
\begin{IEEEeqnarray}{l}
\biggdistof { \max_{m_\ry \in \setM_\ry} d \Bigl( \bm P_V^{(m_\ry)}, \tilde {\bm P}_V^{(m_\ry)} \Bigr) \geq e^{-n \tau} } \nonumber \\
\quad \leq \bigdistof { \{ \indexset \nu \} \notin \setH^\ry_\mu } + \sum_{ \indexsetsre \in \setH^\ry_\mu } \bigdistof { \{ \indexset \nu \} =  \indexsetsre } \, \biggdistsubof { \indexsetsre } { \max_{m_\ry \in \setM_\ry} d \Bigl( \bm P_V^{(m_\ry)}, \tilde {\bm P}_V^{(m_\ry)} \Bigr) \geq e^{-n \tau} }. \label{eq:IDBCToShowNotinAndInH}
\end{IEEEeqnarray}

We consider the two terms on the RHS of \eqref{eq:IDBCToShowNotinAndInH} separately, beginning with $\bigdistof { \{ \indexset \nu \} \notin \setH^\ry_\mu }$. Following the proof of Lemma~\ref{le:prEventEToZero} in Section~\ref{sec:IDCodingTechnique}, we will show that $\bigdistof { \{ \indexset \nu \} \notin \setH^\ry_\mu }$ converges to zero as $n$ tends to infinity. This does not follow from Lemma~\ref{le:prEventEToZero}, because here we require $\mu$ to satisfy \eqref{eq:muIDBCPf} instead of the more restrictive condition \eqref{eq:IDCodeMu} of Section~\ref{sec:IDCodingTechnique}. For every fixed $\nu \in \setM_\ry$ the $e^{n R_\poolre}$ binary random variables $\{ \ind {v \in \indexset \nu} \}_{v \in \setV}$ are IID, and
\begin{IEEEeqnarray}{rCl}
\BiggEx {}{\sum_{v \in \setV} \ind {v \in \indexset \nu}} & = & \sum_{v \in \setV} \distof {v \in \indexset \nu} = e^{ n \tilde R_\ry}.
\end{IEEEeqnarray}
Consequently, by the multiplicative Chernoff bound \eqref{eq:multChernDeltaSm1Sm} in Proposition~\ref{pr:multChernoff},
\begin{IEEEeqnarray}{rCl}
\Bigdistof { |\indexset \nu| \leq (1 - \delta_n) \, e^{n \tilde R_\ry} } & = & \Biggdistof { \sum_{v \in \setV} \ind {v \in \indexset \nu} \leq (1 - \delta_n) \, e^{n \tilde R_\ry} } \\
& \leq & \exp \bigl\{ - \delta_n^2 \, e^{n \tilde R_\ry - \log 2} \bigr\} \\
& = & \exp \bigl\{ - e^{n (\tilde R_\ry - \mu) - \log 2} \bigr\}.
\end{IEEEeqnarray}
The Union-of-Events bound thus implies that
\begin{IEEEeqnarray}{rCl}
\bigdistof { \{ \indexset \nu \} \notin \setH^\ry_\mu } & \leq & |\setM_\ry| \exp \bigl\{ - e^{n (\tilde R_\ry - \mu) - \log 2} \bigr\} \\
& \stackrel{(a)}\rightarrow & 0 \, (n \rightarrow \infty), \label{eq:allBmySuffLarge}
\end{IEEEeqnarray}
where $(a)$ holds because $|\setM_\ry| = \exp (\exp (n R_\ry))$ and by \eqref{eq:muIDBCPf}.

Having established \eqref{eq:allBmySuffLarge}, we return to \eqref{eq:IDBCToShowNotinAndInH} and conclude the proof of \eqref{eq:toShowIDBC4} by showing that
\begin{IEEEeqnarray}{l}
\exists \, \tau > 0 \textnormal{ s.t.\ } \lim_{n \rightarrow \infty} \max_{ \indexsetsre \in \setH^\ry_\mu}  \biggdistsubof { \indexsetsre } { \max_{m_\ry \in \setM_\ry} d \Bigl( \bm P_V^{(m_\ry)}, \tilde {\bm P}_V^{(m_\ry)} \Bigr) \geq e^{-n \tau} } = 0. \label{eq:toShowIDBC5}
\end{IEEEeqnarray}
(The proof of \eqref{eq:toShowIDBC5} ahead exploits the fact that the index-sets $\bigl\{ \indexset {m_\rz} \bigr\}_{m_\rz \in \setM_\rz}$ are drawn at random. Likewise, when we prove \eqref{eq:toShowIDBC2} with $\ry$ replaced by $\rz$, we shall need the fact that the index-sets $\bigl\{ \indexset {m_\ry} \bigr\}_{m_\ry \in \setM_\ry}$ are drawn at random. Hence Remark~\ref{re:drawIndexSetsAtRandom}.) To prove \eqref{eq:toShowIDBC5}, let us henceforth assume that $n$ is large enough so that the following two inequalities hold:
\begin{subequations}\label{bl:IDBCnLargeEnough}
\begin{IEEEeqnarray}{rCl}
(1 - \delta_n) e^{n \tilde R_\ry} & \geq & 1, \label{eq:IDBCnLargeEnough} \\
\delta_n & \leq & 1/2, \label{eq:IDBCnLargeEnough2}
\end{IEEEeqnarray}
 \end{subequations}
where $\delta_n$ is defined in \eqref{eq:deltanIDBCPf}. (This is possible, because $\delta_n$ converges to zero as $n$ tends to infinity and $\tilde R_\ry > 0$.) Fix any realization $\indexsetsre$ in $\setH^\ry_\mu$. Rather than directly upper-bounding the maximum over $m_\ry \in \setM_\ry$ of $d \bigl( \bm P_V^{(m_\ry)}, \tilde {\bm P}_V^{(m_\ry)} \bigr)$ under $\dist_{\indexsetsre}$, we first consider $d \bigl( \bm P_V^{(m_\ry)}, \tilde {\bm P}_V^{(m_\ry)} \bigr)$ for a fixed $m_\ry \in \setM_\ry$. By \eqref{eq:IDBCSetH} (which holds because $\indexsetsre \in \setH^\ry_\mu$) and \eqref{eq:IDBCnLargeEnough}, $\indexsetre {m_\ry}$ is nonempty. For every fixed $v \in \setV \setminus \indexsetre {m_\ry}$ we therefore have that under $\dist_{ \indexsetsre }$ the $\exp (\exp (n R_\rz))$ binary random variables $\bigl\{ \ind {v = V_{m_\ry,m_\rz}} \bigr\}_{m_\rz \in \setM_\rz}$ are IID and of mean
\begin{IEEEeqnarray}{l}
\bigEx { \indexsetsre }{\ind {v = V_{m_\ry,m_\rz}}} \nonumber \\
\quad = \bigdistsubof{ \indexsetsre }{V_{m_\ry,m_\rz} = v}\\
\quad \stackrel{(a)}= \frac{1}{|\setV|} \distsubof { \indexsetsre }{\indexsetre {m_\ry} \cap \indexset {m_\rz} = \emptyset} \\
\quad \stackrel{(b)}= \frac{1}{|\setV|} \distsubof {}{\indexsetre {m_\ry} \cap \indexset {m_\rz} = \emptyset} \\
\quad \stackrel{(c)}= \frac{1}{|\setV|} \bigl( 1 - e^{-n (R_\poolre - \tilde R_\rz)} \bigr)^{|\indexsetre {m_\ry}|} \label{eq:meanVNotinBMy} \\
\quad \stackrel{(d)}\leq \exp \bigl\{ - e^{-n (R_\poolre - \tilde R_\rz)} |\indexsetre {m_\ry}| - n R_\poolre \bigr\} \\
\quad \stackrel{(e)}\leq (1 - \delta_n)^{-1} \exp \bigl\{ - (1 - \delta_n) e^{n (\tilde R_\ry + \tilde R_\rz - R_\poolre)} - n \tilde R_\ry \bigr\}, \quad v \in \setV \setminus \setV_{m_\ry} \label{eq:meanVNotinBMyBd}
\end{IEEEeqnarray}
with the following justification. Equality~$(a)$ holds because $v \notin \indexsetre {m_\ry}$ and $\indexset {m_\ry} = \indexsetre {m_\ry}$ $\dist_{ \indexsetsre }$-almost-surely, and therefore: if $\indexsetre {m_\ry} \cap \indexset {m_\rz} \neq \emptyset$, then $V_{m_\ry,m_\rz} \neq v$, and otherwise $V_{m_\ry,m_\rz}$ is uniform over $\setV$. Equality~$(b)$ holds because $\indexset {m_\rz}$ is independent of $\{ \indexset \nu \}_{\nu \in \setM_\ry}$, and its distribution w.r.t.\ $\dist_{ \indexsetsre }$ is thus the same as w.r.t.\ $\dist$; $(c)$ holds because we have selected each element of $\setV$ for inclusion in $\indexset {m_\rz}$ independently with probability $e^{-n (R_\poolre - \tilde R_\rz)}$; $(d)$ holds because $|\setV| = e^{n R_\poolre}$ and because
\begin{equation}
1 - x \leq e^{-x}, \quad x \in \reals; \label{eq:1MinXSmEMinX}
\end{equation}
and $(e)$ holds because $0 \leq \delta_n < 1$, by \eqref{eq:IDBCSetH} (which holds because $\indexsetsre \in \setH^\ry_\mu$), and because $\tilde R_\ry < R_\poolre$. Similarly, for every fixed $v \in \indexsetre {m_\ry}$ we have that  under $\dist_{ \indexsetsre }$ the $\exp (\exp (n R_\rz))$ binary random variables $\bigl\{ \ind {v = V_{m_\ry,m_\rz}} \bigr\}_{m_\rz \in \setM_\rz}$ are IID and of mean
\begin{IEEEeqnarray}{l}
\bigEx { \indexsetsre }{\ind {v = V_{m_\ry,m_\rz}}} \nonumber \\
\quad = \bigdistsubof{\indexsetsre}{V_{m_\ry,m_\rz} = v}\\
\quad \stackrel{(a)}= \frac{1}{|\indexsetre {m_\ry}|} \bigdistsubof {\indexsetsre}{V_{m_\ry,m_\rz} \in \indexsetre {m_\ry}} \\
\quad = \frac{1}{|\indexsetre {m_\ry}|} \Bigl( 1 - \bigdistsubof {\indexsetsre}{V_{m_\ry,m_\rz} \notin \indexsetre {m_\ry}} \Bigr) \\
\quad \stackrel{(b)}= \frac{1}{|\indexsetre {m_\ry}|} \biggl( 1 - \frac{|\setV| - |\indexsetre {m_\ry}|}{|\setV|} \bigl( 1 - e^{-n (R_\poolre - \tilde R_\rz)} \bigr)^{|\indexsetre {m_\ry}|} \biggr) \\
\quad = \frac{1}{|\indexsetre {m_\ry}|} - \biggl( \frac{1}{|\indexsetre {m_\ry}|} - \frac{1}{|\setV|} \biggr) \bigl( 1 - e^{-n (R_\poolre - \tilde R_\rz)} \bigr)^{|\indexsetre {m_\ry}|} \\
\quad \stackrel{(c)}\in \biggl[ \frac{1}{|\indexsetre {m_\ry}|} \Bigl( 1 - \exp \bigl\{ - e^{-n (R_\poolre - \tilde R_\rz)} |\indexsetre {m_\ry}| \bigr\} \Bigr), \frac{1}{|\indexsetre {m_\ry}|} \biggr] \\
\quad \stackrel{(d)}\subseteq \biggl[ \frac{1}{|\indexsetre {m_\ry}|} - (1 - \delta_n)^{-1} \exp \bigl\{ - (1 - \delta_n) e^{n (\tilde R_\ry + \tilde R_\rz - R_\poolre)} - n \tilde R_\ry \bigr\}, \frac{1}{|\indexsetre {m_\ry}|} \biggr], \quad v \in \setV_{m_\ry}, \label{eq:meanVInBMyBd}
\end{IEEEeqnarray}
where $(a)$ holds by symmetry; $(b)$ holds by \eqref{eq:meanVNotinBMy}, because $\indexset {m_\ry} = \indexsetre {m_\ry}$ $\dist_{ \indexsetsre }$-almost-surely, and hence if $\indexsetre {m_\ry} \cap \indexset {m_\rz} = \emptyset$, then $V_{m_\ry,m_\rz}$ is uniform over $\setV$, and because $|\setV \setminus \indexsetre {m_\ry}| = |\setV| - |\indexsetre {m_\ry}|$; $(c)$ holds by \eqref{eq:1MinXSmEMinX}; and $(d)$ holds by \eqref{eq:IDBCSetH} (which holds because $\indexsetsre \in \setH^\ry_\mu$). Fix some $\kappa$ satisfying
\begin{IEEEeqnarray}{rCl}
0 < \kappa < \min \{ R_\rz, \tilde R_\ry + \tilde R_\rz - R_\poolre \}, \label{eq:IDBCKappaDef}
\end{IEEEeqnarray}
and let
\begin{IEEEeqnarray}{rCl}
\xi_n = 4 \exp \bigl\{ -e^{n \kappa - \log 2} \bigr\}. \label{eq:IDBCKappaTripExpSmall}
\end{IEEEeqnarray}
By \eqref{eq:IDBCnLargeEnough2}
\begin{IEEEeqnarray}{rCl}
\xi_n / 2 > (1 - \delta_n)^{-1} \exp \bigl\{ - (1 - \delta_n) e^{n (\tilde R_\ry + \tilde R_\rz - R_\poolre)} - n \tilde R_\ry \bigr\}. \label{eq:xinGeq2Dev}
\end{IEEEeqnarray}
Consequently, Hoeffding's inequality (Proposition~\ref{pr:hoeffding}) implies that for every fixed $v \in \setV \setminus \indexsetre {m_\ry}$
\begin{IEEEeqnarray}{l}
\biggdistsubof {\indexsetsre}{\Bigl| \bm P_V^{(m_\ry)} (v) - \tilde {\bm P}_V^{(m_\ry)} (v) \Bigr| \geq \xi_n} \nonumber \\
\quad \stackrel{(a)}= \Biggdistsubof {\indexsetsre}{\frac{1}{|\setM_\rz|} \sum_{m_\rz \in \setM_\rz} \ind {v = V_{m_\ry,m_\rz}} \geq \xi_n} \\
\quad \stackrel{(b)}\leq \exp \biggl\{ - 2 \, |\setM_\rz| \Bigl( \xi_n - (1 - \delta_n)^{-1} \exp \bigl\{ - (1 - \delta_n) e^{n (\tilde R_\ry + \tilde R_\rz - R_\poolre)} - n \tilde R_\ry \bigr\} \Bigr)^2 \biggr\} \\
\quad \stackrel{(c)}\leq \exp \bigl\{ - |\setM_\rz| \, \xi_n^2 / 2 \bigr\},  \label{eq:contTotVarDistVNotinBMy}
\end{IEEEeqnarray}
where $(a)$ holds because $\indexset {m_\ry} = \indexsetre {m_\ry}$ $\dist_{ \indexsetsre }$-almost-surely, because $\indexsetre {m_\ry}$ is nonempty (which holds because $\indexsetsre \in \setH^\ry_\mu$ implies \eqref{eq:IDBCSetH} and by \eqref{eq:IDBCnLargeEnough}), by \eqref{bl:remIndMzAndUnifBin}, and because $v \notin \indexsetre {m_\ry}$; $(b)$ follows from Hoeffding's inequality (Proposition~\ref{pr:hoeffding}) and \eqref{eq:meanVNotinBMyBd}; and $(c)$ holds by \eqref{eq:xinGeq2Dev}. Similarly, for every fixed $v \in \indexsetre {m_\ry}$
\begin{IEEEeqnarray}{l}
\biggdistsubof {\indexsetsre}{\Bigl| \bm P_V^{(m_\ry)} (v) - \tilde {\bm P}_V^{(m_\ry)} (v) \Bigr| \geq \xi_n} \nonumber \\
\quad \stackrel{(a)}= \Biggdistsubof {\indexsetsre}{\biggl| \frac{1}{|\setM_\rz|} \sum_{m_\rz \in \setM_\rz} \ind {v = V_{m_\ry,m_\rz}} - \frac{1}{|\indexsetre {m_\ry}|} \biggr| \geq \xi_n} \\
\quad \stackrel{(b)}\leq 2 \exp \biggl\{ - 2 \, |\setM_\rz| \Bigl( \xi_n - (1 - \delta_n)^{-1} \exp \bigl\{ - (1 - \delta_n) e^{n (\tilde R_\ry + \tilde R_\rz - R_\poolre)} - n \tilde R_\ry \bigr\} \Bigr)^2 \biggr\} \\
\quad \stackrel{(c)}\leq 2 \exp \bigl\{ - |\setM_\rz| \, \xi_n^2 / 2 \bigr\}, \label{eq:contTotVarDistVInBMy}
\end{IEEEeqnarray}
where $(a)$ holds because $\indexset {m_\ry} = \indexsetre {m_\ry}$ $\dist_{ \indexsetsre }$-almost-surely, because $\indexsetre {m_\ry}$ is nonempty, by \eqref{bl:remIndMzAndUnifBin}, and because $v \in \indexsetre {m_\ry}$; $(b)$ follows from Hoeffding's inequality (Proposition~\ref{pr:hoeffding}), \eqref{eq:meanVInBMyBd}, and the Union-of-Events bound; and $(c)$ holds by \eqref{eq:xinGeq2Dev}. The Union-of-Events bound, \eqref{eq:contTotVarDistVNotinBMy}, and \eqref{eq:contTotVarDistVInBMy} imply that
\begin{IEEEeqnarray}{l}
\biggdistsubof {\indexsetsre}{\exists \, v \in \setV \colon \Bigl| \bm P_V^{(m_\ry)} (v) - \tilde {\bm P}_V^{(m_\ry)} (v) \Bigr| \geq \xi_n} \nonumber \\
\quad \leq 2 \, |\setV| \exp \bigl\{ - |\setM_\rz| \, \xi_n^2 / 2 \bigr\}. \label{eq:prExistsVTooBig}
\end{IEEEeqnarray}
Therefore,
\begin{IEEEeqnarray}{l}
\biggdistsubof {\indexsetsre}{ d \Bigl( \bm P_V^{(m_\ry)}, \tilde {\bm P}_V^{(m_\ry)} \Bigr) \geq |\setV| \, \xi_n / 2 } \nonumber \\
\quad \stackrel{(a)}= \Biggdistsubof { \indexsetsre }{ \sum_{v \in \setV} \Bigl| \bm P_V^{(m_\ry)} (v) - \tilde {\bm P}_V^{(m_\ry)} (v) \Bigr| \geq |\setV| \, \xi_n} \\
\quad \leq \biggdistsubof {\indexsetsre}{ \exists \, v \in \setV \colon \Bigl| \bm P_V^{(m_\ry)} (v) - \tilde {\bm P}_V^{(m_\ry)} (v) \Bigr| \geq \xi_n} \\
\quad \stackrel{(b)}\leq 2 \, |\setV| \exp \bigl\{ - |\setM_\rz| \, \xi_n^2 / 2 \bigr\}, \quad \indexsetsre \in \setH^\ry_\mu,  \label{eq:prTotVarDistTooBigFixedMy}
\end{IEEEeqnarray}
where $(a)$ holds by definition of the Total-Variation distance; and $(b)$ holds by \eqref{eq:prExistsVTooBig}.

Having obtained \eqref{eq:prTotVarDistTooBigFixedMy} for every fixed $m_\ry \in \setM_\ry$, we are now ready to tackle the maximum over $m_\ry \in \setM_\ry$ and prove \eqref{eq:toShowIDBC5}: By \eqref{eq:ratePairAchRz}, \eqref{eq:blBinAndPoolRateConstraintsBCUBPoolRate}, \eqref{eq:IDBCKappaDef}, and \eqref{eq:IDBCKappaTripExpSmall} there must exist a positive constant $\tau > 0$ and some $\eta_0 \in \naturals$ for which
\begin{IEEEeqnarray}{l}
|\setV| \, \xi_n / 2 \leq e^{-n \tau}, \quad n \geq \eta_0. \label{eq:IDBCXiExpSmall}
\end{IEEEeqnarray}
For every $\tau > 0$ and $\eta_0 \in \naturals$ satisfying \eqref{eq:IDBCXiExpSmall} and for all $n$ exceeding $\eta_0$
\begin{IEEEeqnarray}{l}
\max_{\indexsetsre \in \setH^\ry_\mu} \biggdistsubof {\indexsetsre}{\exists \, m_\ry \in \setM_\ry \colon d \Bigl( \bm P_V^{(m_\ry)}, \bm U^{(m_\ry)}_V \Bigr) \geq e^{-n \tau}} \nonumber \\
\quad \stackrel{(a)}\leq \max_{\indexsetsre \in \setH^\ry_\mu} \biggdistsubof {\indexsetsre}{\exists \, m_\ry \in \setM_\ry \colon d \Bigl( \bm P_V^{(m_\ry)}, \bm U^{(m_\ry)}_V \Bigr) \geq |\setV| \, \xi_n / 2} \\
\quad \stackrel{(b)}\leq \max_{\indexsetsre \in \setH^\ry_\mu} \sum_{m_\ry \in \setM_\ry} \biggdistsubof {\indexsetsre}{  d \Bigl( \bm P_V^{(m_\ry)}, \bm U^{(m_\ry)}_V \Bigr) \geq |\setV| \, \xi_n / 2 } \\
\quad \stackrel{(c)}\leq 2 \, |\setV| \, |\setM_\ry| \exp \bigl\{ - |\setM_\rz| \exp \{ - e^{n \kappa} + 3 \log 2 \} \bigr\} \\
\quad \stackrel{(d)}\rightarrow 0 \, ( n \rightarrow \infty ),
\end{IEEEeqnarray}
where $(a)$ holds by \eqref{eq:IDBCXiExpSmall}, because $n$ exceeds $\eta_0$; $(b)$ follows from the Union-of-Events bound; $(c)$ holds by \eqref{eq:prTotVarDistTooBigFixedMy} and \eqref{eq:IDBCKappaTripExpSmall}; and $(d)$ holds because $|\setV| = e^{n R_\poolre}$, $|\setM_\ry| = \exp (\exp (n R_\ry))$, $|\setM_\rz| = \exp (\exp (n R_\rz))$, and by \eqref{eq:IDBCKappaDef}.
\end{proof}

\subsection{The Converse Part of Theorem~\ref{th:IDBC}}\label{sec:CVIDBC}

In this section we prove a strong converse to Theorem~\ref{th:IDBC}:

\begin{claim}\label{cl:toShowIDBCConv}
For every rate-pair $( R_\ry, R_\rz)$, every positive constants $\lambda_1^\ry, \, \lambda_2^\ry, \, \lambda_1^\rz, \, \lambda_2^\rz$ satisfying
\begin{IEEEeqnarray}{rCl}
\lambda_1^\ry + \lambda_2^\ry + \lambda_1^\rz + \lambda_2^\rz & < & 1, \label{eq:sumMissWrongBCSm1}
\end{IEEEeqnarray}
and every $\epsilon > 0$ there exists some $\eta_0 \in \naturals$ so that, for every blocklength~$n \geq \eta_0$, every size-$\exp (\exp (n R_\ry))$ set $\setM_\ry$ of possible ID messages for Receiver~$\ry$, and every size-$\exp (\exp (n R_\rz))$ set $\setM_\rz$ of possible ID messages for Receiver~$\rz$, a necessary condition for an $\bigl( n, \setM_\ry, \setM_\rz, \lambda_1^\ry, \lambda_2^\ry, \lambda_1^\rz, \lambda_2^\rz \bigr)$ ID code for the BC $\channel {y,z} x$ to exist is that for some PMF $P$ on $\setX$
 \begin{subequations}\label{eq:converseBC}
\begin{IEEEeqnarray}{rCl}
R_\ry & < & \muti {P}{W_\ry} + \epsilon, \\
R_\rz & < & \muti {P}{W_\rz} + \epsilon.
\end{IEEEeqnarray}
\end{subequations}
\end{claim}

To prove Claim~\ref{cl:toShowIDBCConv}, we recall from Remark~\ref{re:eqFormDefIDCodeBC} that the following two conditions are necessary and sufficient for some collection of tuples $$\bigl\{ Q_{m_\ry,m_\rz}, \setD_{m_\ry}, \setD_{m_\rz} \bigr\}_{(m_\ry,m_\rz) \in \setM_\ry \times \setM_\rz}$$ to be an $\bigl( n, \setM_\ry, \setM_\rz, \lambda^{\ry}_1, \lambda^{\ry}_2, \lambda^{\rz}_1, \lambda^{\rz}_2 \bigr)$ ID code for the BC $\channel {y,z} x$: 1) $\bigl\{Q_{m_\ry}, \setD_{m_\ry} \bigr\}_{m_\ry \in \setM_\ry}$ is an $\bigl( n, \setM_\ry, \lambda^{\ry}_1, \lambda^{\ry}_2 \bigr)$ ID code for the marginal channel $\channely y x$; and 2) $\bigl\{ Q_{m_\rz}, \setD_{m_\rz} \bigr\}_{m_\rz \in \setM_\rz}$ is an $\bigl( n, \setM_\rz, \lambda^{\rz}_1, \lambda^{\rz}_2 \bigr)$ ID code for $\channelz z x$, where $\bigl\{ Q_{m_\ry} \bigr\}_{m_\ry \in \setM_\ry}$ and $\bigl\{ Q_{m_\rz} \bigr\}_{m_\rz \in \setM_\rz}$ are defined in \eqref{bl:PMFsMarginalCodesBC}. We shall use these conditions to establish Claim~\ref{cl:toShowIDBCConv} following Han and Verd\'u's proof of the strong converse for identification via the DMC \cite{hanverdu92}. To that end, we shall need some terminology and results from \cite{hanverdu92}. We begin with the following two definitions from \cite{hanverdu92}: 
%

\begin{defin}
An $( n,\setM,\lambda_1,\lambda_2 )$ ID code $\{ Q_m, \setD_m \}_{m \in \setM}$ for the DMC $\channel y x$ is \emph{homogeneous} if for every $n$-type $P$ on $\setX^n$
\begin{IEEEeqnarray}{rCl}
Q_m \bigl( \setT_P^{(n)} \bigr) & = & \frac{1}{\card \setM} \sum_{\nu \in \setM} Q_{ \nu } \bigl( \setT_P^{(n)} \bigr), \quad m \in \setM.
\end{IEEEeqnarray}
\end{defin}

\begin{defin}
Given an $( n,\setM,\lambda_1,\lambda_2 )$ ID code $\{ Q_m, \setD_m \}_{m \in \setM}$ for the DMC $\channel y x$, define for every $n$-type $P$ on $\setX^n$ and $m \in \setM$ the PMF
\begin{IEEEeqnarray}{rCl}
Q_m^{( n, P )} ( \vecx ) & = & \begin{cases} \frac{Q_m ( \vecx )}{Q_m ( \setT_P^{(n)} )} &\textnormal{if } \vecx \in \setT_P^{(n)} \textnormal{ and } Q_m \bigl( \setT_P^{(n)} \bigr) > 0, \\ \frac{1}{|\setT_P^{(n)}|} & \textnormal{if } \vecx \in \setT_P^{(n)} \textnormal{ and } Q_m \bigl( \setT_P^{(n)} \bigr) = 0, \\ 0 &\textnormal{if } \vecx \notin \setT_P^{(n)}. \end{cases}
\end{IEEEeqnarray}
The ID code is \emph{$L$-regular} if for every $n$-type $P$ on $\setX^n$ and $m \in \setM$ satisfying $Q_m \bigl( \setT_P^{(n)} \bigr) > 0$ the PMF $Q_m^{( n, P )}$ on $\setT^{(n)}_P$ is an $L$-type.
\end{defin}

Following the line of arguments in \cite{hanverdu92}, we shall construct from $\bigl\{Q_{m_\ry}, \setD_{m_\ry} \bigr\}_{m_\ry \in \setM_\ry}$ and $\bigl\{Q_{m_\rz}, \setD_{m_\rz} \bigr\}_{m_\rz \in \setM_\rz}$ homogeneous $L$-regular ID codes. For the construction we shall need Proposition~\ref{pr:homogIDCodeApprox} and Lemma~\ref{le:rcodApprox} ahead. Proposition~\ref{pr:homogIDCodeApprox} is a variation on \cite[Proposition~3]{hanverdu92}, and Lemma~\ref{le:rcodApprox} is a generalization of \cite[Lemma~1]{hanverdu92} similar to that in \cite[Lemma~2]{ahlswedezhang95}.

\begin{proposition}\label{pr:homogIDCodeApprox}
For every $( n,\setM,\lambda_1,\lambda_2 )$ ID code $\{ Q_m, \setD_m \}_{m \in \setM}$ for the DMC $\channel y x$ and for every $\delta \geq \log 2 / n$ there exists a subset $\setS$ of $\setM$ with
\begin{IEEEeqnarray}{rCl}
|\setS| & \geq & |\setM| \exp \bigr\{ - e^{ \log (1 + n) (1 + |\setX|) + \log \delta } \bigl\} \label{eq:homogIDCodeApproxSizeSetS}
\end{IEEEeqnarray}
for which we can construct from $\{ Q_m, \setD_m \}_{m \in \setS}$ a homogeneous $( n,\setS,\lambda_1^\prime,\lambda_2^\prime )$ ID code $\{ Q_m^\prime, \setD_m \}_{m \in \setS}$ for $\channel y x$ with
\begin{subequations}\label{bl:homogIDCodeApproxLambdaTilde}
\begin{IEEEeqnarray}{rCl}
\lambda_1^\prime & = & \lambda_1 + e^{ - n \delta + \log ( 1 + n ) |\setX|}, \\
\lambda_2^\prime & = & \lambda_2 + e^{ - n \delta + \log ( 1 + n ) |\setX|}.
\end{IEEEeqnarray}
\end{subequations}
Moreover, if for some $\epsilon, \, \kappa > 0$
\begin{IEEEeqnarray}{rCl}
Q_m \bigl( X^n \in \{ \vecx \in \setX^n \colon \muti {P_\vecx}{W} \leq R - \epsilon \} \bigr) & \geq & \kappa, \quad m \in \setM, \label{eq:convReducedCodeSatisfiesTypeProperty}
\end{IEEEeqnarray}
then 
\begin{IEEEeqnarray}{rCl}
Q_m^\prime \bigl( X^n \in \{ \vecx \in \setX^n \colon \muti {P_\vecx}{W} \leq R - \epsilon \} \bigr) & \geq & \kappa, \quad m \in \setS. \label{eq:convHomogCodeSatisfiesTypeProperty}
\end{IEEEeqnarray}
\end{proposition}

\begin{proof}
The proof is essentially that of \cite[Proposition~3]{hanverdu92}. Additionally, we observe the following: if the PMFs $\{ Q_m \}_{m \in \setM}$ satisfy \eqref{eq:convReducedCodeSatisfiesTypeProperty}, then the PMFs $\{ Q_m^\prime \}_{m \in \setS}$, which are constructed in the proof of \cite[Proposition~3]{hanverdu92}, satisfy \eqref{eq:convHomogCodeSatisfiesTypeProperty}. For the sake of completeness, we provide a proof in Appendix~\ref{app:homogIDCodeApprox}.
\end{proof}

\begin{lemma}\label{le:rcodApprox}
For every DMC $\channel y x$ there exists a positive constant $\delta_0 > 0$, which depends only on $| \setY |$, and a continuous, strictly-increasing function $\rho \colon [0,\delta_0] \rightarrow \reals^+_0$ with $\rho (0) = 0$ so that, for every $\delta \in ( 0, \delta_0 ]$, every $\epsilon \in ( 0, 1 )$, and every blocklength $n \geq \eta_0$ (where $\eta_0 \in \naturals$ depends only on $|\setX |$, $| \setY |$, $\delta$, and $\epsilon$), it holds that for every $n$-type $P$ on $\setX^n$, every PMF $Q$ on $\setT^{(n)}_P \subseteq \setX^n$, every $R \geq I (P,W) + \rho (\delta)$, and every $L = \lceil e^{n R} \rceil$ there exists an $L$-type $Q^\prime$ on $\setT^{(n)}_P$ that satisfies for every subset $\setD$ of $\setY^n$
\begin{subequations} \label{bl:lemma1}
\begin{IEEEeqnarray}{rCl}
( Q^{\prime} W^n ) ( Y^n \in \setD ) & \leq & ( 1 + \epsilon ) ( 1 - e^{ -n \delta } )^{-1} ( Q W^n ) ( Y^n \in \setD ) + e^{ -n \delta }, \\
( Q^{\prime} W^n ) ( Y^n \in \setD ) & \geq & ( 1 - \epsilon ) ( 1 - e^{ -n \delta } ) ( Q W^n ) ( Y^n \in \setD ) - e^{ -n \delta }.
\end{IEEEeqnarray}
\end{subequations}
\end{lemma}

\begin{proof}
The proof is essentially that of \cite[Lemma~1]{hanverdu92} with the differences being pointed out in the proof of \cite[Lemma~2]{ahlswedezhang95}. For the sake of completeness, we provide a proof in Appendix~\ref{app:rcodApprox}.
\end{proof}

Once we have constructed from $\bigl\{Q_{m_\ry}, \setD_{m_\ry} \bigr\}_{m_\ry \in \setM_\ry}$ and $\bigl\{Q_{m_\rz}, \setD_{m_\rz} \bigr\}_{m_\rz \in \setM_\rz}$ homogeneous $L$-regular ID codes, we shall use the following proposition to upper-bound the number of possible ID messages $|\setM_\ry|$ and $|\setM_\rz|$:

\begin{proposition}\cite[Proposition~4]{hanverdu92}\label{prop:mRegRed}
Let $\setM$ be a finite set and $\lambda_1, \, \lambda_2$ positive constants satisfying $\lambda_1 + \lambda_2 < 1$. Every homogeneous $L$-regular $( n, \setM, \lambda_1, \lambda_2 )$ ID code for the DMC $\channel y x$ satisfies
\begin{IEEEeqnarray}{rCl}
\log | \setM | & \leq & n ( 1 + n )^{\card \setX} L \log | \setX |.
\end{IEEEeqnarray}
\end{proposition}

Once we have upper-bounded $|\setM_\ry|$ and $|\setM_\rz|$, we shall infer from the upper bounds that for every $\epsilon > 0$ and $n$ sufficiently large the mixture PMF on $\setX^n$ $$Q = \frac{1}{\card {\setM_\ry} \, \card {\setM_\rz}} \sum_{(m_\ry,m_\rz) \in \setM_\ry \times \setM_\rz} Q_{m_\ry,m_\rz}$$ must assign notable probability mass to some sequence $\vecx \in \setX^n$ that satisfies both $\muti {P_\vecx}{W_\ry} > R_\ry - \epsilon$ and $\muti {P_\vecx}{W_\rz} > R_\rz - \epsilon$. This implies Claim~\ref{cl:toShowIDBCConv}, because it implies that there must exist some PMF $P$ on $\setX$ for which \eqref{eq:converseBC} holds.

We next establish Claim~\ref{cl:toShowIDBCConv}, proceeding as outlined above. In a first step we shall combine Proposition~\ref{pr:homogIDCodeApprox}, Lemma~\ref{le:rcodApprox}, and Proposition~\ref{prop:mRegRed} to obtain the following lemma:

\begin{lemma}\label{le:avgDistWeightOnTypes}
For every DMC $\channel y x$, every ID rate $R$, and every positive constants $\lambda_1, \, \lambda_2, \, \epsilon, \, \kappa$ satisfying $\lambda_1 + \lambda_2 < \kappa < 1$ there exists some $\eta_0 \in \naturals$ so that, for every blocklength~$n \geq \eta_0$ and every size-$\exp (\exp (n R))$ set $\setM$ of possible ID messages, a necessary condition for a collection of tuples $\{ Q_m, \setD_m \}_{m \in \setM}$ to be an $( n,\setM,\lambda_1,\lambda_2 )$ ID code for the DMC $\channel y x$ is that
\begin{IEEEeqnarray}{rCl}
\frac{1}{| \setM |} \sum_{m \in \setM} Q_m \bigl( X^n \in \{ \vecx \in \setX^n \colon \muti {P_\vecx}{W} > R - \epsilon \} \bigr) & > & 1 - \kappa - \exp \bigl\{ e^{n ( R - \epsilon / 2 )} \bigr\} / \exp \bigl\{ e^{n R} \bigr\}.
\end{IEEEeqnarray}
\end{lemma}

\begin{proof}
Choose
\begin{equation}
\gamma = \biggl( 1 - \frac{\lambda_1 + \lambda_2}{\kappa} \biggr) / 2,
\end{equation}
and note that $\gamma > 0$. Pick $\delta > 0$ sufficiently small so that it satisfies the requirement in Lemma~\ref{le:rcodApprox} and so that $\rho (\delta) < \epsilon / 2$, where $\rho (\cdot)$ denotes the same function as in Lemma~\ref{le:rcodApprox}, and let $\epsilon^\prime = \rho (\delta)$. We henceforth assume that $n$ is sufficiently large so that the following four inequalitites hold:
\begin{subequations}\label{bl:convLemmaAvgDistWeightOnTypesNSuffLarge}
\begin{IEEEeqnarray}{rCl}
\log 2 / n & \leq & \delta, \label{eq:convLemmaAvgDistWeightOnTypesNSuffLarge1} \\
(1 + \gamma/4) (1 - e^{-n \delta})^{-1} + e^{-n \delta} & \leq & 1 + \gamma / 2, \label{eq:convLemmaAvgDistWeightOnTypesNSuffLarge3} \\
( \lambda_1 + \lambda_2 + 2 e^{ - n \delta + \log ( 1 + n ) |\setX|} ) / \kappa + \gamma & < & 1, \label{eq:convLemmaAvgDistWeightOnTypesNSuffLarge2} \\
\exp \bigl\{ e^{ n ( R - \epsilon + \epsilon^\prime ) + \log (1 + n) ( 1 + |\setX| ) + \log \log |\setX| } + e^{ \log (1 + n) (1 + |\setX|) + \log \delta } \bigr\} & < & \exp \bigl\{ e^{n ( R - \epsilon / 2 )} \bigr\}. \label{eq:convLemmaAvgDistWeightOnTypesNSuffLarge4}
\end{IEEEeqnarray}
\end{subequations}
Let $\setM$ be some size-$\exp (\exp (n R))$ set, and assume that the collection of tuples $\{ Q_m, \setD_m \}_{m \in \setM}$ is an $( n,\setM,\lambda_1,\lambda_2 )$ ID code for the DMC $\channel y x$. Pick
\begin{IEEEeqnarray}{rCl}
\setK & = & \Bigl\{ m \in \setM \colon Q_m \bigl( X^n \in \{ \vecx \in \setX^n \colon \muti {P_\vecx}{W} \leq R - \epsilon \} \bigr) \geq \kappa \Bigr\}, \label{bl:convLemmaAvgDistWeightOnTypesSetK}
\end{IEEEeqnarray}
and note that $\{ Q_m, \setD_m \}_{m \in \setK}$ is an $( n, \setK, \lambda_1, \lambda_2 )$ ID code for the DMC $\channel y x$. By \eqref{eq:convLemmaAvgDistWeightOnTypesNSuffLarge1}, \eqref{bl:convLemmaAvgDistWeightOnTypesSetK}, and Proposition~\ref{pr:homogIDCodeApprox} there exists a subset $\setS$ of $\setK$ with
\begin{IEEEeqnarray}{rCl}
| \setS | & \geq & | \setK | \exp \bigl\{ - e^{ \log (1 + n) (1 + |\setX|) + \log \delta } \bigr\} \label{eq:subsHomCode}
\end{IEEEeqnarray}
for which we can construct from $\{ Q_m, \setD_m \}_{m \in \setS}$ a homogeneous $( n,\setS,\lambda_1^\prime,\lambda_2^\prime )$ ID code $\{ Q_m^\prime, \setD_m \}_{m \in \setS}$ with
\begin{subequations}\label{bl:leAvgDistWeightOnTypesLambdasPrime}
\begin{IEEEeqnarray}{rCl}
\lambda_1^\prime & = & \lambda_1 + e^{ - n \delta + \log ( 1 + n ) |\setX|}, \\
\lambda_2^\prime & = & \lambda_2 + e^{ - n \delta + \log ( 1 + n ) |\setX|},
\end{IEEEeqnarray}
\end{subequations}
and
\begin{IEEEeqnarray}{rCl}
Q_m^\prime  \bigl( X^n \in \{ \vecx \in \setX^n \colon \muti {P_\vecx}{W} \leq R - \epsilon \} \bigr) & \geq & \kappa, \quad m \in \setS.
\end{IEEEeqnarray}
For every $m \in \setS$ define the PMF on $\setX^n$
\begin{IEEEeqnarray}{rCl}
Q^{\prime\prime}_m ( \vecx ) & = & \begin{cases} \frac{Q^\prime_m ( \vecx ) }{Q^\prime_m ( X^n \in \{ \vecx^\prime \in \setX^n \colon \muti {P_{\vecx^\prime}}{W} \leq R - \epsilon \} ) } &\textnormal{if } \muti {P_{\vecx}}{W} \leq R - \epsilon, \\
0 &\textnormal{otherwise}, \end{cases} \quad \vecx \in \setX^n.
\end{IEEEeqnarray}
Let
\begin{IEEEeqnarray}{l}
\lambda^{\prime\prime}_1 =  \frac{ \lambda_1^\prime }{\kappa} \quad \textnormal{and} \quad \lambda^{\prime\prime}_2 = \frac{ \lambda_2^\prime }{\kappa}, \label{eq:leAvgDistWeightOnTypesLambdasDoublePrime}
\end{IEEEeqnarray}
and note that the collection of tuples $\{ Q^{\prime\prime}_m, \setD_m \}_{m \in \setS}$ is a homogeneous $( n, \setS, \lambda_1^{\prime\prime}, \lambda_2^{\prime\prime} )$ ID code, because for every distinct pair $m, \, m^\prime \in \setS$
\begin{subequations}
\begin{IEEEeqnarray}{rCl}
( Q^{\prime\prime}_m W^n ) ( Y^n \notin \setD_m ) & \leq & \lambda^{\prime\prime}_1, \\
( Q^{\prime\prime}_m W^n ) ( Y^n \in \setD_{m^\prime} ) & \leq & \lambda^{\prime\prime}_2.
\end{IEEEeqnarray}
\end{subequations}

By Lemma~\ref{le:rcodApprox} there exists some $\eta_0^\prime \in \naturals$, which depends only on $|\setX|$, $|\setY|$, $\delta$, and $\gamma$, so that for every $n \geq \eta_0^\prime$ we can, for every $n$-type $P$ on $\setX^n$ for which $$\muti {P}{W} \leq R - \epsilon$$ and for every $m \in \setM$, approximate the PMF $( Q_m^{\prime\prime} )^{(n, P)}$ on $\setT^{(n)}_P$ by an $e^{ n ( R - \epsilon + \epsilon^\prime ) }$-type $( Q_m^{\prime\prime\prime})^{(n,P)}$ on $\setT^{(n)}_P$ that satisfies for every subset $\setD$ of $\setY^n$
\begin{IEEEeqnarray}{rCl}
\bigl( ( Q_m^{\prime\prime\prime})^{(n,P)} W^n \bigr) (Y^n \in \setD) & \leq & ( 1 + \gamma / 4 ) ( 1 - e^{ - n \delta } )^{-1} \bigl( ( Q_m^{\prime\prime})^{(n,P)} W^n \bigr) (Y^n \in \setD) + e^{ - n \delta } \\
& \leq & \bigl( ( Q_m^{\prime\prime})^{(n,P)} W^n \bigr) (Y^n \in \setD) + \gamma / 2, \label{eq:convLemmaAvgDistWeightOnTypesRegAppr}
\end{IEEEeqnarray}
where in the second inequality we used \eqref{eq:convLemmaAvgDistWeightOnTypesNSuffLarge3}. For every $m \in \setS$ define the PMF
\begin{IEEEeqnarray}{rCl}
Q_m^{\prime\prime\prime} (\vecx) & = & Q_m^{\prime\prime} \bigl( \setT^{(n)}_P \bigr) (Q_m^{\prime\prime\prime})^{(n,P)} (\vecx), \quad P \in \Gamma^{(n)}, \, \vecx \in \setT^{(n)}_P.
\end{IEEEeqnarray}
By \eqref{eq:convLemmaAvgDistWeightOnTypesRegAppr} it holds for every subset $\setD$ of $\setY^n$ that
\begin{IEEEeqnarray}{rCl}
( Q_m^{\prime\prime\prime} W^n ) (Y^n \in \setD) & = & \sum_{P \in \Gamma^{(n)}} Q_m^{\prime\prime} \bigl( \setT^{(n)}_P \bigr) \bigl( (Q_m^{\prime\prime\prime})^{(n,P)} W^n \bigr) ( Y^n \in \setD ) \\
& \leq & \sum_{P \in \Gamma^{(n)}} Q_m^{\prime\prime} \bigl( \setT^{(n)}_P \bigr) \Bigl( \bigl( (Q_m^{\prime\prime})^{(n,P)} W^n \bigr) (Y^n \in \setD) + \gamma / 2 \Bigr) \\
 & = & ( Q_m^{\prime\prime} W^n ) (Y^n \in \setD) + \gamma / 2. \label{eq:prQPPPYnInSet}
\end{IEEEeqnarray}
Let
\begin{IEEEeqnarray}{rCl}
\lambda_1^{\prime\prime\prime} = \lambda_1^{\prime\prime} + \frac{\gamma}{2} \quad \textnormal{and} \quad \lambda_2^{\prime\prime\prime} & = & \lambda_2^{\prime\prime} + \frac{\gamma}{2}.
\end{IEEEeqnarray}
By \eqref{eq:prQPPPYnInSet} and because $\{ Q^{\prime\prime}_m, \setD_m \}_{m \in \setS}$ is a homogeneous $( n, \setS, \lambda_1^{\prime\prime}, \lambda_2^{\prime\prime} )$ ID code, the collection of tuples $\{ Q_m^{\prime\prime\prime}, \setD_m \}_{m \in \setS}$ is a homogeneous $e^{ n ( R - \epsilon + \epsilon^\prime ) }$-regular $( n, \setS, \lambda_1^{\prime\prime\prime}, \lambda_2^{\prime\prime\prime} )$ ID code, and by \eqref{eq:convLemmaAvgDistWeightOnTypesNSuffLarge2}, \eqref{bl:leAvgDistWeightOnTypesLambdasPrime}, and \eqref{eq:leAvgDistWeightOnTypesLambdasDoublePrime}
\begin{IEEEeqnarray}{rCl}
\lambda_1^{\prime\prime\prime} + \lambda_2^{\prime\prime\prime} & < & 1.
\end{IEEEeqnarray}
Consequently, Proposition~\ref{prop:mRegRed} implies that
\begin{IEEEeqnarray}{rCl}
\log | \setS| & \leq & n ( 1 + n )^{| \setX |} e^{ n ( R - \epsilon + \epsilon^\prime ) } \log | \setX |,
\end{IEEEeqnarray}
and by \eqref{eq:subsHomCode}
\begin{IEEEeqnarray}{rCl}
| \setK | & \leq & \exp \bigl\{ e^{ n ( R - \epsilon + \epsilon^\prime ) + \log (1 + n) ( 1 + |\setX| ) + \log \log |\setX| } + e^{ \log (1 + n) (1 + |\setX|) + \log \delta } \bigr\} \\
& < & \exp \bigl\{ e^{n ( R - \epsilon / 2 )} \bigr\}, \label{eq:IDBCConvSetKUB}
\end{IEEEeqnarray}
where in the second inequality we used \eqref{eq:convLemmaAvgDistWeightOnTypesNSuffLarge4}. We are now ready to conclude the proof:
\begin{IEEEeqnarray}{l}
\frac{1}{| \setM |} \sum_{m \in \setM} Q_m \bigl( X^n \in \{ \vecx \in \setX^n \colon \muti {P_\vecx}{W} > R - \epsilon \} \bigr) \nonumber \\
\quad \stackrel{(a)}> ( 1 - \kappa ) \frac{|\setM| - | \setK |}{| \setM |} \\
\quad \stackrel{(b)}> 1 - \kappa - \exp \bigl\{ e^{n ( R - \epsilon / 2 )} \bigr\} / \exp \bigl\{ e^{n R} \bigr\}, \quad n \geq \eta_0,
\end{IEEEeqnarray}
where $(a)$ holds by \eqref{bl:convLemmaAvgDistWeightOnTypesSetK}; $(b)$ holds by \eqref{eq:IDBCConvSetKUB}; and we can let $\eta_0$ be the smallest integer no smaller than $\eta_0^\prime$ that satisfies \eqref{bl:convLemmaAvgDistWeightOnTypesNSuffLarge}.
\end{proof}

With Lemma~\ref{le:avgDistWeightOnTypes} at hand, we are now ready to conclude the proof of Claim~\ref{cl:toShowIDBCConv} by establishing that for every $\epsilon > 0$ and $n$ sufficiently large the mixture PMF on $\setX^n$ $$Q = \frac{1}{\card {\setM_\ry} \, \card {\setM_\rz}} \sum_{(m_\ry,m_\rz) \in \setM_\ry \times \setM_\rz} Q_{m_\ry,m_\rz}$$ must assign notable probability mass to some sequence $\vecx \in \setX^n$ that satisfies both $\muti {P_\vecx}{W_\ry} > R_\ry - \epsilon$ and $\muti {P_\vecx}{W_\rz} > R_\rz - \epsilon$:

\begin{proof}[Proof of Claim~\ref{cl:toShowIDBCConv}]
Fix $\kappa^\ry, \, \kappa^\rz > 0$ that satisfy the following three: 1) $\lambda_1^\ry + \lambda_2^\ry < \kappa^\ry$; 2) $\lambda_1^\rz + \lambda_2^\rz < \kappa^\rz$; and 3) $\kappa^\ry + \kappa^\rz< 1$. (This is possible because of \eqref{eq:sumMissWrongBCSm1}.) By Remark~\ref{re:eqFormDefIDCodeBC} and Lemma~\ref{le:avgDistWeightOnTypes} there must exist some $\eta_0^\prime \in \naturals$ so that, for every blocklength~$n \geq \eta_0^\prime$, every size-$\exp (\exp (n R_\ry))$ set $\setM_\ry$ of possible ID messages for Receiver~$\ry$, and every size-$\exp (\exp (n R_\rz))$ set $\setM_\rz$ of possible ID messages for Receiver~$\rz$, the following conditions are necessary for a collection of tuples $$\bigl\{ Q_{m_\ry,m_\rz}, \setD_{m_\ry}, \setD_{m_\rz} \bigr\}_{(m_\ry,m_\rz) \in \setM_\ry \times \setM_\rz}$$ to be an $\bigl( n, \setM_\ry, \setM_\rz, \lambda^{\ry}_1, \lambda^{\ry}_2, \lambda^{\rz}_1, \lambda^{\rz}_2 \bigr)$ ID code for the BC $\channel {y,z} x$: the mixture PMFs on $\setX^n$
\begin{subequations}
\begin{IEEEeqnarray}{rCl}
Q_{m_\ry} & = & \frac{1}{\card {\setM_\rz}} \sum_{m_\rz \in \setM_\rz} Q_{m_\ry,m_\rz}, \quad m_\ry \in \setM_\ry, \\
Q_{m_\rz} & = & \frac{1}{\card {\setM_\ry}} \sum_{m_\ry \in \setM_\ry} Q_{m_\ry,m_\rz}, \quad m_\rz \in \setM_\rz, \\
Q & = & \frac{1}{\card {\setM_\ry} \, \card {\setM_\rz}} \sum_{(m_\ry,m_\rz) \in \setM_\ry \times \setM_\rz} Q_{m_\ry,m_\rz}
\end{IEEEeqnarray}
\end{subequations}
satisfy
\begin{IEEEeqnarray}{l}
Q \bigl( X^n \in \{ \vecx \in \setX^n \colon \muti {P_\vecx}{W_\ry} > R_\ry - \epsilon \} \bigr) \nonumber \\
\quad = \frac{1}{| \setM_\ry |} \sum_{m_\ry \in \setM_\ry} Q_{m_\ry} \bigl( X^n \in \{ \vecx \in \setX^n \colon \muti {P_\vecx}{W_\ry} > R_\ry - \epsilon \} \bigr) \\
\quad \geq 1 - \kappa^\ry - \exp \bigl\{ e^{n ( R_\ry - \epsilon / 2 )} \bigr\} / \exp \bigl\{ e^{n R_\ry} \bigr\} \label{eq:condSinceW1IDCode}
\end{IEEEeqnarray}
and
\begin{IEEEeqnarray}{l}
Q \bigl( X^n \in \{ \vecx \in \setX^n \colon \muti {P_\vecx}{W_\rz} > R_\rz - \epsilon \} \bigr) \nonumber \\
\quad = \frac{1}{| \setM_\rz |} \sum_{m_\rz \in \setM_\rz} Q_{m_\rz} \bigl( X^n \in \{ \vecx \in \setX^n \colon \muti {P_\vecx}{W_\rz} > R_\rz - \epsilon \} \bigr) \\
\quad \geq 1 - \kappa^\rz - \exp \bigl\{ e^{n ( R_\rz - \epsilon / 2 )} \bigr\} / \exp \bigl\{ e^{n R_\rz} \bigr\}. \label{eq:condSinceW2IDCode}
\end{IEEEeqnarray}
The Union-of-Events bound, \eqref{eq:condSinceW1IDCode}, and \eqref{eq:condSinceW2IDCode} imply that
\begin{IEEEeqnarray}{l}
Q \bigl( X^n \in \{ \vecx \in \setX^n \colon \muti {P_\vecx}{W_\ry} > R_\ry - \epsilon, \, \muti {P_\vecx}{W_\rz} > R_\rz - \epsilon \} \bigr) \nonumber \\
\quad \geq 1 - \kappa^\ry - \kappa^\rz - \exp \bigl\{ e^{n ( R_\ry - \epsilon / 2 )} \bigr\} / \exp \bigl\{ e^{n R_\ry} \bigr\} - \exp \bigl\{ e^{n ( R_\rz - \epsilon / 2 )} \bigr\} / \exp \bigl\{ e^{n R_\rz} \bigr\}.  \label{eq:BCExistsPMFGoodForBoth}
\end{IEEEeqnarray}
Now let $\eta_0$ be the smallest integer $n \geq \eta_0^\prime$ for which the RHS of \eqref{eq:BCExistsPMFGoodForBoth} is positive (such an $n$ must exist, because $\epsilon > 0$ and $\kappa^\ry + \kappa^\rz < 1$). Then, for every blocklength $n \geq \eta_0$ a necessary condition for \eqref{eq:BCExistsPMFGoodForBoth} to hold is that for some PMF $P$ on $\setX$ \eqref{eq:converseBC} holds, and hence Claim~\ref{cl:toShowIDBCConv} follows.
\end{proof}

\section{Average- vs.\ Maximum-Error Criterion}\label{sec:maxError}

This section touches on the maximum-error criterion for identification via the BC, which was adopted in \cite{verbovenmeulen90, biliksteinberg01,oohama03,ahlswede08}. We are primarily interested in whether or not the maximum-error ID capacity region differs from the average-error ID capacity region. For Shannon's classical transmission problem this question can be answered in the negative: by Willems' result \cite{willems90} the transmission capacity region of the BC is the same under the average- and the maximum-error criterion. We begin with the basic definitions of a maximum-error ID code for the BC $\channel {y,z} x$:

\begin{defin}\label{def:IDCodeBCMacC}
Fix finite sets $\setM_\ry$ and $\setM_\rz$, a blocklength $n \in \naturals$, and positive constants $\lambda^{\ry}_1, \lambda^{\ry}_2, \lambda^{\rz}_1, \lambda^{\rz}_2$. Associate with every ID message-pair $( m_{\ry}, m_{\rz} ) \in \setM_{\ry} \times \setM_{\rz}$ a PMF $Q_{m_\ry,m_\rz}$ on $\setX^n$, with every $m_\ry \in \setM_\ry$ an ID set $\setD_{m_{\ry}} \subset \setY^n$, and with every $m_\rz \in \setM_\rz$ an ID set $\setD_{m_\rz} \subset \setZ^n$. The collection of tuples $\bigl\{ Q_{m_\ry,m_\rz}, \setD_{m_\ry}, \setD_{m_\rz} \bigr\}_{(m_\ry, m_\rz) \in \setM_\ry \times \setM_\rz}$ is an $\bigl( n, \setM_\ry, \setM_\rz, \lambda^{\ry}_1, \lambda^{\ry}_2, \lambda^{\rz}_1, \lambda^{\rz}_2 \bigr)$ maximum-error ID code for the BC $\channel {y,z} x$ if the maximum probabilities of missed identification at Terminals~$\ry$ and $\rz$
\begin{subequations}
\begin{IEEEeqnarray}{rCl}
p^\ry_{\textnormal{m-missed-ID}} = \max_{(m_\ry,m_\rz) \in \setM_\ry \times \setM_\rz} \bigl( Q_{m_\ry,m_\rz} W^n \bigr) \bigl( Y^n \notin \setD_{m_\ry} \bigr), \\
p^\rz_{\textnormal{m-missed-ID}} = \max_{(m_\ry,m_\rz) \in \setM_\ry \times \setM_\rz} \bigl( Q_{m_\ry,m_\rz} W^n \bigr) \bigl( Z^n \notin \setD_{m_\rz} \bigr)
\end{IEEEeqnarray}
\end{subequations}
satisfy
\begin{subequations}\label{bl:maxPrMissedIDBCMaxC}
\begin{IEEEeqnarray}{rCl}
p^\ry_{\textnormal{m-missed-ID}} & \leq & \lambda^{\ry}_1, \\
p^\rz_{\textnormal{m-missed-ID}} & \leq & \lambda^{\rz}_1,
\end{IEEEeqnarray}
\end{subequations}
and the maximum probabilities of wrong identification at Terminals~$\ry$ and $\rz$
\begin{subequations}
\begin{IEEEeqnarray}{rCl}
p^\ry_{\textnormal{m-wrong-ID}} & = & \max_{(m_\ry,m_\rz) \in \setM_\ry \times \setM_\rz} \max_{m^\prime_\ry \neq m_\ry } \bigl( Q_{m_\ry,m_\rz} W^n \bigr) \bigl( Y^n \in \setD_{m^\prime_\ry} \bigr), \\
p^\rz_{\textnormal{m-wrong-ID}} & = & \max_{(m_\ry,m_\rz) \in \setM_\ry \times \setM_\rz} \max_{m^\prime_\rz \neq m_\rz } \bigl( Q_{m_\ry,m_\rz} W^n \bigr) \bigl( Z^n \in \setD_{m^\prime_\rz} \bigr)
\end{IEEEeqnarray}
\end{subequations}
satisfy
\begin{subequations}\label{bl:maxPrWrongIDBCMacC}
\begin{IEEEeqnarray}{rCl}
p^\ry_{\textnormal{m-wrong-ID}} & \leq & \lambda^{\ry}_2, \\
p^\rz_{\textnormal{m-wrong-ID}} & \leq & \lambda^{\rz}_2.
\end{IEEEeqnarray}
\end{subequations}
A rate-pair $( R_\ry, R_\rz )$ is called maximum-error achievable if for every positive $\lambda^{\ry}_1$, $\lambda^{\ry}_2$, $\lambda^{\rz}_1$, and $\lambda^{\rz}_2$ and for every sufficiently-large blocklength~$n$ there exists an $\bigl( n, \setM_\ry, \setM_\rz, \lambda^{\ry}_1, \lambda^{\ry}_2, \lambda^{\rz}_1, \lambda^{\rz}_2 \bigr)$ maximum-error ID code for the BC with
\begin{IEEEeqnarray*}{rrCll}
& \tfrac{1}{n} \log \log | \setM_\ry | & \geq & R_\ry \quad &\textnormal{if } R_\ry > 0,
\\*[-0.625\normalbaselineskip]
\smash{\left\{
\IEEEstrut[6.39\jot]
\right.} \nonumber
\\*[-0.625\normalbaselineskip]
& |\setM_\ry| & = & 1 \quad &\textnormal{if } R_\ry = 0, \\[0.5\normalbaselineskip]
& \tfrac{1}{n} \log \log | \setM_\rz | & \geq & R_\rz \quad &\textnormal{if } R_\rz > 0,
\\*[-0.625\normalbaselineskip]
\smash{\left\{
\IEEEstrut[6.39\jot]
\right.} \nonumber
\\*[-0.625\normalbaselineskip]
& |\setM_\rz| & = & 1 \quad &\textnormal{if } R_\rz = 0.
\end{IEEEeqnarray*}
The maximum-error ID capacity region $\setC_{\textnormal{m}}$ of the BC is the closure of the set of all maximum-error-achievable rate-pairs.
\end{defin}

While the average-error criterion requires that each receiver identify the message intended for it reliably in expectation over the uniform ID message intended for the other receiver, the maximum-error criterion requires that each receiver identify the message intended for it reliably even if the realization of the ID message for the other receiver assumes the worst possible realization. Consequently, every rate-pair that is maximum-error achievable is also average-error achievable, and the average-error ID capacity region is thus an outer bound on the maximum-error ID capacity region. The maximum-error ID capacity region of the BC is still unknown. To-date the best known inner bound is the common-randomness capacity region $\setR_{\textnormal{cr}}$. It is unknown whether this bound is tight.

\begin{theorem}\cite[Theorem~11]{ahlswede08}\label{th:IDBCMaxC}
The maximum-error ID capacity region $\setC_{\textnormal{m}}$ of the BC $\channel {y,z}{x}$ contains the common-randomness capacity region $\setR_{\textnormal{cr}}$, which is the set of all rate-pairs $( R_\ry, R_\rz ) \in ( \reals^+_0 )^2$ that for some finite set $\setU$ and some PMF $P_{U,X}$ on $\setU \times \setX$ satisfy
\begin{subequations}\label{bl:capacityBCMaxC1}
\begin{IEEEeqnarray}{rCl}
R_\ry & \leq & \muti{P_U}{P_{X|U} W_\ry}, \\
R_\rz & \leq & \min \bigl\{ \muti{P_U}{P_{X|U} W_\ry} + \condmuti{P_{X|U}}{W_\rz}{P_U}, \muti {P_X}{W_\rz} \bigr\},
\end{IEEEeqnarray}
\end{subequations}
or
\begin{subequations}\label{bl:capacityBCMaxC2}
\begin{IEEEeqnarray}{rCl}
R_\ry & \leq & \min \bigl\{ \muti{P_U}{P_{X|U} W_\rz} + \condmuti{P_{X|U}}{W_\ry}{P_U}, \muti {P_X}{W_\ry} \bigr\}, \\
R_\rz & \leq & \muti{P_U}{P_{X|U} W_\rz}.
\end{IEEEeqnarray}
\end{subequations}
\end{theorem}

The region $\setR_{\textnormal{cr}}$ can be achieved by a common-randomness ID code, which---like that of \cite{ahlswededueckfb89} for the DMC---uses a transmission code to establish common randomness between the encoder and each decoder. If the BC is degraded, then Theorem~\ref{th:IDBCMaxC} specializes to \cite[Theorem~1]{biliksteinberg01}. Also for the degraded BC it is unknown whether the common-randomness inner bound is tight.

\begin{remark}\label{re:commRandIBStrCont}
The common-randomness capacity region $\setR_{\textnormal{cr}}$ is contained in the average-error ID capacity region $\setC$, and the containment can be strict.
\end{remark}

\begin{proof}
Every maximum-error-achievable rate-pair is also average-error achievable, and hence, by Theorem~\ref{th:IDBCMaxC}, $\setR_{\textnormal {cr}} \subseteq \setC_{\textnormal m} \subseteq \setC$. To see that $\setR_{\textnormal {cr}}$ can be strictly smaller than $\setC$, consider the binary-symmetric BC of \cite[Example~5.3]{gamalkim11}, whose marginal channels are both binary symmetric. This BC is degraded, and Theorem~\ref{th:IDBCMaxC} thus specializes to \cite[Theorem~1]{biliksteinberg01}, which we can evaluate as in \cite[Example~5.3 and Section 5.4.2]{gamalkim11} to conclude that $\setR_{\textnormal{cr}} \subsetneq \setC$ holds whenever the transition probabilities of the marginal binary-symmetric channels are distinct.
\end{proof}

To-date it is still unknown whether the common-randomness inner bound on the maximum-error ID capacity region of the BC is tight, i.e., whether $\setC_{\textnormal m} = \setR_{\textnormal {cr}}$. Ahlswede argued that it is whenever an additional constraint is imposed on the maximum probabilities of missed and wrong identification, namely, that they decay like $n^{-7}$, where $n$ is the blocklength \cite[Section~15]{ahlswede08}. Since the average-error ID capacity region of the BC is also achievable when we require that the error probabilities decay exponentially in $n$ (Remark~\ref{re:convExpFastBC}), we could thus infer from Remark~\ref{re:commRandIBStrCont} that, for some BCs and subject to the additional constraint that the maximum probabilities of missed and wrong identification decay like $n^{-7}$, the average-error ID capacity region is strictly larger than the maximum-error ID capacity region.

We hesitate to draw this conclusion, because there seems to be a gap in Ahlswede's proof: Ahlswede's proof (that of the converse part of \cite[Theorem~11]{ahlswede08}) builds on his converse to the single-user ID coding theorem \cite[Theorem~9]{ahlswede08}, which applies when for every blocklength~$n$ the maximum probabilities of missed and wrong identification must not exceed $n^{-7}$. The proof of \cite[Theorem~9]{ahlswede08} can be roughly sketched as follows: First, it is shown that for every possible ID message $m$ the PMF $Q_m$ can be represented by a size-$M$ subset of $\setX^n$. Then, it is argued that only few ID messages can have the same representation, and that the ID rate can thus be upper-bounded in terms of the number of possible representations, i.e., in terms of ${|\setX|^n \choose M}$. Since
\begin{IEEEeqnarray}{rCcCl}
{|\setX|^n \choose M} & \leq & |\setX^n|^{M} & = & e^{n \log |\setX| \, M} \approx \exp \bigl\{ e^{\log M} \bigr\},
\end{IEEEeqnarray}
it is concluded that for $n$ sufficiently large the ID rate cannot exceed $\log M / n$, where $M$ can be upper-bounded by \cite[Lemma~7]{ahlswede08}. Ahlswede's converse for the BC is similar (see \cite[Section~15]{ahlswede08}): To upper-bound the ID rate $R_\rz$ of Receiver~$\rz$, an auxiliary random variable $U$ is introduced, which is uniform over the support $\setM_\ry$ of the possible ID messages for Receiver~$\ry$. As in the proof of \cite[Theorem~9]{ahlswede08}, it is shown that for every possible ID message $m_\rz \in \setM_\rz$ for Receiver~$\rz$ the PMF $$Q_{m_\rz} = \frac{1}{|\setM_\ry|} \sum_{m_\ry \in \setM_\ry} Q_{m_\ry,m_\rz}$$ can be represented by a size-$M$ subset of $\setM_\ry$. Like for the single-user channel, it is argued that only few ID messages for Receiver~$\rz$ can have the same representation, and that one can thus upper-bound the ID rate $R_\rz$ in terms of the number of possible representations, i.e., in terms of ${|\setM_\ry| \choose M}$. From this it is concluded that for $n$ sufficiently large the ID rate cannot exceed $\log M / n$. There seems to be a gap in this conclusion, because, unlike $\setX^n$, the cardinality of $\setM_\ry$  grows doubly-exponentially in~$n$, i.e., $|\setM_\ry| = \exp ( \exp (n R_\ry))$, where $R_\ry$ is the ID rate of Receiver~$\ry$; and it is therefore not clear how to conclude that for $n$ sufficiently large $R_\rz$ cannot exceed $\log M / n$, because
\begin{IEEEeqnarray}{rCcCcCl}
{|\setM_\ry| \choose M} & \leq & |\setM_\ry|^{M} & = & \exp \bigl\{ e^{n R_\ry} M \bigr\} & = & \exp \bigl\{ e^{n R_\ry + \log M} \bigr\}.
\end{IEEEeqnarray}

\section{Extensions}\label{sec:extensions}

This section discusses several extensions: identification via the BC with more than two receivers (Section~\ref{sec:3RecBC}), identification via the BC with a common message (Section~\ref{sec:BCCommMsg}), and identification via the BC with one-sided feedback (Section~\ref{sec:BC1FB}). 

\subsection{More than Two Receivers}\label{sec:3RecBC}

In this section we study identification via the BC with more than two receivers. As we shall see, it is easy to adapt the converse of Theorem~\ref{th:IDBC} to this more general scenario, but in the direct part difficulties already arise when the number of receivers increases from two to three. To keep the exposition simple, we shall thus focus on the three-receiver BC. We inner-bound its ID capacity region and show that the bound is in some cases tight.\\

Consider a three-receiver BC of transition law $\channel {y_1,y_2,y_3}{x}$, and for every $k \in \{ 1, 2, 3 \}$ let $\setY_k$ denote the support of the channel output at Receiver~$k$ and $W_k (y_k|x)$ the marginal channel to Receiver~$k$. We begin with the basic definitions of an average-error ID code for the BC $\channel {y_1, y_2, y_3}{x}$:

\begin{defin}
Fix finite sets $\setM_1$, $\setM_2$, and $\setM_3$, a blocklength $n \in \naturals$, and positive constants $$\lambda^{(k)}_1, \, \lambda^{(k)}_2, \quad k \in \{ 1,2,3 \}.$$ Associate with every ID message-triple $( m_1, m_2, m_3 ) \in \setM_1 \times \setM_2 \times \setM_3$ a PMF $Q_{m_1,m_2,m_3}$ on $\setX^n$, and for each $k \in \{ 1,2,3 \}$ associate with every $m_k \in \setM_k$ an ID set $\setD_{m_k} \subset \setY_k^n$. Define the mixture PMFs on $\setX^n$
\begin{subequations}
\begin{IEEEeqnarray}{rCll}
Q_{m_1} & = & \frac{1}{|\setM_2| \, |\setM_3|} \sum_{m_2, m_3} Q_{m_1,m_2,m_3}, &\quad m_1 \in \setM_1, \\
Q_{m_2} & = & \frac{1}{|\setM_1| \, |\setM_3|} \sum_{m_1, m_3} Q_{m_1,m_2,m_3}, &\quad m_2 \in \setM_2, \\
Q_{m_3} & = & \frac{1}{|\setM_1| \, |\setM_2|} \sum_{m_1, m_2} Q_{m_1,m_2,m_3}, &\quad m_3 \in \setM_3.
\end{IEEEeqnarray}
\end{subequations}
The collection of tuples $\bigl\{ Q_{m_1,m_2,m_3}, \setD_{m_1}, \setD_{m_2}, \setD_{m_3} \bigr\}_{(m_1,m_2,m_3) \in \setM_1 \times \setM_2 \times \setM_3}$ is an $\bigl( n, \{ \setM_k, \lambda^{(k)}_1, \lambda^{(k)}_2 \}_{k \in \{ 1,2,3 \}} \bigr)$ ID code for the BC $\channel {y_1,y_2,y_3} x$ if for each $k \in \{ 1, 2, 3 \}$ the collection of tuples $\bigl\{ Q_{m_k}, \setD_{m_k} \bigr\}_{m_k \in \setM_k}$ is an $\bigl( n, \setM_k, \lambda^{(k)}_1, \lambda^{(k)}_2 \bigr)$ ID code for the marginal channel $W_k (y_k|x)$. A rate-triple $( R_1, R_2, R_3 )$ is called achievable if for every positive $\lambda^{(1)}_1$, $\lambda^{(1)}_2$, $\lambda^{(2)}_1$, $\lambda^{(2)}_2$, $\lambda^{(3)}_1$, and $\lambda^{(3)}_2$ and for every sufficiently-large blocklength~$n$ there exists an $\bigl( n, \{ \setM_k, \lambda^{(k)}_1, \lambda^{(k)}_2 \}_{k \in \{ 1,2,3 \}} \bigr)$ ID code for the BC with
\begin{IEEEeqnarray*}{rrClll}
& \tfrac{1}{n} \log \log | \setM_k | & \geq & R_k \quad &\textnormal{if } R_k > 0, \\*[-0.625\normalbaselineskip]
\smash{\left\{
\IEEEstrut[6.39\jot]
\right.} \nonumber &&&&&\qquad k \in \{1,2,3\}.
\\*[-0.625\normalbaselineskip]
& |\setM_k| & = & 1 \quad &\textnormal{if } R_k = 0,
\end{IEEEeqnarray*}
The ID capacity region $\setC_3$ of the three-receiver BC is the closure of the set of all achievable rate-triples.
\end{defin}

Our next result is an outer bound on the ID capacity region of the three-receiver BC:

\begin{theorem}\label{th:obBC3Rec}
The ID capacity region $\setC_3$ of the BC $\channel {y_1, y_2, y_3}{x}$ is contained in the set $\setR_{3 \textnormal{-ob}}$ of all rate-triples $(R_1, R_2, R_3) \in ( \reals^+_0 )^3$ that for some PMF $P$ on $\setX$ satisfy 
\begin{IEEEeqnarray}{rCl}
R_k & \leq & \muti{P}{W_k}, \, \forall \, k \in \{ 1, 2, 3 \}. \label{eq:obBC3Rec}
\end{IEEEeqnarray}
\end{theorem}

\begin{proof}
The proof follows along the line of arguments in Section~\ref{sec:CVIDBC} (see Appendix~\ref{app:obBC3Rec} for the details).
\end{proof}

We can adapt the two-receiver broadcast ID code of Section~\ref{sec:DPIDBC} to obtain the following inner bound on the ID capacity region of the three-receiver BC.

\begin{theorem}\label{th:ibBC3Rec}
The ID capacity region $\setC_3$ of the BC $\channel {y_1, y_2, y_3}{x}$ contains the set $\setR_{3 \textnormal{-ib}}$ of all rate-triples $(R_1, R_2, R_3) \in ( \reals^+_0 )^3$ that for some PMF $P$ on $\setX$ satisfy
\begin{IEEEeqnarray}{rCl}
R_k & \leq & \min \Biggl\{ \muti{P}{W_k}, \sum_{l \in \{ 1,2,3 \} \setminus \{ k \}} \muti{P}{W_l} \Biggr\}, \, \forall \, k \in \{ 1, 2, 3 \}. \label{eq:ibBC3Rec}
\end{IEEEeqnarray}
The interior of $\setR_{3 \textnormal{-ib}}$ is achieved by codes with deterministic encoders.
\end{theorem}

\begin{proof}
See Appendix~\ref{app:ibBC3Rec}.
\end{proof}

By comparing Theorems~\ref{th:IDBC} and \ref{th:ibBC3Rec}, we see that to adapt the broadcast ID code of Section~\ref{sec:DPIDBC} to the three-receiver BC we additionally need the constraints
\begin{IEEEeqnarray}{rCl}
R_k & < & \sum_{l \in \{ 1,2,3 \} \setminus \{ k \}} \muti{P}{W_l}, \, \forall \, k \in \{ 1, 2, 3 \}, \label{eq:3RecAddConst}
\end{IEEEeqnarray}
which have no counterpart in the two-receiver case. We next explain where we use \eqref{eq:3RecAddConst}. To this end, we briefly describe how to extend the random code construction of Section~\ref{sec:DPIDBC} to the three-receiver BC. Fix a PMF $P$ on $\setX$, a blocklength~$n$, ID rates $R_k, \, k \in \{ 1,2,3 \}$, expected bin rates $\tilde R_k, \, k \in \{ 1,2,3 \}$, and a pool rate $R_\poolre$ satisfying
\begin{IEEEeqnarray}{l}
R_k < \tilde R_k < \min \{ \muti {P}{W_k}, R_\poolre \}, \, \forall \, k \in \{ 1,2,3 \}. \label{eq:3RecAddConstExpl1}
\end{IEEEeqnarray}
Draw $e^{n R_\poolre}$ $n$-tuples $\sim P^n$ independently, index them, and place them in a pool $\pool$. For each receiving terminal $k \in \{ 1,2,3 \}$ associate with each ID message $m_k \in \setM_k$ a Bin~$\bin {m_k}$ by randomly selecting each indexed element of the pool for inclusion in $\bin {m_k}$ independently with probability $e^{-n (R_\poolre - \tilde R_k)}$. Associate with every ID message-triple~$(m_1,m_2,m_3)$ an $n$-tuple we call the $(m_1,m_2,m_3)$-codeword as follows. If at least one indexed pool-element is contained in all three bins $\bin {m_1}$, $\bin {m_2}$, and $\bin {m_3}$, then draw the $(m_1,m_2,m_3)$-codeword uniformly over the indexed pool-elements that are contained in all three bins. Otherwise draw the $(m_1,m_2,m_3)$-codeword uniformly over the pool. To send ID message-triple $(m_1,m_2,m_3)$, the encoder transmits the $(m_1,m_2,m_3)$-codeword. For each $k \in \setM_k$ the $m^\prime_k$-focused party at Terminal~$k$ guesses that $m_k^\prime$ was sent if at least one element of the $m^\prime_k$-th bin is jointly typical with the channel outputs that it observes. Therefore, if the $(m_1,m_2,m_3)$-codeword is not an element of Bin~$\bin {m_k}$, then the probability that the $m_k$-focused party at Terminal~$k$ erroneously guesses that $m_k$ was not sent is high.

Note that for every ID message-triple~$(m_1,m_2,m_3)$ the expected number of indexed pool-elements that are contained in all three bins $\bin {m_1}$, $\bin {m_2}$, and $\bin {m_3}$ is $e^{n ( \sum^3_{k = 1} \tilde R_k - 2 R_\poolre )}$ ($= e^{n R_\poolre} \prod^3_{k = 1} e^{-n (R_\poolre - \tilde R_k)}$), which is smaller than one unless
\begin{IEEEeqnarray}{rCl}
2 R_\poolre & \leq & \sum^3_{k = 1} \tilde R_k. \label{eq:3BinsShare1El}
\end{IEEEeqnarray}
Therefore, if \eqref{eq:3BinsShare1El} does not hold, then with high probability the $(m_1,m_2,m_3)$-codeword is not contained in all three bins $\bin {m_1}$, $\bin {m_2}$, and $\bin {m_3}$, and our scheme will thus fail. This, combined with \eqref{eq:3RecAddConstExpl1}, implies that the code can be reliable only if \eqref{eq:3RecAddConst} holds. Note that in the two-receiver scenario the counterpart to \eqref{eq:3BinsShare1El} is 
\begin{IEEEeqnarray}{rCl}
R_\poolre & \leq & \tilde R_\ry + \tilde R_\rz. \label{eq:2BinsShare1El}
\end{IEEEeqnarray}
Unlike \eqref{eq:3BinsShare1El} in the three-receiver scenario, \eqref{eq:2BinsShare1El} in the two-receiver scenario can be satisfied by choosing $R_\poolre$ sufficiently small and hence without constraining the rate-pair $(R_\ry, R_\rz)$.\\

As the following example shows, the inner bound of Theorem~\ref{th:ibBC3Rec} need not be tight:

\begin{example}
Consider a deterministic BC $\channel {y_1,y_2,y_3} x$ with input $X = (X_1,X_2,X_3)$, where for each $k \in \{ 1,2,3 \}$ $X_k$ is binary, and with output $Y = (Y_1,Y_2,Y_3)$, where
\begin{subequations}
\begin{IEEEeqnarray}{rCl} 
Y_k & = & X_k, \quad k \in \{ 1,2 \}, \\
Y_3 & = & X.
\end{IEEEeqnarray}
\end{subequations}
For this channel the inner bound $\setR_{3 \textnormal{-ib}}$ of Theorem~\ref{th:ibBC3Rec} evaluates to the set of all rate-triples $(R_1, R_2, R_3) \in ( \reals^+_0 )^3$ that satisfy
\begin{subequations}
\begin{IEEEeqnarray}{rCl}
R_k & \leq & \log 2, \, \forall \, k \in \{ 1,2 \}, \\
R_3 & \leq & 2 \log 2.
\end{IEEEeqnarray}
\end{subequations}
Since the BC is deterministic, the encoder can compute all outputs from the inputs that it produces, and the ID capacity region $\setC_3$ does thus not increase if the encoder if furnished with perfect feedback. Therefore, Theorem~\ref{th:obBC3Rec} and \cite[Corollary~3]{ahlswedeverboven91}, which holds under the maximum-error criterion, imply that $\setC_3$ is the set of all rate-triples $(R_1, R_2, R_3) \in ( \reals^+_0 )^3$ that satisfy
\begin{subequations}
\begin{IEEEeqnarray}{rCl}
R_k & \leq & \log 2, \, \forall \, k \in \{ 1,2 \}, \\
R_3 & \leq & 3 \log 2.
\end{IEEEeqnarray}
\end{subequations}
Consequently, $\setR_{3 \textnormal{-ib}} \subsetneq \setC_3$.
\end{example}

The inner bound of Theorem~\ref{th:ibBC3Rec} is in some cases tight, e.g., if no receiver is ``much more capable'' than the other two:

\begin{remark}\label{re:3RecTight}
If the BC $\channel {y_1,y_2,y_3} x$ satisfies for every PMF $P$ on $\setX$
\begin{IEEEeqnarray}{rCl}
2 \max_{k \in \{ 1, 2, 3 \}} I (P,W_k) & \leq & \sum_{l \in \{ 1,2,3 \}} \muti{P}{W_l},
\end{IEEEeqnarray}
then its ID capacity region $\setC_3$ is the set of all rate-triples $(R_1, R_2, R_3) \in ( \reals^+_0 )^3$ that for some PMF $P$ on $\setX$ satisfy \eqref{eq:obBC3Rec}.
\end{remark}

\begin{proof}
This follows from Theorems~\ref{th:obBC3Rec} and \ref{th:ibBC3Rec}, because for such a BC $\setR_{3 \textnormal{-ob}} = \setR_{3 \textnormal{-ib}}$.
\end{proof}

\subsection{A Common Message}\label{sec:BCCommMsg}

In this section we consider the two-receiver BC $\channel {y,z} x$ and adapt the coding scheme in Section~\ref{sec:DPIDBC} to solve for the capacity region of a more general scenario where the receivers' ID messages need not be independent but can have a common part. We thus assume that the ID message intended for Terminal~$\ry$ is a tuple comprising a private message and a common message, and likewise for Terminal~$\rz$. We begin with the basic definitions of an average-error ID code for the BC $\channel {y,z} x$ with a common message:

\begin{defin}\label{def:commonMessageIDCodeBC}
Fix finite sets $\setM$, $\setM_\ry$, and $\setM_\rz$, a blocklength $n \in \naturals$, and positive constants $\lambda^\ry_1, \, \lambda^\ry_2, \, \lambda^\rz_1, \, \lambda^\rz_2$. Associate with every ID message-triple $( m, m_\ry, m_\rz ) \in \setM \times \setM_\ry \times \setM_\rz$ a PMF $Q_{m,m_\ry,m_\rz}$ on $\setX^n$, with every $(m,m_\ry) \in \setM \times \setM_\ry$ an ID set $\idsetre {m, m_\ry} \subset \setY^n$, and with every $(m,m_\rz) \in \setM \times \setM_\rz$ an ID set $\idsetre {m, m_\rz} \subset \setZ^n$. Define the mixture PMFs on $\setX^n$
\begin{subequations}
\begin{IEEEeqnarray}{rCll}
Q_{m,m_\ry} & = & \frac{1}{ | \setM_\rz | } \sum_{m_\rz \in \setM_\rz} Q_{m,m_\ry,m_\rz}, & \quad (m,m_\ry) \in \setM \times \setM_\ry,  \\
Q_{m,m_\rz} & = & \frac{1}{ | \setM_\ry | } \sum_{m_\ry \in \setM_\ry} Q_{m,m_\ry,m_\rz}, & \quad (m,m_\rz) \in \setM \times \setM_\rz.
\end{IEEEeqnarray}
\end{subequations}
The collection of tuples $$\bigl\{ Q_{m,m_\ry,m_\rz}, \idsetre {m,m_\ry}, \idsetre {m,m_\rz} \bigr\}_{(m,m_\ry,m_\rz) \in \setM \times \setM_\ry \times \setM_\rz}$$ is an $\bigl( n, \setM, \setM_\ry, \setM_\rz, \lambda^\ry_1, \lambda^\ry_2, \lambda^\rz_1, \lambda^\rz_2 \bigr)$ ID code for the BC $\channel {y,z} x$ with a common message if the following two requirements are met: 1) $\bigl\{ Q_{m,m_\ry}, \setD_{m,m_\ry} \bigr\}_{( m,m_\ry ) \in \setM \times \setM_\ry}$ is an $\bigl( n, \setM \times \setM_\ry, \lambda^{\ry}_1, \lambda^{\ry}_2 \bigr)$ ID code for the marginal channel $\channely y x$; and 2) $\bigl\{ Q_{m,m_\rz}, \setD_{m,m_\rz} \bigr\}_{( m, m_\rz ) \in \setM \times \setM_\rz}$ is an $\bigl( n, \setM \times \setM_\rz, \lambda^{\rz}_1, \lambda^{\rz}_2 \bigr)$ ID code for $\channelz z x$. A rate-triple $( R, R_\ry, R_\rz )$ is called achievable if for every positive $\lambda^\ry_1$, $\lambda^\ry_2$, $\lambda^\rz_1$, and $\lambda^\rz_2$ and for every sufficiently-large blocklength~$n$ there exists an $\bigl( n, \setM, \setM_\ry, \setM_\rz, \lambda^\ry_1, \lambda^\ry_2, \lambda^\rz_1, \lambda^\rz_2 \bigr)$ ID code for the BC with
\begin{IEEEeqnarray*}{rrCll}
& \tfrac{1}{n} \log \log | \setM | & \geq & R \quad &\textnormal{if } R > 0, \\*[-0.625\normalbaselineskip]
\smash{\left\{
\IEEEstrut[6.39\jot]
\right.} \nonumber
\\*[-0.625\normalbaselineskip]
& |\setM| & = & 1 \quad &\textnormal{if } R = 0, \\[0.5\normalbaselineskip]
& \tfrac{1}{n} \log \log | \setM_\ry | & \geq & R_\ry \quad &\textnormal{if } R_\ry > 0,
\\*[-0.625\normalbaselineskip]
\smash{\left\{
\IEEEstrut[6.39\jot]
\right.} \nonumber
\\*[-0.625\normalbaselineskip]
& |\setM_\ry| & = & 1 \quad &\textnormal{if } R_\ry = 0, \\[0.5\normalbaselineskip]
& \tfrac{1}{n} \log \log | \setM_\rz | & \geq & R_\rz \quad &\textnormal{if } R_\rz > 0,
\\*[-0.625\normalbaselineskip]
\smash{\left\{
\IEEEstrut[6.39\jot]
\right.} \nonumber
\\*[-0.625\normalbaselineskip]
& |\setM_\rz| & = & 1 \quad &\textnormal{if } R_\rz = 0.
\end{IEEEeqnarray*}
The ID capacity region $\setC_{\textnormal{cm}}$ of the BC with a common message is the closure of the set of all achievable rate-triples.
\end{defin}

We restrict our analysis to positive ID rates $R_\ry, \, R_\rz$, because if to some receiver we send only the common message, then for the other receiver the imposed average-error criterion will turn into a maximum-error criterion. Theorem~\ref{th:IDBC} allows for the following generalization:

\begin{theorem}\label{th:IDBCCM}
The ID capacity region $\setC_{\textnormal{cm}}$ of the BC $\channel {y,z} x$ with a common message and positive private rates $R_\ry, \, R_\rz$ is the set of all rate-triples $( R, R_\ry, R_\rz ) \in ( \reals^+_0 )^3$ that for some PMF $P$ on $\setX$ satisfy
\begin{subequations}\label{bl:capRegBCComM}
\begin{IEEEeqnarray}{rCl}
R, \, R_\ry & \leq & \muti{P}{W_\ry}, \\
R, \, R_\rz & \leq & \muti{P}{W_\rz}, \\
R_\ry, \, R_\rz & > & 0.
\end{IEEEeqnarray}
\end{subequations}
The interior of $\setC_{\textnormal{cm}}$ is achieved by codes with deterministic encoders.
\end{theorem}

\begin{proof}
The proof is similar to that of Theorem~\ref{th:IDBC} (see Appendix~\ref{app:IDBCCM} for the details).
\end{proof}

Comparing Theorems~\ref{th:IDBCCM} and \ref{th:IDBC} we see that the common message appears to come for free at all rates up to $\min \bigl\{ \muti {P}{W_\ry}, \muti {P}{W_\rz} \bigr\}$. This can be explained as follows. The ID rate is the iterated logarithm of the number of ID messages normalized by the blocklength~$n$, and for $n$ sufficiently large and for all nonnegative real numbers $R_1$ and $R_2$ $$\exp ( \exp (n R_1) ) \exp ( \exp (n R_2) ) \approx \exp \bigl( \exp \bigl( n \max \{ R_1, R_2 \} \bigr) \bigr).$$\\

So far, we assumed that each receiver identifies the common message and its private message jointly. Next, we assume that each receiver identifies the common message and its private message separately. We begin with the basic definitions of an average-error ID code for the BC $\channel {y,z} x$ with a common message and where each receiver identifies the common message and its private message separately:

\begin{defin}\label{def:commonMessageIDCodeBCSeperateID}
Fix finite sets $\setM$, $\setM_\ry$, and $\setM_\rz$, a blocklength $n \in \naturals$, and positive constants $\lambda^\ry_1, \, \lambda^\ry_2, \, \lambda^\rz_1, \, \lambda^\rz_2$. Associate with every ID message-triple $( m, m_\ry, m_\rz ) \in \setM \times \setM_\ry \times \setM_\rz$ a PMF $Q_{m,m_\ry,m_\rz}$ on $\setX^n$, with every $m \in \setM$ ID sets $\idsetre m^\ry \subset \setY^n$ and $\idsetre m^\rz \subset \setZ^n$, with every $m_\ry \in \setM_\ry$ an ID set $\idsetre {m_\ry} \subset \setY^n$, and with every $m_\rz \in \setM_\rz$ an ID set $\idsetre {m_\rz} \subset \setZ^n$. Define the mixture PMFs on $\setX^n$
\begin{subequations}
\begin{IEEEeqnarray}{rCll}
Q_m & = & \frac{1}{ | \setM_\ry | \, | \setM_\rz | } \sum_{m_\ry, m_\rz} Q_{m,m_\ry,m_\rz}, &\quad m \in \setM,  \\
Q_{m_\ry} & = & \frac{1}{ | \setM | \, | \setM_\rz | } \sum_{m, m_\rz} Q_{m,m_\ry,m_\rz}, &\quad m_\ry \in \setM_\ry, \\
Q_{m_\rz} & = & \frac{1}{ | \setM | \, | \setM_\ry | } \sum_{m, m_\ry} Q_{m,m_\ry,m_\rz}, &\quad m_\rz \in \setM_\rz.
\end{IEEEeqnarray}
\end{subequations}
The collection of tuples $$\bigl\{ Q_{m,m_\ry,m_\rz}, \idsetre m^\ry, \idsetre {m_\ry}, \idsetre m^\rz, \idsetre {m_\rz} \bigr\}_{(m,m_\ry,m_\rz) \in \setM \times \setM_\ry \times \setM_\rz}$$ is an $\bigl( n, \setM, \setM_\ry, \setM_\rz, \lambda^\ry_1, \lambda^\ry_2, \lambda^\rz_1, \lambda^\rz_2 \bigr)$ ID code for the BC $\channel {y,z} x$ with a common message and where each receiver identifies the common message and its private message separately if the following four requirements are met: 1) $\bigl\{ Q_m, \setD^{\ry}_m \bigr\}_{m \in \setM}$ is an $\bigl( n, \setM, \lambda^{\ry}_1, \lambda^{\ry}_2 \bigr)$ ID code for the marginal channel $\channely y x$; 2) $\bigl\{ Q_{m_\ry}, \setD_{m_\ry} \bigr\}_{m_\ry \in \setM_\ry}$ is an $\bigl( n, \setM_\ry, \lambda^{\ry}_1, \lambda^{\ry}_2 \bigr)$ ID code for $\channely y x$; 3) $\bigl\{ Q_m, \setD^{\rz}_m \bigr\}_{m \in \setM}$ is an $\bigl( n, \setM, \lambda^{\rz}_1, \lambda^{\rz}_2 \bigr)$ ID code for $\channelz z x$; and 4) $\bigl\{ Q_{m_\rz}, \setD_{m_\rz} \bigr\}_{m_\rz \in \setM_\rz}$ is an $\bigl( n, \setM_\rz, \lambda^{\rz}_1, \lambda^{\rz}_2 \bigr)$ ID code for $\channelz z x$. A rate-triple $( R, R_\ry, R_\rz )$ is called achievable if for every positive $\lambda^\ry_1$, $\lambda^\ry_2$, $\lambda^\rz_1$, and $\lambda^\rz_2$ and for every sufficiently-large blocklength $n$ there exists an $$\bigl( n, \setM, \setM_\ry, \setM_\rz, \lambda^\ry_1, \lambda^\ry_2, \lambda^\rz_1, \lambda^\rz_2 \bigr)$$ ID code for the BC with
\begin{IEEEeqnarray*}{rrCll}
& \tfrac{1}{n} \log \log | \setM | & \geq & R \quad &\textnormal{if } R > 0, \\*[-0.625\normalbaselineskip]
\smash{\left\{
\IEEEstrut[6.39\jot]
\right.} \nonumber
\\*[-0.625\normalbaselineskip]
& |\setM| & = & 1 \quad &\textnormal{if } R = 0, \\[0.5\normalbaselineskip]
& \tfrac{1}{n} \log \log | \setM_\ry | & \geq & R_\ry \quad &\textnormal{if } R_\ry > 0,
\\*[-0.625\normalbaselineskip]
\smash{\left\{
\IEEEstrut[6.39\jot]
\right.} \nonumber
\\*[-0.625\normalbaselineskip]
& |\setM_\ry| & = & 1 \quad &\textnormal{if } R_\ry = 0, \\[0.5\normalbaselineskip]
& \tfrac{1}{n} \log \log | \setM_\rz | & \geq & R_\rz \quad &\textnormal{if } R_\rz > 0,
\\*[-0.625\normalbaselineskip]
\smash{\left\{
\IEEEstrut[6.39\jot]
\right.} \nonumber
\\*[-0.625\normalbaselineskip]
& |\setM_\rz| & = & 1 \quad &\textnormal{if } R_\rz = 0.
\end{IEEEeqnarray*}
The ID capacity region $\setC_{\textnormal{cm-s}}$ of the BC with a common message and where each receiver identifies the common message and its private message separately is the closure of the set of all achievable rate-triples.
\end{defin}

When each receiver identifies the common message and its private message separately, we can argue similarly as for the three-receiver BC to obtain the following result:

\begin{theorem}\label{th:obibBCCMSep}
The ID capacity region $\setC_{\textnormal{cm-s}}$ of the BC $\channel {y, z}{x}$ with a common message and where each receiver identifies the common message and its private message separately is contained in the set of all rate-triples $(R, R_\ry, R_\rz) \in ( \reals^+_0 )^3$ that for some PMF $P$ on $\setX$ satisfy
\begin{subequations}\label{bl:obBCCMSep}
\begin{IEEEeqnarray}{rCl}
R, \, R_\ry & \leq & \muti{P}{W_\ry}, \\
R, \, R_\rz & \leq & \muti{P}{W_\rz},
\end{IEEEeqnarray}
\end{subequations}
and it contains the set of all rate-tiples $(R, R_\ry, R_\rz) \in ( \reals^+_0 )^3$ that for some PMF $P$ on $\setX$ satisfy \eqref{bl:obBCCMSep} and
\begin{subequations}\label{bl:ibBCCMSep}
\begin{IEEEeqnarray}{rCl}
R_\ry & \leq & 2 \muti{P}{W_\rz}, \\
R_\rz & \leq & 2 \muti{P}{W_\ry}.
\end{IEEEeqnarray}
\end{subequations}
\end{theorem}

\begin{proof}
Pretend that the common ID message were intended for a third receiver whose marginal channel is time-invariant but can be either $\channely y x$ or $\channelz z x$. Then, we can argue as in Appendices~\ref{app:obBC3Rec} and \ref{app:ibBC3Rec} to establish the outer and inner bound, respectively.
\end{proof}

\subsection{One-Sided Feedback}\label{sec:BC1FB}

In this section we study identification via the BC $\channel {y,z} x$ with perfect feedback from at least one receiving terminal. Feedback from both terminals $\setY$ and $\setZ$ allows the encoder to choose the Time-$i$ channel-input in dependence on all past channel outputs $Y^{i-1}$ and $Z^{i-1}$: to transmit ID Message-Pair~$( m_{\ry}, m_{\rz} )$ when the past channel inputs are $X^{i-1} = x^{i-1}$ and the past channel outputs are $Y^{i-1} = y^{i-1}$ and $Z^{i-1} = z^{i-1}$, the stochastic encoder generates the Time-$i$ channel-input from a PMF of the form $$Q^{(i)}_{m_\ry,m_\rz} ( x | x^{i-1}, y^{i-1}, z^{i-1} ), \quad x \in \setX.$$ The ID capacity region $\setC_{\textnormal{fb}}$ of the BC with feedback from both terminals is known and can be achieved by a common-randomness ID code similar to that of \cite{ahlswededueckfb89}. It does not depend on the error criterion.

\begin{theorem}\cite[Corollary~3]{ahlswedeverboven91}\label{th:IDBCFB}
The ID capacity region $\setC_{\textnormal{fb}}$ of the BC $\channel {y,z} x$ with feedback from both terminals is the set of all rate-pairs $( R_\ry, R_\rz ) \in ( \reals^+_0 )^2$ that for some PMF $P$ on $\setX$ satisfy
\begin{subequations}\label{eq:IDBC2FB}
\begin{IEEEeqnarray}{rCl}
R_\ry & \leq & \ent {P W_\ry} \ind {\max_{\tilde P} \muti {\tilde P}{W_\ry} > 0}, \\
R_\rz & \leq & \ent {P W_\rz} \ind {\max_{\tilde P} \muti {\tilde P}{W_\rz} > 0}.
\end{IEEEeqnarray}
\end{subequations}
\end{theorem}

Things get more interesting when the encoder is furnished with feedback from only one receiving terminal, say Terminal~$\ry$. In this scenario the encoder can choose the Time-$i$ channel-input in dependence on the past Terminal-$\ry$ outputs $Y^{i-1}$: to transmit ID Message-Pair~$( m_{\ry}, m_{\rz} )$ when the past channel inputs are $X^{i-1} = x^{i-1}$ and the past Terminal-$\ry$ outputs are $Y^{i-1} = y^{i-1}$, the stochastic encoder generates the Time-$i$ channel-input from a PMF of the form $$Q^{(i)}_{m_\ry,m_\rz} ( x | x^{i-1}, y^{i-1} ), \quad x \in \setX.$$ We use the following basic definitions of an average-error ID code with one-sided feedback from Terminal~$\ry$:

\begin{defin}\label{def:IDCodeBC1FB}
Fix finite sets $\setM_\ry$ and $\setM_\rz$, a blocklength $n \in \naturals$, and positive constants $\lambda^{\ry}_1, \lambda^{\ry}_2, \lambda^{\rz}_1, \lambda^{\rz}_2$. Associate with every ID message-pair $( m_{\ry}, m_{\rz} ) \in \setM_{\ry} \times \setM_{\rz}$ conditional PMFs $$Q^{(i)}_{m_\ry,m_\rz} (x|x^{i-1}, y^{i-1}), \quad i \in [1:n], \, (x,x^{i-1},y^{i-1}) \in \setX \times \setX^{i-1} \times \setY^{i-1},$$ with every $m_\ry \in \setM_\ry$ an ID set $\setD_{m_{\ry}} \subset \setY^n$, and with every $m_\rz \in \setM_\rz$ an ID set $\setD_{m_\rz} \subset \setZ^n$. The tuple $$\Bigl\{ \bigl\{ Q^{(i)}_{m_\ry,m_\rz} \bigr\}_{i \in \{ 1, \ldots, n \}}, \setD_{m_\ry}, \setD_{m_\rz} \Bigr\}_{(m_\ry,m_\rz) \in \setM_\ry \times \setM_\rz}$$ is an $\bigl( n, \setM_\ry, \setM_\rz, \lambda^{\ry}_1, \lambda^{\ry}_2, \lambda^{\rz}_1, \lambda^{\rz}_2 \bigr)$ ID code for the BC $\channel {y,z} x$ with one-sided feedback from Terminal~$\ry$ if the maximum probabilities of missed identification
\begin{subequations}
\begin{IEEEeqnarray}{rCl}
p^\ry_{\textnormal{missed-ID}} = \max_{m_\ry \in \setM_\ry} \frac{1}{\card {\setM_\rz}} \sum_{m_\rz \in \setM_\rz} \sum_{\substack{\vecx \in \setX^n, \\ \vecy \notin \setD_{m_\ry}}} \prod^n_{i = 1} Q^{(i)}_{m_\ry,m_\rz} ( x_i | x^{i-1}, y^{i-1} ) \channely {y_i}{x_i}, \\
p^\rz_{\textnormal{missed-ID}} = \max_{m_\rz \in \setM_\rz} \frac{1}{\card {\setM_\ry}} \sum_{m_\ry \in \setM_\ry} \sum_{\substack{\vecx \in \setX^n, \\ \vecy \in \setY^n, \\ \vecz \notin \setD_{m_\rz}}} \prod^n_{i = 1} Q^{(i)}_{m_\ry,m_\rz} ( x_i | x^{i-1}, y^{i-1} ) \channel {y_i,z_i}{x_i}
\end{IEEEeqnarray}
\end{subequations}
satisfy
\begin{subequations}\label{bl:maxPrMissedIDBC1FB}
\begin{IEEEeqnarray}{rCl}
p^\ry_{\textnormal{missed-ID}} & \leq & \lambda^{\ry}_1, \\
p^\rz_{\textnormal{missed-ID}} & \leq & \lambda^{\rz}_1,
\end{IEEEeqnarray}
\end{subequations}
and the maximum probabilities of wrong identification
\begin{subequations}
\begin{IEEEeqnarray}{rCl}
p^\ry_{\textnormal{wrong-ID}} & = & \max_{m_\ry \in \setM_\ry} \max_{m^\prime_\ry \neq m_\ry} \frac{1}{\card {\setM_\rz}} \sum_{m_\rz \in \setM_\rz} \sum_{\substack{\vecx \in \setX^n, \\ \vecy \in \setD_{m^\prime_\ry}}} \prod^n_{i = 1} Q^{(i)}_{m_\ry,m_\rz} ( x_i | x^{i-1}, y^{i-1} ) \channely {y_i}{x_i}, \\
p^\rz_{\textnormal{wrong-ID}} & = & \max_{m_\rz \in \setM_\rz} \max_{m^\prime_\rz \neq m_\rz } \frac{1}{\card {\setM_\ry}} \sum_{m_\ry \in \setM_\ry} \sum_{\substack{\vecx \in \setX^n, \\ \vecy \in \setY^n, \\ \vecz \in \setD_{m^\prime_\rz}}} \prod^n_{i = 1} Q^{(i)}_{m_\ry,m_\rz} ( x_i | x^{i-1}, y^{i-1} ) \channel {y_i,z_i}{x_i}
\end{IEEEeqnarray}
\end{subequations}
satisfy
\begin{subequations}\label{bl:maxPrWrongIDBC1FB}
\begin{IEEEeqnarray}{rCl}
p^\ry_{\textnormal{wrong-ID}} & \leq & \lambda^{\ry}_2, \\
p^\rz_{\textnormal{wrong-ID}} & \leq & \lambda^{\rz}_2.
\end{IEEEeqnarray}
\end{subequations}
A rate-pair $( R_\ry, R_\rz )$ is called achievable if for every positive $\lambda^{\ry}_1$, $\lambda^{\ry}_2$, $\lambda^{\rz}_1$, and $\lambda^{\rz}_2$ and for every sufficiently-large blocklength~$n$ there exists an $\bigl( n, \setM_\ry, \setM_\rz, \lambda^{\ry}_1, \lambda^{\ry}_2, \lambda^{\rz}_1, \lambda^{\rz}_2 \bigr)$ ID code for the BC with
\begin{IEEEeqnarray*}{rrCll}
& \tfrac{1}{n} \log \log | \setM_\ry | & \geq & R_\ry \quad &\textnormal{if } R_\ry > 0,
\\*[-0.625\normalbaselineskip]
\smash{\left\{
\IEEEstrut[6.39\jot]
\right.} \nonumber
\\*[-0.625\normalbaselineskip]
& |\setM_\ry| & = & 1 \quad &\textnormal{if } R_\ry = 0, \\[0.5\normalbaselineskip]
& \tfrac{1}{n} \log \log | \setM_\rz | & \geq & R_\rz \quad &\textnormal{if } R_\rz > 0,
\\*[-0.625\normalbaselineskip]
\smash{\left\{
\IEEEstrut[6.39\jot]
\right.} \nonumber
\\*[-0.625\normalbaselineskip]
& |\setM_\rz| & = & 1 \quad &\textnormal{if } R_\rz = 0.
\end{IEEEeqnarray*}
The ID capacity region $\setC_{1 \textnormal{-fb}}$ of the BC with one-sided feedback from Receiver~$\ry$ is the closure of the set of all achievable rate-pairs.
\end{defin}

One-sided feedback from Terminal~$\ry$ can be viewed as a special case of noisy feedback from Terminal~$\rz$. The ID capacity of the DMC with noisy feedback is to-date unknown. Inner and outer bounds can be found in \cite[Theorem~1]{ahlswedezhang95}. We do not tackle the general problem here, but we adapt the coding scheme in Section~\ref{sec:DPIDBC} to inner-bound the ID capacity region of the BC with one-sided feedback, and we show that the bound is tight whenever the channel outputs are independent conditional on the channel input. In such a scenario feedback from Terminal~$\ry$ does not provide the encoder with information about the channel output at Terminal~$\rz$. We can adapt the broadcast ID code of Section~\ref{sec:DPIDBC} to obtain the following inner bound:

\begin{theorem}\label{th:IDBC1FBIB}
The ID capacity region $\setC_{1\textnormal{-fb}}$ of the BC $\channel {y,z} x$ with one-sided feedback from Terminal~$\ry$ contains the set $\setR_{1\textnormal{-fb-ib}}$ of all rate-pairs $( R_\ry, R_\rz ) \in ( \reals^+_0 )^2$ that for some PMF $P$ on $\setX$ satisfy
\begin{subequations}\label{bl:IDBC1FB}
\begin{IEEEeqnarray}{rCl}
R_\ry & \leq & \ent {P W_\ry} \ind {\max_{\tilde P} \muti {\tilde P}{W_\ry} > 0}, \\
R_\rz & \leq & \muti P{W_\rz}.
\end{IEEEeqnarray}
\end{subequations}
The interior of $\setR_{1\textnormal{-fb-ib}}$ is achieved by codes with deterministic encoders.
\end{theorem}

\begin{proof}
A formal proof can be found in Appendix~\ref{app:IDBC1FBIB}. Here, we provide a rough sketch. To prove the theorem, we extend the random code construction of Section~\ref{sec:DPIDBC} as follows: Fix an input distribution $P \in \mathscr P (\setX)$ and any positive ID rate-pair $(R_\ry, R_\rz)$ satisfying
\begin{subequations}
\begin{IEEEeqnarray}{rCcCl}
0 & < & R_\ry & < & \ent {P W_\ry} \ind {\max_{\tilde P} \muti {\tilde P}{W_\ry} > 0}, \\
0 & < & R_\rz & < & \muti P{W_\rz}.
\end{IEEEeqnarray}
\end{subequations}
Let $\setM_\ry$ be a size-$\exp (\exp (n R_\ry))$ set of possible ID messages for Receiver~$\ry$, and let $\setM_\rz$ be a size-$\exp (\exp (n R_\rz))$ set of possible ID messages for Receiver~$\rz$. Generate an ID code for the marginal channel $\channelz z x$ as in Section~\ref{sec:IDCodingTechnique}, and associate with every ID message-pair $(m_\ry, m_\rz)$ an $n$-tuple we call the $(m_\ry, m_\rz)$-codeword as follows. If Bin~$m_\rz$ is not empty, then draw the codeword uniformly over Bin~$m_\rz$, otherwise let it be some arbitrary but fixed pool element. To send ID Message-Pair~$(m_\ry, m_\rz)$, the encoder transmits during the first $n$ channel uses the $(m_\ry, m_\rz)$-codeword. Similarly as in Secion~\ref{sec:DPIDBC}, we can show that if the ID message that is sent to Terminal~$\ry$ is uniform over its support $\setM_\ry$, the ID message that is sent to Terminal~$\rz$ is $m_\rz$, and Bin~$m_\rz$ is not empty, then the transmitted codeword is nearly uniformly distributed (in terms of Total-Variation distance) over Bin~$m_\rz$. Consequently, by the analysis in Section~\ref{sec:IDCodingTechnique} and because $R_\rz < \muti P {W_\rz}$, Receiver~$\rz$ can identify its ID message reliably after the first $n$ channel uses.

As to Receiver~$\ry$, we can show that if the ID message that is sent to Terminal~$\ry$ is $m_\ry$ and the ID message that is sent to Terminal~$\rz$ is uniform over its support $\setM_\rz$, then the transmitted codeword is nearly uniformly distributed over the pool (in terms of Total-Variation distance). Since the pool contains $e^{n R_\poolre}$ $n$-tuples, which are drawn $\sim P^n$ independently, the results in \cite{hanverdu93} imply that for $R_\poolre > \muti {P}{W_\ry}$ the distribution of the length-$n$ Terminal-$\ry$ output-sequence $Y^n$ is nearly the product distribution $( P W_\ry )^n$ (in terms of Total-Variation distance). Therefore, if we choose $R_\poolre > \muti {P}{W_\ry}$, then the common randomness $Y^n$ that the encoder and Receiver~$\ry$ share after $n$ transmissions is of rate $\ent{ P W_\ry }$. Consequently, we can use the common-randomness argument of \cite{ahlswededueckfb89} to show that an additional $\sqrt n$ channel uses suffice for Receiver~$\ry$ to identify its ID message reliably, because $R_\ry < \ent{ P W_\ry } \ind {\max_{\tilde P} \muti {\tilde P}{W_\ry} > 0}$. To conclude, note that asymptotically $\sqrt n$ additional channel uses cannot decrease the ID rates.
\end{proof}

As the following example shows, the inner bound of Theorem~\ref{th:IDBC1FBIB} need not be tight:

\begin{example}
Consider a BC $\channel {y,z}{x}$ for which $Z = f ( Y )$. On such a channel feedback from Terminal~$\ry$ is as good as feedback from both terminals, and the ID capacity region $\setC_{1 \textnormal{-fb}}$ with one-sided feedback from Terminal~$\ry$ is thus the ID capacity region $\setC_{\textnormal{fb}}$ with feedback from both terminals. To see that in general $\setR_{1\textnormal{-fb-ib}} \subsetneq \setC_{\textnormal{fb}}$, consider for example a binary symmetric BC with identical outputs, whose receiving terminals both observe the output of the same binary symmetric channel.
\end{example}

Denote the conditional PMF of the Terminal-$\rz$ output given the channel input and the Terminal-$\ry$ output by $\channelztild {z}{x,y}$, i.e.,
\begin{equation}
\channelztild {z}{x,y} = \frac{\channel {y,z}{x}}{\channely y x}. \label{eq:channelztild}
\end{equation}
Our next result is an outer bound on the ID capacity region of the BC with one-sided feedback from Terminal~$\ry$:

\begin{theorem}\label{th:IDBC1FBOB}
The ID capacity region $\setC_{1\textnormal{-fb}}$ of the BC $\channel {y,z} x$ with one-sided feedback from Terminal~$\ry$ is contained in the set $\setR_{1\textnormal{-fb-ob}}$ of all rate-pairs $( R_\ry, R_\rz ) \in ( \reals^+_0 )^2$ that for some PMF $P$ on $\setX$ satisfy
\begin{subequations}\label{eq:IDBC1FBOB}
\begin{IEEEeqnarray}{rCl}
R_\ry & \leq & \ent{P W_\ry} \ind {\max_{\tilde P} \muti {\tilde P}{W_\ry} > 0}, \nonumber \\
R_\rz & \leq & \muti{P \times W_\ry}{\widetilde W_\rz} \ind {\max_{\tilde P} \muti {\tilde P}{W_\rz} > 0},
\end{IEEEeqnarray}
\end{subequations}
where $\widetilde W_\rz$ is defined in \eqref{eq:channelztild}.
\end{theorem}

\begin{proof}
See Appendix~\ref{app:IDBC1FBOB}.
\end{proof}

If the outputs of the BC are conditionally independent given its input, i.e., if $\channel {y,z} x = \channely y x \channelz z x$, then the inner bound of Theorem~\ref{th:IDBC1FBIB} coincides with the outer bound of Theorem~\ref{th:IDBC1FBOB}:

\begin{corollary}\label{co:ICBC1FBMarkov}
The ID capacity region $\setC_{1\textnormal{-fb}}$ of the BC $\channel {y,z}x = \channely y x \channelz z x$ with one-sided feedback from Terminal~$\ry$ is the set of all rate-pairs $( R_\ry, R_\rz ) \in ( \reals^+_0 )^2$ that for some PMF $P$ on $\setX$ satisfy
\begin{subequations}\label{eq:capRegBCFB}
\begin{IEEEeqnarray}{rCl}
R_\ry & \leq & \ent{P W_\ry} \ind {\max_{P} \muti {P}{W_\ry} > 0}, \\
R_\rz & \leq & \muti{P}{W_\rz}. 
\end{IEEEeqnarray}
\end{subequations}
\end{corollary}

\begin{proof}
The direct part follows from Theorem~\ref{th:IDBC1FBIB}. And the converse part follows from Theorem~\ref{th:IDBC1FBOB}, because $\channel {y,z} x = \channely y x \channelz z x$ implies that $\channelztild z {x,y} = \channelz z x$, and hence it holds that for every PMF $P$ on $\setX$
\begin{IEEEeqnarray}{rCl}
\muti{P \times W_\ry}{\widetilde W_\rz} & = & \muti {P}{W_\rz}.
\end{IEEEeqnarray}
\end{proof}

\section{Summary}

The ID capacity region of the two-receiver BC is the set of rate-pairs for which, for some distribution on the channel input, each receiver's ID rate does not exceed the mutual information between the channel input and the output that it observes. The capacity region's interior is achieved by codes with deterministic encoders. The results hold under the average-error criterion, which requires that each receiver identify the message intended for it reliably in expectation over the uniform ID message intended for the other receiving terminal. Previously, identification via the BC was studied under the maximum-error criterion, which requires that each receiver identify the message intended for it reliably irrespective of the realization of the ID message intended for the other receiving terminal. Both criteria---average- and maximum-error---consistently extend Ahlswede and Dueck's identification-via-channels problem to the broadcast setting.

The average-error criterion is suitable whenever the receivers' ID messages are independent and uniform over their supports. As we have seen, our coding scheme can be adapted to solve for the capacity region of a more general scenario where the receivers' ID messages are not independent but have a common part. We also discussed extensions to the BC with more than two receivers and the two-receiver BC with one-sided feedback. In particular, we obtained the ID capacity region of the three-receiver BC whenever no receiver is ``much more capable'' than the other two and that of the two-receiver BC with one-sided feedback whenever the channel outputs are independent conditional on the channel input.

The question whether for some BCs the average-error ID capacity region can be strictly larger than the maximum-error ID capacity region remains open. We do know that the ID capacity regions differ when only deterministic encoders are allowed: under the average-error criterion deterministic encoders can achieve every rate-pair in the interior of the ID capacity region, but under the maximum-error criterion they cannot achieve any positive ID rates.

\begin{appendix}
%

\section{A Proof of Lemma~\ref{le:prEventEToZero}}\label{app:lePrEventEToZero}

We use the Union-of-Events bound to show that $\bigdistof {\{ \indexset m \}_{m \in \setM} \notin \setG_\mu}$ converges to zero. We begin with the events $| \indexset m | \leq ( 1 - \delta_n) e^{n \tilde R}$ and $| \indexset {m^\prime} | \geq ( 1 + \delta_n) e^{n \tilde R}$. For every $\nu \in \setM$ the binary random variables $\bigl\{ \ind {v \in \indexset \nu} \bigr\}_{v \in \setV}$ are IID, and
\begin{IEEEeqnarray}{rCl}
\BiggEx {}{\sum_{v \in \setV} \ind {v \in \indexset \nu}} & = & \sum_{v \in \setV} \distof {v \in \indexset \nu} = e^{ n \tilde R}.
\end{IEEEeqnarray}
Consequently, by the multiplicative Chernoff bounds in Proposition~\ref{pr:multChernoff},
\begin{IEEEeqnarray}{rCl}
\Bigdistof { |\indexset m| \leq (1 - \delta_n) \, e^{n \tilde R} } & = & \Biggdistof { \sum_{v \in \setV} \ind {v \in \indexset m} \leq (1 - \delta_n) \, e^{n \tilde R} } \\
& \leq & \exp \bigl\{ - \delta_n^2 \, e^{n \tilde R - \log 2} \bigr\} \\
& = & \exp \bigl\{ - e^{n (\tilde R - \mu) - \log 2} \bigr\}, \label{eq:missedIdImSuffLarge}
\end{IEEEeqnarray}
and
\begin{IEEEeqnarray}{rCl}
\Bigdistof { |\indexset {m^\prime}| \geq (1 + \delta_n) \, e^{n \tilde R} } & \leq & \exp \bigl\{ - e^{n (\tilde R - \mu) - \log 3} \bigr\}. \label{eq:ImSuffSmall}
\end{IEEEeqnarray}
As to $| \indexset {m,m^\prime} | \geq e^{n ( \tilde R - \mu / 2 ) + \log 2}$, note that for every $v \in \setV$ $$\ind {v \in \indexset {m,m^\prime}} = \ind {v \in \indexset m} \ind {v \in \indexset {m^\prime}},$$ where $\ind {v \in \indexset m}$ and $\ind {v \in \indexset {m^\prime}}$ are independent because $m \neq m^\prime$. Hence, the binary random variables $\bigl\{ \ind {v \in \indexset {m,m^\prime}} \bigr\}_{v \in \setV}$ are IID of mean
\begin{IEEEeqnarray}{C}
\BiggEx {}{\sum_{v \in \setV} \ind {v \in \indexset {m,m^\prime}}} = \sum_{v \in \setV} \distof {v \in \indexset m} \, \distof {v \in \indexset {m^\prime}} = e^{n ( 2 \tilde R - R_\poolre )}.
\end{IEEEeqnarray}
Fix some $\xi$ satisfying
\begin{equation}
R_\poolre - \tilde R - \mu \leq \xi \leq R_\poolre - \tilde R - \mu/2, \label{eq:IDCodeXi}
\end{equation}
and let
\begin{equation}
\kappa_n = e^{n \xi}. \label{eq:IDCodeKappan}
\end{equation}
Observe that
\begin{IEEEeqnarray}{rCl}
\Bigdistof { |\indexset {m,m^\prime}| \geq e^{n ( \tilde R - \mu / 2 ) + \log 2} } & \stackrel{(a)}\leq & \Bigdistof { |\indexset {m,m^\prime}| \geq e^{n ( 2 \tilde R - R_\poolre + \xi) + \log 2} } \\
& \stackrel{(b)}\leq & \Bigdistof { |\indexset {m,m^\prime}| \geq (1 + \kappa_n) \, e^{n ( 2 \tilde R - R_{\poolre} )} } \\
& = & \Biggdistof { \sum_{v \in \setV} \ind {v \in \indexset {m,m^\prime}} \geq (1 + \kappa_n) \, e^{n ( 2 \tilde R - R_\poolre ) }} \\
& \stackrel{(c)}\leq & \exp \bigl\{ - \kappa_n \, e^{n ( 2 \tilde R - R_\poolre ) - \log 3} \bigr\} \\
& \stackrel{(d)}\leq & \exp \bigl\{ - e^{n ( \tilde R - \mu ) - \log 3} \bigr\}, \label{eq:probWrongIDVmhatm}
\end{IEEEeqnarray}
where $(a)$ holds because \eqref{eq:IDCodeXi} implies that $\tilde R - R_\poolre + \xi \leq - \mu / 2$; $(b)$ holds by \eqref{eq:IDCodeKappan} and because \eqref{eq:IDCodeMu} implies that $\mu < R_\poolre - \tilde R$, and hence it follows from \eqref{eq:IDCodeXi} that $\xi > 0$; $(c)$ follows from the multiplicative Chernoff bound \eqref{eq:multChernDeltaGeq1} in Proposition~\ref{pr:multChernoff}; and $(d)$ holds by \eqref{eq:IDCodeKappan} and because \eqref{eq:IDCodeXi} implies that $-\mu \leq \tilde R - R_\poolre + \xi$. The Union-of-Events bound, \eqref{eq:missedIdImSuffLarge}, \eqref{eq:ImSuffSmall}, and \eqref{eq:probWrongIDVmhatm} imply that
\begin{IEEEeqnarray}{rCl}
\bigdistof {\{ \indexset m \}_{m \in \setM} \notin \setG_\mu} & \leq & \card \setM \Bigl( \exp \bigl\{ -e^{n ( \tilde R - \mu ) - \log 2} \bigr\} + \card \setM \exp \bigl\{ - e^{n ( \tilde R - \mu ) - \log 3} \bigr\} \Bigr) \\
& \stackrel{(a)}\rightarrow & 0 \, (n \rightarrow \infty), \label{eq:probEventESmall}
\end{IEEEeqnarray}
where $(a)$ holds because $|\setM| = \exp (\exp (n R))$ and by \eqref{eq:IDCodeMu}.

\section{A Proof of Proposition~\ref{pr:homogIDCodeApprox}}\label{app:homogIDCodeApprox}

Let $\setL = \bigl\{ 0, \ldots, \lfloor e^{n \delta} / 2 \rfloor \bigr\}$, and partition the collection of PMFs $\{ Q_m \}_{m \in \setM}$ into $|\setL|^{|\Gamma^{(n)}|}$ subsets so that two PMFs $Q_m$ and $Q_{m^\prime}$ are in the same subset iff for every $n$-type $P$ on $\setX^n$ there exists an $\ell \in \setL$ for which $$Q_m \bigl( \setT^{(n)}_P \bigr), \, Q_{m^\prime} \bigl( \setT^{(n)}_P \bigr) \in \bigl[ 2 \ell e^{-n \delta}, 2 (\ell + 1) e^{- n \delta} \bigr).$$ Pick a largest subset, say $\setS$, and note that $\setS$ satisfies \eqref{eq:homogIDCodeApproxSizeSetS}:
\begin{IEEEeqnarray}{rCl}
|\setS| & \geq & |\setM| / |\setL|^{|\Gamma^{(n)}|} \\
& \geq & |\setM| / \exp \bigl\{ (1 + n)^{|\setX|} \log ( 1 + e^{ n \delta } / 2 ) \bigr\} \\
& \geq & |\setM| \exp \bigl\{ - e^{ \log (1 + n) (1 + |\setX|) + \log \delta } \bigr\},
\end{IEEEeqnarray}
where the last inequality holds because $e^{n \delta} \geq 2$. Pick $m^\star \in \setS$, and for each $m \in \setS$ define the PMF 
\begin{IEEEeqnarray}{rCl}
Q_m^\prime (\vecx) & = & Q_{m^\star} \bigl( \setT^{(n)}_P \bigr) Q_m^{(n,P)} (\vecx), \quad P \in \Gamma^{(n)}, \, \vecx \in \setT^{(n)}_P.
\end{IEEEeqnarray}
Note that for every $m \in \setS$
\begin{IEEEeqnarray}{rCl}
Q_m^\prime \bigl( \setT^{(n)}_P \bigr) = Q_{m^\star} \bigl( \setT^{(n)}_P \bigr), \quad P \in \Gamma^{(n)},
\end{IEEEeqnarray}
and therefore
\begin{IEEEeqnarray}{rCl}
Q_m^\prime \bigl( X^n \in \{ \vecx \in \setX^n \colon \muti {P_\vecx}{W} \leq R - \epsilon \} \bigr) & = & Q_{m^\star} \bigl( X^n \in \{ \vecx \in \setX^n \colon \muti {P_\vecx}{W} \leq R - \epsilon \} \bigr).
\end{IEEEeqnarray}
Consequently, \eqref{eq:convReducedCodeSatisfiesTypeProperty} implies \eqref{eq:convHomogCodeSatisfiesTypeProperty}. For every $m \in \setS$ we obtain from $m^\star \in \setS$ that
\begin{IEEEeqnarray}{rCl}
Q_m \bigl( \setT^{(n)}_P \bigr) - 2 e^{ - n \delta} < Q_m^\prime \bigl( \setT^{(n)}_P \bigr) < Q_m \bigl( \setT^{(n)}_P \bigr) + 2 e^{ - n \delta}, \quad P \in \Gamma^{(n)}.
\end{IEEEeqnarray}
This implies that for every subset $\setD$ of $\setY^n$
\begin{IEEEeqnarray}{rCl}
\bigl| ( Q_m^\prime W^n ) (\setD) - ( Q_m W^n ) (\setD) \bigr| & \leq & d ( Q_m^\prime W^n, Q_m W^n ) \\
& \stackrel{(a)}\leq & d ( Q_m^\prime, Q_m) \\
& = & \frac{1}{2} \sum_{P \in \Gamma^{(n)}} \sum_{\vecx \in \setT^{(n)}_P} \bigl| Q_{m^\star} \bigl( \setT^{(n)}_P \bigr) - Q_m \bigl( \setT^{(n)}_P \bigr) \bigr| Q_m^{(n,P)} (\vecx) \\
& \leq & \frac{1}{2} \sum_{P \in \Gamma^{(n)}} 2 e^{ - n \delta} \\
& \leq & e^{ - n \delta + \log ( 1 + n ) |\setX|},
\end{IEEEeqnarray}
where $(a)$ follows from the Data-Processing inequality for the Total-Variation distance \cite[Lemma~1]{cannoneronservedio15}. Hence, $\{ Q_m^\prime, \setD_m \}_{m \in \setS}$ is a homogeneous $( n,\setS,\lambda_1^\prime,\lambda_2^\prime )$ ID code for $\channel y x$, where $\setS$ satisfies \eqref{eq:homogIDCodeApproxSizeSetS} and $\lambda_1^\prime, \, \lambda_2^\prime$ are defined in \eqref{bl:homogIDCodeApproxLambdaTilde}.

\section{A Proof of Lemma~\ref{le:rcodApprox}}\label{app:rcodApprox}

Let $g (\cdot)$ be the continuous function that maps every nonnegative real number $u$ to
\begin{IEEEeqnarray}{rCl}
g (u) & = & \begin{cases} - \sqrt{ 2 u } \log \! \sqrt{2 u} &\textnormal{if } u > 0, \\ 0 &\textnormal{if } u = 0, \end{cases} \label{eq:lemma1FunG}
\end{IEEEeqnarray}
and let the function $\rho (\cdot)$ map every nonnegative real number $u$ to
\begin{IEEEeqnarray}{rCl}
\rho (u) & = & 6 u + 2 g (3 u) + \sqrt {3 u} \log |\setY|. \label{eq:lemma1FunRho}
\end{IEEEeqnarray}
There exists a positive constant $\delta_0$, which depends only on $|\setY|$, satisfying $3 \delta_0 < 1/64$ and that $g (\cdot)$ is continuous and strictly increasing on the interval $[ 0, 3 \delta_0 ]$. Because $g (\cdot)$ is continuous and strictly increasing on $[ 0, 3 \delta_0 ]$, $\rho (\cdot)$ is continuous and strictly increasing on $[0,\delta_0]$. Fix $\delta \in ( 0, \delta_0 ]$ and $\epsilon \in ( 0, 1 )$. Let $\eta_0$ be the smallest positive integer satisfying that for all $n \geq \eta_0$
\begin{subequations}\label{bl:lemma1NProp1}
\begin{IEEEeqnarray}{rCl}
e^{ - 3 n \delta + \log (1 + n) |\setX| \, |\setY| } & < & e^{ - 3 n \delta / 2}, \label{eq:lemma1NProp11} \\
e^{ - 3 n \delta / 2 + \log 2} & \leq & e^{ - n \delta}, \label{eq:lemma1NProp12}
\end{IEEEeqnarray}
\end{subequations}
and
\begin{IEEEeqnarray}{l}
2 \exp \bigl\{ - \epsilon^2 e^{ 3 n \delta - \log 3 } + n \log |\setY| + \log (1 + n) |\setX| \, |\setY| \bigr\} + \exp \bigl\{ - 3 n \delta / 2 + 2 \log (1 + n) |\setX| \, |\setY| \bigr\} < 1. \label{eq:lemma1NProp2}
\end{IEEEeqnarray}
Fix a blocklength~$n \geq \eta_0$, an $n$-type $P$ on $\setX$, a PMF $Q$ on $\setT^{(n)}_P \subseteq \setX^n$, a nonnegative real number $R \geq I (P,W) + \rho (\delta)$, and $L = \lceil e^{n R} \rceil$. We next show that there exists an $L$-type $Q^\prime$ on $\setT^{(n)}_P$ that satisfies \eqref{bl:lemma1} for every subset $\setD$ of $\setY^n$. The proof is essentially that of \cite[Lemma~1]{hanverdu92}:

\subparagraph*{Canonical Decomposition into Equitype Channels:} For every transition law $V (y|x)$ from $\setX$ to $\setY$ and every $n$-tuple $\vecx \in \setX^n$ let $\setT_{P \times V}^{(n)} (\vecx)$ denote the set of $n$-tuples $\vecy \in \setY^n$ for which $(\vecx, \vecy)$ has empirical type $P \times V$, so $$\setT_{P \times V}^{(n)} (\vecx) = \bigl\{ \vecy \in \setY^n \colon (\vecx, \vecy) \in \setT^{(n)}_{P \times V} \bigr\}.$$ Note that $\bigl| \setT_{P \times V}^{(n)} (\vecx) \bigr|$ is the same for all $\vecx \in \setT_P^{(n)}$, and denote it $L^{(n)}_{V|P}$, so $$L^{(n)}_{V|P} = \bigl| \setT_{P \times V}^{(n)} (\vecx) \bigr|, \quad \vecx \in \setT^{(n)}_P.$$ Let $\Lambda^{(n)}_P$ denote the set of all the transition laws $V (y|x)$ from $\setX$ to $\setY$ satisfying $L^{(n)}_{V|P} > 0$ and $V (y|x) = \channel y x$ whenever $P (x) = 0$, so $$\Lambda^{(n)}_P = \bigl\{ V \in \mathscr V (\setY|\setX) \colon L^{(n)}_{V|P} > 0 \textnormal{ and } V (y|x) = \channel y x, \,\, \forall \,  ( x, y ) \in \setX \times \setY \textnormal { s.t. } P (x) = 0 \bigr\},$$ where $\mathscr V (\setY|\setX)$ denotes the set of all transition laws from $\setX$ to $\setY$. Define for every $V \in \Lambda^{(n)}_P$ the transition law
\begin{IEEEeqnarray}{rCl}
W^{(n)}_{V|P} (\vecy | \vecx) & = & \begin{cases} \frac{1}{L^{(n)}_{V|P}} &\textnormal{if } \vecx \in \setT^{(n)}_P \textnormal{ and } \vecy \in \setT^{(n)}_{P \times V} (\vecx), \\ 0 &\textnormal{otherwise}. \end{cases}
\end{IEEEeqnarray}
Following the terminology of \cite{hanverdu92} we call $W^{(n)}_{V|P} (\vecy | \vecx)$ an \emph{equitype channel}, because it connects inputs of type $P$ to outputs of type $P V$, and because all positive transition probabilities are the same. The equitype channels $W^{(n)}_{V|P}, \,\, V \in \Lambda^{(n)}_P$ are distinct, because each $V \in \Lambda^{(n)}_P$ satisfies $V (y|x) = \channel y x$ whenever $P (x) = 0$.

Since $W^n (\vecy|\vecx)$ depends on the input sequence $\vecx \in \setX^n$ and the output sequence $\vecy \in \setY^n$ only via the type of $\vecx$ and the conditional type of $\vecy$ given $\vecx$, we can define $$c^{(n)}_{V|P} = W^n \bigl(\setT^{(n)}_{P \times V} (\vecx) \bigl| \vecx \bigr), \quad \vecx \in \setT^{(n)}_P$$ to obtain for every $\vecx \in \setT^{(n)}_P$ and every $\vecy \in \setT^{(n)}_{P \times V} (\vecx)$
\begin{IEEEeqnarray}{rCl}
W^n (\vecy | \vecx) & = & \frac{c^{(n)}_{V|P}}{L^{(n)}_{V|P}} \\
& = & c^{(n)}_{V|P} W^{(n)}_{V|P} (\vecy | \vecx).
\end{IEEEeqnarray}
Note that
\begin{IEEEeqnarray}{rCl}
\sum_{V \in \Lambda^{(n)}_P} c^{(n)}_{V|P} & = & 1.
\end{IEEEeqnarray}
Since for every pair $(\vecx, \vecy) \in \setT^{(n)}_P \times \setY^n$ there exists exactly one $V \in \Lambda^{(n)}_P$ for which $\vecy \in \setT_{P \times V}^{(n)} (\vecx)$, we can write
\begin{IEEEeqnarray}{rCl}
W^n (\vecy | \vecx) & = & \sum_{V \in \Lambda^{(n)}_P} c^{(n)}_{V|P} W^{(n)}_{V|P} (\vecy | \vecx), \quad (\vecx, \vecy) \in \setT^{(n)}_P \times \setY^n. \label{eq:canonicalDecomp}
\end{IEEEeqnarray}
Following the terminology of \cite{hanverdu92} we call this the \emph{canonical decomposition into equitype channels} of the transition law $W^n (\vecy|\vecx)$ from $\setT^{(n)}_P$ to $\setY^n$. The canonical decomposition is useful, because it allows us to first focus attention on each equitype channel separately, and to then take the weighted average \eqref{eq:canonicalDecomp} of the resulting approximations.

\subparagraph*{Estimating the Probability of Inverse Images:} For every $V \in \Lambda^{(n)}_P$ the subset of $n$-tuples $\vecx \in \setT^{(n)}_P$ that are connected to a specific $\vecy \in \setY^n$ by the equitype channel $W^{(n)}_{V|P}$ is denoted $H^{(n)}_{P \times V} (\vecy)$, so $$H^{(n)}_{P \times V} (\vecy) = \bigl\{ \vecx \in \setT^{(n)}_P \colon W^{(n)}_{V|P} (\vecy | \vecx) > 0 \bigr\}.$$ Note that for every PMF $\tilde Q$ on $\setT^{(n)}_P$
\begin{IEEEeqnarray}{rCl}
\bigl( \tilde Q W^{(n)}_{V|P} \bigr) (\vecy) = \frac{\tilde Q \bigl( H^{(n)}_{P \times V} (\vecy) \bigr)}{L^{(n)}_{V|P}}. \label{eq:lemma1PrEst}
\end{IEEEeqnarray}

\begin{lemma}\cite[Lemma~2]{hanverdu92} \label{le:2HanVerdu92}
For every $V \in \Lambda^{(n)}_P$ and every $\delta^\prime > 0$ define $$G^{(n)}_{\delta^\prime} (V|P) = \Bigl\{ \vecy \in \setY^n \colon Q \bigl( H^{(n)}_{P \times V} (\vecy) \bigr) \geq e^{ - n ( \muti P V + \delta^\prime ) } \Bigr\}.$$ Then, for every $n \in \naturals$
\begin{IEEEeqnarray}{rCl}
\bigl( Q W^{(n)}_{V|P} \bigr) \bigl( G^{(n)}_{\delta^\prime} (V|P) \bigr) \geq 1 - e^{-n \delta^\prime + \log (1 + n) |\setX| \, |\setY| }.
\end{IEEEeqnarray}
\end{lemma}

\subparagraph*{Channel Clipping:} For every $\delta^\prime > 0$ denote $$\Lambda^{(n)}_{\delta^\prime} (P) = \bigl\{ V \in \Lambda^{(n)}_P \colon D (V || W | P) \leq \delta^\prime \bigr\},$$ and define the transition law $W^{(n)}_{P, \delta^\prime}$ from $\setT^{(n)}_P$ to $\setY^n$ by
\begin{IEEEeqnarray}{rCl}
W^{(n)}_{P, \delta^\prime} (\vecy|\vecx) = \sum_{V \in \Lambda^{(n)}_P} \bar c^{(n)}_{V|P} W^{(n)}_{V|P} (\vecy|\vecx), \label{eq:lemma1ClippedCh}
\end{IEEEeqnarray}
where
\begin{IEEEeqnarray}{rCl}
\bar c^{(n)}_{V|P} = \begin{cases} \frac{c^{(n)}_{V|P}}{\sum_{V^\prime \in \Lambda^{(n)}_{\delta^\prime} (P)} c^{(n)}_{V^\prime|P}} &\textnormal{if } V \in \Lambda^{(n)}_{\delta^\prime} (P), \\ 0 &\textnormal{otherwise.} \end{cases}
\end{IEEEeqnarray}
As the following lemma shows, $W^{(n)}_{P, \delta^\prime}$ closely approximates the transition law $W^n$ from $\setT^{(n)}_P$ to $\setY^n$:

\begin{lemma}\cite[Lemma~3]{hanverdu92} \label{le:3HanVerdu92}
For every $\delta^\prime > 0$, every $n$-tuple $\vecx \in \setT^{(n)}_P$, and every subset $\setD$ of $\setY^n$
\begin{subequations}\label{bl:lemma3}
\begin{IEEEeqnarray}{rCl}
W^n (\setD|\vecx) & \geq & \bigl( 1 - e^{ - n \delta^\prime + \log (1 + n) |\setX| \, |\setY|  } \bigr) W^{(n)}_{P, \delta^\prime} (\setD | \vecx), \\
W^n (\setD|\vecx) & \leq & W^{(n)}_{P, \delta^\prime} (\setD | \vecx) + e^{ - n \delta^\prime + \log (1 + n) |\setX| \, |\setY| }.
\end{IEEEeqnarray}
\end{subequations}
\end{lemma}

As we argue next, Lemma~\ref{le:3HanVerdu92} reduces the proof to verifying that, whenever $n \geq \eta_0$, there exists an $L$-type $Q^\prime$ on $\setT^{(n)}_P$ that satisfies for $\delta^\prime = 3 \delta$ and for every subset $\setD$ of $\setY^n$
\begin{subequations}\label{bl:simplificationLemma1}
\begin{IEEEeqnarray}{rCl}
\bigl( Q^\prime W^{(n)}_{P, \delta^\prime} \bigr) \bigl( \setD \bigr) & \leq & ( 1 + \epsilon ) \bigl( Q W^{(n)}_{P, \delta^\prime} \bigr) ( \setD ) + e^{ -n \delta^\prime / 2 }, \\
\bigl( Q^\prime W^{(n)}_{P, \delta^\prime} \bigr) \bigl( \setD \bigr) & \geq & ( 1 - \epsilon ) \bigl( Q W^{(n)}_{P, \delta^\prime} \bigr) ( \setD ) - e^{ -n \delta^\prime / 2 }.
\end{IEEEeqnarray}
\end{subequations}
Indeed, \eqref{bl:lemma3} and \eqref{bl:simplificationLemma1} imply that
\begin{IEEEeqnarray}{rCl}
( Q^\prime W^n ) ( \setD ) & \leq & \bigl( Q^\prime W^{(n)}_{P, \delta^\prime} \bigr) ( \setD ) + e^{ - n \delta^\prime + \log (1 + n) |\setX| \, |\setY| }\\
& \leq & ( 1 + \epsilon ) \bigl( Q W^{(n)}_{P, \delta^\prime} \bigr) ( \setD ) + e^{ -n \delta^\prime / 2 } + e^{ - n \delta^\prime + \log (1 + n) |\setX| \, |\setY| }\\
& \leq & \frac{1 + \epsilon}{1 - e^{ - n \delta^\prime + \log (1 + n) |\setX| \, |\setY| }} ( Q W^n ) ( \setD ) + e^{ -n \delta^\prime / 2 } + e^{ - n \delta^\prime + \log (1 + n) |\setX| \, |\setY| }, \label{eq:simplificationLemma1Imp1}
\end{IEEEeqnarray}
and
\begin{IEEEeqnarray}{rCl}
( Q^\prime W^n ) ( \setD ) & \geq & \bigl( 1 - e^{ - n \delta^\prime + \log (1 + n) |\setX| \, |\setY| } \bigr) \bigl( Q^\prime W^{(n)}_{P, \delta^\prime} \bigr) ( \setD ) \\
& \geq & ( 1 - \epsilon ) \bigl( 1 - e^{ - n \delta^\prime + \log (1 + n) |\setX| \, |\setY| } \bigr) \bigl( Q W^{(n)}_{P, \delta^\prime} \bigr) ( \setD ) - \bigl( 1 - e^{ - n \delta^\prime + \log (1 + n) |\setX| \, |\setY| } \bigr) e^{ -n \delta^\prime / 2 } \\
& \geq & ( 1 - \epsilon ) \bigl( 1 - e^{ - n \delta^\prime + \log (1 + n) |\setX| \, |\setY| } \bigr) ( Q W^n ) ( \setD ) \nonumber \\
& & - ( 1 - \epsilon ) \bigl( 1 - e^{ - n \delta^\prime + \log (1 + n) |\setX| \, |\setY| } \bigr) e^{ -n \delta^\prime + \log (1 + n) |\setX| \, |\setY| } - \bigl( 1 - e^{ - n \delta^\prime + \log (1 + n) |\setX| \, |\setY| } \bigr) e^{ -n \delta^\prime / 2 }. \label{eq:simplificationLemma1Imp2}
\end{IEEEeqnarray}
For $\delta^\prime = 3 \delta$ we obtain from \eqref{bl:lemma1NProp1} (which holds because $n \geq \eta_0$) that
\begin{IEEEeqnarray}{rCl}
e^{ - n \delta} & \geq & e^{ -n \delta^\prime + \log (1 + n) |\setX| \, |\setY| } + e^{ -n \delta^\prime / 2 },
\end{IEEEeqnarray}
and hence that
\begin{subequations}
\begin{IEEEeqnarray}{rCl}
e^{ - n \delta} & > & e^{ - n \delta^\prime + \log (1 + n) |\setX| \, |\setY| }, \\
e^{ - n \delta} & > & ( 1 - \epsilon ) \bigl( 1 - e^{ - n \delta^\prime  + \log (1 + n) |\setX| \, |\setY|} \bigr) e^{ -n \delta^\prime + \log (1 + n) |\setX| \, |\setY| } + \bigl( 1 - e^{ - n \delta^\prime + \log (1 + n) |\setX| \, |\setY| } \bigr) e^{ -n \delta^\prime / 2 }.
\end{IEEEeqnarray}
\end{subequations}
Consequently, \eqref{bl:lemma1} follows from \eqref{eq:simplificationLemma1Imp1} and \eqref{eq:simplificationLemma1Imp2}. In the following, we let $\delta^\prime = 3 \delta$ and conclude the proof by showing that there exists an $L$-type $Q^\prime$ on $\setT^{(n)}_P$ that satisfies \eqref{bl:simplificationLemma1} for every subset $\setD$ of $\setY^n$.

\subparagraph*{Required Fineness of Approximations for the Clipped Channel:} For every $V \in \Lambda^{(n)}_{\delta^\prime} (P)$ we can upper-bound $\muti P V$ in terms of $\muti P W$:

\begin{lemma}\cite[Lemma~4]{hanverdu92}
If $\sqrt{ D ( V || W | P ) } < 1/8$, then
\begin{IEEEeqnarray}{rCl}
\bigl| \muti P V - \muti P W \bigr| \leq 2 g \bigl( D ( V || W | P ) \bigr) + \sqrt{ D ( V || W | P ) } \log |\setY|,
\end{IEEEeqnarray}
where $g (\cdot)$ is defined in \eqref{eq:lemma1FunG}.
\end{lemma}

For every $V \in \Lambda^{(n)}_{\delta^\prime} (P)$ the lemma, the fact that $\delta^\prime = 3 \delta$ satisfies $\sqrt{ \delta^\prime } < 1 / 8$, and the fact that $g (\cdot)$ is strictly increasing on $[0,\delta^\prime]$ imply that
\begin{IEEEeqnarray}{rCl}
\muti P V + 2 \delta^\prime \leq \muti P W + \rho (\delta). \label{eq:lemma1MutiPV}
\end{IEEEeqnarray}
Hence, if $\vecy \in G^{(n)}_{\delta^\prime} (V|P)$ for some $V \in \Lambda^{(n)}_{\delta^\prime} (P)$, then the definitions of $G^{(n)}_{\delta^\prime} (V|P)$ and $H^{(n)}_{P \times V} (\vecy)$ imply that
\begin{IEEEeqnarray}{rCl}
Q \bigl( H^{(n)}_{P \times V} (\vecy) \bigr) & \geq & e^{ - n ( \muti P W + \rho (\delta) - \delta^\prime ) }. \label{eq:lemma1MutiPVQHn}
\end{IEEEeqnarray}

\subparagraph*{The $L$-type Approximation $Q^\prime$:} We next show by random construction that the desired $L$-type $Q^\prime$ on $\setT^{(n)}_P$ exists. Draw $L$ $n$-tupes $\sim Q$ independently and place them in a pool $\pool$. Note that $\pool \subset \setT^{(n)}_P$. Index the $n$-tuples in the pool by the elements of a size-$L$ set $\setV$, e.g., $\{ 1, \ldots, L \}$, and denote by $\poolel v$ the $n$-tuple in $\pool$ that is indexed by $v \in \setV$. Define the $L$-type $\bm Q^\prime$ on $\setT^{(n)}_P$ by
\begin{IEEEeqnarray}{rCl}
\bm Q^\prime (\vecx) & = & \frac{1}{L} \sum_{v \in \setV} \ind {\vecx = \poolel v}, \quad \vecx \in \setX^n. \label{eq:lemma1TildeQ}
\end{IEEEeqnarray}

\begin{lemma}\cite[essentially Lemma~5]{hanverdu92}\label{le:lemma5}
With positive probability the $L$-type $\bm Q^\prime$ on $\setT^{(n)}_P$ satisfies for every $V \in \Lambda^{(n)}_{\delta^\prime} (P)$
\begin{subequations}\label{bl:lemma5}
\begin{IEEEeqnarray}{rCl}
\bm Q^\prime \bigl( H^{(n)}_{P \times V} (\vecy) \bigr) & < & (1 + \epsilon) Q \bigl( H^{(n)}_{P \times V} (\vecy) \bigr), \quad \vecy \in G^{(n)}_{\delta^\prime} (V|P), \label{bl:eqlemma51} \\
\bm Q^\prime \bigl( H^{(n)}_{P \times V} (\vecy) \bigr) & > & (1 - \epsilon) Q \bigl( H^{(n)}_{P \times V} (\vecy) \bigr), \quad \vecy \in G^{(n)}_{\delta^\prime} (V|P), \label{bl:eqlemma52}  \\
\bigl( \bm Q^\prime W^{(n)}_{V|P} \bigr) \bigl( \setY^n \setminus G^{(n)}_{\delta^\prime} ( V | P ) \bigr) & < & e^{ - n \delta^\prime / 2}. \label{bl:eqlemma53} 
\end{IEEEeqnarray}
\end{subequations}
\end{lemma}

\begin{proof}
We use the Union-of-Events bound to show that with positive probability $\bm Q^\prime$ satisfies \eqref{bl:lemma5} for every $V \in \Lambda^{(n)}_{\delta^\prime} (P)$. We begin with \eqref{bl:eqlemma51} and \eqref{bl:eqlemma52}. For every $V \in \Lambda^{(n)}_{\delta^\prime} (P)$ and $\vecy \in G^{(n)}_{\delta^\prime} (V | P)$
\begin{IEEEeqnarray}{l}
\Bigdistof { \bm Q^\prime \bigl( H^{(n)}_{P \times V} (\vecy) \bigr) \geq (1 + \epsilon) Q \bigl( H^{(n)}_{P \times V} (\vecy) \bigr) } \\
\quad \stackrel{(a)}= \Biggdistof { \frac{1}{L} \sum_{v \in \setV} \ind { \poolel v \in H^{(n)}_{P \times V} (\vecy) } \geq (1 + \epsilon) Q \bigl( H^{(n)}_{P \times V} (\vecy) \bigr) } \\
\quad \stackrel{(b)}\leq \exp \bigl\{ - \epsilon^2 Q \bigl( H^{(n)}_{P \times V} (\vecy) \bigr) L / 3 \bigr\} \\
\quad \stackrel{(c)}\leq \exp \bigl\{ - \epsilon^2 e^{ n \delta^\prime - \log 3 } \bigr\},
\end{IEEEeqnarray}
where $(a)$ is due to \eqref{eq:lemma1TildeQ}; $(b)$ follows from the multiplicative Chernoff bound \eqref{eq:multChernDeltaSm1La} in Proposition~\ref{pr:multChernoff}; and $(c)$ holds by \eqref{eq:lemma1MutiPVQHn} and because $L \geq e^{n ( \muti P W + \rho (\delta) ) }$. By the Union-of-Events bound and because $\bigl| G^{(n)}_{\delta^\prime} (V|P) \bigr| \leq |\setY|^n$
\begin{IEEEeqnarray}{l}
\Bigdistof { \exists \, \vecy \in G^{(n)}_{\delta^\prime} ( V | P ) \colon \bm Q^\prime \bigl( H^{(n)}_{P \times V} (\vecy) \bigr) \geq (1 + \epsilon) Q \bigl( H^{(n)}_{P \times V} (\vecy) \bigr) } \\
\leq \exp \bigl\{ - \epsilon^2 e^{ n \delta^\prime - \log 3 } + n \log \! |\setY| \bigr\}. \label{eq:lemma51}
\end{IEEEeqnarray}
Similarly, the multiplicative Chernoff bound \eqref{eq:multChernDeltaSm1Sm} in Proposition~\ref{pr:multChernoff} and the Union-of-Events bound imply that for every $V \in \Lambda^{(n)}_{\delta^\prime} (P)$
\begin{IEEEeqnarray}{l}
\Bigdistof { \exists \, \vecy \in G^{(n)}_{\delta^\prime} (V | P) \colon \bm Q^\prime \bigl( H^{(n)}_{P \times V} (\vecy) \bigr) \leq (1 - \epsilon) Q \bigl( H^{(n)}_{P \times V} (\vecy) \bigr) } \\
\quad \leq \exp \bigl\{ - \epsilon^2 e^{ n \delta^\prime - \log 2 } + n \log \! |\setY| \bigr\}. \label{eq:lemma52}
\end{IEEEeqnarray}
As to \eqref{bl:eqlemma53}, for every $V \in \Lambda^{(n)}_{\delta^\prime} (P)$
\begin{IEEEeqnarray}{l}
\BigEx {}{ \bigl( \bm Q^\prime W^{(n)}_{V | P} \bigr) \bigl( \setY^n \setminus G^{(n)}_{\delta^\prime} (V | P) \bigr)} \\
\quad = \sum_{\vecx \in \setX^n} \sum_{\vecy \in \setY^n \setminus G^{(n)}_{\delta^\prime} (V | P)} \frac{1}{L} \sum_{v \in \setV} \bigEx {}{ \ind {\vecx = \poolel v} } W^{(n)}_{V|P} (\vecy | \vecx) \\
\quad = \sum_{\vecx \in \setX^n} \sum_{\vecy \in \setY^n \setminus G^{(n)}_{\delta^\prime} (V | P)} \frac{1}{L} \sum_{v \in \setV} Q (\vecx) W^{(n)}_{V|P} (\vecy | \vecx) \\
\quad = \bigl( Q W^{(n)}_{V | P} \bigr) \bigl( \setY^n \setminus G^{(n)}_{\delta^\prime} (V | P) \bigr) \\
\quad \leq e^{-n \delta^\prime + \log (1 + n) |\setX| \, |\setY| },
\end{IEEEeqnarray}
where the last inequality is due to Lemma~\ref{le:2HanVerdu92}. Hence, Markov's inequality implies that
\begin{IEEEeqnarray}{rCl}
\Bigdistof { \bigl( \bm Q^\prime W^{(n)}_{V|P} \bigr) \bigl( \setY^n \setminus G^{(n)}_{\delta^\prime} (V | P) \bigr) \geq e^{-n \delta^\prime / 2}} & \leq & e^{-n \delta^\prime / 2 + \log (1 + n) |\setX| \, |\setY| }. \label{eq:lemma53}
\end{IEEEeqnarray}
Because $\bigl| \Lambda^{(n)}_{\delta^\prime} (P) \bigr| \leq (1 + n)^{|\setX| \, |\setY|}$ and by the Union-of-Events bound, \eqref{eq:lemma51}, \eqref{eq:lemma52}, and \eqref{eq:lemma53}, the probability that there exists a $V \in \Lambda^{(n)}_{\delta^\prime} (P)$ for which $\bm Q^\prime$ does not satisfy \eqref{bl:lemma5} is upper-bounded by
\begin{IEEEeqnarray}{l}
\exp \bigl\{ - \epsilon^2 e^{ n \delta^\prime - \log 3 } + n \log |\setY| + \log (1 + n) |\setX| \, |\setY| \bigr\} + \exp \bigl\{ - \epsilon^2 e^{ n \delta^\prime - \log 2 } + n \log |\setY| + \log (1 + n) |\setX| \, |\setY| \bigr\} \nonumber \\
\quad + \exp \bigl\{-n \delta^\prime / 2 + 2 \log (1 + n) |\setX| \, |\setY| \bigr\} < 1,
\end{IEEEeqnarray}
where the inequality holds because $\delta^\prime = 3 \delta$, by \eqref{eq:lemma1NProp2}, and because $n \geq \eta_0$.
\end{proof}

Fix a realization $Q^\prime$ of the random $L$-type $\bm Q^\prime$ on $\setT^{(n)}_P$ that satisfies \eqref{bl:lemma5} for all $V \in \Lambda^{(n)}_{\delta^\prime} (P)$. (By Lemma~\ref{le:lemma5} such a realization must exist.)

\subparagraph*{Approximation of $Q W^{(n)}_{P,\delta^\prime}$ by $Q^\prime W^{(n)}_{P,\delta^\prime}$:} It remains to show that the $L$-type $Q^\prime$ on $\setT^{(n)}_P$ satisfies \eqref{bl:simplificationLemma1}. For every $V \in \Lambda^{(n)}_{\delta^\prime} (P)$ and $\vecy \in G^{(n)}_{\delta^\prime} (V | P)$
\begin{IEEEeqnarray}{rCl}
\bigl( Q^\prime W^{(n)}_{V|P} \bigr) ( \vecy ) & \stackrel{(a)}= & \frac{ Q^\prime \bigl( H^{(n)}_{P \times V} (\vecy) \bigr) }{L^{(n)}_{V | P}} \\
& \stackrel{(b)}< & (1 + \epsilon) \frac{ Q \bigl( H^{(n)}_{P \times V} (\vecy) \bigr) }{L^{(n)}_{V | P}} \\
& \stackrel{(c)}= & (1 + \epsilon) \bigl( Q W^{(n)}_{V | P} \bigr) ( \vecy ), \label{eq:lemma1Approx1}
\end{IEEEeqnarray}
where $(a)$ and $(c)$ follow from \eqref{eq:lemma1PrEst}; and where $(b)$ holds because $Q^\prime$ satisfies \eqref{bl:lemma5}. For every $V \in \Lambda^{(n)}_{\delta^\prime} (P)$ and subset $\setD$ of $\setY^n$ we thus have
\begin{IEEEeqnarray}{rCl}
\bigl( Q^\prime W^{(n)}_{V | P} \bigr) ( \setD ) & \stackrel{(a)}= & \bigl( Q^\prime W^{(n)}_{V | P} \bigr) \bigl( \setD \cap G^{(n)}_{\delta^\prime} (V | P) \bigr) + \bigl( Q^\prime W^{(n)}_{V | P} ) \bigl( \setD \cap \bigl( \setY^n \setminus G^{(n)}_{\delta^\prime} (V | P) \bigr) \bigr) \\
& \stackrel{(b)}\leq & \bigl( Q^\prime W^{(n)}_{V|P} \bigr) \bigl( \setD \cap G^{(n)}_{\delta^\prime} (V | P) \bigr) + \bigl( Q^\prime W^{(n)}_{V | P} \bigr) \bigl( \setY^n \setminus G^{(n)}_{\delta^\prime} (V | P) \bigr) \\
& \stackrel{(c)}\leq  & (1 + \epsilon) \bigl( Q W^{(n)}_{V | P} \bigr) \bigl( \setD \cap G^{(n)}_{\delta^\prime} (V | P) \bigr) + e^{ - n \delta^\prime / 2} \\
&\stackrel{(d)}\leq & (1 + \epsilon) \bigl( Q W^{(n)}_{V | P} \bigr) ( \setD ) + e^{ - n \delta^\prime / 2}, \label{eq:lemma1Approx2}
\end{IEEEeqnarray}
where $(a)$ follows from the law of total probability; $(b)$ and $(d)$ are due to the monotonicity of probability; and $(c)$ holds by \eqref{eq:lemma1Approx1} and because $Q^\prime$ satisfies \eqref{bl:lemma5}. Similarly,
\begin{IEEEeqnarray}{rCl}
(1 - \epsilon) \bigl( Q W^{(n)}_{V|P} \bigr) ( \setD ) & \stackrel{(a)}\leq & \bigl( Q^\prime W^{(n)}_{V | P} \bigr) ( \setD ) + (1 - \epsilon) \bigl( Q W^{(n)}_{V | P} \bigr) \bigl(  \setY^n \setminus G^{(n)}_{\delta^\prime} (V|P) \bigr) \\
& \stackrel{(b)}\leq & \bigl( Q^\prime W^{(n)}_{V|P} \bigr) ( \setD ) + (1 - \epsilon) e^{-n \delta^\prime + \log (1 + n) |\setX| |\setY| } \\
& \stackrel{(c)}\leq & \bigl( Q^\prime W^{(n)}_{V|P} \bigr) ( \setD ) + e^{ - n \delta^\prime / 2}, \label{eq:lemma1Approx3}
\end{IEEEeqnarray}
where $(a)$ follows from the law of total probability, the monotonicity of probability, and the fact that $Q^\prime$ satisfies \eqref{bl:lemma5}; $(b)$ is due to Lemma~\ref{le:2HanVerdu92}; and $(c)$ holds because $\delta^\prime = 3 \delta$ and by \eqref{eq:lemma1NProp11} (which holds because $n \geq \eta_0$). On account of \eqref{eq:lemma1ClippedCh}, we can now conclude the proof of \eqref{bl:simplificationLemma1} by computing the weighted average of \eqref{eq:lemma1Approx2} and \eqref{eq:lemma1Approx3} w.r.t.\ $V \in \Lambda^{(n)}_{\delta^\prime} (P)$ and with the weights being $\bigl\{ \bar c^{(n)}_{V|P} \bigr\}_{ V \in \Lambda^{(n)}_{\delta^\prime} (P) }$.

\section{A Proof of Theorem~\ref{th:obBC3Rec}}\label{app:obBC3Rec}

We prove the following strong converse:

\begin{claim}\label{cl:toShowBC3RecConv}
For every rate-triple $( R_1, R_2, R_3 )$, every positive constants $$\lambda_1^{(k)}, \, \lambda_2^{(k)}, \quad k \in \{ 1,2,3 \}$$ satisfying
\begin{IEEEeqnarray}{rCl}
\sum_{k = 1}^3 \Bigl( \lambda_1^{(k)} + \lambda_2^{(k)} \Bigr) < 1, \label{eq:sumMissWrongBC3RecSm1}
\end{IEEEeqnarray}
and every $\epsilon > 0$ there exists some $\eta_0 \in \naturals$ so that, for every blocklength~$n \geq \eta_0$, every size-$\exp (\exp (n R_1))$ set $\setM_1$ of possible ID messages for Receiver~$1$, every size-$\exp (\exp (n R_2))$ set $\setM_2$ of possible ID messages for Receiver~$2$, and every size-$\exp (\exp (n R_3))$ set $\setM_3$ of possible ID messages for Receiver~$3$, a necessary condition for an $\bigl( n, \{ \setM_k, \lambda_1^{(k)}, \lambda_2^{(k)} \}_{k \in \{ 1,2,3 \}} \bigr)$ ID code for the BC $\channel {y_1,y_2,y_3} x$ to exist is that for some PMF $P$ on $\setX$
\begin{IEEEeqnarray}{rCl}
R_k &< \muti {P}{W_k} + \epsilon, \, \forall \, k \in \{ 1,2,3 \}. \label{eq:converseBC3Rec}
\end{IEEEeqnarray}
\end{claim}

\begin{proof}
The proof is similar to that of Claim~\ref{cl:toShowIDBCConv}. Fix $\kappa^{(1)}, \, \kappa^{(2)}, \, \kappa^{(3)} > 0$ that satisfy
\begin{subequations}
\begin{IEEEeqnarray}{rCl}
\lambda_1^{(k)} + \lambda_2^{(k)} & < & \kappa^{(k)}, \, \forall \, k \in \{ 1,2,3 \}, \\
\sum_{k = 1}^3 \kappa^{(k)} & < & 1.
\end{IEEEeqnarray}
\end{subequations}
(This is possible because of \eqref{eq:sumMissWrongBC3RecSm1}.) By Lemma~\ref{le:avgDistWeightOnTypes} there must exist some $\eta_0^\prime \in \naturals$ so that, for every blocklength~$n \geq \eta_0^\prime$, every size-$\exp (\exp (n R_1))$ set $\setM_1$ of possible ID messages for Receiver~$1$, every size-$\exp (\exp (n R_2))$ set $\setM_2$ of possible ID messages for Receiver~$2$, and every size-$\exp (\exp (n R_3))$ set $\setM_3$ of possible ID messages for Receiver~$3$, the following is necessary for a collection of tuples $$\bigl\{ Q_{m_1,m_2,m_3}, \setD_{m_1}, \setD_{m_2}, \setD_{m_3} \bigr\}_{(m_1,m_2,m_3) \in \setM_1 \times \setM_2 \times \setM_3}$$ to be an $\bigl( n, \{ \setM_k,  \lambda^{(k)}_1, \lambda^{(k)}_2 \}_{k \in \{ 1,2,3 \}} \bigr)$ ID code for the BC $\channel {y_1,y_2,y_3} x$: the mixture PMFs on $\setX^n$
\begin{subequations}
\begin{IEEEeqnarray}{rCl}
Q_{m_1} & = & \frac{1}{|\setM_2| \, |\setM_3|} \sum_{m_2, m_3} Q_{m_1,m_2,m_3}, \quad m_1 \in \setM_1, \\
Q_{m_2} & = & \frac{1}{|\setM_1| \, |\setM_3|} \sum_{m_1, m_3} Q_{m_1,m_2,m_3}, \quad m_2 \in \setM_2, \\
Q_{m_3} & = & \frac{1}{|\setM_1| \, |\setM_2|} \sum_{m_1, m_2} Q_{m_1,m_2,m_3}, \quad m_3 \in \setM_3, \\
Q & = & \frac{1}{|\setM_1| \, |\setM_2| \, |\setM_3|} \sum_{m_1, m_2, m_3} Q_{m_1,m_2,m_3}
\end{IEEEeqnarray}
\end{subequations}
satisfy
\begin{IEEEeqnarray}{l}
Q \bigl( X^n \in \{ \vecx \in \setX^n \colon \muti {P_\vecx}{W_k} > R_k - \epsilon \} \bigr) \nonumber \\
\quad = \frac{1}{| \setM_k |} \sum_{m_k \in \setM_k} Q_{m_k} \bigl( X^n \in \{ \vecx \in \setX^n \colon \muti {P_\vecx}{W_k} > R_k - \epsilon \} \bigr) \\
\quad \geq 1 - \kappa^{(k)} - \exp \bigl\{ e^{n ( R_k - \epsilon )} \bigr\} / \exp \bigl\{ e^{n R_k} \bigr\}, \quad k \in \{ 1,2,3 \}. \label{eq:condSinceWkIDCode}
\end{IEEEeqnarray}
The Union-of-Events bound and \eqref{eq:condSinceWkIDCode} imply that
\begin{IEEEeqnarray}{l}
Q \bigl( X^n \in \bigl\{ \vecx \in \setX^n \colon \muti {P_\vecx}{W_k} > R_k - \epsilon, \, \forall \, k \in \{ 1, 2, 3 \} \bigr\} \bigr) \nonumber \\
\quad \geq 1 - \sum_{k = 1}^3 \Bigl( \kappa^{(k)} + \exp \bigl\{ e^{n ( R_k - \epsilon / 2 )} \bigr\} / \exp \bigl\{ e^{n R_k} \bigr\} \Bigr).  \label{eq:BCExistsPMFGoodForAllThree}
\end{IEEEeqnarray}
Now let $\eta_0$ be the smallest integer $n \geq \eta_0^\prime$ for which the RHS of \eqref{eq:BCExistsPMFGoodForAllThree} is positive (such an $n$ must exist, because $\epsilon > 0$ and $\sum_{k = 1}^3 \kappa^{(k)} < 1$). Then, for every blocklength $n \geq \eta_0$ a necessary condition for \eqref{eq:BCExistsPMFGoodForAllThree} to hold is that for some PMF $P$ on $\setX$ \eqref{eq:converseBC3Rec} holds, and hence Claim~\ref{cl:toShowBC3RecConv} follows.
\end{proof}

\section{A Proof of Theorem~\ref{th:ibBC3Rec}}\label{app:ibBC3Rec}

The proof is similar to that in Section~\ref{sec:DPIDBC}. We prove Theorem~\ref{th:ibBC3Rec} by fixing any input distribution $P \in \mathscr P (\setX)$ and any positive ID rate-triple $(R_1,R_2,R_3)$ satisfying
\begin{subequations}\label{bl:rateTripleAch3Rec}
\begin{equation}
0 < R_k < \min \biggl\{ \muti P {W_k}, \sum_{l \in \{ 1,2,3 \} \setminus \{ k \}} \muti P {W_l} \biggr\}, \, \forall \, k \in \{ 1,2,3 \}
\end{equation}
\end{subequations}
and showing that the rate-triple $(R_1,R_2,R_3)$ is achievable. We assume that $$\muti P {W_k}, \quad k \in \{ 1,2,3 \}$$ are all positive; when they are not, the result follows from Theorem~\ref{th:IDBC}. For each $k \in \{ 1,2,3 \}$ let $\setM_k$ be a size-$\exp (\exp (n R_k))$ set of possible ID messages for Terminal~$k$. We next describe our random code construction and show that, for every positive $$\lambda_1^{(k)}, \, \lambda_2^{(k)}, \quad k \in \{ 1,2,3 \}$$ and every sufficiently-large blocklength~$n$, it produces with high probability an $\bigl( n, \{ \setM_k, \lambda_1^{(k)}, \lambda_2^{(k)} \}_{k \in \{ 1,2,3 \}} \bigr)$ ID code for the BC $\channel {y_1,y_2,y_3} x$.

\subparagraph*{Code Generation:} Fix expected bin rates $$\tilde R_k, \quad k \in \{ 1,2,3 \}$$ and a pool rate $R_\poolre$ satisfying
\begin{subequations}\label{bl:expBinRatesPoolRate3Rec}
\begin{IEEEeqnarray}{rCcCl}
R_k & < & \tilde R_k & < & \min \biggl\{ \muti P {W_k}, \sum_{l \in \{ 1,2,3 \} \setminus \{ k \}} \muti P {W_l} \biggr\}, \\
\tilde R_k & < & R_\poolre, \\
2 R_\poolre & < & \sum_{k \in \{ 1,2,3 \}} \tilde R_k. \label{eq:expBinRatesPoolRate3Rec}
\end{IEEEeqnarray}
\end{subequations}
This is possible by \eqref{bl:rateTripleAch3Rec}. Draw $e^{n R_\poolre}$ $n$-tuples $\sim P^n$ independently and place them in a pool $\pool$. Index the $n$-tuples in the pool by the elements of a size-$e^{n R_\poolre}$ set $\setV$, e.g., $\{ 1, \ldots, e^{n R_\poolre} \}$, and denote by $\poolel v$ the $n$-tuple in $\pool$ that is indexed by~$v \in \setV$. For each receiving terminal $k \in \{ 1,2,3 \}$ associate with each ID message $m_k \in \setM_k$ an index-set $\indexset {m_k}$ and a bin $\bin {m_k}$ as follows. Select each element of $\setV$ for inclusion in $\indexset {m_k}$ independently with probability $e^{-n( R_\poolre - \tilde R_k )}$, and let Bin~$\bin {m_k}$ be the multiset that contains all the $n$-tuples in the pool that are indexed by $\indexset {m_k}$, $$\bin {m_k} = \bigl\{ \poolel v, \, v \in \indexset {m_k} \bigr\}.$$ (Bin~$\bin {m_k}$ is thus of expected size $e^{n \tilde R_k}$.) Associate with each ID message-triple $(m_1,m_2,m_3) \in \setM_1 \times \setM_2 \times \setM_3$ an index $V_{m_1,m_2,m_3}$ as follows. If $\indexset {m_1} \cap \indexset {m_2} \cap \indexset {m_3}$ is not empty, then draw $V_{m_1,m_2,m_3}$ uniformly over $\indexset {m_1} \cap \indexset {m_2} \cap \indexset {m_3}$. Otherwise draw $V_{m_1,m_2,m_3}$ uniformly over $\setV$. Reveal the pool $\pool$, the index-sets $$\bigl\{ \indexset {m_k} \bigr\}_{m_k \in \setM_k}, \quad k \in \{ 1,2,3 \},$$ the corresponding bins $$\bigl\{ \bin {m_k} \bigr\}_{m_k \in \setM_k}, \quad k \in \{ 1,2,3 \},$$ and the indices $\bigl\{ V_{m_1,m_2,m_3} \bigr\}_{(m_1,m_2,m_3) \in \setM_1 \times \setM_2 \times \setM_3}$ to all parties. The encoding and decoding are determined by
\begin{IEEEeqnarray}{l}
\rcode = \Bigl( \pool, \bigl\{ \indexset {m_1} \bigr\}_{m_1 \in \setM_1}, \bigl\{ \indexset {m_2} \bigr\}_{m_2 \in \setM_2}, \bigl\{ \indexset {m_3} \bigr\}_{m_3 \in \setM_3}, \bigl\{ V_{m_1, m_2, m_3} \bigr\}_{ (m_1, m_2, m_3 ) \in \setM_1 \times \setM_2 \times \setM_3 } \Bigr). \label{eq:randCodeBC3Rec}
\end{IEEEeqnarray}

\subparagraph*{Encoding:} To send ID Message-Triple~$(m_1, m_2, m_3) \in \setM_1 \times \setM_2 \times \setM_3$, the encoder transmits the sequence $\poolel {V_{m_1,m_2,m_3}}$. ID Message-Triple~$(m_1, m_2, m_3)$ is thus associated with the $\{ 0,1 \}$-valued PMF
\begin{IEEEeqnarray}{rCl}
\bm Q_{m_1, m_2, m_3} (\vecx) & = & \ind {\vecx = \poolel {V_{m_1, m_2, m_3}}}, \quad \vecx \in \setX^n. \label{eq:defPMFRndCodeBC3Rec}
\end{IEEEeqnarray}
Note that once the code \eqref{eq:randCodeBC3Rec} has been constructed, the encoder is deterministic: it maps ID Message-Triple~$(m_1, m_2, m_3)$ to the $(m_1, m_2, m_3)$-codeword $\poolel {V_{m_1, m_2, m_3}}$. 

\subparagraph*{Decoding:} In this section the function $\delta (\cdot)$ maps every nonnegative real number $u$ to $u  \ent {P \times W}$. The decoders choose $\epsilon > 0$ sufficiently small so that $$2 \delta ( \epsilon ) < \muti {P}{W_k} - \tilde R_k, \quad k \in \{ 1,2,3 \}.$$ For each $k \in \{ 1,2,3 \}$ the $m^\prime_k$-focused party at Terminal~$k$ guesses that $m^\prime_k$ was sent iff for some index $v \in \indexset {m_k^\prime}$ the $n$-tuple $\poolel v$ in Bin~$\bin {m_k^\prime}$ is jointly $\epsilon$-typical with the Terminal-$k$ output-sequence $Y_{k,1}^n$, i.e., iff $( \poolel v, Y^n_{k,1} ) \in \eptyp (P \times W_k)$ for some $v \in \indexset {m_k^\prime}$. The set $\idset {m^\prime_k}$ of Terminal-$k$ output-sequences $\vecy_k \in (\setY_k)^n$ that result in the guess ``$m^\prime_k$ was sent'' is thus
\begin{IEEEeqnarray}{rCl}
\idset {m^\prime_k} & = & \bigcup_{v \in \indexset {m^\prime_k}} \setT^{ ( n  )}_\epsilon \bigl( P \times W_k \bigl| \poolel v \bigr), \quad k \in \{ 1,2,3 \}. \label{eq:DefIDSetMk3Rec}
\end{IEEEeqnarray}

\subparagraph*{Analysis of the Probabilities of Missed and Wrong Identification:} We first note that $\rcode$ of \eqref{eq:randCodeBC3Rec} (together with the fixed blocklength $n$ and the chosen $\epsilon$) fully specifies the encoding and guessing rules. That is, the randomly constructed ID code
\begin{IEEEeqnarray}{l}
\bigl\{ \bm Q_{m_1,m_2,m_3}, \idset {m_1}, \idset {m_2}, \idset {m_3} \bigr\}_{(m_1,m_2,m_3) \in \setM_1 \times \setM_2 \times \setM_3} \label{eq:randCodeBC23Rec}
\end{IEEEeqnarray}
is fully specified by $\rcode$. Let $\dist$ be the distribution of  $\rcode$, and let $\Exop$ denote expectation w.r.t.\ $\dist$. Subscripts indicate conditioning on the event that some of the chance variables assume the values indicated by the subscripts, e.g., $\dist_{\indexsetre {m_1}}$ denotes the distribution conditional on $\indexset {m_1} = \indexsetre {m_1}$, and $\Exop_{\indexsetre {m_1}}$ denotes the expectation w.r.t.\ $\dist_{\indexsetre {m_1}}$.\\

The maximum probabilities of missed and wrong identification of the randomly constructed ID code are the random variables
\begin{subequations}
\begin{IEEEeqnarray}{rCl}
P^{(k)}_{\textnormal{missed-ID}} & = & \max_{m_k \in \setM_k} \frac{1}{\card {\setM_j} \, \card {\setM_\ell}} \sum_{(m_j,m_\ell) \in \setM_j \times \setM_\ell} \bigl( \bm Q_{m_1,m_2,m_3} W^n \bigr) \bigl( Y^n_{k,1} \notin \idset {m_k} \bigr), \\
P^{(k)}_{\textnormal{wrong-ID}} & = & \max_{m_k \in \setM_k} \max_{m^\prime_k \neq m_k } \frac{1}{\card {\setM_j} \, \card {\setM_\ell}} \sum_{(m_j, m_\ell) \in \setM_j \times \setM_\ell} \bigl( \bm Q_{m_1,m_2,m_3} W^n \bigr) \bigl( Y^n_{k,1} \in \idset {m^\prime_k} \bigr),
\end{IEEEeqnarray}
\end{subequations}
where $k \in \{ 1, 2, 3 \}$ and $\ell, \, j$ is the pair of elements of $\{ 1, 2, 3 \} \setminus \{ k \}$ that satisfies $\ell < j$. They are fully specified by $\rcode$, because they are fully specified by the randomly constructed ID code \eqref{eq:randCodeBC23Rec}, which is in turn fully specified by $\rcode$. To prove that for every choice of $$\lambda_1^{(k)}, \, \lambda_2^{(k)} > 0, \quad k \in \{ 1,2,3 \}$$ and $n$ sufficiently large the collection of tuples \eqref{eq:randCodeBC23Rec} is with high probability an $\bigl( n, \{ \setM_k, \lambda^{(k)}_1, \lambda^{(k)}_2 \}_{k \in \{ 1,2,3 \}} \bigr)$ ID code for the BC $\channel {y_1,y_2,y_3} x$, we prove the following stronger result:

\begin{claim}\label{cl:toShowIDBC3Rec}
The probabilities $$P^{(k)}_{\textnormal{missed-ID}}, P^{(k)}_{\textnormal{wrong-ID}}, \quad k \in \{ 1,2,3 \}$$ of the randomly constructed ID code \eqref{eq:randCodeBC23Rec} converge in probability to zero exponentially in the blocklength~$n$, i.e.,
\begin{IEEEeqnarray}{l}
\exists \, \tau > 0 \textnormal{ s.t.\ } \lim_{n \rightarrow \infty} \biggdistof { \max_{k \in \{ 1,2,3 \}} \Bigl\{ P^{(k)}_{\textnormal{missed-ID}}, P^{(k)}_{\textnormal{wrong-ID}} \Bigr\} \geq e^{-n \tau} } = 0. \label{eq:toShowIDBC3Rec}
\end{IEEEeqnarray}
\end{claim}

\begin{proof}
We will prove that
\begin{IEEEeqnarray}{l}
\exists \, \tau > 0 \textnormal{ s.t.\ } \lim_{n \rightarrow \infty} \Bigdistof { \max \bigl\{ P^{(1)}_{\textnormal{missed-ID}}, P^{(1)}_{\textnormal{wrong-ID}} \bigr\} \geq e^{-n \tau} } = 0. \label{eq:toShowIDBC23Rec}
\end{IEEEeqnarray}
By swapping $1$ and $2$ or $3$ throughout the proof it will then follow that \eqref{eq:toShowIDBC23Rec} also holds when we replace $1$ with $2$ or $3$, respectively, and \eqref{eq:toShowIDBC3Rec} will then follow using the Union-of-Events bound. To prove \eqref{eq:toShowIDBC23Rec} we consider for each $m_1 \in \setM_1$ two distributions on the set $\setV$, which indexes the pool $\pool$. We fix some $v^\star \in \setV$ and define for every $m_1 \in \setM_1$ the PMFs on $\setV$
\begin{subequations}\label{bl:remIndM2M3AndUnifBin3Rec}
\begin{IEEEeqnarray}{rCl}
\bm P_V^{(m_1)} (v) & = & \frac{1}{|\setM_2| \, |\setM_3|} \sum_{(m_2,m_3) \in \setM_2 \times \setM_3} \ind {v = V_{m_1,m_2,m_3}}, \quad v \in \setV, \label{eq:remIndM2M3Bin3Rec} \\
\tilde {\bm P}_V^{(m_1)} (v) & = & \begin{cases} \frac{1}{|\indexset {m_1}|} \sum_{v^\prime \in \indexset {m_1}} \ind {v = v^\prime} &\textnormal{if } \indexset {m_1} \neq \emptyset, \\ \ind {v = v^\star} &\textnormal{otherwise}, \end{cases} \quad v \in \setV. \label{eq:remUnifBin3Rec}
\end{IEEEeqnarray}
\end{subequations}
The latter PMF is reminiscent of the distribution we encountered in \eqref{eq:distIDMsgDMC} and \eqref{eq:distIDMsgDMC2} in the single-user case. The former is related to the three-receiver BC setting when we view the pair $(M_2,M_3)$ as uniform over $\setM_2 \times \setM_3$. Like the proof of Claim~\ref{cl:toShowIDBC}, to establish \eqref{eq:toShowIDBC3Rec} it suffices to show that the two PMFs are similar in the sense that
\begin{IEEEeqnarray}{l}
\exists \, \tau > 0 \textnormal{ s.t.\ } \lim_{n \rightarrow \infty} \biggdistof { \max_{m_1 \in \setM_1} d \Bigl( \bm P_V^{(m_1)}, \tilde {\bm P}_V^{(m_1)} \Bigr) \geq e^{-n \tau} } = 0. \label{eq:toShowIDBC43Rec}
\end{IEEEeqnarray}

Establishing \eqref{eq:toShowIDBC43Rec} requires more work than establishing \eqref{eq:toShowIDBC4} in the proof of Claim~\ref{cl:toShowIDBC}. The reason for this is that---unlike the index-sets $\bigl\{ \indexset {m_\rz} \bigr\}_{m_\rz \in \setM_\rz}$ of Section~\ref{sec:DPIDBC}---the intersections $\bigl\{ \indexset {m_2} \cap \indexset {m_3} \bigr\}_{(m_2,m_3) \in \setM_2 \times \setM_3}$ are not independent. To overcome this difficulty, we shall first view only $M_3$ as uniform over $\setM_3$ while fixing $M_2 = m_2$ for some $m_2 \in \setM_2$. Later, we shall view also $M_2$ as uniform over $\setM_2$.

We define for every pair $(m_1,m_2) \in \setM_1 \times \setM_2$ the PMFs on $\setV$
\begin{subequations}\label{bl:remIndM3AndUnifIntBin3Rec}
\begin{IEEEeqnarray}{rCl}
\bm P_V^{(m_1,m_2)} (v) & = & \frac{1}{|\setM_3|} \sum_{m_3 \in \setM_3} \ind {v = V_{m_1,m_2,m_3}}, \quad v \in \setV, \\ \label{eq:remIndM3Bin3Rec}
\hat {\bm P}_V^{(m_1,m_2)} (v) & = & \begin{cases} \frac{1}{|\indexset {m_1} \cap \indexset {m_2}|} \sum_{v^\prime \in \indexset {m_1} \cap \indexset {m_2}} \ind {v = v^\prime} &\textnormal{if } \indexset {m_1} \cap \indexset {m_2} \neq \emptyset, \\ \ind {v = v^\star} &\textnormal{otherwise}, \end{cases} \quad v \in \setV. \label{eq:remUnifIntBin3Rec}
\end{IEEEeqnarray}
\end{subequations}
The latter PMF is reminiscent of the distribution in \eqref{eq:remUnifBin3Rec}. The former is related to the three-receiver BC setting when we view $M_3$ as uniform over $\setM_3$, and for every $m_1 \in \setM_1$ it relates to the distribution in \eqref{eq:remIndM2M3Bin3Rec} through
\begin{IEEEeqnarray}{rCl}
\bm P_V^{(m_1)} (v) & = & \frac{1}{|\setM_2|} \sum_{m_2 \in \setM_2} \bm P_V^{(m_1,m_2)} (v), \quad v \in \setV. \label{eq:remIndM2M3FromIndM3Bin3Rec}
\end{IEEEeqnarray}
For every $m_1 \in \setM_1$ define the PMF on $\setV$
\begin{IEEEeqnarray}{rCl}
\hat {\bm P}_V^{(m_1)} (v) & = & \frac{1}{|\setM_2|} \sum_{m_2 \in \setM_2} \hat {\bm P}_V^{(m_1,m_2)} (v), \quad v \in \setV. \label{eq:remIndM2UnifIntBin3Rec}
\end{IEEEeqnarray}
We can now upper-bound $d \bigl( \bm P_V^{(m_1)}, \tilde {\bm P}_V^{(m_1)} \bigr)$ by
\begin{IEEEeqnarray}{l}
d \Bigl( \bm P_V^{(m_1)}, \tilde {\bm P}_V^{(m_1)} \Bigr) \nonumber \\
\quad \stackrel{(a)}\leq d \Bigl( \bm P_V^{(m_1)}, \hat {\bm P}_V^{(m_1)} \Bigr) + d \Bigl( \hat {\bm P}_V^{(m_1)}, \tilde {\bm P}_V^{(m_1)} \Bigr) \\
\quad \stackrel{(b)}\leq \frac{1}{|\setM_2|} \sum_{m_2 \in \setM_2} d \Bigl( \bm P_V^{(m_1,m_2)}, \hat {\bm P}_V^{(m_1,m_2)} \Bigr) + d \Bigl( \hat {\bm P}_V^{(m_1)}, \tilde {\bm P}_V^{(m_1)} \Bigr), \label{eq:ubToShowIDBC43Rec}
\end{IEEEeqnarray}
where $(a)$ follows from the Triangle inequality; and $(b)$ holds because
\begin{IEEEeqnarray}{l}
d \Bigl( \bm P_V^{(m_1)}, \hat {\bm P}_V^{(m_1)} \Bigr) \nonumber \\
\quad \stackrel{(c)}= \frac{1}{2} \sum_{v \in \setV} \Bigl| \bm P_V^{(m_1)} (v) - \hat {\bm P}_V^{(m_1)} (v) \Bigr| \\
\quad \stackrel{(d)}= \frac{1}{2} \sum_{v \in \setV} \biggl| \frac{1}{|\setM_2|} \sum_{m_2 \in \setM_2} \bm P_V^{(m_1,m_2)} (v) - \hat {\bm P}_V^{(m_1,m_2)} (v) \biggr| \\
\quad \stackrel{(e)}\leq \frac{1}{2} \sum_{v \in \setV} \frac{1}{|\setM_2|} \sum_{m_2 \in \setM_2} \Bigl| \bm P_V^{(m_1,m_2)} (v) - \hat {\bm P}_V^{(m_1,m_2)} (v) \Bigr| \\
\quad \stackrel{(f)}= \frac{1}{|\setM_2|} \sum_{m_2 \in \setM_2} d \Bigl( \bm P_V^{(m_1,m_2)}, \hat {\bm P}_V^{(m_1,m_2)} \Bigr),
\end{IEEEeqnarray}
where $(c)$ and $(f)$ hold by definition of the Total-Variation distance; $(d)$ holds by \eqref{eq:remIndM2M3FromIndM3Bin3Rec} and \eqref{eq:remIndM2UnifIntBin3Rec}; and $(e)$ follows from the Triangle inequality. For every $\tau_1, \, \tau_2$, and $\tau < \min \{ \tau_1, \tau_2 \}$ we have for all sufficiently-large $n$,
\begin{equation}
e^{-n \tau_1} + e^{-n \tau_2} \leq e^{- n \tau}. \label{eq:tau1tau2tau3Rec}
\end{equation}
This, combined with the Union-of-Events bound and \eqref{eq:ubToShowIDBC43Rec}, implies that to establish \eqref{eq:toShowIDBC43Rec} it suffices to show the following two:
\begin{subequations}\label{bl:toShowIDBC53Rec}
\begin{IEEEeqnarray}{l}
\exists \, \tau > 0 \textnormal{ s.t.\ } \lim_{n \rightarrow \infty} \biggdistof { \max_{m_1 \in \setM_1} \frac{1}{|\setM_2|} \sum_{m_2 \in \setM_2} d \Bigl( \bm P_V^{(m_1,m_2)}, \hat {\bm P}_V^{(m_1,m_2)} \Bigr) \geq e^{-n \tau} } = 0, \label{eq:toShowIDBC5a3Rec} \\
\exists \, \tau > 0 \textnormal{ s.t.\ } \lim_{n \rightarrow \infty} \biggdistof { \max_{m_1 \in \setM_1} d \Bigl( \hat {\bm P}_V^{(m_1)}, \tilde {\bm P}_V^{(m_1)} \Bigr) \geq e^{-n \tau} } = 0. \label{eq:toShowIDBC5b3Rec}
\end{IEEEeqnarray}
\end{subequations}

We next establish \eqref{bl:toShowIDBC53Rec}, beginning with \eqref{eq:toShowIDBC5a3Rec}. For every fixed ID message-pair $(m_1,m_2) \in \setM_1 \times \setM_2$ the pair $\bigl( \bm P_V^{(m_1,m_2)}, \hat {\bm P}_V^{(m_1,m_2)} \bigr)$ of \eqref{bl:remIndM3AndUnifIntBin3Rec} has the same law as the pair $\bigl( \bm P_V^{(m_\ry)}, \tilde {\bm P}_V^{(m_\ry)} \bigr)$ of \eqref{bl:remIndMzAndUnifBin} in Section~\ref{sec:DPIDBC} with expected bin rates $\tilde R_\ry = \tilde R_1 + \tilde R_2 - R_\poolre$ and $\tilde R_\rz = \tilde R_3$, pool rate $R_\poolre$, rate $R_\rz = R_3$, index-set $\setV$, and where $m_\ry \in \setM_\ry$. (To see this, note that the index-sets $\indexset {m_1} \cap \indexset {m_2}$ and $\indexset {m_\ry}$ are constructed by selecting each element of $\setV$ for inclusion in $\indexset {m_1} \cap \indexset {m_2}$ or $\indexset {m_\ry}$, respectively, independently with probability $e^{- n (2 R_\poolre - \tilde R_1 - \tilde R_2)} = e^{-n (R_\poolre - \tilde R_\ry)}$ ($= e^{-n (R_\poolre - \tilde R_1)} e^{-n (R_\poolre - \tilde R_2)}$), and that for every $m_3 \in \setM_3$ and $m_\rz \in \setM_\rz$ the indices $V_{m_1,m_2,m_3}$ and $V_{m_\ry,m_\rz}$ are of the same law.) To establish \eqref{eq:toShowIDBC5a3Rec}, we can thus adopt some of the arguments leading to \eqref{eq:toShowIDBC4} in the proof of Claim~\ref{cl:toShowIDBC}.

Let $\delta_n$ be positive and converge to zero as $n$ tends to infinity, and let us henceforth assume that $n$ is large enough so that the following two inequalities hold:
\begin{subequations}\label{bl:IDBCnLargeEnough3Rec}
\begin{IEEEeqnarray}{rCl}
(1 - \delta_n) e^{n (\tilde R_1 + \tilde R_2 - R_\poolre)} & \geq & 1, \\
\delta_n & \leq & 1/2. \label{eq:IDBCnLargeEnough23Rec}
\end{IEEEeqnarray}
\end{subequations}
(This is possible, because $\delta_n$ converges to zero as $n$ tends to infinity and, by \eqref{bl:expBinRatesPoolRate3Rec}, $\tilde R_1 + \tilde R_2 - R_\poolre > 0$.)
For every $(m_1,m_2) \in \setM_1 \times \setM_2$ we upper-bound $d \bigl( \bm P_V^{(m_1,m_2)}, \hat {\bm P}_V^{(m_1,m_2)} \bigr)$ differently depending on whether or not
\begin{equation}
|\indexset {m_1} \cap \indexset {m_2}| > (1 - \delta_n) e^{n (\tilde R_1 + \tilde R_2 - R_\poolre)}. \label{eq:indexsetM1M2SuffLarge3Rec}
\end{equation}
If \eqref{eq:indexsetM1M2SuffLarge3Rec} does not hold, then we upper-bound it by one (which is an upper bound on the Total-Variation distance between any two probability measures) to obtain
\begin{IEEEeqnarray}{l}
\max_{m_1 \in \setM_1} \frac{1}{|\setM_2|} \sum_{m_2 \in \setM_2} d \Bigl( \bm P_V^{(m_1,m_2)}, \hat {\bm P}_V^{(m_1,m_2)} \Bigr) \nonumber \\
\quad \leq \max_{m_1 \in \setM_1} \frac{1}{|\setM_2|} \bigl| \bigl\{ m_2 \in \setM_2 \colon |\indexset {m_1} \cap \indexset {m_2}| \leq (1 - \delta_n) e^{n (\tilde R_1 + \tilde R_2 - R_\poolre)} \bigr\} \bigr| \nonumber \\
\qquad + \max_{(m_1,m_2) \in \setM_1 \times \setM_2} d \Bigl( \bm P_V^{(m_1,m_2)}, \hat {\bm P}_V^{(m_1,m_2)} \Bigr) \ind {|\indexset {m_1} \cap \indexset {m_2}| > (1 - \delta_n) e^{n (\tilde R_1 + \tilde R_2 - R_\poolre)}}.
\end{IEEEeqnarray}
This, combined with the Union-of-Events bound and \eqref{eq:tau1tau2tau3Rec} (which holds for every $\tau_1$, $\tau_2$, and $\tau < \min \{ \tau_1, \tau_2 \}$, and for all sufficiently-large $n$) implies that to establish \eqref{eq:toShowIDBC5a3Rec} it suffices to show the following two:
\begin{subequations}\label{bl:toShowIDBC63Rec}
\begin{IEEEeqnarray}{l}
\exists \, \tau > 0 \textnormal{ s.t.\ } \lim_{n \rightarrow \infty} \biggdistof { \max_{m_1 \in \setM_1} \frac{1}{|\setM_2|} \bigl| \bigl\{ m_2 \in \setM_2 \colon |\indexset {m_1} \cap \indexset {m_2}| \leq (1 - \delta_n) e^{n (\tilde R_1 + \tilde R_2 - R_\poolre)} \bigr\} \bigr| \geq e^{-n \tau} } = 0, \label{eq:toShowIDBC6a3Rec} \\
\exists \, \tau > 0 \textnormal{ s.t.\ } \lim_{n \rightarrow \infty} \biggdistof { \! \max_{(m_1,m_2) \in \setM_1 \times \setM_2} \!\!\!\!\! d \Bigl( \bm P_V^{(m_1,m_2)}, \hat {\bm P}_V^{(m_1,m_2)} \Bigr) \ind {|\indexset {m_1} \cap \indexset {m_2}| > (1 - \delta_n) e^{n (\tilde R_1 + \tilde R_2 - R_\poolre)}} \geq e^{-n \tau} } = 0. \label{eq:toShowIDBC6b3Rec}
\end{IEEEeqnarray}
\end{subequations}

We next establish \eqref{bl:toShowIDBC63Rec}, beginning with \eqref{eq:toShowIDBC6b3Rec}. As in \eqref{eq:IDBCKappaDef}, fix some $\kappa$ satisfying
\begin{IEEEeqnarray}{rCl}
0 < \kappa < \min \bigl\{ R_3, \tilde R_1 + \tilde R_2 + \tilde R_3 - 2 R_\poolre \bigr\}, \label{eq:IDBCKappaDef3Rec}
\end{IEEEeqnarray}
and let
\begin{IEEEeqnarray}{rCl}
\xi_n = 4 \exp \bigl\{ -e^{n \kappa - \log 2} \bigr\}. \label{eq:IDBCKappaTripExpSmall3Rec}
\end{IEEEeqnarray}
By \eqref{eq:IDBCnLargeEnough23Rec}
\begin{IEEEeqnarray}{l}
\xi_n / 2 > (1 - \delta_n)^{-1} \exp \bigl\{ - (1 - \delta_n) e^{n (\tilde R_1 + \tilde R_2 + \tilde R_3  - 2 R_\poolre)} - n (\tilde R_1 + \tilde R_2 - R_\poolre) \bigr\}.
\end{IEEEeqnarray}
For a fixed pair $(m_1,m_2) \in \setM_1 \times \setM_2$ fix any realization $\indexsetre {m_1} \cap \indexsetre {m_2}$ of the intersection $\indexset {m_1} \cap \indexset {m_2}$ satisfying that
\begin{equation}
|\indexsetre {m_1} \cap \indexsetre {m_2}| > (1 - \delta_n) e^{n (\tilde R_1 + \tilde R_2 - R_\poolre)}.
\end{equation}
The line of arguments leading to \eqref{eq:prTotVarDistTooBigFixedMy} in the proof of Claim~\ref{cl:toShowIDBC} implies that
\begin{IEEEeqnarray}{l}
\biggdistsubof {\indexsetre {m_1} \cap \indexsetre {m_2}}{ d \Bigl( \bm P_V^{(m_1,m_2)}, \hat {\bm P}_V^{(m_1,m_2)} \Bigr) \geq |\setV| \, \xi_n / 2 } \nonumber \\
\quad \leq 2 \, |\setV| \exp \bigl\{ - |\setM_3| \xi_n^2 / 2 \bigr\}, \quad |\indexsetre {m_1} \cap \indexsetre {m_2}| > (1 - \delta_n) e^{n (\tilde R_1 + \tilde R_2 - R_\poolre)}. \label{eq:prTotVarDistTooBigFixedM1M23Rec}
\end{IEEEeqnarray}
By \eqref{eq:expBinRatesPoolRate3Rec}, \eqref{eq:IDBCKappaDef3Rec}, and \eqref{eq:IDBCKappaTripExpSmall3Rec} there must exist a positive constant $\tau > 0$ and some $\eta_0 \in \naturals$ for which
\begin{IEEEeqnarray}{l}
|\setV| \, \xi_n / 2 \leq e^{-n \tau}, \quad n \geq \eta_0. \label{eq:IDBCXiExpSmall3Rec}
\end{IEEEeqnarray}
For every $\tau > 0$ and $\eta_0 \in \naturals$ satisfying \eqref{eq:IDBCXiExpSmall3Rec} and for all $n$ exceeding $\eta_0$
\begin{IEEEeqnarray}{l}
\biggdistof { \max_{(m_1,m_2) \in \setM_1 \times \setM_2 } d \Bigl( \bm P_V^{(m_1,m_2)}, \hat {\bm P}_V^{(m_1,m_2)} \Bigr) \ind {|\indexset {m_1} \cap \indexset {m_2}| > (1 - \delta_n) e^{n (\tilde R_1 + \tilde R_2 - R_\poolre)}} \geq e^{- n \tau} } \nonumber \\
\quad \stackrel{(a)} \leq |\setM_1| \, |\setM_2| \max_{|\indexsetre {m_1} \cap \indexsetre {m_2}| > (1 - \delta_n) e^{n (\tilde R_1 + \tilde R_2 - R_\poolre)}} \biggdistsubof {\indexsetre {m_1} \cap \indexsetre {m_2}}{ d \Bigl( \bm P_V^{(m_1,m_2)}, \hat {\bm P}_V^{(m_1,m_2)} \Bigr) \geq e^{- n \tau} } \\
\quad \stackrel{(b)}\leq |\setM_1| \, |\setM_2| \max_{|\indexsetre {m_1} \cap \indexsetre {m_2}| > (1 - \delta_n) e^{n (\tilde R_1 + \tilde R_2 - R_\poolre)}} \biggdistsubof {\indexsetre {m_1} \cap \indexsetre {m_2}}{ d \Bigl( \bm P_V^{(m_1,m_2)}, \hat {\bm P}_V^{(m_1,m_2)} \Bigr) \geq |\setV| \, \xi_n / 2 } \\
\quad \stackrel{(c)}\leq 2 \, |\setV| \, |\setM_1| \, |\setM_2| \exp \bigl\{ - |\setM_3| \exp \{ - e^{n \kappa} + 3 \log 2 \} \bigr\} \nonumber \\
\quad \stackrel{(d)}\rightarrow 0 \, (n \rightarrow \infty),  \label{eq:prTotVarDistTooBigAnyM1M23Rec}
\end{IEEEeqnarray}
where $(a)$ follows from the Union-of-Events bound; $(b)$ holds by \eqref{eq:IDBCXiExpSmall3Rec}, because $n$ exceeds $\eta_0$; $(c)$ holds by \eqref{eq:prTotVarDistTooBigFixedM1M23Rec} and \eqref{eq:IDBCKappaTripExpSmall3Rec}; and $(d)$ holds by \eqref{eq:IDBCKappaDef3Rec}, because $|\setV| = e^{n R_\poolre}$, and because $|\setM_k| = \exp (\exp (n R_k)), \,\, k \in \{ 1,2,3 \}$.

Having established \eqref{eq:toShowIDBC6b3Rec} for every $\delta_n$ that converges to zero as $n$ tends to infinity, we return to \eqref{bl:toShowIDBC63Rec} and conclude the proof of \eqref{eq:toShowIDBC5a3Rec} by establishing \eqref{eq:toShowIDBC6a3Rec} for some $\delta_n$ that converges to zero as $n$ tends to infinity. To that end, fix some $\mu$ satisfying
\begin{equation}
0 < \mu < \tilde R_1 - R_1, \label{eq:muIDBCPf3Rec}
\end{equation}
and let
\begin{equation}
\alpha_n = e^{- n \mu / 2}. \label{eq:alphaIDBCPf3Rec}
\end{equation}
Introduce the set $\setH^{(1)}_\mu$ comprising the realizations $\{ \indexsetre \nu \}_{\nu \in \setM_1}$ of the index-sets $\{ \indexset \nu \}_{\nu \in \setM_1}$ satisfying that
\begin{equation}
|\indexsetre \nu| > (1 - \alpha_n) e^{n \tilde R_1}, \,\, \forall \, \nu \in \setM_1. \label{eq:IDBCSetH3Rec}
\end{equation}
We upper-bound $$\bigl| \bigl\{ m_2 \in \setM_2 \colon |\indexset {m_1} \cap \indexset {m_2}| \leq (1 - \delta_n) e^{n (\tilde R_1 + \tilde R_2 - R_\poolre)} \bigr\} \bigr|$$ differently depending on whether or not $\{ \indexset \nu \}$ is in $\setH^{(1)}_\mu$, where $\{ \indexset \nu \}$ is short for $\{ \indexset \nu \}_{\nu \in \setM_1}$. If $\{ \indexset \nu \} \notin \setH^{(1)}_\mu$, then we upper-bound it by $|\setM_2|$ to obtain for every $\tau > 0$
\begin{IEEEeqnarray}{l}
\biggdistof { \max_{m_1 \in \setM_1} \frac{1}{|\setM_2|} \bigl| \bigl\{ m_2 \in \setM_2 \colon |\indexset {m_1} \cap \indexset {m_2}| \leq (1 - \delta_n) e^{n (\tilde R_1 + \tilde R_2 - R_\poolre)} \bigr\} \bigr| \geq e^{-n \tau} } \nonumber \\
\quad \leq \bigdistof { \{ \indexset \nu \} \notin \setH^{(1)}_\mu } + \sum_{ \indexsetsre \in \setH^{(1)}_\mu} \bigdistof { \{ \indexset \nu \} = \indexsetsre } \nonumber \\
\qquad \times \, \biggdistsubof {\indexsetsre}{ \max_{m_1 \in \setM_1} \frac{1}{|\setM_2|} \bigl| \bigl\{ m_2 \in \setM_2 \colon |\indexsetre {m_1} \cap \indexset {m_2}| \leq (1 - \delta_n) e^{n (\tilde R_1 + \tilde R_2 - R_\poolre)} \bigr\} \bigr| \geq e^{-n \tau} }. \label{eq:IDBCToShowNotinAndInH3Rec}
\end{IEEEeqnarray}

We consider the two terms on the RHS of \eqref{eq:IDBCToShowNotinAndInH3Rec} separately, beginning with $\bigdistof { \{ \indexset \nu \} \notin \setH^{(1)}_\mu }$. By the line of arguments leading to \eqref{eq:allBmySuffLarge} in the proof of Claim~\ref{cl:toShowIDBC}
\begin{IEEEeqnarray}{rCl}
\bigdistof { \{ \indexset \nu \} \notin \setH^{(1)}_\mu } & \leq & |\setM_1| \exp \bigl\{ - e^{n (\tilde R_1 - \mu) - \log 2} \bigr\} \\
& \stackrel{(a)}\rightarrow & 0 \, (n \rightarrow \infty), \label{eq:allBmySuffLarge3Rec}
\end{IEEEeqnarray}
where $(a)$ holds because $|\setM_1| = \exp (\exp (n R_1))$ and by \eqref{eq:muIDBCPf3Rec}.

Having established \eqref{eq:allBmySuffLarge3Rec}, we return to \eqref{eq:IDBCToShowNotinAndInH3Rec} and conclude the proof of \eqref{eq:toShowIDBC6a3Rec} by showing that
\begin{IEEEeqnarray}{l}
\exists \, \tau > 0 \textnormal{ s.t.\ } \nonumber \\
\lim_{n \rightarrow \infty} \max_{\indexsetsre \in \setH^{(1)}_\mu} \biggdistsubof {\indexsetsre}{ \max_{m_1 \in \setM_1} \frac{1}{|\setM_2|} \bigl| \bigl\{ m_2 \in \setM_2 \colon |\indexsetre {m_1} \cap \indexset {m_2}| \leq (1 - \delta_n) e^{n (\tilde R_1 + \tilde R_2 - R_\poolre)} \bigr\} \bigr| \geq e^{-n \tau} } = 0. \label{eq:toShowIDBC6b23Rec}
\end{IEEEeqnarray}
To prove \eqref{eq:toShowIDBC6b23Rec}, let us henceforth assume that $n$ is large enough so that the following two inequalities hold:
\begin{subequations}\label{bl:IDBCnLargeEnough23Rec}
\begin{IEEEeqnarray}{rCl}
(1 - \alpha_n) e^{n \tilde R_1} & \geq & 1, \label{eq:IDBCnLargeEnough213Rec} \\
\alpha_n & \leq & 1/2, \label{eq:IDBCnLargeEnough223Rec}
\end{IEEEeqnarray}
 \end{subequations}
where $\alpha_n$ is defined in \eqref{eq:alphaIDBCPf3Rec}. (This is possible, because $\alpha_n$ converges to zero as $n$ tends to infinity and $\tilde R_1 > 0$.) Fix any realization $\indexsetsre$ in $\setH^{(1)}_\mu$. Rather than directly upper-bounding the maximum over $m_1 \in \setM_1$ of $$\frac{1}{|\setM_2|} \bigl| \bigl\{ m_2 \in \setM_2 \colon |\indexsetre {m_1} \cap \indexset {m_2}| \leq (1 - \delta_n) e^{n (\tilde R_1 + \tilde R_2 - R_\poolre)} \bigr\} \bigr|$$ under $\dist_{\indexsetsre}$, we first consider $$\frac{1}{|\setM_2|} \bigl| \bigl\{ m_2 \in \setM_2 \colon |\indexsetre {m_1} \cap \indexset {m_2}| \leq (1 - \delta_n) e^{n (\tilde R_1 + \tilde R_2 - R_\poolre)} \bigr\} \bigr|$$ for a fixed $m_1 \in \setM_1$. By \eqref{eq:IDBCSetH3Rec} (which holds because $\indexsetsre \in \setH^{(1)}_\mu$) and \eqref{eq:IDBCnLargeEnough213Rec}, $\indexsetre {m_1}$ is nonempty. For every fixed $m_2 \in \setM_2$ we therefore have that under $\dist_{\indexsetsre}$ the $|\indexsetre {m_1}|$ binary random variables $\bigl\{ \ind {v \in \indexset {m_2}} \bigr\}_{v \in \indexsetre {m_1}}$ are IID and of mean
\begin{IEEEeqnarray}{rCl}
\bigEx {\indexsetsre}{\ind {v \in \indexset {m_2}}} & = & e^{-n (R_\poolre - \tilde R_2)}. \label{eq:meanIndVInM23Rec}
\end{IEEEeqnarray}
Fix some $\lambda$ satisfying
\begin{equation}
0 < \lambda < \tilde R_1 + \tilde R_2 - R_\poolre, \label{eq:lambdaIDBCPf3Rec}
\end{equation}
let
\begin{equation}
\beta_n = e^{- n \lambda / 2}, \label{eq:betaIDBCPf3Rec}
\end{equation}
and let
\begin{equation}
\delta_n = \alpha_n + \beta_n - \alpha_n \beta_n,
\end{equation}
where $\alpha_n$ is defined in \eqref{eq:alphaIDBCPf3Rec}. Note that $\delta_n$ satisfies
\begin{equation}
1 - \delta_n = (1 - \alpha_n) (1 - \beta_n). \label{eq:relDeltaAlphaBetaPf3Rec}
\end{equation}
Because $\alpha_n$ and $\beta_n$ are positive, smaller than one, and converge to zero as $n$ tends to infinity, also $\delta_n$ is positive, smaller than one, and converges to zero as $n$ tends to infinity. For every $m_2 \in \setM_2$ the multiplicative Chernoff bound \eqref{eq:multChernDeltaSm1Sm} implies that
\begin{IEEEeqnarray}{l}
\BigEx {\indexsetsre}{\ind{|\indexsetre {m_1} \cap \indexset {m_2}| \leq (1 - \delta_n) e^{n (\tilde R_1 + \tilde R_2 - R_\poolre)}}} \nonumber \\
\quad = \Bigdistsubof {\indexsetsre}{|\indexsetre {m_1} \cap \indexset {m_2}| \leq (1 - \delta_n) e^{n (\tilde R_1 + \tilde R_2 - R_\poolre)}} \\
\quad \stackrel{(a)}\leq \Biggdistsubof {\indexsetsre}{\sum_{v \in \indexsetre {m_1}} \ind {v \in \indexset {m_2}} \leq (1 - \beta_n) |\indexsetre {m_1}| e^{ - n (R_\poolre - \tilde R_2)}} \\
\quad \stackrel{(b)}\leq \exp \bigl\{ - \beta_n^2 (1 - \alpha_n) e^{ n (\tilde R_1 + \tilde R_2 - R_\poolre) - \log 2} \bigr\}, \quad \indexsetsre \in \setH^{(1)}_\mu, \label{eq:intIndexsetsM1M2Small3Rec}
\end{IEEEeqnarray}
where $(a)$ holds by \eqref{eq:IDBCSetH3Rec} (which holds because $\indexsetsre \in \setH^{(1)}_\mu$) and \eqref{eq:relDeltaAlphaBetaPf3Rec}; and $(b)$ holds by \eqref{eq:meanIndVInM23Rec}, \eqref{eq:multChernDeltaSm1Sm}, and \eqref{eq:IDBCSetH3Rec}. By \eqref{eq:lambdaIDBCPf3Rec}, \eqref{eq:betaIDBCPf3Rec}, and because $\alpha_n$ converges to zero as $n$ tends to infinity, there must exist a positive constant $\tau > 0$ and some $\eta_0 \in \naturals$ for which
\begin{equation}
\exp \bigl\{ - \beta_n^2 (1 - \alpha_n) e^{ n (\tilde R_1 + \tilde R_2 - R_\poolre) - \log 2} \bigr\} \leq e^{-n \tau} / 2, \quad n \geq \eta_0. \label{eq:tauAlphaBetapf3Rec}
\end{equation}
Since the $\exp (\exp (n R_2))$ binary random variables $$\Bigl\{ \ind{|\indexsetre {m_1} \cap \indexset {m_2}| \leq (1 - \delta_n) e^{n (\tilde R_1 + \tilde R_2 - R_\poolre)}} \Bigr\}_{m_2 \in \setM_2}$$ are IID, Hoeffding's inequality (Proposition~\ref{pr:hoeffding}) implies that for every $\tau > 0$ and $\eta_0 \in \naturals$ satisfying \eqref{eq:tauAlphaBetapf3Rec} and for all $n$ exceeding $\eta_0$
\begin{IEEEeqnarray}{l}
\Biggdistsubof {\indexsetsre}{ \frac{1}{|\setM_2|} \bigl| \bigl\{ m_2 \in \setM_2 \colon |\indexsetre {m_1} \cap \indexset {m_2}| \leq (1 - \delta_n) e^{n (\tilde R_1 + \tilde R_2 - R_\poolre)} \bigr\} \bigr| \geq e^{- n \tau}}  \nonumber \\
\quad = \Biggdistsubof {\indexsetsre}{\frac{1}{|\setM_2|} \sum_{m_2 \in \setM_2} \ind{|\indexsetre {m_1} \cap \indexset {m_2}| \leq (1 - \delta_n) e^{n (\tilde R_1 + \tilde R_2 - R_\poolre)}} \geq e^{- n \tau}} \\
\quad \leq \exp \bigl\{ - |\setM_2| e^{- 2 n \tau} / 2 \bigr\}, \quad \indexsetsre \in \setH^{(1)}_\mu, \label{eq:intIndexsetsM1M2SmallTooManyM23Rec}
\end{IEEEeqnarray}
where in the last inequality we used \eqref{eq:intIndexsetsM1M2Small3Rec} and \eqref{eq:tauAlphaBetapf3Rec}.

Having obtained \eqref{eq:intIndexsetsM1M2SmallTooManyM23Rec} for every fixed $m_1 \in \setM_1$, we are now ready to tackle the maximum over $m_1 \in \setM_1$ and prove \eqref{eq:toShowIDBC6a3Rec}: For every $\tau > 0$ and $\eta_0 \in \naturals$ satisfying \eqref{eq:tauAlphaBetapf3Rec} and for all $n$ exceeding $\eta_0$
\begin{IEEEeqnarray}{l}
\max_{\indexsetsre \in \setH^{(1)}_\mu} \biggdistsubof {\indexsetsre}{ \max_{m_1 \in \setM_1} \frac{1}{|\setM_2|} \bigl| \bigl\{ m_2 \in \setM_2 \colon |\indexsetre {m_1} \cap \indexset {m_2}| \leq (1 - \delta_n) e^{n (\tilde R_1 + \tilde R_2 - R_\poolre)} \bigr\} \bigr| \geq e^{- n \tau}}  \nonumber \\
\quad \stackrel{(a)}\leq |\setM_1| \exp \bigl\{ - |\setM_2| e^{- 2 n \tau} / 2 \bigr\} \\
\quad \stackrel{(b)}\rightarrow 0 \, (n \rightarrow \infty), \label{eq:intIndexsetsM1M2SmallTooManyM2MaxM13Rec}
\end{IEEEeqnarray}
where $(a)$ follows from the Union-of-Events bound and \eqref{eq:intIndexsetsM1M2SmallTooManyM23Rec}; and $(b)$ holds because $|\setM_k| = \exp (\exp (n R_k)), \,\, k \in \{ 1,2 \}$. This concludes the proof of \eqref{bl:toShowIDBC63Rec} and hence that of \eqref{eq:toShowIDBC5a3Rec}.\\

Having established \eqref{eq:toShowIDBC5a3Rec}, we return to \eqref{bl:toShowIDBC53Rec} and conclude the proof of Claim~\ref{cl:toShowIDBC3Rec} by establishing \eqref{eq:toShowIDBC5b3Rec}. To that end, we argue similarly as when establishing \eqref{eq:toShowIDBC4} in the proof of Claim~\ref{cl:toShowIDBC}. Recall that $\setH^{(1)}_\mu$ is the set comprising the realizations $\indexsetsre$ of the index-sets $\{ \indexset \nu \}$ satisfying \eqref{eq:IDBCSetH3Rec}, where $\mu$ is defined in \eqref{eq:muIDBCPf3Rec} and $\alpha_n$ in \eqref{eq:alphaIDBCPf3Rec}. We upper-bound $d \bigl( \hat {\bm P}_V^{(m_1)}, \tilde {\bm P}_V^{(m_1)} \bigr)$ differently depending on whether or not $\{ \indexset \nu \}$ is in $\setH^{(1)}_\mu$. If $\{ \indexset \nu \} \notin \setH^{(1)}_\mu$, then we upper-bound it by one (which is an upper bound on the Total-Variation distance between any two probability measures) to obtain for every $\tau > 0$
\begin{IEEEeqnarray}{l}
\biggdistof { \max_{m_1 \in \setM_1} d \Bigl( \hat {\bm P}_V^{(m_1)}, \tilde {\bm P}_V^{(m_1)} \Bigr) \geq e^{-n \tau} } \nonumber \\
\quad \leq \bigdistof { \{ \indexset \nu \} \notin \setH^{(1)}_\mu } + \sum_{ \indexsetsre \in \setH^{(1)}_\mu} \bigdistof { \{ \indexset \nu \} = \indexsetsre} \biggdistsubof { \indexsetsre }{ \max_{m_1 \in \setM_1} d \Bigl( \hat {\bm P}_V^{(m_1)}, \tilde {\bm P}_V^{(m_1)} \Bigr) \geq e^{-n \tau} }.
\end{IEEEeqnarray}
This and \eqref{eq:allBmySuffLarge3Rec} imply that to establish \eqref{eq:toShowIDBC5b3Rec} it suffices to show that
\begin{IEEEeqnarray}{l}
\exists \, \tau > 0 \textnormal{ s.t.\ } \lim_{n \rightarrow \infty} \max_{ \indexsetsre \in \setH^{(1)}_\mu } \biggdistsubof { \indexsetsre }{ \max_{m_1 \in \setM_1} d \Bigl( \hat {\bm P}_V^{(m_1)}, \tilde {\bm P}_V^{(m_1)} \Bigr) \geq e^{-n \tau} } = 0. \label{eq:toShowIDBC5b23Rec}
\end{IEEEeqnarray}
To prove \eqref{eq:toShowIDBC5b23Rec}, let us henceforth assume that $n$ is large enough so that \eqref{bl:IDBCnLargeEnough23Rec} holds. Fix any realization $\indexsetsre$ in $\setH^{(1)}_\mu$. Rather than directly upper-bounding the maximum over $m_1 \in \setM_1$ of $d \bigl( \hat {\bm P}_V^{(m_1)}, \tilde {\bm P}_V^{(m_1)} \bigr)$, we first consider $d \bigl( \hat {\bm P}_V^{(m_1)}, \tilde {\bm P}_V^{(m_1)} \bigr)$ for a fixed $m_1 \in \setM_1$. By \eqref{eq:IDBCSetH3Rec} (which holds because $\indexsetsre \in \setH^{(1)}_\mu$) and \eqref{eq:IDBCnLargeEnough213Rec}, $\indexsetre {m_1}$ is nonempty. We therefore have that under $\dist_{\indexsetre {m_1}}$
\begin{subequations}\label{bl:remUnifIndOthers}
\begin{IEEEeqnarray}{rCl}
\tilde {\bm P}_V^{(m_1)} (v) & = & \frac{1}{|\indexsetre {m_1}|} \sum_{v^\prime \in \indexsetre {m_1}} \ind {v = v^\prime}, \quad v \in \setV, \label{eq:remUnifBinM1} \\
\hat {\bm P}_V^{(m_1)} (v) & = & \frac{1}{|\setM_2|} \sum_{m_2 \in \setM_2} \Biggl( \frac{\ind {v \in \indexsetre {m_1} \cap \indexset {m_2}}}{|\indexsetre {m_1} \cap \indexset {m_2}| \vee 1}  + \ind {\indexsetre {m_1} \cap \indexset {m_2} = \emptyset} \, \ind {v = v^\star} \Biggr), \quad v \in \setV, \label{eq:remIndM2Bin}
\end{IEEEeqnarray}
\end{subequations}
where for every fixed $v \in \setV$ the $\exp (\exp (n R_2))$ $[0,1]$-valued random variables $$\Biggl\{ \frac{\ind {v \in \indexsetre {m_1} \cap \indexset {m_2}}}{|\indexsetre {m_1} \cap \indexset {m_2}| \vee 1}  + \ind {\indexsetre {m_1} \cap \indexset {m_2} = \emptyset} \, \ind {v = v^\star} \Biggr\}_{m_2 \in \setM_2}$$ are IID and have mean
\begin{IEEEeqnarray}{l}
\begin{cases} \frac{1}{|\indexsetre {m_1}|} \Bigl( 1 - \bigl( 1 - e^{-n (R_\poolre - \tilde R_2)} \bigr)^{|\indexsetre {m_1}|} \Bigr) &\textnormal{if } v \in \indexsetre {m_1} \setminus \{ v^\star \}, \\ \frac{1}{|\indexsetre {m_1}|} \Bigl( 1 + (|\indexsetre {m_1}| - 1) \bigl( 1 - e^{-n (R_\poolre - \tilde R_2)} \bigr)^{|\indexsetre {m_1}|} \Bigr) &\textnormal{if } v \in \indexsetre {m_1} \cap \{ v^\star \}, \\ \bigl( 1 - e^{-n (R_\poolre - \tilde R_2)} \bigr)^{|\indexsetre {m_1}|} &\textnormal{if } v \in \{ v^\star \} \setminus \indexsetre {m_1}, \\ 0 &\textnormal{if } v \notin \indexsetre {m_1} \cup \{ v^\star \}, \end{cases} \label{eq:meanHatPM1v3Rec}
\end{IEEEeqnarray}
where we used that
\begin{equation}
\distsubof {\indexsetsre}{\indexsetre {m_1} \cap \indexset {m_2} = \emptyset} = \bigl( 1 - e^{-n (R_\poolre - \tilde R_2)} \bigr)^{|\indexsetre {m_1}|}
\end{equation}
and that for every $m_2 \in \setM_2$ the $[0,1]$-valued random variables $$\Biggl\{ \frac{\ind {v \in \indexsetre {m_1} \cap \indexset {m_2}}}{|\indexsetre {m_1} \cap \indexset {m_2}| \vee 1} \Biggr\}_{v \in \indexsetre {m_1}}$$ are IID and sum to
\begin{equation}
\sum_{v \in \indexsetre {m_1}} \frac{\ind {v \in \indexsetre {m_1} \cap \indexset {m_2}}}{|\indexsetre {m_1} \cap \indexset {m_2}| \vee 1} = \ind {\indexsetre {m_1} \cap \indexset {m_2} \neq \emptyset}.
\end{equation}
With \eqref{eq:meanHatPM1v3Rec} at hand, we can establish \eqref{eq:toShowIDBC5b23Rec} essentially along the line of arguments leading to \eqref{eq:toShowIDBC5} in the proof of Claim~\ref{cl:toShowIDBC}.
\end{proof}

\section{A Proof of Theorem~\ref{th:IDBCCM}}\label{app:IDBCCM}

The proof consists of a direct and a converse part.

\subsection{The Direct Part of Theorem~\ref{th:IDBCCM}}

The proof of the direct part is similar to that in Section~\ref{sec:DPIDBC}. We prove the direct part of Theorem~\ref{th:IDBCCM} by fixing any input distribution $P \in \mathscr P (\setX)$ and any positive ID rate-triple $(R,R_\ry,R_\rz)$ satisfying
\begin{subequations}\label{bl:rateTripleAchCM}
\begin{IEEEeqnarray}{rCcCl}
0 & < & R, \, R_\ry & < & \muti{P}{W_\ry}, \\
0 & < & R, \, R_\rz & < & \muti{P}{W_\rz}
\end{IEEEeqnarray}
\end{subequations}
and showing that the rate-triple $(R,R_\ry,R_\rz)$ is achievable. The restriction to positive rates $R_\ry$ and $R_\rz$ is that of Theorem~\ref{th:IDBCCM}. Moreover, we assume that $R$ is positive; when it is not, the result follows from Theorem~\ref{th:IDBC}. Let $\setM$ be a size-$\exp (\exp (n R))$ set of possible common ID messages, let $\setM_\ry$ be a size-$\exp (\exp (n R_\ry))$ set of possible ID messages for Terminal~$\ry$, and let $\setM_\rz$ be a size-$\exp (\exp (n R_\rz))$ set of possible ID messages for Terminal~$\rz$. We next describe our random code construction and show that, for every positive $\lambda^{\ry}_1$, $\lambda^{\ry}_2$, $\lambda^{\rz}_1$, and $\lambda^{\rz}_2$ and for every sufficiently-large blocklength~$n$, it produces with high probability an $\bigl( n, \setM, \setM_\ry, \setM_\rz, \lambda^{\ry}_1, \lambda^{\ry}_2, \lambda^{\rz}_1, \lambda^{\rz}_2 \bigr)$ ID code for the BC $\channel {y,z} x$.

\subparagraph*{Code Generation:} Fix an expected bin rate $\tilde R_\ry$ for Terminal~$\ry$, an expected bin rate $\tilde R_\rz$ for Terminal~$\rz$, and a pool rate $R_\poolre$ satisfying
\begin{subequations}\label{bl:binAndPoolRateConstraintsBCCM}
\begin{IEEEeqnarray}{rCcCl}
R, \, R_\ry & < & \tilde R_\ry & < & \muti {P}{W_\ry}, \\
R, \, R_\rz & < & \tilde R_\rz & < & \muti {P}{W_\rz}, \\
&& \tilde R_\ry & < & R_\poolre, \\
&& \tilde R_\rz & < & R_\poolre, \\
R_\poolre & < & \tilde R_\ry + \tilde R_\rz. && \label{eq:blBinAndPoolRateConstraintsBCCMUBPoolRate}
\end{IEEEeqnarray}
\end{subequations}
This is possible by \eqref{bl:rateTripleAchCM}. Draw $e^{n R_\poolre}$ $n$-tuples $\sim P^n$ independently and place them in a pool $\pool$. Index the $n$-tuples in the pool by the elements of a size-$e^{n R_\poolre}$ set $\setV$, e.g., $\{ 1, \ldots, e^{n R_\poolre} \}$, and denote by $\poolel v$ the $n$-tuple in $\pool$ that is indexed by~$v \in \setV$. For each receiving terminal $\Psi \in \{ \ry, \rz \}$ associate with each ID message-pair $(m,m_\Psi) \in \setM \times \setM_\Psi$ an index-set $\indexset {m,m_\Psi}$ and a bin $\bin {m,m_\Psi}$ as follows. Select each element of $\setV$ for inclusion in $\indexset {m,m_\Psi}$ independently with probability $e^{-n( R_\poolre - \tilde R_\Psi )}$, and let Bin~$\bin {m,m_\Psi}$ be the multiset that contains all the $n$-tuples in the pool that are indexed by $\indexset {m,m_\Psi}$, $$\bin {m,m_\Psi} = \bigl\{ \poolel v, \, v \in \indexset {m,m_\Psi} \bigr\}.$$ (Bin~$\bin {m,m_\Psi}$ is thus of expected size $e^{n \tilde R_\Psi}$.) Associate with each ID message-triple $(m,m_\ry,m_\rz) \in \setM \times \setM_\ry \times \setM_\rz$ an index $V_{m,m_\ry,m_\rz}$ as follows. If $\indexset {m,m_\ry} \cap \indexset {m,m_\rz}$ is not empty, then draw $V_{m,m_\ry,m_\rz}$ uniformly over $\indexset {m,m_\ry} \cap \indexset {m,m_\rz}$. Otherwise draw $V_{m,m_\ry,m_\rz}$ uniformly over $\setV$. Reveal the pool $\pool$, the index-sets $\bigl\{ \indexset {m,m_\ry} \bigr\}_{(m,m_\ry) \in \setM \times \setM_\ry}$ and $\bigl\{ \indexset {m,m_\rz} \bigr\}_{(m,m_\rz) \in \setM \times \setM_\rz}$, the corresponding bins $\bigl\{ \bin {m,m_\ry} \bigr\}_{(m,m_\ry) \in \setM \times \setM_\ry}$ and $\bigl\{ \bin {m,m_\rz} \bigr\}_{(m,m_\rz) \in \setM \times \setM_\rz}$, and the indices $\bigl\{ V_{m,m_\ry,m_\rz} \bigr\}_{(m,m_\ry,m_\rz) \in \setM \times \setM_\ry \times \setM_\rz}$ to all parties. The encoding and decoding are determined by
\begin{IEEEeqnarray}{l}
\rcode = \Bigl( \pool, \bigl\{ \indexset {m,m_\ry} \bigr\}_{(m,m_\ry) \in \setM \times \setM_\ry}, \bigl\{ \indexset {m,m_\rz} \bigr\}_{(m,m_\rz) \in \setM \times \setM_\rz}, \bigl\{ V_{m, m_\ry, m_\rz} \bigr\}_{ (m, m_\ry, m_\rz ) \in \setM \times \setM_\ry \times \setM_\rz } \Bigr). \label{eq:randCodeBCCM}
\end{IEEEeqnarray}

\subparagraph*{Encoding:} To send ID Message-Triple~$(m, m_\ry, m_\rz) \in \setM \times \setM_\ry \times \setM_\rz$, the encoder transmits the sequence $\poolel {V_{m,m_\ry,m_\rz}}$. ID Message-Triple~$(m,m_\ry,m_\rz)$ is thus associated with the $\{ 0,1 \}$-valued PMF
\begin{IEEEeqnarray}{rCl}
\bm Q_{m,m_\ry,m_\rz} (\vecx) & = & \ind {\vecx = \poolel {V_{m,m_\ry,m_\rz}}}, \quad \vecx \in \setX^n. \label{eq:defPMFRndCodeBCCM}
\end{IEEEeqnarray}
Note that once the code \eqref{eq:randCodeBCCM} has been constructed, the encoder is deterministic: it maps ID Message-Triple~$(m,m_\ry,m_\rz)$ to the $(m,m_\ry,m_\rz)$-codeword $\poolel {V_{m,m_\ry,m_\rz}}$. 

\subparagraph*{Decoding:} In this section the function $\delta (\cdot)$ maps every nonnegative real number $u$ to $u  \ent {P \times W}$. The decoders choose $\epsilon > 0$ sufficiently small so that $2 \delta ( \epsilon ) < \muti {P}{W_\ry} - \tilde R_\ry$ and $2 \delta ( \epsilon ) < \muti {P}{W_\rz} - \tilde R_\rz$. The $(m^\prime,m^\prime_\ry)$-focused party at Terminal~$\ry$ guesses that $(m^\prime,m^\prime_\ry)$ was sent iff for some index $v \in \indexset {m^\prime,m_\ry^\prime}$ the $n$-tuple $\poolel v$ in Bin~$\bin {m^\prime,m_\ry^\prime}$ is jointly $\epsilon$-typical with the Terminal-$\ry$ output-sequence $Y^n$, i.e., iff $( \poolel v, Y^n ) \in \eptyp (P \times W_\ry)$ for some $v \in \indexset {m^\prime,m_\ry^\prime}$. The set $\idset {m^\prime,m^\prime_\ry}$ of Terminal-$\ry$ output-sequences $\vecy \in \setY^n$ that result in the guess ``$(m^\prime,m^\prime_\ry)$ was sent'' is thus
\begin{IEEEeqnarray}{rCl}
\idset {m^\prime,m^\prime_\ry} & = & \bigcup_{v \in \indexset {m^\prime,m^\prime_\ry}} \setT^{ ( n  )}_\epsilon \bigl( P \times W_\ry \bigl| \poolel v \bigr). \label{eq:DefIDSetMMyCM}
\end{IEEEeqnarray}
Likewise, the $(m^\prime,m^\prime_\rz)$-focused party at Terminal~$\rz$ guesses that $(m^\prime,m^\prime_\rz)$ was sent iff $( \poolel v, Z^n ) \in \eptyp (P \times W_\rz)$ for some $v \in \indexset {m^\prime,m_\rz^\prime}$. The set $\idset {m^\prime,m^\prime_\rz}$ of Terminal-$\rz$ output-sequences $\vecz \in \setZ^n$ that result in the guess ``$(m^\prime,m^\prime_\rz)$ was sent'' is thus
\begin{IEEEeqnarray}{rCl}
\idset {m^\prime,m^\prime_\rz} & = & \bigcup_{v \in \indexset {m^\prime,m^\prime_\rz}} \setT^{( n )}_\epsilon \bigl( P \times W_\rz \bigl| \poolel v \bigr). \label{eq:DefIDSetMMzCM}
\end{IEEEeqnarray}

\subparagraph*{Analysis of the Probabilities of Missed and Wrong Identification:} We first note that $\rcode$ of \eqref{eq:randCodeBCCM} (together with the fixed blocklength $n$ and the chosen $\epsilon$) fully specifies the encoding and guessing rules. That is, the randomly constructed ID code
\begin{IEEEeqnarray}{l}
\bigl\{ \bm Q_{m,m_\ry,m_\rz}, \idset {m,m_\ry}, \idset {m,m_\rz} \bigr\}_{(m,m_\ry,m_\rz) \in \setM \times \setM_\ry \times \setM_\rz} \label{eq:randCodeBC2CM}
\end{IEEEeqnarray}
is fully specified by $\rcode$. Let $\dist$ be the distribution of  $\rcode$, and let $\Exop$ denote expectation w.r.t.\ $\dist$.\\

The maximum probabilities of missed and wrong identification of the randomly constructed ID code are the random variables

\begin{subequations}
\begin{IEEEeqnarray}{rCl}
P^\ry_{\textnormal{missed-ID}} & = & \max_{(m,m_\ry) \in \setM \times \setM_\ry} \frac{1}{\card {\setM_\rz}} \sum_{m_\rz \in \setM_\rz} \bigl( \bm Q_{m,m_\ry,m_\rz} W^n \bigr) \bigl( Y^n \notin \idset {m,m_\ry} \bigr), \\
P^\rz_{\textnormal{missed-ID}} & = & \max_{(m,m_\rz) \in \setM \times \setM_\rz} \frac{1}{\card {\setM_\ry}} \sum_{m_\ry \in \setM_\ry} \bigl( \bm Q_{m,m_\ry,m_\rz} W^n \bigr) \bigl( Z^n \notin \idset {m,m_\rz} \bigr), \\
P^\ry_{\textnormal{wrong-ID}} & = & \max_{(m,m_\ry) \in \setM \times \setM_\ry} \max_{(m^\prime,m^\prime_\ry) \neq (m,m_\ry) } \frac{1}{\card {\setM_\rz}} \sum_{m_\rz \in \setM_\rz} \bigl( \bm Q_{m,m_\ry,m_\rz} W^n \bigr) \bigl( Y^n \in \idset {m^\prime,m^\prime_\ry} \bigr), \\
P^\rz_{\textnormal{wrong-ID}} & = & \max_{(m,m_\rz) \in \setM \times \setM_\rz} \max_{(m^\prime,m^\prime_\rz) \neq (m,m_\rz) } \frac{1}{\card {\setM_\ry}} \sum_{m_\ry \in \setM_\ry} \bigl( \bm Q_{m,m_\ry,m_\rz} W^n \bigr) \bigl( Z^n \in \idset {m^\prime,m^\prime_\rz} \bigr).
\end{IEEEeqnarray}
\end{subequations}
They are fully specified by $\rcode$, because they are fully specified by the randomly constructed ID code \eqref{eq:randCodeBC2CM}, which is in turn fully specified by $\rcode$. To prove that for every choice of $\lambda^{\ry}_1, \, \lambda^{\ry}_2, \, \lambda^{\rz}_1, \, \lambda^{\rz}_2 > 0$ and $n$ sufficiently large the collection of tuples \eqref{eq:randCodeBC2CM} is with high probability an $\bigl( n, \setM, \setM_\ry, \setM_\rz, \lambda^{\ry}_1, \lambda^{\ry}_2, \lambda^{\rz}_1, \lambda^{\rz}_2 \bigr)$ ID code for the BC $\channel {y,z} x$, we prove the following stronger result:

\begin{claim}\label{cl:toShowIDBCCM}
The probabilities $P^\ry_{\textnormal{missed-ID}}$, $P^\rz_{\textnormal{missed-ID}}$, $P^\ry_{\textnormal{wrong-ID}}$, and $P^\rz_{\textnormal{wrong-ID}}$ of the randomly constructed ID code \eqref{eq:randCodeBC2CM} converge in probability to zero exponentially in the blocklength~$n$, i.e.,
\begin{IEEEeqnarray}{l}
\exists \, \tau > 0 \textnormal{ s.t.\ } \lim_{n \rightarrow \infty} \Bigdistof { \max \bigl\{ P^\ry_{\textnormal{missed-ID}}, P^\rz_{\textnormal{missed-ID}}, P^\ry_{\textnormal{wrong-ID}}, P^\rz_{\textnormal{wrong-ID}} \bigr\} \geq e^{-n \tau} } = 0. \label{eq:toShowIDBCCM}
\end{IEEEeqnarray}
\end{claim}

\begin{proof}
We will prove that
\begin{IEEEeqnarray}{l}
\exists \, \tau > 0 \textnormal{ s.t.\ } \lim_{n \rightarrow \infty} \Bigdistof { \max \bigl\{ P^\ry_{\textnormal{missed-ID}}, P^\ry_{\textnormal{wrong-ID}} \bigr\} \geq e^{-n \tau} } = 0. \label{eq:toShowIDBC2CM}
\end{IEEEeqnarray}
By swapping $\rz$ and $\ry$ throughout the proof it will then follow that \eqref{eq:toShowIDBC2CM} also holds when we replace $\ry$ with $\rz$, and \eqref{eq:toShowIDBCCM} will then follow using the Union-of-Events bound.

To prove \eqref{eq:toShowIDBC2CM} we consider for each $m_\ry \in \setM_\ry$ two distributions on the set $\setV$, which indexes the pool $\pool$. We fix some $v^\star \in \setV$ and define for every $m_\ry \in \setM_\ry$ the PMFs on $\setV$
\begin{subequations}\label{bl:remIndM2M3AndUnifBinCM}
\begin{IEEEeqnarray}{rCl}
\bm P_V^{(m,m_\ry)} (v) & = & \frac{1}{|\setM_\rz|} \sum_{m_\rz \in \setM_\rz} \ind {v = V_{m,m_\ry,m_\rz}}, \quad v \in \setV, \\ \label{eq:remIndMzBinCM}
\tilde {\bm P}_V^{(m,m_\ry)} (v) & = & \begin{cases} \frac{1}{|\indexset {m,m_\ry}|} \sum_{v^\prime \in \indexset {m,m_\ry}} \ind {v = v^\prime} &\textnormal{if } \indexset {m,m_\ry} \neq \emptyset, \\ \ind {v = v^\star} &\textnormal{otherwise}, \end{cases} \quad v \in \setV. \label{eq:remUnifBinCM}
\end{IEEEeqnarray}
\end{subequations}
The latter PMF is reminiscent of the distribution we encountered in \eqref{eq:distIDMsgDMC} and \eqref{eq:distIDMsgDMC2} in the single-user case. The former is related to the common-message BC setting when we view $M_\setZ$ as uniform over $\setM_\setZ$. Like the proof of Claim~\ref{cl:toShowIDBC}, to establish \eqref{eq:toShowIDBCCM} it suffices to show that the two are similar in the sense that
\begin{IEEEeqnarray}{l}
\exists \, \tau > 0 \textnormal{ s.t.\ } \lim_{n \rightarrow \infty} \biggdistof { \max_{(m,m_\ry) \in \setM \times \setM_\ry} d \Bigl( \bm P_V^{(m,m_\ry)}, \tilde {\bm P}_V^{(m,m_\ry)} \Bigr) \geq e^{-n \tau} } = 0, \label{eq:toShowIDBC4CM}
\end{IEEEeqnarray}
which follows essentially along the line of arguments leading to \eqref{eq:toShowIDBC4} in the proof of Claim~\ref{cl:toShowIDBC}.
\end{proof}

\subsection{The Converse Part of Theorem~\ref{th:IDBCCM}}

We prove the following strong converse:

\begin{claim}\label{cl:toShowIDBCConvCM}
For every rate-triple $( R, R_\ry, R_\rz )$, every positive constants $\lambda_1^\ry, \, \lambda_2^\ry, \, \lambda_1^\rz, \, \lambda_2^\rz$ satisfying
\begin{IEEEeqnarray}{rCl}
\lambda_1^\ry + \lambda_2^\ry + \lambda_1^\rz + \lambda_2^\rz & < & 1, \label{eq:sumMissWrongBCCM}
\end{IEEEeqnarray}
and every $\epsilon > 0$ there exists some $\eta_0 \in \naturals$ so that, for every blocklength~$n \geq \eta_0$, every size-$\exp (\exp (n R))$ set $\setM$ of possible common ID messages, every size-$\exp (\exp (n R_\ry))$ set $\setM_\ry$ of possible ID messages for Receiver~$\ry$, and every size-$\exp (\exp (n R_\rz))$ set $\setM_\rz$ of possible ID messages for Receiver~$\rz$, a necessary condition for an $\bigl( n, \setM, \setM_\ry, \setM_\rz, \lambda_1^\ry, \lambda_2^\ry, \lambda_1^\rz, \lambda_2^\rz \bigr)$ ID code for the BC $\channel {y,z} x$ to exist is that for some PMF $P$ on $\setX$
\begin{subequations}\label{bl:converseBCCM}
\begin{IEEEeqnarray}{rCl}
R, \, R_\ry & < & \muti {P}{W_\ry} + \epsilon, \\
R, \, R_\rz & < & \muti {P}{W_\rz} + \epsilon.
\end{IEEEeqnarray}
\end{subequations}
\end{claim}

\begin{proof}
The proof is similar to that of Claim~\ref{cl:toShowIDBCConv}. Fix $\kappa^\ry, \, \kappa^\rz > 0$ that satisfy the following three: 1) $\lambda_1^\ry + \lambda_2^\ry < \kappa^\ry$; 2) $\lambda_1^\rz + \lambda_2^\rz < \kappa^\rz$; and 3) $\kappa^\ry + \kappa^\rz< 1$. (This is possible because of \eqref{eq:sumMissWrongBCCM}.) By Lemma~\ref{le:avgDistWeightOnTypes} there must exist some $\eta_0^\prime \in \naturals$ so that, for every blocklength~$n \geq \eta_0^\prime$, every size-$\exp (\exp (n R))$ set $\setM$ of possible common ID messages, every size-$\exp (\exp (n R_\ry))$ set $\setM_\ry$ of possible ID messages for Receiver~$\ry$, and every size-$\exp (\exp (n R_\rz))$ set $\setM_\rz$ of possible ID messages for Receiver~$\rz$, the following conditions are necessary for a collection of tuples $$\bigl\{ Q_{m,m_\ry,m_\rz}, \idsetre {m,m_\ry}, \idsetre {m,m_\rz} \bigr\}_{(m,m_\ry,m_\rz) \in \setM \times \setM_\ry \times \setM_\rz}$$ to be an $\bigl( n, \setM, \setM_\ry, \setM_\rz, \lambda_1^\ry, \lambda_2^\ry, \lambda_1^\rz, \lambda_2^\rz \bigr)$ ID code for the BC $\channel {y,z} x$: for
\begin{subequations}\label{bl:cmMessRyRzPrime}
\begin{IEEEeqnarray}{rCl}
R_\ry^\prime (n) & = & \frac{1}{n} \log \log \bigl( |\setM| \, |\setM_\ry| \bigr), \\
R_\rz^\prime (n) & = & \frac{1}{n} \log \log \bigl( |\setM| \, |\setM_\rz| \bigr)
\end{IEEEeqnarray}
\end{subequations}
the mixture PMFs on $\setX^n$
\begin{subequations}
\begin{IEEEeqnarray}{rCl}
Q_{m,m_\ry} & = & \frac{1}{|\setM_\rz|} \sum_{m_\rz \in \setM_\rz} Q_{m,m_\ry,m_\rz}, \quad (m,m_\ry) \in \setM \times \setM_\ry, \\
Q_{m,m_\rz} & = & \frac{1}{|\setM_\ry|} \sum_{m_\ry \in \setM_\ry} Q_{m,m_\ry,m_\rz}, \quad (m,m_\rz) \in \setM \times \setM_\rz, \\
Q & = & \frac{1}{|\setM| \, |\setM_\ry| \, |\setM_\rz|} \sum_{(m, m_\ry, m_\rz) \in \setM \times \setM_\ry \times \setM_\rz} Q_{m,m_\ry,m_\rz}
\end{IEEEeqnarray}
\end{subequations}
satisfy
\begin{IEEEeqnarray}{l}
Q \bigl( X^n \in \{ \vecx \in \setX^n \colon \muti {P_\vecx}{W_\ry} > R_\ry^\prime (n) - \epsilon \} \bigr) \nonumber \\
\quad = \frac{1}{|\setM| \, | \setM_\ry |} \sum_{(m,m_\ry) \in \setM \times \setM_\ry} Q_{m, m_\ry} \bigl( X^n \in \{ \vecx \in \setX^n \colon \muti {P_\vecx}{W_\ry} > R_\ry^\prime (n) - \epsilon \} \bigr) \\
\quad \geq 1 - \kappa^\ry - \exp \bigl\{ e^{n ( R_\ry^\prime (n) - \epsilon / 2 )} \bigr\} / \exp \bigl\{ e^{n R_\ry^\prime (n)} \bigr\} \label{eq:condSinceW1IDCodeCM}
\end{IEEEeqnarray}
and
\begin{IEEEeqnarray}{l}
Q \bigl( X^n \in \{ \vecx \in \setX^n \colon \muti {P_\vecx}{W_\rz} > R_\rz^\prime (n) - \epsilon \} \bigr) \nonumber \\
\quad = \frac{1}{| \setM | \, | \setM_\rz |} \sum_{(m,m_\rz) \in \setM \times \setM_\rz} Q_{m, m_\rz} \bigl( X^n \in \{ \vecx \in \setX^n \colon \muti {P_\vecx}{W_\rz} > R_\rz^\prime (n) - \epsilon \} \bigr) \\
\quad \geq 1 - \kappa^\rz - \exp \bigl\{ e^{n ( R_\rz^\prime (n) - \epsilon / 2 )} \bigr\} / \exp \bigl\{ e^{n R_\rz^\prime (n)} \bigr\}. \label{eq:condSinceW2IDCodeCM}
\end{IEEEeqnarray}
The Union-of-Events bound, \eqref{eq:condSinceW1IDCodeCM}, and \eqref{eq:condSinceW2IDCodeCM} imply that
\begin{IEEEeqnarray}{l}
Q \bigl( X^n \in \{ \vecx \in \setX^n \colon \muti {P_\vecx}{W_\ry} > R_\ry^\prime (n) - \epsilon, \, \muti {P_\vecx}{W_\rz} > R_\rz^\prime (n) - \epsilon \} \bigr) \nonumber \\
\quad \geq 1 - \kappa^\ry - \kappa^\rz - \exp \bigl\{ e^{n ( R_\ry^\prime (n) - \epsilon / 2 )} \bigr\} / \exp \bigl\{ e^{n R_\ry^\prime (n)} \bigr\} - \exp \bigl\{ e^{n ( R_\rz^\prime (n) - \epsilon / 2 )} \bigr\} / \exp \bigl\{ e^{n R_\rz^\prime (n)} \bigr\}.  \label{eq:BCExistsPMFGoodForBothCM}
\end{IEEEeqnarray}
Now let $\eta_0$ be the smallest integer $n \geq \eta_0^\prime$ for which the RHS of \eqref{eq:BCExistsPMFGoodForBothCM} is positive (such an $n$ must exist, because $\epsilon > 0$ and $\kappa^\ry + \kappa^\rz < 1$). By \eqref{bl:cmMessRyRzPrime}
\begin{subequations}
\begin{IEEEeqnarray}{rCl}
R_\ry^\prime (n) & \geq & \max \{ R, R_\ry \}, \\
R_\rz^\prime (n) & \geq & \max \{ R, R_\rz \},
\end{IEEEeqnarray}
\end{subequations}
and hence Claim~\ref{cl:toShowIDBCConvCM} follows: for every blocklength~$n \geq \eta_0$ a necessary condition for \eqref{eq:BCExistsPMFGoodForBothCM} to hold is that for some PMF $P$ on $\setX$ \eqref{bl:converseBCCM} holds.
\end{proof}

\section{A Proof of Theorem~\ref{th:IDBC1FBIB}}\label{app:IDBC1FBIB}

We prove Theorem~\ref{th:IDBC1FBIB} by fixing any input distribution $P \in \mathscr P (\setX)$ and any positive ID rate-pair $(R_\ry,R_\rz)$ satisfying
\begin{subequations}\label{bl:ratePairAch1FB}
\begin{IEEEeqnarray}{rCcCl}
0 & < & R_\ry & < & \ent {P W_\ry} \ind {\max_{\tilde P} \muti {\tilde P}{W_\ry} > 0}, \\
0 & < & R_\rz & < & \muti P{W_\rz}
\end{IEEEeqnarray}
\end{subequations}
and showing that the rate-pair $(R_\ry,R_\rz)$ is achievable. We assume that $\max_{\tilde P} \muti {\tilde P}{W_\ry}$, $\ent {P W_\ry}$, and $\muti P {W_\rz}$ are all positive; when they are not, the result follows from Theorem~\ref{th:IDDMC} and \cite{ahlswededueckfb89}. Let $\setM_\ry$ be a size-$\exp (\exp (n R_\ry))$ set of possible ID messages for Terminal~$\ry$, and let $\setM_\rz$ be a size-$\exp (\exp (n R_\rz))$ set of possible ID messages for Terminal~$\rz$. We next describe our random code construction and show that, for every positive $\lambda_1^\ry, \, \lambda_2^\ry, \, \lambda_1^\rz, \, \lambda_2^\rz$ and for every sufficiently-large $n$, it produces with high probability an $\bigl( n + \sqrt n, \setM_\ry, \setM_\rz, \lambda^{\ry}_1, \lambda^{\ry}_2, \lambda^{\rz}_1, \lambda^{\rz}_2 \bigr)$ ID code for the BC $\channel {y,z} x$ with one-sided feedback from Terminal~$\ry$. A rough description of the coding scheme that we propose can be found in Section~\ref{sec:BC1FB}.

\subparagraph*{Code Generation:} Fix an expected bin rate $\tilde R_\rz$ for Terminal~$\rz$, a pool rate $R_\poolre$, and a transmission rate $\hat R_\ry$ for Terminal~$\ry$ satisfying
\begin{subequations}\label{bl:binAndPoolRateConstraintsBC1FB}
\begin{IEEEeqnarray}{rCcCl}
R_\rz & < & \tilde R_\rz & < & \muti {P}{W_\rz}, \\
&& \muti {P}{W_\ry} & < & R_\poolre, \label{eq:poolRateConstYConstraints1FB} \\
&& \tilde R_\rz & < & R_\poolre, \\
0 & < & \hat R_\ry & < & \max_{\tilde P} \muti {P}{W_\ry}. \label{eq:transRateTerYConstraints1FB}
\end{IEEEeqnarray}
\end{subequations}
This is possible by \eqref{bl:ratePairAch1FB}. Draw $e^{n R_\poolre}$ $n$-tuples $\sim P^n$ independently and place them in a pool $\pool$. Index the $n$-tuples in the pool by the elements of a size-$e^{n R_\poolre}$ set $\setV$, e.g., $\{ 1, \ldots, e^{n R_\poolre} \}$, and denote by $\poolel v$ the $n$-tuple in $\pool$ that is indexed by~$v \in \setV$. Associate with each ID message $m_\rz \in \setM_\rz$ an index-set $\indexset {m_\rz}$ and a bin $\bin {m_\rz}$ as follows. Select each element of $\setV$ for inclusion in $\indexset {m_\rz}$ independently with probability $e^{-n( R_\poolre - \tilde R_\rz )}$, and let Bin~$\bin {m_\rz}$ be the multiset that contains all the $n$-tuples in the pool that are indexed by $\indexset {m_\rz}$, $$\bin {m_\rz} = \bigl\{ \poolel v, \, v \in \indexset {m_\rz} \bigr\}.$$ (Bin~$\bin {m_\rz}$ is thus of expected size $e^{n \tilde R_\rz}$.) Associate with each ID message-pair $(m_\ry,m_\rz) \in \setM_\ry \times \setM_\rz$ an index $V_{m_\ry,m_\rz}$ as follows. If $\indexset {m_\rz}$ is not empty, then draw $V_{m_\ry,m_\rz}$ uniformly over $\indexset {m_\rz}$. Otherwise let $V_{m_\ry,m_\rz} = v^\star$, where $v^\star$ is an arbitrary but fixed element of $\setV$. Let $\bigl\{ ( f_k, \phi_k ) \bigr\}_{k \in \naturals}$ be a sequence of blocklength-$k$, rate-$\hat R_\ry$ transmission codes for the marginal channel $W_\ry$ satisfying that the maximum error probability $\epsilon_k$ converges to zero as the blocklength~$k$ tends to infinity. (By \eqref{eq:transRateTerYConstraints1FB} such a transmission-code sequence exists.) For the code construction we use the blocklength-$\sqrt n$ transmission code $( f_{\sqrt n}, \phi_{\sqrt n} )$, which we denote by $( f, \phi )$. Denote the size-$2^{\sqrt n \hat R_\ry}$ set of possible transmission messages by $\setU$, so $f \colon \setU \rightarrow \setX^{\sqrt n}$, $\phi \colon \setY^{\sqrt n} \rightarrow \setU$, and
\begin{equation}
\epsilon_{\sqrt n} = \max_{u \in \setU} W^{\sqrt n}_\ry \Bigl( Y^{\sqrt n} \notin \phi^{-1} ( u ) \Bigl| f ( u ) \Bigr). \label{eq:epsilonSqrtn1FB}
\end{equation}
Associate with each pair $(\vecy, m_\ry) \in \setY^n \times \setM_\ry$ a transmission message $\cwind \vecy {m_\ry}$ by drawing the transmission messages independently and uniformly over $\setU$. Reveal the pool $\pool$, the index-sets $\bigl\{ \indexset {m_\rz} \bigr\}_{m_\rz \in \setM_\rz}$, the corresponding bins $\bigl\{ \bin {m_\rz} \bigr\}_{m_\rz \in \setM_\rz}$, the indices $\bigl\{ V_{m_\ry,m_\rz} \bigr\}_{(m_\ry,m_\rz) \in \setM_\ry \times \setM_\rz}$, the transmission code $( f, \phi )$, and the transmission messages $\bigl\{ \cwind \vecy {m_\ry} \bigr\}_{(\vecy, m_\ry) \in \setY^n \times \setM_\ry}$ to all parties. The encoding and decoding are determined by
\begin{IEEEeqnarray}{l}
\rcode = \Bigl( \pool, \bigl\{ \indexset {m_\rz} \bigr\}_{m_\rz \in \setM_\rz}, \bigl\{ V_{m_\ry, m_\rz} \bigr\}_{ (m_\ry, m_\rz ) \in \setM_\ry \times \setM_\rz }, (f, \phi), \bigl\{ \cwind \vecy {m_\ry} \bigr\}_{(\vecy, m_\ry) \in \setY^n \times \setM_\ry} \Bigr). \label{eq:randCodeBC1FB}
\end{IEEEeqnarray}

\subparagraph*{Encoding:} To send ID Message-Pair~$(m_\ry, m_\rz) \in \setM_\ry \times \setM_\rz$, the encoder transmits the sequence $\poolel {V_{m_\ry,m_\rz}} \circ f \bigl( \cwind {Y^n} {m_\ry} \bigr)$. Note that once the code \eqref{eq:randCodeBC1FB} has been constructed, the encoder is deterministic: The encoder first maps ID Message-Pair~$(m_\ry,m_\rz)$ to the $(m_\ry,m_\rz)$-codeword $\poolel {V_{m_\ry,m_\rz}}$, which it transmits during the first $n$ channel uses; it then observes the first $n$ channel outputs $Y^n$ at Receiver~$\ry$ through the feedback link; from $Y^n$ and ID Message~$m_\ry$, the encoder computes the $(Y^n,m_\ry)$-transmission-codeword $f \bigl( \cwind {Y^n} {m_\ry} \bigr)$, which it transmits during the remaining $\sqrt n$ channel uses.

\subparagraph*{Decoding:} In this section the function $\delta (\cdot)$ maps every nonnegative real number $u$ to $u  \ent {P \times W}$. The decoders choose $\epsilon > 0$ sufficiently small so that $3 \delta ( \epsilon ) < \ent {P W_\ry} - R_\ry$ and $2 \delta ( \epsilon ) < \muti {P}{W_\rz} - \tilde R_\rz$. The $m^\prime_\ry$-focused party at Terminal~$\ry$ guesses that $m^\prime_\ry$ was sent iff the Terminal-$\ry$ output-sequence $Y^{n + \sqrt n}$ satisfies that the decoding function $\phi$ maps $Y^{n + \sqrt n}_{n + 1}$ to the $(Y^n,m_\ry^\prime)$-transmission-message $\cwind {Y^n} {m_\ry^\prime}$, i.e., iff $\phi \bigl( Y_{n+1}^{n+\sqrt n} \bigr) = \cwind {Y^n}{m^\prime_\ry}$. The set $\idset {m^\prime_\ry}$ of Terminal-$\ry$ output-sequences $\vecy \in \setY^{n + \sqrt n}$ that result in the guess ``$m^\prime_\ry$ was sent'' is thus
\begin{IEEEeqnarray}{rCl}
\idset {m^\prime_\ry} & = & \Bigl\{\vecy \in \setY^{n+\sqrt n} \colon \phi \bigl( y_{n+1}^{n+\sqrt n} \bigr) = \cwind {y^n}{m^\prime_\ry} \Bigr\}. \label{eq:DefIDSetMy1FB}
\end{IEEEeqnarray}
The $m^\prime_\rz$-focused party at Terminal~$\rz$ guesses that $m^\prime_\rz$ was sent iff for some index $v \in \indexset {m_\rz^\prime}$ the $n$-tuple $\poolel v$ in Bin~$\bin {m_\rz^\prime}$ is jointly $\epsilon$-typical with the first $n$ channel outputs at Terminal-$\rz$, i.e., iff $( \poolel v, Z^n ) \in \eptyp (P \times W_\rz)$ for some $v \in \indexset {m_\rz^\prime}$. The set $\idset {m^\prime_\rz}$ of Terminal-$\rz$ output-sequences $\vecz \in \setZ^{n + \sqrt n}$ that result in the guess ``$m^\prime_\rz$ was sent'' is thus
\begin{IEEEeqnarray}{rCl}
\idset {m^\prime_\rz} & = & \Biggl( \bigcup_{v \in \indexset {m^\prime_\rz}} \setT^{ ( n  )}_\epsilon \bigl( P \times W_\rz \bigl| \poolel v \bigr) \Biggr) \times \setZ^{ \sqrt n }. \label{eq:DefIDSetMz1FB}
\end{IEEEeqnarray}

\subparagraph*{Analysis of the Probabilities of Missed and Wrong Identification:} We first note that $\rcode$ of \eqref{eq:randCodeBC1FB} (together with the fixed $n$ and the chosen $\epsilon$) fully specifies the encoding and guessing rules. Let $\dist$ be the distribution of  $\rcode$, and let $\Exop$ denote expectation w.r.t.\ $\dist$. Subscripts indicate conditioning on the event that some of the chance variables assume the values indicated by the subscripts, e.g., $\dist_{\indexsetre {m_\ry}}$ denotes the distribution conditional on $\indexset {m_\ry} = \indexsetre {m_\ry}$, and $\Exop_{\indexsetre {m_\ry}}$ denotes the expectation w.r.t.\ $\dist_{\indexsetre {m_\ry}}$.\\

The maximum probabilities of missed and wrong identification of the randomly constructed ID code are the random variables
\begin{subequations}\label{bl:prMissedWrongIDIDBC1FB}
\begin{IEEEeqnarray}{rCl}
P^\ry_{\textnormal{missed-ID}} & = & \max_{m_\ry \in \setM_\ry} \frac{1}{|\setM_\rz|} \sum_{m_\rz \in \setM_\rz} \sum_{\vecy \notin \idset {m_\ry}} W^n_\ry \bigl( y^n \bigl| \poolel {V_{m_\ry,m_\rz}} \bigr) W^{\sqrt n}_\ry \Bigl( y^{n + \sqrt n}_{n + 1} \Bigl| f \bigl( \cwind {y^n}{m_\ry} \bigr) \Bigr), \label{eq:prMissedIDIDBC1FB} \\
P^\rz_{\textnormal{missed-ID}} & = & \max_{m_\rz \in \setM_\rz} \frac{1}{|\setM_\ry|} \sum_{m_\ry \in \setM_\ry} W^n \bigl( Z^n \notin \idset {m_\rz} \bigl| \poolel {V_{m_\ry,m_\rz}} \bigr), \\
P^\ry_{\textnormal{wrong-ID}} & = & \max_{m_\ry \in \setM_\ry} \max_{m_\ry^\prime \neq m_\ry} \frac{1}{|\setM_\rz|} \! \sum_{m_\rz \in \setM_\rz} \sum_{\vecy \in \idset {m_\ry^\prime}} W^n_\ry \bigl( y^n \bigl| \poolel {V_{m_\ry,m_\rz}} \bigr) W^{\sqrt n}_\ry \Bigl( y^{n + \sqrt n}_{n + 1} \Bigl| f \bigl( \cwind {y^n}{m_\ry} \bigr) \Bigr), \label{eq:prWrongIDIDBC1FB} \\
P^\rz_{\textnormal{wrong-ID}} & = & \max_{m_\rz \in \setM_\rz} \max_{m_\rz^\prime \neq m_\rz} \frac{1}{|\setM_\ry|} \sum_{m_\ry \in \setM_\ry} W^n \bigl( Z^n \in \idset {m_\rz^\prime} \bigl| \poolel {V_{m_\ry,m_\rz}} \bigr).
\end{IEEEeqnarray}
\end{subequations}
They are fully specified by $\rcode$. To prove that for every choice of $\lambda_1^\ry, \, \lambda_2^\ry, \, \lambda_1^\rz, \, \lambda_2^\rz > 0$ and $n$ sufficiently large the constructed code is with high probability an $\bigl( n + \sqrt n, \setM_\ry, \setM_\rz, \lambda^{\ry}_1, \lambda^{\ry}_2, \lambda^{\rz}_1, \lambda^{\rz}_2 \bigr)$ ID code for the BC $\channel {y,z} x$ with one-sided feedback from Terminal~$\ry$, we prove the following stronger result:

\begin{claim}\label{cl:toShowIDBC1FB}
The probabilities $P^\ry_{\text{missed-ID}}$, $P^\rz_{\text{missed-ID}}$, $P^\ry_{\text{wrong-ID}}$, and $P^\rz_{\text{wrong-ID}}$ satisfy
\begin{subequations}\label{bl:claimIDBC1FB}
\begin{IEEEeqnarray}{rl}
\exists \, \{ \kappa_n \}_{n \in \naturals} \textnormal{ s.t. } & \lim_{n \rightarrow \infty} \kappa_n = 0 \textnormal{ and } \lim_{n \rightarrow \infty} \Bigdistof {\max \bigl\{ P^\ry_{\textnormal{missed-ID}}, P^\ry_{\textnormal{wrong-ID}} \bigr\} \geq \kappa_n} = 0, \label{eq:claimIDBC1FBY} \\
\exists \, \tau > 0 \textnormal{ s.t. } & \lim_{n \rightarrow \infty} \Bigdistof {\max \bigl\{ P^\rz_{\textnormal{missed-ID}}, P^\rz_{\textnormal{wrong-ID}} \bigr\} \geq e^{- n \tau}} = 0. \label{eq:claimIDBC1FBZ}
\end{IEEEeqnarray}
\end{subequations}
\end{claim}

\begin{proof}
We begin with \eqref{eq:claimIDBC1FBZ}. To prove \eqref{eq:claimIDBC1FBZ} we consider for each $m_\rz \in \setM_\rz$ two distributions on the set $\setV$, which indexes the pool $\pool$. We define for every $m_\rz \in \setM_\rz$ the PMFs on $\setV$
\begin{subequations}\label{bl:remIndMyAndUnifBin1FB}
\begin{IEEEeqnarray}{rCl}
\bm P_V^{(m_\rz)} (v) & = & \frac{1}{|\setM_\ry|} \sum_{m_\ry \in \setM_\ry} \ind {v = V_{m_\ry,m_\rz}}, \quad v \in \setV, \\ \label{eq:remIndMyBin1FB}
\tilde {\bm P}_V^{(m_\rz)} (v) & = & \begin{cases} \frac{1}{|\indexset {m_\rz}|} \sum_{v^\prime \in \indexset {m_\rz}} \ind {v = v^\prime} &\textnormal{if } \indexset {m_\rz} \neq \emptyset, \\ \ind {v = v^\star} &\textnormal{otherwise}, \end{cases} \quad v \in \setV. \label{eq:remUnifBin1FB}
\end{IEEEeqnarray}
\end{subequations}
The latter PMF is reminiscent of the distribution we encountered in \eqref{eq:distIDMsgDMC} and \eqref{eq:distIDMsgDMC2} in the single-user case. The former is related to the BC setting when we view $M_\ry$ as uniform over $\setM_\ry$. Like the proof of Claim~\ref{cl:toShowIDBC}, to establish \eqref{bl:claimIDBC1FB} it suffices to show that the two are similar in the sense that
\begin{IEEEeqnarray}{l}
\exists \, \tau > 0 \textnormal{ s.t.\ } \lim_{n \rightarrow \infty} \biggdistof { \max_{m_\rz \in \setM_\rz} d \Bigl( \bm P_V^{(m_\rz)}, \tilde {\bm P}_V^{(m_\rz)} \Bigr) \geq e^{-n \tau} } = 0. \label{eq:toShowIDBC41FB}
\end{IEEEeqnarray}

To establish \eqref{eq:toShowIDBC41FB}, we adapt the line of arguments leading to \eqref{eq:toShowIDBC4} in the proof of Claim~\ref{cl:toShowIDBC}. Fix some $\mu$ satisfying
\begin{IEEEeqnarray}{l}
0 < \mu < \tilde R_\rz - R_\rz, \label{eq:muIDBCPf1FB}
\end{IEEEeqnarray}
and let
\begin{IEEEeqnarray}{l}
\delta_n = e^{-n \mu / 2}. \label{eq:deltanIDBCPf1FB}
\end{IEEEeqnarray}
Introduce the set $\setH^\rz_\mu$ comprising the realizations $\{ \indexsetre \nu \}_{\nu \in \setM_\rz}$ of the index-sets $\{ \indexset {\nu} \}_{\nu \in \setM_\rz}$ satisfying that
 \begin{IEEEeqnarray}{rCl}
|\indexsetre \nu| > (1 - \delta_n) e^{n \tilde R_\rz}, \, \forall \, \nu \in \setM_\rz. \label{eq:IDBCSetH1FB}
\end{IEEEeqnarray}
We upper-bound $\max_{m_\rz \in \setM_\rz} d \bigl( \bm P_V^{(m_\rz)}, \tilde {\bm P}_V^{(m_\rz)} \bigr)$ differently depending on whether or not $\{ \indexset \nu \}$ is in $\setH^\rz_\mu$, where $\{ \indexset \nu \}$ is short for $\{ \indexset {\nu} \}_{\nu \in \setM_\rz}$. If $\{ \indexset \nu \} \notin \setH^\rz_\mu$, then we upper-bound it by one (which is an upper bound on the Total-Variation distance between any two probability measures) to obtain for every $\tau > 0$
\begin{IEEEeqnarray}{l}
\biggdistof { \max_{m_\rz \in \setM_\rz} d \Bigl( \bm P_V^{(m_\rz)}, \tilde {\bm P}_V^{(m_\rz)} \Bigr) \geq e^{-n \tau} } \nonumber \\
\quad \leq \bigdistof { \{ \indexset \nu \} \notin \setH^\rz_\mu } + \sum_{ \indexsetsre \in \setH^\rz_\mu } \bigdistof { \{ \indexset \nu \} =  \indexsetsre } \, \biggdistsubof { \indexsetsre } { \max_{m_\rz \in \setM_\rz} d \Bigl( \bm P_V^{(m_\rz)}, \tilde {\bm P}_V^{(m_\rz)} \Bigr) \geq e^{-n \tau} }. \label{eq:IDBCToShowNotinAndInH1FB}
\end{IEEEeqnarray}

We consider the two terms on the RHS of \eqref{eq:IDBCToShowNotinAndInH1FB} separately, beginning with $\bigdistof { \{ \indexset \nu \} \notin \setH^\rz_\mu }$. By the line of arguments leading to \eqref{eq:allBmySuffLarge} in the proof of Claim~\ref{cl:toShowIDBCConv}
\begin{IEEEeqnarray}{rCl}
\bigdistof { \{ \indexset \nu \} \notin \setH^\rz_\mu } & \leq & |\setM_\rz| \exp \bigl\{ - e^{n (\tilde R_\rz - \mu) - \log 2} \bigr\} \\
& \stackrel{(a)}\rightarrow & 0 \, (n \rightarrow \infty), \label{eq:allBmzSuffLarge1FB}
\end{IEEEeqnarray}
where $(a)$ holds because $|\setM_\rz| = \exp (\exp (n R_\rz))$ and by \eqref{eq:muIDBCPf1FB}.

Having established \eqref{eq:allBmzSuffLarge1FB}, we return to \eqref{eq:IDBCToShowNotinAndInH1FB} and conclude the proof of \eqref{eq:toShowIDBC41FB} by showing that
\begin{equation}
\exists \, \tau > 0 \textnormal{ s.t. } \lim_{n \rightarrow \infty} \max_{\indexsetsre \in \setH^\rz_\mu} \biggdistsubof { \indexsetsre } { \max_{m_\rz \in \setM_\rz} d \Bigl( \bm P_V^{(m_\rz)}, \tilde {\bm P}_V^{(m_\rz)} \Bigr) \geq e^{-n \tau} } = 0. \label{eq:toShowIDBC51FB}
\end{equation}
To prove \eqref{eq:toShowIDBC51FB}, let us henceforth assume that $n$ is large enough so that the following two inequalities hold:
\begin{subequations}\label{bl:IDBCnLargeEnough1FB}
\begin{IEEEeqnarray}{rCl}
(1 - \delta_n) e^{n \tilde R_\rz} & \geq & 1, \label{eq:IDBCnLargeEnough1FB} \\
\delta_n & \leq & 1/2, \label{eq:IDBCnLargeEnough21FB}
\end{IEEEeqnarray}
\end{subequations}
where $\delta_n$ is defined in \eqref{eq:deltanIDBCPf1FB}. (This is possible, because $\delta_n$ converges to zero as $n$ tends to infinity and $\tilde R_\rz > 0$.) Fix any realization $\indexsetsre$ in $\setH^\rz_\mu$. By \eqref{eq:IDBCSetH1FB} (which holds because $\indexsetsre \in \setH^\rz_\mu$) and \eqref{eq:IDBCnLargeEnough1FB}, $\indexsetre {m_\rz}$ is nonempty. For every fixed $v \in \setV$ we therefore have that under $\dist_{ \indexsetsre }$ the $\exp (\exp (n R_\ry))$ binary random variables $\bigl\{ \ind {v = V_{m_\ry,m_\rz}} \bigr\}_{m_\ry \in \setM_\ry}$ are IID and of mean
\begin{IEEEeqnarray}{l}
\begin{cases}
\frac{1}{|\indexsetre {m_\rz}|} &\textnormal{if } v \in \indexsetre {m_\rz}, \\
0 &\textnormal{if } v \notin \indexsetre {m_\rz}. \label{eq:IDBC1FBMeans}
\end{cases}
\end{IEEEeqnarray}
With \eqref{eq:IDBC1FBMeans} at hand, we can establish \eqref{eq:toShowIDBC51FB} essentially along the line of arguments leading to \eqref{eq:toShowIDBC5} in the proof of Claim~\ref{cl:toShowIDBC}.\\

Having established \eqref{eq:claimIDBC1FBZ}, we return to \eqref{bl:claimIDBC1FB} and conclude the proof by establishing \eqref{eq:claimIDBC1FBY}. We first observe that if the ID message that is sent to Terminal~$\rz$ is drawn uniformly over $\setM_\rz$, then the ID code that is used to send the ID message intended for Receiver~$\ry$ is similar to the common-randomness ID code \cite[Section~IV]{ahlswededueckfb89} for the DMC $\channely y x$ with perfect feedback. The difference is that---unlike the common-randomness ID code \cite{ahlswededueckfb89}---the common randomness $Y^n$ is not generated by drawing the first $n$ channel inputs $X^n \sim P^n$, irrespective of the ID message that is sent to Receiver~$\ry$. Instead, if the ID message that is sent to Receiver~$\ry$ is $m_\ry$ and the ID message that is sent to Receiver~$\rz$ is drawn uniformly over $\setM_\rz$, then $X^n$ is drawn from the PMF on $\setX^n$
\begin{IEEEeqnarray}{l}
\bm P^{(m_\ry)}_{X^n} ( \vecx ) = \frac{1}{| \setM_\rz |} \sum_{m_\rz \in \setM_\rz} \ind {\vecx = \poolel {V_{m_\ry,m_\rz}}}, \quad \vecx \in \setX^n. \label{eq:distXCondBin1FB}
\end{IEEEeqnarray}
As we argue next, the reasoning of \cite{ahlswededueckfb89} nevertheless applies.

The maximum probability of missed identification satisfies the upper bound
\begin{IEEEeqnarray}{l}
P^\ry_{\textnormal{missed-ID}} \nonumber \\
\quad \stackrel{(a)}= \max_{m_\ry \in \setM_\ry} \frac{1}{|\setM_\rz|} \sum_{m_\rz \in \setM_\rz} \sum_{\vecy \notin \idset {m_\ry}} W^n_\ry \bigl( y^n \bigl| \poolel {V_{m_\ry,m_\rz}} \bigr) W^{\sqrt n}_\ry \Bigl( y^{n + \sqrt n}_{n + 1} \Bigl| f \bigl( \cwind {y^n}{m_\ry} \bigr) \Bigr) \\
\quad \stackrel{(b)}= \max_{m_\ry \in \setM_\ry} \frac{1}{|\setM_\rz|} \sum_{m_\rz \in \setM_\rz} \sum_{\vecy^\prime \in \setY^n} W^n_\ry \bigl( \vecy^\prime \bigl| \poolel {V_{m_\ry,m_\rz}} \bigr) W^{\sqrt n}_\ry \Bigl( Y^{\sqrt n} \notin \phi^{-1} \bigl( \cwind {\vecy^\prime}{m_\ry} \bigr) \Bigl| f \bigl( \cwind {\vecy^\prime}{m_\ry} \bigr) \Bigr) \\
\quad \stackrel{(c)}\leq \max_{m_\ry \in \setM_\ry} \frac{1}{|\setM_\rz|} \sum_{m_\rz \in \setM_\rz} \sum_{\vecy^\prime \in \setY^n} W^n_\ry \bigl( \vecy^\prime \bigl| \poolel {V_{m_\ry,m_\rz}} \bigr) \, \epsilon_{\sqrt n} \\
\quad = \epsilon_{\sqrt n},
\end{IEEEeqnarray}
where $(a)$ holds by \eqref{eq:prMissedIDIDBC1FB}; $(b)$ holds by \eqref{eq:DefIDSetMy1FB}; and $(c)$ holds by \eqref{eq:epsilonSqrtn1FB}. This, combined with the Union-of-Events bound and the fact that $\epsilon_{\sqrt n}$ converges to zero as $n$ tends to infinity, implies that to establish \eqref{eq:claimIDBC1FBY} it suffices to show that
\begin{equation}
\exists \, \{ \kappa_n \}_{n \in \naturals} \textnormal{ s.t. } \lim_{n \rightarrow \infty} \kappa_n = 0 \textnormal{ and } \lim_{n \rightarrow \infty} \bigdistof { P^\ry_{\textnormal{wrong-ID}} \geq \kappa_n } = 0. \label{eq:claimIDBC1FBY2}
\end{equation}

Before we establish \eqref{eq:claimIDBC1FBY2}, we first show that
\begin{equation}
\exists \, \{ \lambda_n \}_{n \in \naturals} \textnormal{ s.t. } \lim_{n \rightarrow \infty} \lambda_n = 0 \textnormal{ and } \lim_{n \rightarrow \infty} \biggdistof {\max_{m_\ry \in \setM_\ry} d \Bigl( \bm P^{(m_\ry)}_{X^n} W_\ry^n, ( P W_\ry )^n \Bigr) \geq \lambda_n } = 0. \label{eq:claimIDBC1FBYTV}
\end{equation}
(This is useful, because, if the ID message that is sent to Receiver~$\ry$ is $m_\ry$ and the ID message that is sent to Terminal~$\rz$ is drawn uniformly over $\setM_\rz$, then we generate the common randomness $Y^n \sim \bm P^{(m_\ry)}_{X^n} W_\ry^n$, whereas the common-randomness ID code \cite{ahlswededueckfb89} for the DMC $\channely y x$ with perfect feedback generates the common randomness $Y^n \sim (P W_\ry)^n$ irrespective of $m_\ry \in \setM_\ry$.) For every $m_\ry \in \setM_\ry$ define the PMF on $\setV$
\begin{IEEEeqnarray}{rCl}
\bm P_V^{(m_\ry)} (v) & = & \frac{1}{|\setM_\rz|} \sum_{m_\rz \in \setM_\rz} \ind {v = V_{m_\ry,m_\rz}}, \quad v \in \setV, \label{eq:remIndMzBin1FB}
\end{IEEEeqnarray}
let $U_V$ denote the uniform distribution on $\setV$, define the conditional PMF
\begin{equation}
\bm P_{X^n|V} ( \vecx | v ) = \ind {\vecx = \poolel v}, \quad (\vecx,v) \in \setX^n \times \setV,
\end{equation}
and note that for every $m_\ry \in \setM_\ry$
\begin{IEEEeqnarray}{rCl}
\Bigl( \bm P_{X^n}^{(m_\ry)} W_\ry^n \Bigr) ( \vecy ) & = & \Bigl( \bm P_V^{(m_\ry)} \bm P_{X^n|V} W_\ry^n \Bigr) (\vecy), \quad \vecy \in \setY^n. \label{eq:remIndMzXFromBin1FB}
\end{IEEEeqnarray}
This implies that $d \bigl( \bm P_{X^n}^{(m_\ry)} W_\ry^n, ( P W_\ry )^n \bigr)$ satisfies the upper bound
\begin{IEEEeqnarray}{l}
d \Bigl( \bm P_{X^n}^{(m_\ry)} W_\ry^n, ( P W_\ry )^n \Bigr) \nonumber \\
\quad \stackrel{(a)}\leq d \Bigl( \bm P_{X^n}^{(m_\ry)} W_\ry^n, U_V \bm P_{X^n|V} W_\ry^n \Bigr) + d \bigl( U_V \bm P_{X^n|V} W_\ry^n, ( P W_\ry )^n \bigr) \\
\quad \stackrel{(b)}\leq d \Bigl( \bm P_V^{(m_\ry)} \bm P_{X^n|V} W_\ry^n, U_V \bm P_{X^n|V} W_\ry^n \Bigr) + d \bigl( U_V \bm P_{X^n|V} W_\ry^n, ( P W_\ry )^n \bigr) \\
\quad \stackrel{(c)}\leq d \Bigl( \bm P_V^{(m_\ry)}, U_V \Bigr) + d \bigl( U_V \bm P_{X^n|V} W_\ry^n, ( P W_\ry )^n \bigr), \quad m_\ry \in \setM_\ry, \label{eq:totVarDistTriangDataPr1FB}
\end{IEEEeqnarray}
where $(a)$ follows from the Triangle inequality; $(b)$ holds by \eqref{eq:remIndMzXFromBin1FB}; and $(c)$ follows from the Data-Processing inequality for the Total-Variation distance \cite[Lemma~1]{cannoneronservedio15}. In \cite{hanverdu93} it is shown that by \eqref{eq:poolRateConstYConstraints1FB}
\begin{IEEEeqnarray}{l}
\BigEx {}{d \bigl( U_V \bm P_{X^n|V} W_\ry^n, ( P W_\ry )^n \bigr)} \rightarrow 0 \, ( n \rightarrow \infty ).
\end{IEEEeqnarray}
Consequently, Markov's inequality implies that
\begin{IEEEeqnarray}{rCl}
\exists \, \{ \lambda_n \}_{n \in \naturals} \textnormal{ s.t. } \lim_{n \rightarrow \infty} \lambda_n = 0 \textnormal{ and } \lim_{n \rightarrow \infty} \Bigdistof { d \bigl( U_V \bm P_{X^n|V} W_\ry^n, ( P W_\ry )^n \bigr) \geq \lambda_n } = 0.
\end{IEEEeqnarray}
This, combined with \eqref{eq:totVarDistTriangDataPr1FB} and the Union-of-Events bound, implies that to establish \eqref{eq:claimIDBC1FBYTV} it suffices to show that
\begin{equation}
\exists \, \{ \lambda_n \}_{n \in \naturals} \textnormal{ s.t. } \lim_{n \rightarrow \infty} \lambda_n = 0 \textnormal{ and } \lim_{n \rightarrow \infty} \biggdistof { \max_{m_\ry \in \setM_\ry} d \Bigl( \bm P_V^{(m_\ry)}, U_V \Bigr) \geq \lambda_n } = 0. \label{eq:claimIDBC1FBYTV2}
\end{equation}

Fix some $\lambda$ satisfying
\begin{equation}
0 < \lambda < R_\rz, \label{eq:lambdaFB1}
\end{equation}
and let
\begin{equation}
\xi_n = \exp \bigl\{ -e^{n \lambda} \bigr\}. \label{eq:xiFB1}
\end{equation}
For every $v \in \setV$ the $\exp(\exp (n R_\rz))$ binary random variables $\bigl\{ \ind {v = V_{m_\ry,m_\rz}} \bigr\}_{m_\rz \in \setM_\rz}$ are IID and have mean $1 / |\setV|$. Consequently, Hoeffding's inequality (Proposition~\ref{pr:hoeffding}) and the Union-of-Events bound imply that for every fixed $v \in \setV$
\begin{IEEEeqnarray}{l}
\biggdistof { \Bigl| \bm P_V^{(m_\ry)} (v) - U_V (v) \Bigr| \geq \xi_n } \nonumber \\
\quad = \Biggdistof { \biggl| \frac{1}{|\setM_\rz|} \sum_{m_\rz \in \setM_\rz} \ind {v = V_{m_\ry,m_\rz}} - \bigEx {}{\ind {v = V_{m_\ry,m_\rz}}} \biggr| \geq \xi_n } \\
\quad \leq 2 \exp \bigl\{ - 2 \, |\setM_\rz| \, \xi_n^2 \bigr\}, \quad v \in \setV,
\end{IEEEeqnarray}
where the first equality holds because $U_V (v)$ and $\bigEx {}{\ind {v = V_{m_\ry,m_\rz}}}$ both equal $1 / |\setV|$. This, combined with the Union-of-Events bound, implies that
\begin{IEEEeqnarray}{l}
\biggdistof { \exists \, v \in \setV \colon \Bigl| \bm P_V^{(m_\ry)} (v) - U_V (v) \Bigr| \geq \xi_n } \nonumber \\
\quad \leq 2 \, |\setV| \exp \bigl\{ - 2 \, |\setM_\rz| \, \xi_n^2 \bigr\}. \label{eq:claimIDBC1FBYTV2ExistsVTooLarge}
\end{IEEEeqnarray}
Consequently,
\begin{IEEEeqnarray}{l}
\biggdistof { d \Bigl( \bm P_V^{(m_\ry)}, U_V \Bigr) \geq |\setV| \, \xi_n / 2 } \\
\quad \stackrel{(a)}= \Biggdistof { \sum_{v \in \setV} \Bigl| \bm P_V^{(m_\ry)} (v) - U_V (v) \Bigr| \geq |\setV| \, \xi_n } \\
\quad \leq \biggdistof { \exists \, v \in \setV \colon \Bigl| \bm P_V^{(m_\ry)} (v) - U_V (v) \Bigr| \geq \xi_n } \\
\quad \stackrel{(b)}\leq 2 \, |\setV| \exp \{ - 2 \, |\setM_\rz| \, \xi_n^2 \}, \label{eq:claimIDBC1FBYTV2TooLarge}
\end{IEEEeqnarray}
where $(a)$ holds by definition of the Total-Variation distance; and $(b)$ holds by \eqref{eq:claimIDBC1FBYTV2ExistsVTooLarge}. Having obtained \eqref{eq:claimIDBC1FBYTV2TooLarge} for every fixed $m_\ry \in \setM_\ry$, we are now ready to tackle the maximum over $m_\ry \in \setM_\ry$ and prove \eqref{eq:claimIDBC1FBYTV2} and hence \eqref{eq:claimIDBC1FBYTV}:
\begin{IEEEeqnarray}{l}
\biggdistof { \exists \, m_\ry \in \setM_\ry \colon d \Bigl( \bm P_V^{(m_\ry)}, U_V \Bigr) \geq |\setV | \, \xi_n / 2 } \nonumber \\
\quad \stackrel{(a)}\leq \sum_{m_\ry \in \setM_\ry} \biggdistof { d \Bigl( \bm P_V^{(m_\ry)}, U_V \Bigr) \geq |\setV| \, \xi_n / 2 } \\
\quad \stackrel{(b)}\leq 2 \, |\setV| \, |\setM_\ry| \exp \bigl\{ - 2 \, |\setM_\rz| \, \xi_n^2 \bigr\} \\
\quad \stackrel{(c)}\rightarrow 0 \, (n \rightarrow \infty),
\end{IEEEeqnarray}
where $(a)$ follows from the Union-of-Events bound; $(b)$ holds by \eqref{eq:claimIDBC1FBYTV2TooLarge}; and $(c)$ holds because $|\setV| = e^{n R_\poolre}$, $|\setM_\ry| = \exp (\exp (n R_\ry))$, $|\setM_\rz| = \exp (\exp (n R_\rz))$, and by \eqref{eq:lambdaFB1} and \eqref{eq:xiFB1}.

We next conclude the proof of Claim~\ref{cl:toShowIDBC1FB} by establishing \eqref{eq:claimIDBC1FBY2}. To that end, we use \eqref{eq:claimIDBC1FBYTV}, which allows us to follow Ahlswede and Dueck's line of arguments \cite{ahlswededueckfb89}. We begin by upper-bounding $$\frac{1}{|\setM_\rz|} \! \sum_{m_\rz \in \setM_\rz} \sum_{\vecy \in \idset {m_\ry^\prime}} W^n_\ry \bigl( y^n \bigl| \poolel {V_{m_\ry,m_\rz}} \bigr) W^{\sqrt n}_\ry \Bigl( y^{n + \sqrt n}_{n + 1} \Bigl| f \bigl( \cwind {y^n}{m_\ry} \bigr) \Bigr)$$ for fixed distinct $m_\ry, \, m^\prime_\ry \in \setM_\ry$. Later we will maximize over such $m_\ry, \, m^\prime_\ry$. For every fixed distinct $m_\ry, \, m^\prime_\ry \in \setM_\ry$
\begin{IEEEeqnarray}{l}
\frac{1}{|\setM_\rz|} \! \sum_{m_\rz \in \setM_\rz} \sum_{\vecy \in \idset {m_\ry^\prime}} W^n_\ry \bigl( y^n \bigl| \poolel {V_{m_\ry,m_\rz}} \bigr) W^{\sqrt n}_\ry \Bigl( y^{n + \sqrt n}_{n + 1} \Bigl| f \bigl( \cwind {y^n}{m_\ry} \bigr) \Bigr) \nonumber \\
\quad \stackrel{(a)}= \sum_{\vecy^\prime \in \setY^n} \Bigl( \bm P^{(m_\ry)}_{X^n} W_\ry^n \Bigr) \bigl( \vecy^\prime \bigr) W^{\sqrt n}_\ry \Bigl( Y^{\sqrt n} \in \phi^{-1} \bigl( \cwind {\vecy^\prime}{m_\ry^\prime} \bigr) \Bigl| f \bigl( \cwind {\vecy^\prime}{m_\ry} \bigr) \Bigr) \\
\quad \stackrel{(b)}\leq \sum_{\vecy^\prime \in \setY^n} \Bigl( \bm P^{(m_\ry)}_{X^n} W_\ry^n \Bigr) \bigl( \vecy^\prime \bigr) \ind { \cwind {\vecy^\prime}{m_\ry^\prime} = \cwind {\vecy^\prime}{m_\ry} } \nonumber \\
\qquad + \sum_{\vecy^\prime \in \setY^n} \Bigl( \bm P^{(m_\ry)}_{X^n} W_\ry^n \Bigr) \bigl( \vecy^\prime \bigr) W^{\sqrt n}_\ry \Bigl( Y^{\sqrt n} \notin \phi^{-1} \bigl( \cwind {\vecy^\prime}{m_\ry} \bigr) \Bigl| f \bigl( \cwind {\vecy^\prime}{m_\ry} \bigr) \Bigr) \\
\quad \stackrel{(c)}\leq \sum_{\vecy^\prime \in \setY^n} \ind { \cwind {\vecy^\prime}{m_\ry^\prime} = \cwind {\vecy^\prime}{m_\ry} } \Bigl( \bm P^{(m_\ry)}_{X^n} W_\ry^n \Bigr) \bigl( \vecy^\prime \bigr) + \epsilon_{\sqrt n}, \label{eq:ubPrWrongIDY1FB}
\end{IEEEeqnarray}
where $(a)$ holds by \eqref{eq:DefIDSetMy1FB} and \eqref{eq:distXCondBin1FB}; $(b)$ follows from the monotonicity of probability and the Union-of-Events bound; and $(c)$ holds by \eqref{eq:epsilonSqrtn1FB}. Let $\eptyp$ be short for $\eptyp ( P W_\ry )$. The first term in \eqref{eq:ubPrWrongIDY1FB} satisfies the upper bound
\begin{IEEEeqnarray}{l}
\sum_{\vecy^\prime \in \setY^n} \ind { \cwind {\vecy^\prime}{m_\ry^\prime} = \cwind {\vecy^\prime}{m_\ry} } \Bigl( \bm P^{(m_\ry)}_{X^n} W_\ry^n \Bigr) (\vecy^\prime) \\
\quad = \Bigl( \bm P^{(m_\ry)}_{X^n} W_\ry^n \Bigr) \Bigl( Y^n \in \bigl\{ \vecy \in \setY^n \colon \cwind {\vecy}{m_\ry^\prime} = \cwind {\vecy}{m_\ry} \bigr\} \Bigr) \\
\quad \stackrel{(a)}\leq ( P W_\ry )^n \Bigl( Y^n \in \bigl\{ \vecy \in \setY^n \colon \cwind {\vecy}{m_\ry^\prime} = \cwind {\vecy}{m_\ry} \bigr\} \Bigr) + d \Bigl( \bm P^{(m_\ry)}_{X^n} W^n_\ry, ( P W_\ry )^n \Bigr) \\
\quad \stackrel{(b)}\leq ( P W_\ry )^n \Bigl( Y^n \in \bigl\{ \vecy \in \eptyp \colon \cwind {\vecy}{m_\ry^\prime} = \cwind {\vecy}{m_\ry} \bigr\} \Bigr) + ( P W_\ry )^n \bigl( Y^n \notin \eptyp \bigr) \nonumber \\
\qquad + d \Bigl( \bm P^{(m_\ry)}_{X^n} W^n_\ry, ( P W_\ry )^n \Bigr), \label{eq:ubPrWrongIDY1FB2}
\end{IEEEeqnarray}
where$(a)$ holds by definition of the Total-Variation distance; and $(b)$ follows from the monotonicity of probability and the Union-of-Events bound. Using \eqref{eq:claimIDBC1FBYTV}, that $\epsilon_{\sqrt n}$ converges to zero as $n$ tends to infinity, and that $( P W_\ry )^n \bigl( Y^n \notin \eptyp \bigr)$ decays exponentially in $n$, we obtain from \eqref{eq:ubPrWrongIDY1FB}, \eqref{eq:ubPrWrongIDY1FB2}, and the Union-of-Events bound that to establish \eqref{eq:claimIDBC1FBY2} it suffices to show that
\begin{IEEEeqnarray}{l}
\exists \, \{ \kappa_n \}_{n \in \naturals} \textnormal{ s.t. } \lim_{n \rightarrow \infty} \kappa_n = 0 \textnormal{ and } \nonumber \\
\quad \lim_{n \rightarrow \infty} \biggdistof {\exists \, m_\ry, m_\ry^\prime \in \setM_\ry, \, m_\ry \neq m_\ry^\prime \colon ( P W_\ry )^n \Bigl( Y^n \in \bigl\{ \vecy \in \eptyp \colon \cwind {\vecy}{m_\ry^\prime} = \cwind {\vecy}{m_\ry} \bigr\} \Bigr) \geq \kappa_n } = 0. \label{eq:claimIDBC1FBY3}
\end{IEEEeqnarray}

Fix some $\rho$ satisfying
\begin{equation}
0 < \rho < \ent {P W_\ry} - R_\ry - 3 \delta ( \epsilon ), \label{eq:rhoFB1}
\end{equation}
and let
\begin{equation}
\alpha_n = \max \bigl\{ 2 / |\setU|, e^{-n \rho / 2} \bigr\}. \label{eq:alphaFB1}
\end{equation}
The binary random variables $\bigl\{ \ind {\cwind {\vecy}{m_\ry^\prime} = \cwind {\vecy}{m_\ry}} \bigr\}_{\vecy \in \setY^n}$ are IID with mean
\begin{equation}
\bigEx {}{\ind {\cwind {\vecy}{m_\ry^\prime} = \cwind {\vecy}{m_\ry}}} = \frac{1}{|\setU|}, \quad \vecy \in \setY^n. \label{eq:meanTransCodesSameFB1}
\end{equation}
Consequently, Hoeffding's inequality (Proposition~\ref{pr:hoeffding}) implies that
\begin{IEEEeqnarray}{l}
\biggdistof { \exists \, m_\ry, m^\prime_\ry \in \setM_\ry, \, m_\ry \neq m^\prime_\ry \colon ( P W_\ry )^n \Bigl( Y^n \in \bigl\{ \vecy \in \eptyp \colon \cwind {\vecy}{m_\ry^\prime} = \cwind {\vecy}{m_\ry} \bigr\} \Bigr) \geq \alpha_n } \nonumber \\
\quad = \Biggdistof { \exists \, m_\ry, m^\prime_\ry \in \setM_\ry, \, m_\ry \neq m^\prime_\ry \colon \sum_{\vecy \in \eptyp} ( P W_\ry )^n ( \vecy ) \ind{ \cwind {\vecy}{m_\ry} = \cwind {\vecy}{m^\prime_\ry} } \geq \alpha_n } \\
\quad \stackrel{(a)}\leq \sum_{m_\ry \in \setM_\ry} \sum_{m^\prime_\ry \neq m_\ry} \Biggdistof { \sum_{\vecy \in \eptyp} ( P W_\ry )^n ( \vecy ) \ind{ \cwind {\vecy}{m_\ry} = \cwind {\vecy}{m^\prime_\ry} } \geq \alpha_n } \\
\quad \stackrel{(b)}\leq | \setM_\ry |^2 \exp \Biggl\{ - \frac{ 2 \bigl( \alpha_n - 1 / | \setU | \bigr)^2 } { \sum_{\vecy \in \eptyp} \bigl( ( P W_\ry )^n ( \vecy ) \bigr)^2 } \Biggr\} \\
\quad \stackrel{(c)}\leq | \setM_\ry |^2 \exp \bigl\{ - e^{ n ( \ent {P W_\ry} - \rho - 3 \delta ( \epsilon ) ) - \log 2} \bigr\} \\
\quad \stackrel{(d)}\rightarrow 0 \, ( n \rightarrow \infty ), \label{eq:smallProbTheSame1FB}
\end{IEEEeqnarray}
where $(a)$ follows from the Union-of-Events bound; $(b)$ follows from \eqref{eq:meanTransCodesSameFB1} and Hoeffding's inequality (Proposition~\ref{pr:hoeffding}); $(c)$ holds by \eqref{eq:alphaFB1} and because
\begin{IEEEeqnarray}{rCl}
( P W_\ry )^n ( \vecy ) & \leq & e^{ - n ( \ent {P W_\ry} - \delta ( \epsilon ) ) }, \nonumber \\
|\eptyp| & \leq & e^{ n ( \ent {P W_\ry} + \delta (\epsilon) ) };
\end{IEEEeqnarray}
and $(d)$ holds because $| \setM_\ry | = \exp ( \exp ( n R_\ry ))$ and by \eqref{eq:rhoFB1}. Since $\alpha_n$ of \eqref{eq:alphaFB1} converges to zero as $n$ tends to infinity, this implies \eqref{eq:claimIDBC1FBY3} and hence concludes the proof.
\end{proof}

\section{A Proof of Theorem~\ref{th:IDBC1FBOB}}\label{app:IDBC1FBOB}

\subsection{A Useful Lemma}\label{le:IDCodesFB}

\begin{lemma}\cite[Lemma~4.1]{venkatesananantharam98}\label{le:typeFB}
For some DMC $\channel y x$, let $\dist$ be some distribution of the pair $(X^n,Y^n)$ of length-$n$ input- and output-sequence satisfying that
\begin{IEEEeqnarray}{l}
\bigdistof { (X^n,Y^n) = (\vecx,\vecy) } = \prod^n_{i = 1} \bigdistof {X_i = x_i \bigl| (X^{i-1},Y^{i-1}) = (x^{i-1},y^{i-1})} \channel {y_i}{x_i}, \quad ( \vecx, \vecy) \in \setX^n \times \setY^n, \label{eq:feedbackDist}
\end{IEEEeqnarray}
and for every pair $( \vecx, \vecy ) \in \setX^n \times \setY^n$ define the PMF on $\setX$
\begin{IEEEeqnarray}{l}
P^{\vecx, \vecy} ( x ) = \frac{1}{n} \sum^n_{i = 1} \bigdistof {X_i = x \bigl| (X^{i-1},Y^{i-1}) = (x^{i-1},y^{i-1})}, \quad x \in \setX. \label{eq:PDOnInput}
\end{IEEEeqnarray}
Then, for any $\nu > 0$
\begin{IEEEeqnarray}{l}
\Bigdistof{ \exists \, ( x,y ) \in \setX \times \setY \colon \bigl| P_{X^n,Y^n} (x,y) - P^{X^n,Y^n} \! ( x ) \channel {y}{x} \bigr| \geq \sqrt { \channel {y}{x} } \, \nu } \leq \frac{ | \setX | \, | \setY | }{n \nu^2}, \label{eq:typeFBProbBound}
\end{IEEEeqnarray}
where $P_{X^n,Y^n}$ is the empirical type of the pair $(X^n, Y^n)$, so $P_{X^n,Y^n} (x,y) = N ( x, y | X^n, Y^n ) / n, \, (x,y) \in \setX \times \setY$.
\end{lemma}

\begin{proof}
For every pair $(x,y) \in \setX \times \setY$ define the binary random variables
\begin{IEEEeqnarray}{l}
E^{x,y}_i = \ind {( X_i, Y_i ) = ( x,y )} = \ind {X_i = x} \ind {Y_i = y}, \quad i \in [ 1 : n ] \label{eq:RVEi}
\end{IEEEeqnarray}
with mean
\begin{IEEEeqnarray}{l}
\bigEx {}{ E^{x,y}_i \bigl| X^{i-1}, Y^{i-1} } \nonumber \\
\quad \stackrel{(a)}= \BigEx {}{ \bigEx {}{\ind {X_i = x} \ind {Y_i = y} \bigl| X^i, Y^{i-1} } \Bigl| X^{i-1}, Y^{i-1} } \\
\quad \stackrel{(b)}= \BigEx {}{ \ind {X_i = x} \bigEx {X_i = x}{ \ind {Y_i = y} \bigl| X^{i-1}, Y^{i-1} } \Bigl| X^{i-1}, Y^{i-1}  } \\
\quad \stackrel{(c)}= \bigEx {}{ \ind {X_i = x} \channel y x \bigl| X^{i-1}, Y^{i-1} } \\
\quad = \bigdistof { X_i = x \bigl| X^{i-1}, Y^{i-1} } \channel {y}{x}, \label{eq:condExpEi}
\end{IEEEeqnarray}
where $(a)$ follows from \eqref{eq:RVEi} and the Tower property of conditional expectation; $(b)$ holds because $\ind {X_i = x}$ is $\sigof {X_i}$-measurable and because $\ind{X_i = x}$ is zero unless $X_i = x$; and $(c)$ holds by \eqref{eq:feedbackDist}. Define the centered random variables
\begin{IEEEeqnarray}{l}
\tilde E^{x,y}_i = E^{x,y}_i - \bigEx {}{ E^{x,y}_i \bigl| X^{i-1}, Y^{i-1} }, \quad i \in [ 1:n ]. \label{eq:RVTildeEi}
\end{IEEEeqnarray}
By \eqref{eq:PDOnInput} and \eqref{eq:condExpEi}
\begin{IEEEeqnarray}{l}
\sum_{i = 1}^n \tilde E^{x,y}_i = N ( x,y | X^n, Y^n ) - n P^{X^n,Y^n} \! ( x ) \channel {y}{x}. \label{eq:RVTildeEiSum}
\end{IEEEeqnarray}
As we shall see, the centered random variables $\bigl\{ \tilde E^{x,y}_i \bigr\}_{i \in [ 1 : n ]}$ are uncorrelated and of variance $\bigEx {}{ ( \tilde E^{x,y}_i )^2 } \leq \channel {y}{x}$. Consequently, Chebyshev's inequality implies that
\begin{equation}
\Biggdistof { \frac{1}{n} \biggl| \sum_{i = 1}^n \tilde E^{x,y}_i \biggr| \geq \sqrt{\channel {y}{x}} \, \nu} \leq \frac{1}{n \nu^2}, \quad (x,y) \in \setX \times \setY,
\end{equation}
and \eqref{eq:typeFBProbBound} thus follows from \eqref{eq:RVTildeEiSum} and the Union-of-Events bound:
\begin{IEEEeqnarray}{l}
\Bigdistof{ \exists \, ( x,y ) \in \setX \times \setY \colon \bigl| P_{X^n,Y^n} (x,y) - P^{X^n,Y^n} \! ( x ) \channel {y}{x} \bigr| \geq \sqrt { \channel {y}{x} } \, \nu } \nonumber \\
\quad = \biggdistof{ \exists \, ( x,y ) \in \setX \times \setY \colon \biggl| \frac{N ( x, y | X^n, Y^n ) }{n} - P^{X^n,Y^n} \! ( x ) \channel {y}{x} \biggr| \geq \sqrt { \channel {y}{x} } \, \nu } \\
\quad = \Biggdistof {\exists \, ( x,y ) \in \setX \times \setY \colon \frac{1}{n} \biggl| \sum_{i = 1}^n \tilde E^{x,y}_i \biggr| \geq \sqrt{\channel {y}{x}} \, \nu} \\
\quad \leq \frac{| \setX | \, | \setY | }{n \nu^2}.
\end{IEEEeqnarray}

To conclude the proof, it remains to show that the centered random variables $\bigl\{ \tilde E^{x,y}_i \bigr\}_{i \in [ 1 : n ]}$ are uncorrelated and of variance $\bigEx {}{ ( \tilde E^{x,y}_i )^2 } \leq \channel {y}{x}$. We first prove the former: For every $l, \, k \in [1:n]$ satisfying $l < k$
\begin{IEEEeqnarray}{rCl}
\bigEx {}{ \tilde E^{x,y}_l \tilde E^{x,y}_k } & \stackrel{(a)}= & \BigEx{}{\bigEx{}{ \tilde E^{x,y}_l \tilde E^{x,y}_k \bigl| X^{k-1}, Y^{k-1} }} \\
& \stackrel{(b)}= & \BigEx{}{\tilde E^{x,y}_l \bigEx{}{\tilde E^{x,y}_k \bigl| X^{k-1}, Y^{k-1} }} \\
& \stackrel{(c)}= & 0,
\end{IEEEeqnarray}
where $(a)$ follows from the Tower property of conditional expectation; $(b)$ holds because $\tilde E^{x,y}_l$ is $\sigof {X^l, Y^l}$-measurable and $l \leq k-1$; and $(c)$ holds by \eqref{eq:RVTildeEi}. Having established that the centered random variables $\bigl\{ \tilde E^{x,y}_i \bigr\}_{i \in [ 1 : n ]}$ are uncorrelated, it remains to show that their variance is upper-bounded by $\channel {y}{x}$. For every $i \in [1:n]$
\begin{IEEEeqnarray}{rCl}
\BigEx{}{ \bigl( \tilde E^{x,y}_i \bigr)^2} & \stackrel{(a)}= & \BigEx {}{ \bigl( E^{x,y}_i - \bigEx{}{E^{x,y}_i \bigl| X^{i-1}, Y^{i-1} } \bigr)^2} \\
& \stackrel{(b)}= & \biggEx{}{\BigEx{}{ \bigl( E^{x,y}_i - \bigEx{}{E^{x,y}_i \bigl| X^{i-1}, Y^{i-1} } \bigr)^2 \Bigr| X^{i-1}, Y^{i-1} }} \\
& \stackrel{(c)}= & \BigEx{}{ \bigdistof{ X_i = x \bigl| X^{i-1}, Y^{i-1} } \channel {y}{x} \bigl( 1 - \bigdistof{ X_i = x \bigl| X^{i-1}, Y^{i-1} } \channel {y}{x} \bigr)} \\
& \stackrel{(d)}\leq & \channel y x,
\end{IEEEeqnarray}
where $(a)$ holds by \eqref{eq:RVTildeEi}; $(b)$ follows from the Tower property of conditional expectation; $(c)$ holds because $(E^{x,y}_i)^2 = E^{x,y}_i$ (which holds by \eqref{eq:RVEi}), because $\bigEx {}{ E^{x,y}_i \bigl| X^{i-1}, Y^{i-1} }$ is $\sigof {X^{i-1}, Y^{i-1}}$-measurable, and by \eqref{eq:condExpEi}; and $(d)$ holds because conditional probability cannot exceed one.
\end{proof}

\subsection{A Proof of Theorem~\ref{th:IDBC1FBOB}}

If $\max_{\tilde P} \muti {\tilde P}{W_\ry} = 0$, then the transition law $\channely y x$ does not depend on $x$, and hence $\lambda^{\ry}_1 + \lambda^{\ry}_2 \geq 1$ whenever $R_\ry > 0$. Likewise, if $\max_{\tilde P} \muti {\tilde P}{W_\rz} = 0$, then $\lambda^{\rz}_1 + \lambda^{\rz}_2 \geq 1$ whenever $R_\rz > 0$. Consequently, if suffices to prove the following strong converse:

\begin{claim}\label{cl:toShowIDBC1FBConv}
For every rate-pair $(R_\ry, R_\rz)$, every positive constants $\lambda_1^\ry, \, \lambda_2^\ry, \, \lambda_1^\rz, \, \lambda_2^\rz$ satisfying
\begin{equation}
\lambda_1^\ry + \lambda_2^\ry + \lambda_1^\rz + \lambda_2^\rz < 1, \label{eq:lambdaSumSm11FB}
\end{equation}
and every $\epsilon > 0$ there exists some $\eta_0 \in \naturals$ so that, for every blocklength $n \geq \eta_0$, every size-$\exp (\exp (n R_\ry))$ set $\setM_\ry$ of possible ID messages for Receiver~$\ry$, and every size-$\exp (\exp (n R_\rz))$ set $\setM_\rz$ of possible ID messages for Receiver~$\rz$, a necessary condition for an $\bigl( n, \setM_\ry, \setM_\rz, \lambda_1^\ry, \lambda_2^\ry, \lambda_1^\rz, \lambda_2^\rz \bigr)$ ID code for the BC $\channel {y,z} x$ with one-sided feedback from Terminal~$\ry$ to exist is that for some PMF $P$ on $\setX$
\begin{subequations}\label{bl:FBConv4}
\begin{IEEEeqnarray}{rCl}
R_\ry & < & \ent {P W_\ry} + \epsilon, \\
R_\rz & < & \muti {P \times W_\ry}{\widetilde W_\rz} + \epsilon,
\end{IEEEeqnarray}
\end{subequations}
where $\widetilde W_\rz$ is defined in \eqref{eq:channelztild}.
\end{claim}

\begin{proof}
Suppose that the collection of tuples $$\Bigl\{ \bigl\{ Q^{(i)}_{m_\ry,m_\rz} \bigr\}_{i \in \{ 1, \ldots, n \}}, \idset {m_\ry}, \idset {m_\rz} \Bigr\}_{(m_\ry,m_\rz) \in \setM_\ry \times \setM_\rz}$$ is an $\bigl( n, \setM_\ry, \setM_\rz, \lambda^\ry_1, \lambda^\ry_2, \lambda^\rz_1, \lambda^\rz_2 \bigr)$ ID code for the BC $\channel {y,z} x$ with one-sided feedback from Terminal~$\ry$. For every pair $( m_\ry, m_\ry ) \in \setM_\ry \times \setM_\rz$ define the PMF on $\setX^n \times \setY^n$
\begin{equation}
Q_{m_\ry,m_\rz} ( \vecx, \vecy ) = \prod^n_{i = 1} Q^{(i)}_{m_\ry,m_\rz} ( x_i | x^{i-1}, y^{i-1} ) \channely {y_i}{x_i}, \quad (\vecx,\vecy) \in \setX^n \times \setY^n, \label{eq:formInputDist}
\end{equation}
and note that $Q_{m_\ry,m_\rz}$ is the distribution of the pair $(X^n,Y^n)$ of length-$n$ input- and output-sequence if ID Message-Pair~$(m_\ry,m_\rz)$ is sent. Introduce the BC $\channeltild {y,z} {x,\tilde y}$ whose outputs are the outputs of the BC $\channel {y,z} x$ and whose inputs are the input and the output at Receiver~$\ry$ of the BC $\channel {y,z} x$, so
\begin{IEEEeqnarray}{l}
\channeltild {y,z}{x,\tilde y} = \ind {y = \tilde y} \channelztild {z}{x, \tilde y}, \quad (x, \tilde y, y, z) \in \setX \times \setY \times \setY \times \setZ.
\end{IEEEeqnarray}
(The marginal channels of the BC $\channeltild {y,z}{x,\tilde y}$ are $\ind {y = \tilde y}$ and $\channelztild {z}{x,\tilde y}$.) Because $$\Bigl\{ \bigl\{ Q^{(i)}_{m_\ry,m_\rz} \bigr\}_{i \in \{ 1, \ldots, n \}}, \idset {m_\ry}, \idset {m_\rz} \Bigr\}_{(m_\ry,m_\rz) \in \setM_\ry \times \setM_\rz}$$ is an $\bigl( n, \setM_\ry, \setM_\rz, \lambda^\ry_1, \lambda^\ry_2, \lambda^\rz_1, \lambda^\rz_2 \bigr)$ ID code for the BC $\channel {y,z} x$ with one-sided feedback from Terminal~$\ry$, the collection of tuples $\bigl\{ Q_{m_\ry, m_\rz}, \setD_{m_\ry}, \setD_{m_\rz} \bigr\}_{(m_\ry,m_\rz) \in \setM_\ry \times \setM_\rz}$ is an $\bigl( n, \setM_\ry, \setM_\rz, \lambda^\ry_1, \lambda^\ry_2, \lambda^\rz_1, \lambda^\rz_2 \bigr)$ ID code for the BC $\channeltild {y,z}{x,\tilde y}$ without feedback. To prove Claim~\ref{cl:toShowIDBC1FBConv}, we can thus adopt some of the arguments in the proof of Claim~\ref{cl:toShowIDBCConv}.

Fix some $\epsilon > 0$, and choose $\mu \in (0,1/2)$ sufficiently small so that
\begin{equation}
\mu + \mu \max \biggl\{ \log \frac{| \setY |}{\mu}, \log \frac{| \setZ |^2}{\mu} \biggr\} < \epsilon. \label{eq:defMu1FB}
\end{equation}
(This is possible, because $\mu \log \mu$ converges to zero as $\mu$ tends to zero.) Introduce the set $\setK_{\mu}$ comprising the realizations $(\vecx,\vecy) \in \setX^n \times \setY^n$ of the pair $(X^n,Y^n)$ that satisfy the following two conditions:
\begin{subequations}\label{bl:FBConvSetS}
\begin{IEEEeqnarray}{rCl}
\muti {P_{\vecx,\vecy}}{\widetilde W_\ry} & > & R_\ry - \mu, \\
\muti {P_{\vecx,\vecy}}{\widetilde W_\rz} & > & R_\rz - \mu,
\end{IEEEeqnarray}
\end{subequations}
where $P_{\vecx,\vecy}$ is the empirical type of the pair $(\vecx,\vecy)$. Moreover, introduce the set $\setL_{\epsilon,\mu}$ comprising the realizations $(\vecx,\vecy) \in \setX^n \times \setY^n$ of the pair $(X^n,Y^n)$ that for some PMF $P$ on $\setX$ satisfy the following two conditions:
\begin{subequations}\label{bl:FBConvSetL}
\begin{IEEEeqnarray}{rCl}
\bigl| \muti {P_{\vecx,\vecy}}{\widetilde W_\ry} - \ent {P W_\ry } \bigr| & \leq & \epsilon - \mu , \\
\bigl| \muti {P_{\vecx,\vecy}}{\widetilde W_\rz} - \muti {P \times W_\ry }{\widetilde W_\rz} \bigr| & \leq & \epsilon - \mu.
\end{IEEEeqnarray}
\end{subequations}
As we shall see, there exists some $\eta_0 \in \naturals$ so that for every blocklength $n \geq \eta_0$ the mixture PMF on $\setX^n \times \setY^n$
\begin{equation}
Q = \frac{1}{ | \setM_\ry | \, | \setM_\rz | } \sum_{(m_\ry, m_\rz) \in \setM_\ry \times \setM_\rz} Q_{m_\ry, m_\rz} \label{eq:FBConvMixPMF}
\end{equation}
satisfies
\begin{equation}
Q \bigl ( (X^n,Y^n) \in \setK_{\mu} \cap \setL_{\epsilon,\mu} \bigr) > 0. \label{eq:FBConv3}
\end{equation}
By \eqref{bl:FBConvSetS} and \eqref{bl:FBConvSetL} the intersection $\setK_{\mu} \cap \setL_{\epsilon,\mu}$ contains only realizations $(\vecx,\vecy) \in \setX^n \times \setY^n$ of the pair $(X^n,Y^n)$ that for some PMF $P$ on $\setX$ satisfy the following two conditions:
\begin{subequations}\label{bl:FBConvSetLCapS}
\begin{IEEEeqnarray}{rCl}
\ent {P W_\ry } & \geq & \muti {P_{\vecx,\vecy}}{\widetilde W_\ry} - \epsilon + \mu > R_\ry - \epsilon, \\
\muti {P \times W_\ry }{\widetilde W_\rz} & \geq & \muti {P_{\vecx,\vecy}}{\widetilde W_\rz} - \epsilon + \mu > R_\rz - \epsilon.
\end{IEEEeqnarray}
\end{subequations}
This implies that for every blocklength~$n \geq \eta_0$ a necessary condition for \eqref{eq:FBConv3} to hold is that for some PMF $P$ on $\setX$ \eqref{bl:FBConv4} holds, and hence Claim~\ref{cl:toShowIDBC1FBConv} follows.

It remains to establish \eqref{eq:FBConv3}. We begin by upper-bounding the probability $Q \bigl( (X^n,Y^n) \notin \setK_\mu \bigr)$. Fix $\kappa^\ry, \, \kappa^\rz > 0$ that satisfy the following three: 1) $\lambda_1^\ry + \lambda_2^\ry < \kappa^\ry$; 2) $\lambda_1^\rz + \lambda_2^\rz < \kappa^\rz$; and 3) $\kappa^\ry + \kappa^\rz < 1$. (This is possible because of \eqref{eq:lambdaSumSm11FB}.) Because $\bigl\{ Q_{m_\ry, m_\rz}, \setD_{m_\ry}, \setD_{m_\rz} \bigr\}_{(m_\ry,m_\rz) \in \setM_\ry \times \setM_\rz}$ is an $\bigl( n, \setM_\ry, \setM_\rz, \lambda^\ry_1, \lambda^\ry_2, \lambda^\rz_1, \lambda^\rz_2 \bigr)$ ID code for the BC $\channeltild {y,z}{x,\tilde y}$ without feedback, \eqref{eq:BCExistsPMFGoodForBoth} in the proof of Claim~\ref{cl:toShowIDBCConv} implies that there must exist some $\eta_0^\prime \in \naturals$ so that for every blocklength $n \geq \eta_0^\prime$
\begin{IEEEeqnarray}{l}
Q \bigl( (X^n,Y^n) \notin \setK_\mu \bigr) \nonumber \\
\quad \leq \kappa^\ry + \kappa^\rz + \exp \bigl\{ e^{n ( R_\ry - \mu / 2 )} \bigr\} / \exp \bigl\{ e^{n R_\ry} \bigr\} + \exp \bigl\{ e^{n ( R_\rz - \mu / 2 )} \bigr\} / \exp \bigl\{ e^{n R_\rz} \bigr\}. \label{eq:FBConv2}
\end{IEEEeqnarray}

Having established \eqref{eq:FBConv2}, we conclude the proof of \eqref{eq:FBConv3} by showing that the probability $Q \bigl( (X^n,Y^n) \notin \setL_{\epsilon,\mu} \bigr)$ satisfies the upper bound
\begin{IEEEeqnarray}{rCl}
Q \bigl( (X^n,Y^n) \notin \setL_{\epsilon,\mu} \bigr) & \leq & \frac{ |\setX|^3 |\setY|^3 }{n \mu^2}. \label{eq:FBConv1}
\end{IEEEeqnarray}
This implies \eqref{eq:FBConv3}, because, combined with the Union-of-Events bound and \eqref{eq:FBConv2}, it implies that
\begin{IEEEeqnarray}{l}
Q \bigl ( (X^n,Y^n) \in \setK_{\mu} \cap \setL_{\epsilon,\mu} \bigr) \nonumber \\
\quad \geq 1 - Q \bigl( (X^n,Y^n) \notin \setK_\mu \bigr) - Q \bigl( (X^n,Y^n) \notin \setL_{\epsilon,\mu} \bigr) \\
\quad \geq 1 - \kappa^\ry - \kappa^\rz - \exp \bigl\{ e^{n ( R_\ry - \mu / 2 )} \bigr\} / \exp \bigl\{ e^{n R_\ry} \bigr\} - \exp \bigl\{ e^{n ( R_\rz - \mu / 2 )} \bigr\} / \exp \bigl\{ e^{n R_\rz} \bigr\} - \frac{ |\setX|^3 |\setY|^3 }{n \mu^2}, \label{eq:FBConv1Comb2}
\end{IEEEeqnarray}
and we can let $\eta_0$ be the smallest integer $n \geq \eta_0^\prime$ for which the RHS of \eqref{eq:FBConv1Comb2} is positive (such an $n$ must exist, because $\mu > 0$ and $\kappa^\ry + \kappa^\rz < 1$).

To conclude the proof of Claim~\ref{cl:toShowIDBC1FBConv}, it remains to establish \eqref{eq:FBConv1}. For every pair $(m_\ry,m_\rz) \in \setM_\ry \times \setM_\rz$ define for every $(\vecx, \vecy) \in \setX^n \times \setY^n$ the PMF on $\setX$
\begin{IEEEeqnarray}{l}
P^{\vecx, \vecy}_{m_\ry,m_\rz} ( x ) = \frac{1}{n} \sum^n_{i = 1} Q^{(i)}_{m_\ry,m_\rz} (x), \quad x \in \setX, \label{eq:PDOnInputConv1FB}
\end{IEEEeqnarray}
and introduce the set $\setL^{m_\ry,m_\rz}_{\epsilon,\mu}$ comprising the realizations $(\vecx,\vecy) \in \setX^n \times \setY^n$ of the pair $(X^n,Y^n)$ that satisfy the following two conditions:
\begin{subequations}\label{bl:FBConvSetLMyMz}
\begin{IEEEeqnarray}{rCl}
\bigl| \muti {P_{\vecx,\vecy}}{\widetilde W_\ry} - \ent {P^{\vecx,\vecy}_{m_\ry,m_\rz} W_\ry } \bigr| & \leq & \epsilon - \mu , \\
\bigl| \muti {P_{\vecx,\vecy}}{\widetilde W_\rz} - \muti {P^{\vecx,\vecy}_{m_\ry,m_\rz} \times W_\ry }{\widetilde W_\rz} \bigr| & \leq & \epsilon - \mu.
\end{IEEEeqnarray}
\end{subequations}
By comparing \eqref{bl:FBConvSetLMyMz} and \eqref{bl:FBConvSetL} we see that
\begin{equation}
\setL^{m_\ry,m_\rz}_{\epsilon,\mu} \subseteq \setL_{\epsilon,\mu}, \quad (m_\ry,m_\rz) \in \setM_\ry \times \setM_\rz.
\end{equation}
This, combined with \eqref{eq:FBConvMixPMF}, implies that
\begin{IEEEeqnarray}{l}
Q \bigl( (X^n,Y^n) \notin \setL_{\epsilon,\mu} \bigr) \nonumber \\
\quad = \frac{1}{|\setM_\ry| \, |\setM_\rz|} \sum_{(m_\ry,m_\rz) \in \setM_\ry \times \setM_\rz} Q_{m_\ry,m_\rz} \bigl( (X^n,Y^n) \notin \setL_{\epsilon,\mu} \bigr) \\
\quad \leq \max_{(m_\ry,m_\rz) \in \setM_\ry \times \setM_\rz} Q_{m_\ry,m_\rz} \bigl( (X^n,Y^n) \notin \setL^{m_\ry,m_\rz}_{\epsilon,\mu} \bigr),
\end{IEEEeqnarray}
and to establish \eqref{eq:FBConv1} it thus suffices to show that
\begin{IEEEeqnarray}{l}
Q_{m_\ry,m_\rz} \bigl( (X^n,Y^n) \notin \setL^{m_\ry,m_\rz}_{\epsilon,\mu} \bigr) \leq \frac{ |\setX|^3 |\setY|^3 }{n \mu^2}, \quad (m_\ry,m_\rz) \in \setM_\ry \times \setM_\rz. \label{eq:FBConv1MaxMyMz}
\end{IEEEeqnarray}
To that end, let
\begin{equation}
\nu = \frac{\mu}{| \setX | \, | \setY |}, \label{eq:defNu1FB}
\end{equation}
and for every $(m_\ry,m_\rz) \in \setM_\ry \times \setM_\rz$ introduce the set $\setN^{m_\ry,m_\rz}_\mu$ comprising the realizations $(\vecx, \vecy) \in \setX^n \times \setY^n$ of the pair $(X^n,Y^n)$ satisfying that
\begin{IEEEeqnarray}{l}
\bigl| P_{\vecx, \vecy} (x,y) - P^{\vecx,\vecy}_{m_\ry,m_\rz} ( x ) \channely {y}{x} \bigr| < \sqrt { \channely {y}{x} } \, \nu, \,\, \forall \, (x,y) \in \setX \times \setY. \label{eq:BC1FBxySatType}
\end{IEEEeqnarray}
As we shall see,
\begin{equation}
\setN^{m_\ry,m_\rz}_\mu \subseteq \setL^{m_\ry,m_\rz}_{\epsilon,\mu}, \label{eq:conv1FBSetNSubsSetL}
\end{equation}
and to establish \eqref{eq:FBConv1MaxMyMz} it thus suffices to show that
\begin{IEEEeqnarray}{rCl}
Q_{m_\ry,m_\rz} \bigl( (X^n,Y^n) \notin \setN^{m_\ry,m_\rz}_{\mu} \bigr) & \leq & \frac{ |\setX|^3 |\setY|^3 }{n \mu^2}, \quad (m_\ry,m_\rz) \in \setM_\ry \times \setM_\rz. \label{eq:FBConv1SetN}
\end{IEEEeqnarray}
But this in an immediate consequence of Lemma~\ref{le:typeFB} in Appendix~\ref{le:IDCodesFB}: For every pair $(m_\ry,m_\rz) \in \setM_\ry \times \setM_\rz$ the PMF $Q_{m_\ry,m_\rz}$ of \eqref{eq:formInputDist} is of the form \eqref{eq:feedbackDist}, and by comparing \eqref{eq:PDOnInputConv1FB} to \eqref{eq:PDOnInput} we see that $P^{\vecx,\vecy}_{m_\ry,m_\rz}$ is the corresponding PMF $P^{\vecx,\vecy}$ on $\setX$ of \eqref{eq:PDOnInput}. Consequently, \eqref{eq:BC1FBxySatType} and Lemma~\ref{le:typeFB} imply that
\begin{IEEEeqnarray}{l}
Q_{m_\ry,m_\rz} \bigl( (X^n,Y^n) \notin \setN^{m_\ry,m_\rz}_\mu \bigr) \nonumber \\
\quad \leq \frac{ | \setX | \, | \setY | }{n \nu^2} \\
\quad \leq \frac{ |\setX|^3 |\setY|^3 }{n \mu^2}, \quad (m_\ry,m_\rz) \in \setM_\ry \times \setM_\rz, \label{eq:BC1FBType}
\end{IEEEeqnarray}
where the last inequality holds by \eqref{eq:defNu1FB}.

Having established \eqref{eq:FBConv1SetN}, we can now conclude the proof of Claim~\ref{cl:toShowIDBC1FBConv} by establishing \eqref{eq:conv1FBSetNSubsSetL}. To that end, fix any pair $(\vecx, \vecy) \in \setN^{m_\ry,m_\rz}_\mu$. By \eqref{eq:BC1FBxySatType} (which holds because $(\vecx, \vecy) \in \setN^{m_\ry,m_\rz}_\mu$)
\begin{IEEEeqnarray}{l}
P_{\vecx, \vecy} ( x,y ) \in \Bigl[ P^{\vecx,\vecy}_{m_\ry,m_\rz} ( x ) \channely {y}{x} \pm \sqrt { \channely {y}{x} } \, \nu \Bigr], \quad ( x,y ) \in \setX \times \setY. \label{eq:BC1FBxySatType2}
\end{IEEEeqnarray}
Consequently, $d ( P_{\vecx, \vecy}, P^{\vecx,\vecy}_{m_\ry,m_\rz} \times W_\ry )$ satisfies the upper bound
\begin{IEEEeqnarray}{l}
d ( P_{\vecx, \vecy}, P^{\vecx,\vecy}_{m_\ry,m_\rz} \times W_\ry ) \nonumber \\
\quad \stackrel{(a)}= \frac{1}{2} \sum_{(x,y) \in \setX \times \setY} \bigl| P_{\vecx, \vecy} ( x,y ) - P^{\vecx,\vecy}_{m_\ry,m_\rz} ( x ) \channely {y}{x} \bigr| \\
\quad \stackrel{(b)}\leq \frac{1}{2} \sum_{(x,y) \in \setX \times \setY} \sqrt { \channely {y}{x} } \, \nu \\
\quad = \frac{1}{2} \sqrt { \channely {y}{x} } \, \card \setX \, \card \setY \, \nu \\
\quad \stackrel{(c)}\leq \mu/2, \quad (\vecx, \vecy) \in \setN_\mu^{m_\ry,m_\rz}, \label{eq:totVarDistFB}
\end{IEEEeqnarray}
where $(a)$ holds by definition of the Total-Variation distance; $(b)$ holds by \eqref{eq:BC1FBxySatType2}; and $(c)$ holds by \eqref{eq:defNu1FB} and because $\sqrt { \channely {y}{x} } \leq 1$. Using this we can upper-bound $d \bigl( P_{\vecx, \vecy} \times \widetilde W, P^{\vecx,\vecy}_{m_\ry,m_\rz} \times W_\ry \times \widetilde W \bigr)$ by
\begin{IEEEeqnarray}{l}
d \bigl( P_{\vecx, \vecy} \times \widetilde W, P^{\vecx,\vecy}_{m_\ry,m_\rz} \times W_\ry \times \widetilde W \bigr) \nonumber \\
\quad = \frac{1}{2} \sum_{(x, \tilde y, y, z) \in \setX \times \setY \times \setY \times \setZ} \bigl| P_{\vecx, \vecy} ( x,\tilde y ) \channeltild {y,z}{x,\tilde y} - P^{\vecx,\vecy}_{m_\ry,m_\rz} ( x ) \channely {\tilde y}{x} \channeltild {y,z}{x,\tilde y} \bigr| \\
\quad = \frac{1}{2} \sum_{(x, \tilde y) \in \setX \times \setY} \bigl| P_{\vecx, \vecy} ( x,\tilde y ) - P^{\vecx,\vecy}_{m_\ry,m_\rz} ( x ) \channely {\tilde y}{x} \bigr| \sum_{(y,z) \in \setY \times \setZ} \channeltild {y,z}{x,\tilde y} \\
\quad = \frac{1}{2} \sum_{(x, \tilde y) \in \setX \times \setY} \bigl| P_{\vecx, \vecy} ( x,\tilde y ) - P^{\vecx,\vecy}_{m_\ry,m_\rz} ( x ) \channely {\tilde y}{x} \bigr| \\
\quad \leq \mu/2, \quad (\vecx, \vecy) \in \setN^{m_\ry,m_\rz}_\mu. \label{eq:totVarDistFBNewCh}
\end{IEEEeqnarray}
Consequently, the Data-Processing inequality for the Total-Variation distance \cite[Lemma~1]{cannoneronservedio15} implies that
\begin{subequations}
\begin{IEEEeqnarray}{rCl}
d ( P_{\vecy}, P^{\vecx,\vecy}_{m_\ry,m_\rz} W_\ry ) & \leq & \mu/2, \quad (\vecx, \vecy) \in \setN^{m_\ry,m_\rz}_\mu, \label{eq:totVarDistFBNewChY} \\
d \bigl( P_{\vecx, \vecy} \widetilde W_\rz, ( P^{\vecx,\vecy}_{m_\ry,m_\rz} \times W_\ry ) \widetilde W_\rz \bigr) & \leq & \mu/2, \quad (\vecx, \vecy) \in \setN^{m_\ry,m_\rz}_\mu. \label{eq:totVarDistFBNewChZ}
\end{IEEEeqnarray}
\end{subequations}
This, combined with the fact that entropy is continuous, implies that
\begin{IEEEeqnarray}{l}
\bigl| \muti {P_{\vecx, \vecy}}{\widetilde W_\ry} - \ent {P^{\vecx, \vecy}_{m_\ry,m_\rz} W_\ry } \bigr| \nonumber \\
\quad \stackrel{(a)}= \bigl| \ent {P_{\vecy}} - \ent {P^{\vecx, \vecy}_{m_\ry,m_\rz} W_\ry } \bigr| \\
\quad \stackrel{(b)}\leq \mu \log \frac{| \setY |}{\mu} \\
\quad \stackrel{(c)}\leq \epsilon - \mu, \quad (\vecx, \vecy) \in \setN^{m_\ry,m_\rz}_\mu, \label{eq:BC1FBMutChY}
\end{IEEEeqnarray}
where $(a)$ holds because $\channelytild {y}{x,\tilde y} = \ind {y = \tilde y}$; $(b)$ holds by \eqref{eq:totVarDistFBNewChY}, \cite[Lemma~2.7]{csiszarkoerner11}, and the fact that $\mu < 1/2$; and $(c)$ holds by \eqref{eq:defMu1FB}. Similarly,
\begin{IEEEeqnarray}{l}
\bigl| \muti {P_{\vecx, \vecy}}{\widetilde W_\rz} - \muti {P^{\vecx, \vecy}_{m_\ry,m_\rz} \times W_\ry }{\widetilde W_\rz} \bigr| \nonumber \\
\quad \stackrel{(a)}\leq \bigl| \ent{ P_{\vecx, \vecy} \widetilde W_\rz } - \bigent{ ( P^{\vecx, \vecy}_{m_\ry,m_\rz} \times W_\ry ) \widetilde W_\rz } \bigr| + \sum_{x,\tilde y} \bigl| P_{\vecx, \vecy} ( x,\tilde y ) - P^{\vecx, \vecy}_{m_\ry,m_\rz} ( x ) \channely {\tilde y}{x} \bigr| \bigent { \channelztild {\cdot}{x, \tilde y} } \\
\quad \stackrel{(b)}\leq \mu \log \frac{| \setZ |}{\mu} + 2 d ( P_{\vecx, \vecy}, P^{\vecx,\vecy}_{m_\ry,m_\rz} \times W_\ry ) \log | \setZ | \\
\quad \stackrel{(c)}\leq \mu \log \frac{| \setZ |^2}{\mu} \\
\quad \stackrel{(d)}\leq \epsilon - \mu, \quad (\vecx, \vecy) \in \setN^{m_\ry,m_\rz}_\mu, \label{eq:BC1FBMutChZ}
\end{IEEEeqnarray}
where $(a)$ holds by definition of mutual information and the Triangle inequality; $(b)$ holds by \eqref{eq:totVarDistFBNewChZ}, \cite[Lemma~2.7]{csiszarkoerner11}, the fact that $\mu < 1/2$, because the uniform distribution maximizes entropy, and by definition of the Total-Variation distance; $(c)$ holds by \eqref{eq:totVarDistFB}; and $(d)$ holds by \eqref{eq:defMu1FB}. From \eqref{bl:FBConvSetLMyMz}, \eqref{eq:BC1FBMutChY}, and \eqref{eq:BC1FBMutChZ} we conclude that \eqref{eq:conv1FBSetNSubsSetL} holds.
\end{proof}

%

\lhead[\fancyplain{\scshape Appendix}
{\scshape Appendix}]
{\fancyplain{\scshape \leftmark}
  {\scshape \leftmark}}
\rhead[\fancyplain{\scshape \leftmark}
{\scshape \leftmark}]
{\fancyplain{\scshape Appendix}
  {\scshape Appendix}}

%
\end{appendix}



\end{document}